\documentclass[a4paper,12pt,openany,oneside]{book}
\usepackage[T1]{fontenc}
\usepackage[utf8]{inputenc}
\usepackage[english]{babel}
\usepackage{textcomp}
\usepackage{amssymb}
\usepackage{amsmath}
\usepackage{epigraph}
\usepackage[numbers]{natbib}
\usepackage{fancyhdr}
\usepackage{pdfpages}
\usepackage{amsthm}
\usepackage{graphicx}

\theoremstyle{plain} 
\newtheorem{thm}{Theorem}[section]
\newtheorem{prop}{Proposition}
\newtheorem{cor}[thm]{Corollary} 
\newtheorem{lem}[thm]{Lemma}
\newtheorem{defn}{Definition}
\newtheorem{obs}{Observation}

\begin{document}
\includepdf{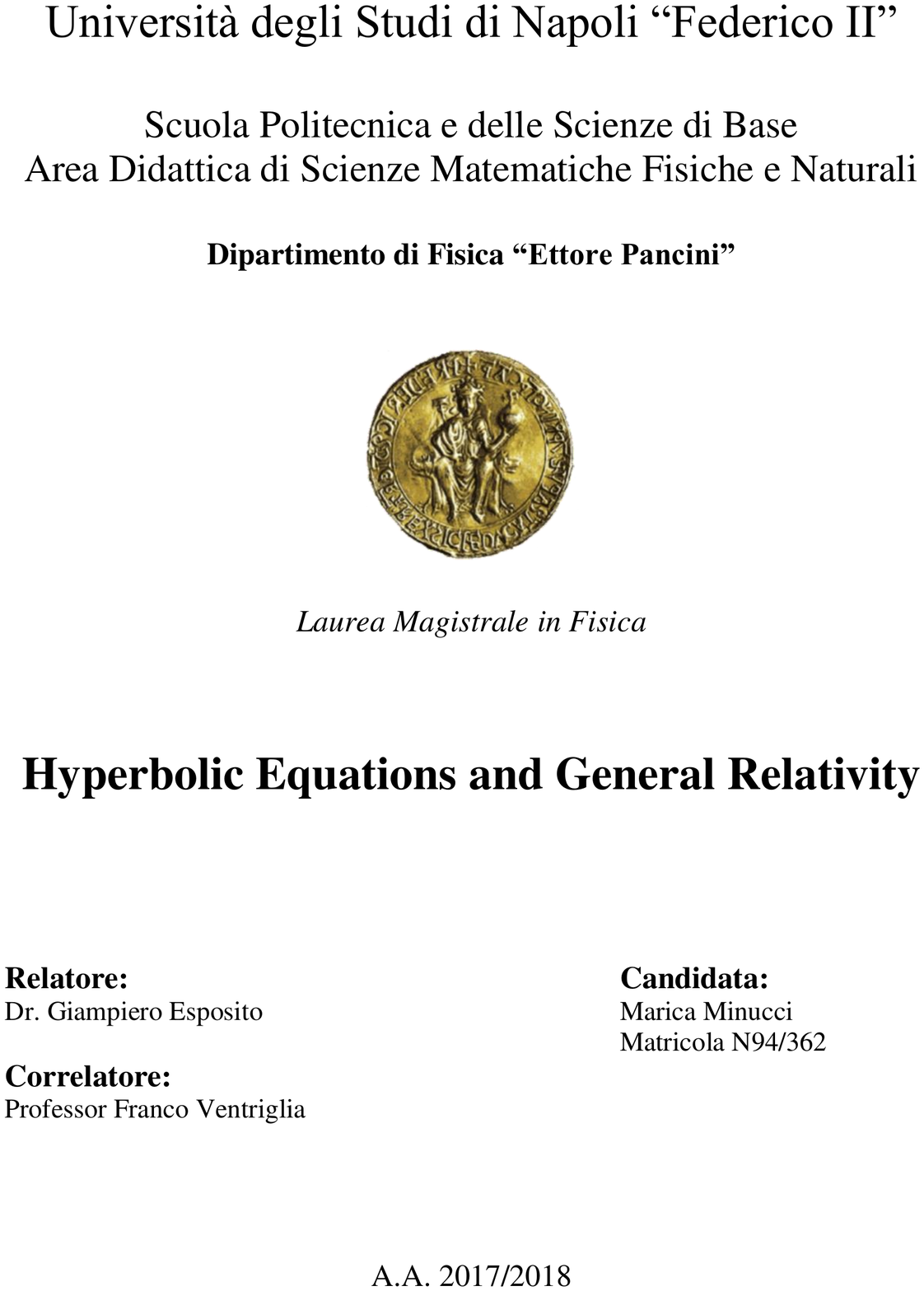}

\linespread{1.125}\selectfont

\thispagestyle{empty}
\begin{flushright}
\null\vspace{\stretch{1}}
\textit{ad Anna,\\
per il suo immancabile supporto.}
\vspace{\stretch{2}}\null

\end{flushright}

\tableofcontents 
\cleardoublepage
\pagenumbering{arabic}

\chapter*{Introduction}
\addcontentsline{toc}{chapter}{Introduction}
\markboth{}{}
The purpose of this work is to demonstrate that it is possible to formulate Einstein's equations as an initial value problem, that is a Cauchy problem. The idea of viewing the field equations of general relativity as a system of evolution equations, see Ringström \cite{ringstrom2015origins}, goes back to Einstein himself, in an argument justifying that gravitational waves propagate at the speed of light. In his papers \cite{einstein2005naherungsweise, einstein1918gravitationswellen}, Einstein considers a situation in which the metric is close to that of Minkowski space, in practice, he studied the linearized problem. Using a special choice of coordinates, he derived a wave equation for the perturbation, a result he used to justify the statement that gravitational waves propagate at the speed of light. The arguments of Einstein give an indication that the field equations of general relativity are a system of wave equations, and thus the problem to pose is an initial value problem. Despite that, the role of the choice of coordinates was not entirely clear at the time. In fact, in his criticism Eddington \cite{Eddington1930-EDDTMT} pointed out that if the coordinates are chosen so as to satisfy a certain condition which has no very geometrical importance, the speed is that of light, but any other choice of coordinates would give a different speed. Thus, the result stands or falls by the choice of coordinates. Furthermore, the choice of the type of coordinates to use was made in order to obtain the simplification which results from representing the propagation occurring with the speed of light. 

One way to approach the objections of Eddington is to argue that gravitational waves propagate at the speed of light without appealing to a specific choice of coordinates. In a paper \cite{vessiot1918propagation}, Vessiot argued that the desired statement follows from the observation that discontinuities in the derivatives of the metric of order strictly higher than one are only allowed along null hypersurfaces. On the other hand, the work of Darmois \cite{darmois1927equations} stressed the fact that characteristic hypersurfaces play a special role in the process of solving the field equations. One particular consequence of Darmois' analysis is that given a metric and its first normal derivatives on a spacelike hypersurface, all the derivatives of the metric are determined on the hypersurface. This yields a local uniqueness result in the real analytic setting. Moreover, there is a linear homogeneous system of equations for the components of the Ricci tensor corresponding to the constraints. Thus, it is not only necessary, but also sufficient, that the constraints be satisfied for the existence of a real analytic solution to the field equations. Furthermore, Darmois, making use of the coordinate choice made by de Donder, known as $\textit{isothermal}$ coordinates, proved that Einstein's argumentation is successful to demonstrate that the gravitational fields propagate at the speed of light. In addition, Darmois states the naturalness of these coordinates because they satisfy the scalar wave equation. Despite of this, a fundamental question remains since, given a solution to the field equations, there are two notions of $\textit{causality}$. There is the causality associated with the metric and there is the notion of domain of dependence associated with solving Einstein's equations considered as a partial differential equation. Then, it is of interest to know if these two notions coincide. This cannot be addressed in the real analytic setting, since real analytic functions have the unique continuation property. This question was addressed by Stellmacher \cite{stellmacher1938anfangswertproblem}, whose argument was based on the use of isothermal coordinates. In fact, given two solutions of Einstein-Maxwell equations, Stellmacher constructs isothermal coordinates such that the PDE techniques can be applied. Then, the conclusion is that two solutions coincide up to a coordinate transformation. Moreover, his work constitutes a justification of the statement that the gravitational field propagates at a speed bounded by that of light. The argument is such that Eddington's objections do not apply.

Acknowledging the results of Stellmacher, Lichnerowicz \cite{lichnerowicz1939problemes} stated the initial value problem as that of finding the solution to Einstein's equations on the basis of the metric and its first derivatives on a hypersurface. Hence, he solves the problem in the real analytic setting for spacelike hypersurfaces and notes the importance of the constraints. Furthermore, he point out the importance to generalize the existence result to the non-real analytic set. 

The work of Yvonne Choquet-Bruhat \cite{foures1952theoreme} provides this generalization by showing that not only does local uniqueness hold in the class of $C^k$-functions for $k$ large enough but, given initial data, there is a unique local solution. Thus, the Cauchy problem in general relativity stands on a solid basis in the $C^k$-setting. It is natural to ask why is the specific regularity class of importance and why it is not sufficient to consider the class of real analytic functions. 

A large part of the difficulty in obtaining the desired result lies in proving the local existence of solution to Einstein's equations in the prescribed regularity. Moreover, it is necessary to use coordinates with respect to which the equations become hyperbolic. Finally, it is necessary to connect the problem of solving the reduced equations with the constraint equations and the problem of solving Einstein's equations.

Hence, by following Yvonne Choquet-Bruhat \cite{foures1952theoreme}, we will show how to construct solutions to Einstein's vacuum equations, given initial data. 

In Chapter 1, by considering a system of partial differential equations, we will give the definition of characteristic manifold, the concept of wavelike propagation and we will introduce and prove the existence of the Riemann Kernel. 

In Chapter 2, we will stress the relation between Riemann's kernel and the fundamental solution, moreover we will introduce the concept of characteristic conoid and the geodesic equations. 

In Chapter 3, we will show how to build the fundamental solution with some examples, in particular we will study the scalar wave equation with smooth initial conditions. 

In Chapter 4, by considering linear systems of normal hyperbolic form we will see on which assumptions they can be solved and we will find the solution.

In Chapter 5, we will see under which assumptions a non-linear hyperbolic system can turn into a linear system such that a solution can be found by making use of the results obtained for them. 

In Chapter 6, by making use of the isothermal coordinates we will see how the previous method applies to Einstein vacuum equations to find their solutions and we will discuss the causal structure of space-time. 

Eventually, in Chapter 7, we will give an useful application by studying the Green functions in the gravitational radiation theory, more precisely, we will use the Riemann Kernel to find a solution of the problem of black hole collisions at the speed of light.

\chapter{Hyperbolic Equations} \label{Chap:1} 
\epigraph{In nature's infinite book of secrecy a little I can read.}{William Shakespeare, Antony and Cleopatra }

\section{Systems of Partial Differential Equations} 
To begin with, following Esposito \cite{esposito2017ordinary}, let us consider a system of $m$ partial differential equations in the unknown functions $ \varphi_1, \varphi_2,..., \varphi_n $ of $ n+1$ independent variables $ x^1,x^2,..., x^n $ that reads as \cite{civita1931caratteristiche}

\begin{equation} \label{eq:1} E_{\mu} = 0,       \hspace{3cm}               \mu =1,2,...,m,  \end{equation}
\\
the $ E_{\mu} $ being functions that depend on the x, on the $ \varphi $ and on the partial derivatives of the $\varphi $ with respect to the x. Such a system is said to be $ \textit{normal} $ with respect to the variable $ x^0 $ if it can be reduced to the form:

\begin{equation} \label{eq:2} \frac {\partial^{r_{\nu}} \varphi}{\partial (x^0)^{r_{\nu}}} = \Phi_{\nu} (x|\varphi|\psi|\chi), \hspace{1cm}    \nu=1,2,...,m, \end{equation}
\\
where the $\psi$ occurring on the right-hand side are partial derivatives of each $\varphi_{\nu}$ with respect to $x^0$ only of order less than $r_{\nu}$, and the $\chi$ are partial derivatives of the $\varphi$ with respect to the $x$ of arbitrary finite order, provided that, with respect to $x^0$, they are of order less than $r_{\nu}$ for the corresponding $\varphi_{\nu}$.\\
The functions $\Phi_{\nu}$ are taken to be real-analytic in the neighbourhood of a set of values of Cauchy data. Before stating the associated Cauchy-Kowalevsky theorem, it is appropriate to recall the existence theorem for integrals of a system of ordinary differential equations. Hence we consider the differential system (having  set $x^0 = t $)

\begin{equation} \label{eq:3} \frac{d^{r_{\nu}} \varphi_{\nu}}{ dt^{r_{\nu}}} = \Phi_{\nu} (t| \varphi|\psi),  \hspace{2cm}  \nu = 1,2,...,m. \end{equation}
\\
This system can be re-expressed in canonical form, involving only first-order equations, by defining

\begin{equation} \label{eq:4} \frac{d}{dt} \varphi_{\nu}  \equiv  \varphi'_{\nu}, \hspace{5mm} \frac{d}{dt} \varphi'_{\nu} \equiv \varphi''_{\nu}, \hspace{5mm} ... \hspace{5mm}  \frac{d}{dt} \equiv {\varphi_{\nu}}^{(r_{\nu} - 1)}, \end{equation}
\\
from which replacement  we obtain
\\
\begin{equation} \label{eq:5} \frac{d}{dt} {\varphi_{\nu}}^{(r_{\nu} -1)} = \Phi_{\nu}(t|\varphi|\psi), \hspace{2cm} \nu=1,2,...,m. \end{equation}
\\
One can also denote by $y_\rho $ the generic element of a table displaying $\varphi_1$ and its derivatives up to the order $(r_1 - 1)$ on the first column, $\varphi_2$ and its derivatives up to $(r_2 - 1)$ on the second column, ..., $\varphi_m$ and its derivatives up to the order $(r_m - 1)$ on the last column. With such a notation, the canonical form ($\ref{eq:4}$) is further re-expressed as

\begin{equation} \label{eq:6} \frac{d}{dt} y_\rho = Y_\rho (t|y), \hspace{3 cm} \rho=1,2,...,r; \hspace{5mm} r\equiv \sum\limits_{k=1}^m r_k. \end{equation}
\\
\\
If each $Y_\rho$ is real-analytic in the neighbourhood of $t=t_0$, $y_\rho = b_\rho $, there exists a unique set of functions $y_\rho$, analytic in the $ t$ variable, which take the value $b_\rho$ at $t=t_0$. In order to prove such a theorem, one begins by remarking that the differential equations make it possible, by means of subsequent differentiations, to evaluate the derivatives of any order of an unknown function $y_\rho$ at the point $ t=t_0$ and hence to write, for each $y_\rho$, the Taylor expansion pertaining to such a point. The essential point of the proof consists in showing that such series converge in a suitable neighbourhood of $t=t_0$. For this purpose one considers some appropriate majorizing functions $Y_\rho$; the corresponding differential system (\ref{eq:5}), which can be integrated by elementary methods, defines some real-analytic functions in the neighbourhood of $t=t_0$, whose Taylor expansions majorize the Taylor expansions of the $y_\rho$ functions.
The Cauchy theorem for the differential systems (\ref{eq:5}) holds also when the right-hand side depends on a certain number of parameters that can be denoted by $x_1$,$x_2$,..., $x_n$ provided that they vary in such a way that the functions $Y_\rho$ are real-analytic. One can then state the following:

\begin{thm} Given the differential system
\begin{equation} \label{eq:7} \frac{d^{r_\nu} \varphi_{\nu}}{d{t}^{r_\nu}} = \Phi_\nu (t|x|\varphi|\psi), \hspace{1cm} \nu=1,2,..., m,  \end{equation}
\\
by assigning at will, at $t= t_0$, the values of each $\varphi_\nu$ and of the subsequent derivatives up to the order $ (r_\nu - 1)$ as functions of the parameters $x_1$,$x_2$,..., $x_n$, there exists a unique set of functions $\varphi$, real-analytic, of the variable t and of the parameters, satisfying Eq.(\ref{eq:7}) and equal to the assigned values at $t=t_0$.
\end{thm} 

This theorem admits an extension to systems of partial differential equations (\ref{eq:2}) in normal form, the new feature being that, on the right-hand side of Eq. (\ref{eq:7}), there occur also derivatives of the unknown functions with respect to the parameters, so that one deals with partial differential equations.
The Cauchy problem consists in finding the functions $\varphi$ satisfying the system (\ref{eq:2}) in normal form, and the initial conditions given by the values of the unknown functions and their partial derivatives with respect to the variable $x^0$, of order less than the maximal order. Let $S$ be the space of the variables, from now on denoted by $x^0$,$x^1$,...,$x^n$. In order to fix the ideas, one can assume that $S$ is endowed with an Euclidean metric, and interpret the $x$ as Cartesian coordinates. Let $\omega$ be the hyperplane with equation
\\
\begin{equation} \label{eq:8} x^0 = a^0.  \end{equation}
\\
The Cauchy existence theorem states that, in the neighbourhood of the hyperplane $\omega$, which is said to be the $\textit{carrier hyperplane}$, one can find the values taken by the $\varphi$ functions, once the initial values of $\varphi$ and $\psi$ functions are freely specified at each point of $\omega$. An easy generalization of the theorem is obtained by replacing the hyperplane $\omega$ with a hypersurface $\sigma$ of $S$. For example, if

\begin{equation} \label{eq:9} z(x^0,x^1,...,x^n)=z^0 = constant \end{equation}
\\
is the equation of the hypersurface $\sigma$, it is enough to replace the $x$ with $n + 1$ independent combinations of the $x$, here denoted by $z,z^1,..., z^n$, in such a way that one of them, i.e. $z$, is precisely the left-hand side of the Eq. (\ref{eq:9}) here written for $\sigma$.

\section{Characteristic Manifolds} 
Let us consider differential systems for which the maximal order of derivation is $s=1$ or $s=2$. Such systems can be made explicit by writing them in the form:
\begin{equation} \label{eq:10} E_\mu = \sum\limits_{\nu =1}^m \sum\limits_{i=0}^n E^i_{\mu \nu} \frac{\partial \varphi_\nu}{\partial x^i} + \Phi_\mu (x|\varphi) =0, \hspace{2cm} \mu=1,2,..., m, \end{equation}
\\
and
\begin{equation} \label{eq:11}  E_\mu = \sum\limits_{\nu =1}^m \sum\limits_{i,k=0}^n E^{ik}_{\mu \nu} \frac{\partial^2 \varphi_\nu}{\partial x^i \partial x^k} + \Phi_\mu (x|\varphi| \chi) =0, \hspace{1cm} \mu=1,2,..., m, \end{equation}
\\
respectively. In Eq. ($\ref{eq:10}$) the $E^i_{\mu \nu}$ and $ \Phi_\mu $ depend on the $x$ and $\varphi$, whereas in Eq. ($ \ref{eq:11}$) the $E^{ik}_{\mu \nu} $ and $\Phi_\mu$ depend on the $x$, $\varphi$ and on the first-order partial derivatives of the $\varphi$ with respect to the $x$. Since the $\varphi_\nu$ are taken to fulfill the conditions  under which one can exchange the order of derivatives, one can always assume that $E^{ik}_{\mu \nu}$ is symmetric in the lower case Latin indices. 

In the particular case of a single unknown function $\varphi$, Eqs. ($\ref{eq:11}$) reduce to the single equation:
\begin{equation} \label{eq:12} E =  \sum\limits_{i,k=0}^n E^{ik} \frac{\partial^2 \varphi_\nu}{\partial x^i \partial x^k} + \Phi(x|\varphi| \chi) =0,  \end{equation} \\
where $\chi$ is a concise notation for the first-order partial derivatives of $\varphi$ with respect to $x^0$, $x^1$, ..., $x^n$. 

A remarkable equation of the type ($\ref{eq:12}$) is the scalar wave equation (having set $x^0=t$ in $c=1$ units): 
\begin{equation} \label{eq:13} \Box \varphi = \bigg{(} \frac{1}{V^2}\frac{\partial^2}{\partial t^2} - \Delta \bigg{)} \varphi = 0, \end{equation} 
where $V$ is a constant and $\Delta = \sum \limits_{i=1}^3 \frac{\partial^2}{\partial (x^i)^2}$ the standard notation for minus the Laplacian in Euclidean three-dimensional space $\mathbb{R}^3$. The $\Box$ operator in Eq. ($\ref{eq:13}$) is the familiar D'Alembert operator for the wave equation in Minkowski space-time. 

The systems ($\ref{eq:10}$) and ($\ref{eq:11}$) are not yet written in normal form, and we now aim at finding the conditions under which such systems are normal with respect to the variable $x^0$. For this purpose, we begin with the system ($\ref{eq:10}$) and point out that, since we are only interested in first-order partial derivatives with respect to $x^0$, we can re-express such equations in the form
\begin{equation} \label{eq:14} \sum\limits_{\nu =1}^m E^0_{\mu \nu} \frac{\partial \varphi_\nu}{\partial x^0} + ... =0, \hspace{3cm} \mu= 1,2,...,m. \end{equation} \\
This system can be solved with respect to the derivatives $\frac{\partial \varphi_\nu}{\partial x^0}$ if the determinant of the matrix $E^0_{\mu \nu}$ does not vanish, i.e.
\begin{equation} \label{eq:15} \Omega = det E^0_{\mu \nu} \neq 0, \hspace{2cm} \mu,\nu=1,2,...,m. \end{equation}
Such a determinant involves the independent variables $x^0$, $x^1$, ..., $x^n$ and also, in general, the unknown functions $\varphi_1$, $\varphi_2$, ..., $\varphi_m$. \\
Let us now consider the Eq. ($\ref{eq:11}$) of the second system, which are written more conveniently in the form
\begin{equation} \label{eq:16} \sum\limits_{\nu=1}^m E^{00}_{\mu \nu} \frac{ \partial^2 \varphi_\nu}{\partial (x^0)^2} + ... = 0, \hspace{2cm} \mu=1,2,...,m, \end{equation} 
and are hence solvable with respect to $\frac{\partialì2 \varphi_\nu}{\partial (x^0)^2}$ if the determinant of the matrix $E^{00}_{\mu \nu}$ does not vanish, i.e.
\begin{equation} \label{eq:17} \Omega = det E^{00}_{\mu \nu} \neq 0, \hspace{2cm} \mu,\nu=1,2,...,m. \end{equation}
Furthermore, the single equation Eq. ($\ref{eq:12}$) can be put in normal form provided that
\begin{equation} \label{eq:18} E^{00} \neq 0. \end{equation}
If the normality conditions ($\ref{eq:15}$), ($\ref{eq:17}$) and ($\ref{eq:18}$) are satisfied, for a given carrier hyperplane having equation $x^0 = a^0$, one can apply the Cauchy theorem, and the functions $\varphi_\nu$, or the single function $\varphi$ of Eq. ($\ref{eq:12}$), are uniquely determined in the neighbourhood of such hyperplane.\\
It is now necessary to investigate under which conditions the normal character is preserved, if the independent variables $x^0$, $x^1$, ..., $x^n$ are mapped into new variables $z$, $z^1$, ..., $z^n$, so that the hyperplane of equation $x^0=a^0$ is turned into a hypersurface $\sigma$ of the space $S$ having equation

\begin{equation} \label{eq:19} z(x^0,x^1,...,x^n)=z^0, \end{equation}
\\
starting from which one can determine (at least in a neighbourhood) the $\varphi$ functions. \\
For this purpose, one defines
\begin{equation} \label{eq:20} p_i \equiv \frac{\partial z}{\partial x^i}, \hspace{2cm} i=0,1,...,n, \end{equation}
from which one obtains
\begin{equation} \label{eq:21} \frac{\partial \varphi_\nu}{\partial x^i} = \frac{\partial \varphi_\nu}{\partial z} p_i + \sum\limits_{j=1}^3 \frac{\partial \varphi_\nu}{\partial z^j} \frac{\partial z_j}{\partial x^i}, \hspace{2cm} \nu=1,2,...,m, \end{equation}
where we need, on the right-hand side, only the first term, so that we write:
\begin{equation} \label{eq:22} \frac{\partial \varphi_\nu}{\partial x^i} = \frac{\partial \varphi_\nu}{\partial z}p_i + ..., \hspace{2cm} \nu=1,2,...,m. \end{equation}
The insertion of ($\ref{eq:22}$) into the system ($\ref{eq:10}$) yields 
\begin{equation}\label{eq:23} \sum\limits_{\nu=1}^m \frac{\partial \varphi_\nu}{\partial z} \sum\limits_{i=0}^n E^i_{\mu \nu} p_i + ... = 0, \hspace{2cm} \mu=1,2,...,m. \end{equation}
If now one sets
\begin{equation} \label{eq:24} \omega^{(1)}_{\mu \nu} \equiv \sum\limits_{i=0}^n E^i_{\mu \nu}p_i, \end{equation}
the transformed system turns out to be normal provided that
\begin{equation} \label{eq:25} \Omega^{(1)} \equiv det \omega^{(1)}_{\mu \nu} \neq 0, \hspace{2cm} \mu,\nu=1,2,..,m. \end{equation}
As far as the system ($\ref{eq:11}$) is concerned, one finds in analogous way
\begin{equation}\label{eq:26} \frac{\partial^2 \varphi_\nu}{\partial x^i \partial x^k} = \frac{\partial^2 \varphi_\nu}{\partial z^2} p_i p_k + ..., \end{equation}
and Eqs. $(\ref{eq:11})$ are turned into
\begin{equation*} \sum_{\nu =1}^m \frac{\partial^2 \varphi_\nu}{\partial z^2} \sum_{i, k=0}^n E^{ik}_{\mu \nu} p_i p_k + ...=0, \hspace{2cm} \mu=1, 2, ..., m. \end{equation*}
If one defines the matrix
\begin{equation} \label{eq:27} \omega^{(2)}_{\mu \nu} \equiv \sum\limits_{i,k=0}^n E^{ik}_{\mu \nu} p_i p_k, \end{equation}
the condition of normality of the system is expressed by non-singularity of this matrix, i.e.
\begin{equation} \label{eq:28} \Omega^{(2)} \equiv det \omega^{(2)}_{\mu \nu} \neq 0, \hspace{2cm} \mu, \nu=1,2,...,m. \end{equation}
Note that, in Eq. ($\ref{eq:25}$), the $\omega^{(1)}_{\mu \nu}$ are linear forms of the variables $p_0$, $p_1$, ..., $p_n$, and hence $\Omega^{(1)}$ is a form of degree $m$ in such arguments, while in Eq. ($\ref{eq:28}$) the $\omega^{(2)}_{\mu \nu}$ are quadratic forms of the $p$, and hence $\Omega^{(2)}$ is a form of degree $2m$ of the argumets $p_0$, $p_1$, ..., $p_n$.
\\
In the case of the unique Eq. ($\ref{eq:12}$), the determinant reduces to the only element
\begin{equation} \label{eq:29} \Omega = \sum\limits_{i,k=0}^n E^{ik}p_i p_k \end{equation}
To sum up, to every function $z(x^0, x^1, ..., x^n)$ for which $\Omega$ does not vanish identically, there corresponds a family of hypersurfaces $z=z^0$, starting from each of which it is still possible to solve the Cauchy problem. This consists in determining the unknown functions when the initial data are relative to the hypersurface itself. This holds by virtue of the normal character of the transformed system with respect to $z$.
\\
When the function $z(x^0, x^1, ..., x^n)$ satisfies the equation
\begin{equation} \label{eq:30} \Omega=0, \end{equation}
it is no longer possible to apply (regardless of the value taken by the constant $z^0$) the Cauchy theorem starting from the carrier hypersurfaces $z=z^0$. In such a case, the carrier hypersurfaces are said to be $\textit{characteristic manifolds}$ \cite{Friedlander:2010eqa, civita1931caratteristiche}. \\ Equation ($\ref{eq:30}$) warns us that the system formed by the variables $z, z^1, ..., z^n$ is not normal with respect to $z$ and makes it possible to assign the manifolds, in correspondence to which one cannot state that the unknown functions can be determined, once the values of the unknown functions and their derivatives of order less than the maximum have been assigned on the manifold.
\\
For the case of the single Eq. ($\ref{eq:12}$), the characteristic manifold is the one satisfying the equation

\begin{equation} \label{eq:31} \sum\limits_{i,k=0}^n E^{ik}p_i p_k =0. \end{equation}
\\
Such a manifold is necessarily complex if the quadratic form on the left-hand side of ($\ref{eq:31}$)  is positive-definite. Otherwise the manifold is real, if the initial data, called Cauchy data, are real. Equation ($\ref{eq:31}$) can then be viewed as expressing the vanishing of the square of the pseudo-norm of the normal vector, and is therefore a null hypersurface. In other words, with the nomenclature of relativity and pseudo-Riemannian geometry, $\textit{characteristic manifolds}$ are null hypersurfaces.

\section{The Concept of Wavelike Propagation}
The scal wave equation ($\ref{eq:13}$) can be applied, in particular, to the air's acoustic vibrations, or to the vibrations of other gaseous masses, since one can neglect in a first approximation any dissipative effect and hence consider the motion as if it were irrotational, without heat exchange among particles (this behaviour is called $\textit{adiabatic}$). If the velocity potential $\varphi$ in Eq. ( $\ref{eq:13}$) describes sound vibrations in the air, the three partial derivatives represent the speed of the air molecule located in $(x^1, x^2, x^3)$ at time $t$. More precisely, what is vibrating at a generic time $t$ is a certain air's stratum, placed in between the two surfaces

\begin{equation} \label{eq:32} z(t|x)=c_1, \hspace{1cm} z(t|x) = c_2 . \end{equation}
Outside this stratum there is rest; i.e. the solution of Eq. ($\ref{eq:13}$) vanishes, whereas within the stratum the acoustic phenomenon is characterized by a non-vanishing solution $\varphi(t|x)$.
\\
From now on, without insisting on the acoustic interpretation of the solutions of Eq. ($\ref{eq:13}$), we assume that $\varphi(t|x)$ and $\varphi^*(t|x)$ are solutions of this equation within and outside of the stratum determined by the surfaces in Eq. ($\ref{eq:32}$), respectively.\\
The phenomenon described by Eq. ($\ref{eq:13}$) is characterized by two distinct functions, depending on whether it is studied inside or outside the stratum. Throughout the surface of Eq. ($\ref{eq:32}$) the derivatives of various orders of $\varphi$ will undergo, in general, sudden variations and, for this reason, one says we are dealing with $\textit{discontinuity surfaces}$. Now it may happen that such discontinuities vary with time, in which case the discontinuity that undergoes propagation is said to be a $\textit{wave}$.\\
Thus, if Eq. ($\ref{eq:13}$) is interpreted as characterizing a possible wavelike propagation, the discontinuity surface (or $\textit{wave surface}$) bounds a stratum that undergoes displacement and, possibly, deformation with time. We shall assume that, during the motion, no interpenetration or molecular cavitation occurs, so that, on passing through a wave surface, the normal component of the velocity of a generic particle does not undergo any discontinuity. We also rule out the possible occurence of sliding phenomena of molecules on such surfaces, which would lead to tangential discontinuities of the velocity of particles. Moreover, in light of the postulates on the pressure that the mechanics of continua relies upon, the pressure cannot, under ordinary conditions, undergo sudden jumps, even if the state of motion were to change abruptly. The density $\mu$ is related to the pressure $p$ by the equation of state $f(\mu,p)=0$, which is the same on both sides of the discontinuity surface. The continuity of $p$ implies therefore that also $\mu$ is continuous.

On the other hand, we have
\begin{equation}\label{eq:33} \frac{\partial \varphi}{\partial t} + V^2 \sigma =0 \end{equation} 
the derivatives of $\varphi$ and $\varphi^*$ with respect to $t$ represent a density up to a constant factor, hence also such derivatives must be continuous across the wave surface.\\
By virtue of all previous considerations one can say that, for the Eq. ($\ref{eq:13}$) to describe a wavelike propagation, one has to assume the existence of two different solutions, say $\varphi$ and $\varphi^*$, taken to characterize the physical phenomenon inside and outside of a stratum, that match each other, i.e. have equal first-order derivatives in time and space, through the wave surface which bounds the stratum at every instant of time. The second derivatives undergo instead sudden variations.\\
Let us now consider one of the wave surfaces $\sigma$ which, at time $t$, bound the stratum where the pertubation is taking place, and let $n$ be the outward-pointing normal to such a stratum at a generic point $P$. The surface undergoes motion and, at time $t + dt$, intersects the normal $n$ at a point $Q$. The measure of the $PQ$ segment, counted positively towards the exterior, can be denoted $d\textit{n}$. The ratio

\begin{equation} \label{eq:34} a \equiv \frac{d \textit{n}}{dt} \end{equation}
\\
is said to be the $\textit{progression velocity}$ of the wave surface at the point $P$ at the instant of time under consideration. Under ordinary circumstances, at all points of one of the two limiting surfaces of the stratum, $\textit{a}$ is positive, while at all points of the other limiting surface $\textit{a}$ is negative. The former surface is said to be a $\textit{wave front}$ or a $\textit{bow}$, while the latter is said to be a $\textit{poop}$. The difference

\begin{equation} \label{eq:35} v(P) \equiv a - \frac{d \varphi}{d \textit{n}} \end{equation}
\\
between the progression velocity and the component orthogonal to $\sigma$ of the velocity of the fluid particle placed at the point $P$ at the instant $t$ is said to be the $\textit{normal propagation velocity}$ of the surface $\sigma$ at the point $P$. This velocity measures the rate at which the surface is moving with respect to the medium (and not with respect to the fixed axes!).

If outside the stratum there is a rest condition, the solution $\varphi^*$ vanishes and therefore, by virtue of the matching conditions at $\sigma$, one can write that

\begin{equation} \label{eq:36} \frac{d \varphi}{ d \textit{n}} =0 \Rightarrow v(P)= a. \end{equation}
\\
In this particular case the propagation velocity coincides with the progression velocity.\\
Note now that the surface $\sigma$ is a characteristic manifold of Eq. ($\ref{eq:13}$), i.e. an integral of the equation 
\begin{equation} \label{eq:37} \frac{1}{V^2}(p_0)^2 - \sum\limits_{i=1}^3(p_i)^2=0. \end{equation}
\\
Indeed, if this were not true, a unique solution of Eq. ($\ref{eq:13}$) would be determined in the neighbourhood of $\sigma$ by the mere knowledge of the values taken upon $\sigma$ by $\varphi$ and $\frac{\partial \varphi}{\partial t}$, in light of Cauchy's theorem. The wavelike propagation is therefore possible  because the wave surfaces are characteristic manifolds.

In order to further appreciate how essential is the consideration of characteristic manifols, let us study the following example \cite{civita1931caratteristiche}. Let us assume for simplicity that we study the wave equation ($\ref{eq:13}$) in two-dimensional Minkowski space-time, with $x^1$ denoted by $x$. Hence we write it in the form

\begin{equation} \label{eq:38} \bigg{(} \frac{1}{V^2} \frac{\partial^2}{\partial t^2} - \frac{\partial^2}{\partial x^2} \bigg{)} \varphi = \bigg{(} \frac{1}{V} \frac{\partial}{\partial t} + \frac{\partial}{\partial x} \bigg{)} \bigg{(} \frac{1}{V} \frac{\partial}{\partial t} - \frac{\partial}{\partial x} \bigg{)}=0. \end{equation}
\\
The form of Eq. ($\ref{eq:38}$) suggests defining the new variables
\begin{equation} \label{eq:39} z \equiv x - Vt, \hspace{1cm} z_1 \equiv x + Vt, \end{equation}
from which the original variables are re-expressed as
\begin{equation}\label{eq:40} x= \frac{1}{2} (z+z_1), \hspace{1cm} t= \frac{1}{2} \frac{(z_1-z)}{V}. \end{equation}
Moreover, the standard rules for differentiation of composite functions lead now to

\begin{equation}\label{eq:41} \frac{\partial}{\partial z}= \frac{1}{2} \bigg{(} \frac{\partial}{\partial x}- \frac{1}{V} \frac{\partial}{\partial t} \bigg{)}, \hspace{1cm} \frac{\partial}{\partial z_1} = \frac{1}{2} \bigg{(} \frac{\partial}{\partial x} + \frac{1}{V} \frac{\partial}{\partial t} \bigg{)}, \end{equation}
\\
and hence Eq. ($\ref{eq:38}$) reads as
\begin{equation} \label{eq:42} \frac{\partial^2 \varphi}{\partial z \partial z_1} = 0, \end{equation}
which is solved by a sum of arbitrary smooth functions
\begin{equation} \label{eq:43} \varphi(z,z_1)= \alpha(z) + \beta (z_1) \end{equation}
depending only on $z$ and on $z_1$, respectively. Thus, it is not possible in general to solve the Cauchy problem for a carrier line $z=c$, but it is necessary that the data satisfy a compatibility condition. In our case, from the solution ($\ref{eq:43}$) one finds

\begin{equation} \label{eq:44} \varphi(c,z_1)= \alpha(c) + \beta(z_1), \hspace{1cm} \bigg{(} \frac{\partial \varphi}{\partial z}\bigg{)}_{z=c} = \alpha'(c). \end{equation}
\\
The functions $\varphi_0=\varphi(z=c)$ and $\varphi_1= \big{(}\frac{\partial \varphi}{\partial z} \big{)}_{z=c}$ of the variable $z_1$ cannot be therefore chosen at will, but the function $\varphi_1(z_1)$ must be a constant, in which case there exist infinitely many forms of the solution of the Cauchy problem for the scalar wave equation. 

\section{The Concept of Hyperbolic Equation}
The scalar wave equation ($\ref{eq:13}$) is a good example of hyperbolic equation, but before we go on it is appropriate to define what is an equation of hyperbolic type. Following Leray \cite{leray1955hyperbolic}, we first define this concept on a vector space and then on a manifold.\\
We consider a $l$-dimensional vector space $X$ over the field of real numbers, whose dual vector space is denoted by $\Xi$. The point $x=(x^1, ..., x^l) \in X$, and the point $p=\big{(} \frac{\partial}{\partial x^1},..., \frac{\partial}{\partial x^l}\big{)} \in \Xi$. A differential equation of order $m$ can be therefore written in the form

\begin{equation} \label{eq:45} a(x,p)u(x)=v(x), \end{equation}
\\
where $a(x,\xi)$ is a given real polynomial in $\xi$ of degree $m$ whose coefficients are functions defined on $X$, $u(x)$ is the unknown function and $v(x)$ a given function. Let $h(x,\xi)$ be the sum of the homogeneous terms of $a(x,\xi)$ of degree $m$ (also called the $\textit{leading symbol}$ of the differential operator $a(x,p)$), and let $V_x(h)$ be the cone defined in $\Xi$ by the equation

\begin{equation} \label{eq:46} h(x,\xi) =0. \end{equation}
\\
The differential operator $a(x,p)$ is said to be $\textit{hyperbolic at the point x}$ if $\Xi$ contains points $\xi$ such that any real line through $\xi$ cuts the cone $V_x(h)$ at $m$ real and distinct points. These points $\xi$ constitute the interior of two opposite convex and closed half-cones $\Gamma_x(a)$ and $- \Gamma_x(a)$, whose boundaries belong to $V_x(h)$.\\
Suppose that the following conditions hold:

\begin{description} 
\item[(i)] The operator $a(x,p)$ is hyperbolic at each point $x$ of the vector space $X$. 
\item[(ii)] The set
\begin{equation} \label{eq:47} \Gamma_X \equiv \cap_{x \in X} \Gamma_x \end{equation} 
has a non-empty interior.
\item[(iii)] No limit of $h(x,\xi)$ as the norm of $x$ approaches 0 is vanishing.
\item[(iv)] No limit of the cones $V_x(h)$ as the norm of $x$ approaches infinity has singular generator. 
\end{description} 

Under such circumstances, the operator $a(x,p)$ is said to be $\textit{regularly}$ $\textit{hyperbolic}$ on $X$.
When $X$ is instead a $l-dimensional$ $(m+ M)$-smooth manifold, not necessarily complete, the operator $a(x,p)$ is said to be $\textit{hyperbolic on}$ $X$ when the following conditions hold:

\begin{description} 
\item[(1)] $a(x,p)$ is hyperbolic at any point $x$ of $X$, in the sense specified above.
\item[(2)] The set of timelike paths (i.e. with timelike tangent vector) from $y$ to $z$ is compact or empty for any $y$ and $z$ $\in$ $X$.
\item[(3)] Either the coefficients of $a(x,p)$ have $\textit{locally bounded}$ (which means boundedness on any compact subset of $X$) derivatives of order $M$ such that $1 \leq M \leq l$, or they have locally bounded derivatives of order $\leq l'$ and locally square integrable derivatives of order $>l'$ and $\leq M$, $l'$ being the smallest integer $> \frac{l}{2}$. This technical condition will became clear in one of the following chapters.
\item[(4)] The total curvature of the interior of $\Gamma_x$ is positive. If $M=1$, then the first derivatives of the coefficients of $h(x,\xi)$ are continuous. 
\end{description} 

\section{Riemann Kernel}
The modern theory of hyperbolic equations was initiated by Riemann’s representation of the solution of the initial-value problem for an equation of second order. Riemann was motivated by a very concrete problem in acoustics, but here we focus on the mathematical ingredients of his conceptual construction. \\
Given a differential expression $\varphi(x, y, y^1, ..., y^n)$ of the variable $x$, a function $y$ and its derivatives up to the $n^{\rm th}$-order, the equation 
\begin{equation} \label{eq:48} \frac{\partial \varphi}{\partial y} - \frac{d}{dx}\bigg{(} \frac{\partial \varphi}{\partial y^1}\bigg{)} + \frac{d^2}{dx^2}\bigg{(} \frac{\partial \varphi}{\partial y}\bigg{)} - ... =0 \end{equation}
expresses the necessary and sufficient condition such that the $\varphi$ function is the derivative of a function $\psi$ which contains, at the same time, the  independent variable $x$, the function $y$ and its $(n-1)$ first derivatives. In the same way, if we consider an expression $\varphi \big{(}x, y, z, \frac{\partial z}{\partial x}, \frac{\partial z}{\partial y}, \frac{\partial^2 z}{\partial x^2}, \frac{\partial z}{\partial x \partial y}, \frac{\partial^2 z}{\partial y^2}, \dots \big{)}$ which contains two independent variables, a function $z$ of them and its partial derivatives up to any $n^{\rm th}$-order, the equation
 \begin{equation} \label{eq:49} \begin{split}   &\frac{\partial \varphi}{\partial z} - \frac{\partial}{\partial x} \bigg{(} \frac{\partial \varphi}{\partial ( \frac{\partial z}{\partial x}) }\bigg{)}  - \frac{\partial}{\partial y} \bigg{(} \frac{\partial \varphi}{\partial ( \frac{\partial z}{\partial y}) }\bigg{)} + \frac{\partial^2}{\partial x^2} \bigg{(} \frac{\partial \varphi}{\partial ( \frac{\partial^2 z}{\partial x^2}) }\bigg{)}\\
&+ \frac{\partial^2}{\partial y^2} \bigg{(} \frac{\partial \varphi}{\partial ( \frac{\partial^2 z}{\partial y^2}) }\bigg{)} + \frac{\partial^2}{\partial x \partial y} \bigg{(} \frac{\partial \varphi}{\partial ( \frac{\partial^2 z}{\partial x \partial y}) }\bigg{)}- ... =0 \end{split}\end{equation} 
expresses the necessary and sufficient condition such that $\varphi$ can be read as $\frac{\partial M}{\partial x} + \frac{\partial N}{\partial y}$, where $M$ and $N$ are functions of $x$, $y$, $z$ and of their partial derivatives up to an order that can be reduced to $n$ or to $(n-1)$. Now, we consider a linear hyperbolic equation of order $n$

\begin{equation} \label{eq:50} L[z] = \sum \limits_{i,k=0}^n A_{ik} \frac{\partial^{i+k} z}{\partial x^i \partial y^k} = 0. \end{equation}
\\
If we multiply the left-hand side by an unknown $u$ and if the Eq. ($\ref{eq:49}$) is verified, we have the linear equation:

\begin{equation} \label{eq:51} L^*[u] = \sum\limits_{i,k=0}^n (-1)^{i+k} \frac{\partial^{i+k}}{\partial x^i \partial y^k} (A_{ik} u) =0, \end{equation}
\\
which defines $u$. This equation is the adjoint of the proposed equation. For any $z$ and $u$, a series of integrations by parts lead us to the identity

\begin{equation} \label{eq:52}  u L[z] - z L^*[u] = \frac{\partial M}{\partial x} + \frac{\partial N}{\partial y} \end{equation}
\\
where $M$ and $N$ have the following values:

\begin{equation} \label{eq:53} 
\left\{\begin{array} {l}
M= A_{10} zu + A_{20} u \frac{\partial z}{\partial x} - z \frac{\partial (A_{20} u)}{ \partial x} + \frac{1}{2} A_{11}u \frac{\partial z}{\partial y} - \frac{1}{2} z \frac{\partial (A_{11} u)}{\partial y} + ... \\
N= A_{01}zu + A_{02} u \frac{\partial z}{\partial y} - z \frac{\partial (A_{02} u)}{ \partial y} + \frac{1}{2} A_{11}u \frac{\partial z}{\partial x} - \frac{1}{2} z \frac{\partial (A_{11} u)}{\partial x} + ...
\end{array}\right.
\end{equation}
and depend on $z$, $u$ and their partial derivatives up to the $(n-1)^{\rm th}$ order.\\
It is important to remark that the expressions of $M$ and $N$ are not completely specified. The right-hand side of the identity ($\ref{eq:52}$), has the same expression if we replace $M$ and $N$ with $M - \frac{\partial \theta}{\partial x} $ and $ N - \frac{\partial \theta}{\partial y}$ and we can take as $\theta$ a linear function of $z$, $u$ and their partial derivatives up to the $(n-2)^{\rm th}$ order, without changing the general form of the values of $M$ and $N$. We can deduce from the previous identity that the relation between $L[z]$ and $L^*[u]$ is mutual, meaning that each equation is the adjoint of the other. \\
To estabilish the identity $ u L[z] - z L^*[u] = \frac{ \partial M}{\partial x} + \frac{\partial N}{\partial y}$ in all its generality, we can make the following calculation.\\
We say that the expression $u \frac{\partial^i v}{\partial x^i} - (-1)^i v \frac{\partial^i u}{\partial x^i} $ is the exact derivative of a function of $u$ and $v$ and of their derivatives up to the $(i-1)^{\rm th}$ order. If we replace $v$ with $\frac{\partial^k v}{\partial y^k}$, we have 

\begin{equation} \label{eq:54} u \frac{\partial^{i+k} v}{\partial x^i \partial y^k} - (-1)^i \frac{\partial^k v}{\partial y^k} \frac{\partial^i u}{\partial x^i} = \frac{\partial P}{\partial x} \end{equation}
\\
$P$ contains the derivatives of $u$ and $v $ up to the order $(i+k-1)$.\\
If we replace, in the previous equation, $u$ with $v$, $x$ with $y$ and $i$ with $k$, we have

\begin{equation} \label{eq:55} v \frac{\partial^{i+k}u}{\partial x^i \partial y^k} - (-1)^k \frac{\partial^i u}{\partial x^i}\frac{\partial^k v}{\partial y^k} = \frac{\partial Q}{\partial y}. \end{equation}
\\
The combination of equations $(\ref{eq:54})$ and $(\ref{eq:55})$ gives us the most general identity

\begin{equation} \label{eq:56} u \frac{\partial^{i+k} v}{\partial x^i \partial y^k} - (-1)^{i-k}v \frac{\partial^{i+k} u}{\partial x^i \partial y^k} = \frac{\partial P}{\partial x} - (-1)^{i - k} \frac{\partial Q}{\partial y} \end{equation}
\\
where $P$ and $Q$ contain the $u$ and $v$ partial derivatives up to the $(i+k-1)^{\rm th}$ order. We have

\begin{equation} \label{eq:57} u L[z] - z L^*[u] = \sum\limits_{i,k=0}^n \bigg{(} u A_{ik} \frac{\partial^{i+k} z}{\partial x^i \partial y^k} - (-1)^{i-k}z \frac{\partial^{i+k}(A_{ik}u)}{\partial x^i \partial y^k}\bigg{)} \end{equation}
and it is possible to use the identity ($\ref{eq:52}$), by replacing $u$ with $A_{ik}u$ and $v$ with $z$, to recognize that the right-hand side of the previous equation can be written as $\frac{\partial M}{\partial x} + \frac{\partial N}{\partial y}$, where $M$ and $N$ contain the derivatives up to the order $(n-1)$.
Let us now consider the double integral
\begin{equation} \label{eq:58} \int \int dx dy (u L[z] - z L^*[u]) = \int \int dx dy \bigg{(} \frac{\partial M}{\partial x} + \frac{\partial N}{\partial y}\bigg{)} \end{equation}
extended to a plane's area $A$ which we suppose to be simply connected and bounded in $S$; this double integral has the same value of the simple integral $\int ( M dy - N dx)$ extended to the bound $S$ walked in the strict sense. Thus, Eq. ($\ref{eq:57}$) may be written in the form

\begin{equation} \label{eq:59} \int \int dx dy (u L[z] - z L^*[u]) = \int_S (M dy - N dx) \end{equation}
that is equivalent to the identity ($\ref{eq:52}$). It is possible to recognize that the indeterminacy stated above for the values of $M$ and $N$ does not affect the previous equation. Indeed, if we replace in Eq. ($\ref{eq:59}$) $M$ and $N$ with their more general values $M + \frac{\partial \theta}{\partial y}$ and $ N - \frac{\partial \theta}{\partial x}$, the right-hand side of the previous equation increases of the integral $\int_S d \theta$ which clearly vanishes everytime that $\theta$ is a finite and uniform function inside the area $A$.\\
To discuss Riemann's method, let us consider a second-order linear hyperbolic equation in two variables that can be written in two equivalent forms 
\begin{equation}\label{eq:60} L[z] \equiv \bigg{(} \frac{\partial^2}{\partial x \partial y} + a(x,y) \frac{\partial}{\partial x} + b(x,y) \frac{\partial}{\partial y} + c(x,y) \bigg{)} z = f(x,y), \end{equation}

\begin{equation} \label{eq:61} L[z] \equiv \bigg{(} \frac{\partial^2}{\partial y^2} - \frac{\partial^2}{\partial x^2} + d(x,y) \frac{\partial}{\partial x} + h(x,y) \frac{\partial}{\partial y} + e(x,y) \bigg{)} z = f(x,y), \end{equation}
where $a$, $b$, $c$, $d$, $h$, $e$ and $f$ are of a suitable differentiability class. The initial curve $C$ is taken to be nowhere tangent to a characteristic direction; the characteristics pertaining to Eq. ($\ref{eq:60}$) are straight lines parallel to the coordinate axes; the characteristics in Eq. ($\ref{eq:61}$) are the lines $x + y = const.$ and $x - y = const$.

The aim is to represent a solution $z$ at a point $P$ in terms of $f$ and of the initial data, i.e. the values taken by $z$ and one derivative of $z$ on $C$. If the initial curve degenerates into a right angle formed by the characteristics $x = \gamma$ and $y= \delta$, it is no longer possible to prescribe two conditions on the initial curve $C$, but it is necessary to consider the $\textit{characteristic initial-value problem}$, in which only the values of $u$ on $x = \gamma$ and $y=\delta$ are prescribed.

Now we choose to consider Eq. ($\ref{eq:60}$) and to use the Riemann's method which consists in multiplying the hyperbolic equation by a function $u$, integrating over a region $A$, transforming the integral by Green's formula such that $z$ appears as a factor of the integrand, then to try to determine $u$ in such a way that the required representation is obtained. This procedure is implemented by introducing the adjoint operator $L^*$, defined to give, as we have seen before, $u L[z] - zL^*[u]$, which is a divergence. 

For the hyperbolic equation in the form ($\ref{eq:60}$), the adjoint operator $L^*$ turns out to be
\begin{equation} \label{62} L^*[u] = \bigg{(} \frac{\partial^2}{\partial x \partial y} - \frac{\partial a}{\partial x} - a(x,y) \frac{\partial}{\partial x} - \frac{\partial b}{\partial y} - b(x,y) \frac{\partial}{\partial y} + c(x,y) \bigg{)} u, \end{equation}
and hence we have
\begin{equation} \label{eq:63} 
\left\{\begin{array} {l}
L[z] = \frac{\partial^2 z}{\partial x \partial y} + a \frac{\partial z}{\partial x} + b \frac{\partial z}{\partial y} + cz, \\
L^*[u] = \frac{\partial^2 u}{\partial x \partial y} - a \frac{\partial u}{\partial x} - b \frac{\partial u}{\partial y} + \bigg{(} c - \frac{\partial a}{\partial x} - \frac{\partial b}{\partial y} \bigg{)} u, \\
M = auz + \frac{1}{2} \bigg{(} u \frac{\partial z}{\partial y} - z \frac{\partial u}{\partial y} \bigg{)}, \\
N= buz + \frac{1}{2} \bigg{(} u \frac{\partial z}{\partial x} - z \frac{\partial u}{\partial x} \bigg{)}.
\end{array}\right.
\end{equation}
More precisely, the identity $ u L[z] - z L^*[u] = \frac{\partial M}{\partial x} + \frac{\partial N}{\partial y}$ reads as

\begin{equation*} \begin{split} & u L[z] - zL^*[u] = u\frac{\partial^2 z}{\partial x \partial y} + au \frac{\partial z}{\partial x}  + bu\frac{\partial z}{\partial y}  + cu z - z \frac{\partial^2 u }{\partial x \partial y} + az \frac{\partial u}{\partial x} \\
&+ bz \frac{\partial u}{\partial y} - z \bigg{(} c - \frac{\partial a}{\partial x} - \frac{\partial b}{\partial y} \bigg{)} u = \frac{\partial}{\partial x}(azu) + \frac{\partial}{\partial y}(bzu) + u \frac{\partial^2 z}{\partial x \partial y} - z \frac{\partial^2 u}{\partial x \partial y}\\
& = \frac{\partial}{\partial x} \bigg{[} azu + \frac{1}{2} \bigg{(} u \frac{\partial z}{\partial y} - z \frac{\partial u}{\partial y} \bigg{)}\bigg{]} + \frac{\partial}{\partial y} \bigg{[} bzu + \frac{1}{2} \bigg{(} u \frac{\partial z}{\partial x} - z \frac{\partial u}{\partial x} \bigg{)} \bigg{]}. \end{split}\end{equation*}
This equation can be re-expressed in the form
\begin{equation}\label{eq:64}u L [z] - z L^*[u] = \frac{\partial}{\partial y} \bigg{(} u \frac{\partial z}{\partial x} + b zu \bigg{)} - \frac{\partial}{\partial x} \bigg{(} z \frac{\partial u}{\partial y} - a z u \bigg{)}. \end{equation}
Let us suppose to take as $z$ and $u$ any integrals of the equation proposed and of its adjoint equation. The integration over a two-dimensional domain $A$ with boundary $S$ and Gauss' formula lead to
\begin{equation} \label{eq:65} - \int\int_A dx dy ( u L[z] - z L^*[u] ) = \int_S \bigg{[} \bigg{(} u \frac{\partial z}{\partial x} + bzu\bigg{)} dx + \bigg{(} z \frac{\partial u}{\partial y} - azu \bigg{)} dy \bigg{]}. \end{equation}
Now, if $L[z]=0$ and $L^*[u]=0$, the left-hand side of Eq. ($\ref{eq:65}$) will always be equal to zero and then it reads as

\begin{equation} \label{eq:66} \int_S \bigg{[}\bigg{(} u \frac{\partial z}{\partial x} + bzu \bigg{)} dx + \bigg{(} z \frac{\partial u}{\partial y} - azu \bigg{)} dy \bigg{]}=0, \end{equation} 
\\
or similarly

\begin{equation*} \int_S \big{(} M dy - N dx \big{)} =0. \end{equation*}
\\
Let $A$ be a point of the plane and $B'C'$ an arbitrary curve in this plane. If we draw from $A$ two straight lines $AB$ and $AC$ parallel to the axes which intersect the curve, and suppose that the integrals $z$, $u$, as well as the coefficients of the proposed equation and their derivatives, are finite and continuous inside the area $ABC$. By integration of the previous equation along the path $ACBA$, in fig.($\ref{fig:1}$), we have:
\begin{equation} \label{eq:67} \int\limits_A^C M dy + \int\limits_C^B (M dy - N dx) - \int\limits_B^A N dx = 0 \end{equation}
where 
\begin{equation} \label{eq:68} \int\limits_A^C M dy = \int\limits_A^C \bigg{[} \frac{1}{2} \frac{\partial (uz)}{\partial y} dy - z \bigg{(} \frac{\partial u}{\partial y} - au \bigg{)}dy \bigg{]} \end{equation}
\begin{equation} \label{eq:69} \int\limits_A^B N dx = \int\limits_A^B \bigg{[} \frac{1}{2} \frac{\partial (uz)}{\partial x} dx - z \bigg{(} \frac{\partial u}{\partial x} - bu \bigg{)}dx \bigg{]} \end{equation}

\begin{figure}

\centering

\includegraphics{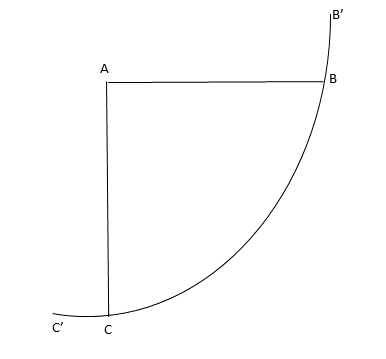}

\caption{}\label{fig:1}

\end{figure}
Then, if we denote by $\varphi_P$ the value of a function $\varphi$ at $P$, the previous equations read as
\begin{equation} \label{eq:70} \int\limits_A^C M dy = \frac{ (uz)_C - (uz)_A}{2} - \int\limits_A^C z \bigg{(} \frac{\partial u}{\partial y} - au \bigg{)} dy \end{equation}
\begin{equation} \label{eq:71} \int\limits_A^B N dx = \frac{ (uz)_B - (uz)_A}{2} - \int\limits_A^B z \bigg{(} \frac{\partial u}{\partial x} - bu \bigg{)} dx \end{equation}
If we insert Eqs. ($\ref{eq:70}$) and ($\ref{eq:71}$) inside Eq. ($\ref{eq:67}$), we have
\begin{equation} \label{eq:72} \begin{split} & (uz)_A =  \frac{(uz)_B + (uz)_C}{2} - \int\limits_B^C (M dy - N dx)\\
& - \int\limits_A^B z \bigg{(} \frac{\partial u}{\partial x} - bu \bigg{)} dx  - \int\limits_A^C z \bigg{(} \frac{\partial u}{\partial y} - au \bigg{)} dy \end{split} \end{equation}

Let us study first the right-hand side of the previous equation. Our aim is to determine, using Riemann's method, the solution $z$ of the partial differential equation proposed, which assumes given values, as well as one of its two derivatives, along all points of the $B'C'$ curve. The equation $ dz= \frac{\partial z}{\partial x} dx + \frac{\partial z}{\partial y}dy $, applied to a path along this curve, clearly determines the two first derivatives which are not given at priori; then we can consider the two derivatives of $z$ as known at each point of the $B'C'$ curve. It follows that, if we choose the solution $u$ of the adjoint equation, the three terms $ (uz)_B$, $(uz)_C$ and $\int\limits_B^C (M dy - N dx) $, are perfectly known and depend only on the boundary conditions imposed  upon $z$. Therefore, we try to evaluate the two integrals on the righ-hand side of Eq. ($\ref{eq:72}$)  with $z_A$ as unknown. These two integrals depend, in general, from the unknown values of $z$ along the straight lines $AB$ and $AC$. To avoid these values, it is necessary that the solution $u$ has to be chosen in such a way that we have
\begin{equation} \label{eq:73} 
\left\{\begin{array} {l}
\frac{\partial u}{\partial x} - bu =0  \; \; {\rm along} \; 
 {\rm every} \; {\rm point} \; {\rm of} \; {\rm AB} \\
\frac{\partial u}{\partial y} - au=0 \; \; {\rm along} \; {\rm every} \; {\rm point} \; {\rm of} \; {\rm AC}
\end{array}\right.
\end{equation}
If these conditions hold, the fundamental equation reads as

\begin{equation} \label{eq:74} (uz)_A = \frac{(uz)_B + (uz)_C}{2} - \int\limits_C^B (N dx - M dy) \end{equation}
and it will give us the value of $z$ for each point of the plane's region $A$, depending only on the boundary conditions. Thus, we have to determine the solution $u$ of the adjoint equation in order to satify the previously stated conditions. To represent $z_A = z(A)= z(\xi,\eta)$, we choose as $u$ a two-point function or $\textit{kernel}$ $R(x,y;\xi, \eta)$, where $x$ and $y$ are the coordinates of any point while $\xi$ and $\eta$ are the coordinates of $A$, subject to the following conditions:

\begin{description} 
\item[(i)] As a functon of $x$ and $y$, $R$ satisfies the homogeneous equation
\begin{equation} \label{eq:75} L^*_{(x,y)}[R]=0. \end{equation}
\item[(ii)] $\frac{\partial R}{\partial x} = bR$ on the segment $AB$ parallel to the $x$-axis, and $\frac{\partial R}{\partial y} = aR$ on the segment $AC$ parallel to the $y$-axis. 

More precisely, one has to write
\begin{equation} \label{eq:76} \frac{\partial}{\partial x} R(x,y; \xi,\eta) - b(x,\eta)R(x,y;\xi,\eta) = 0 \hspace{1mm} \; {\rm on} \; {\rm  y = \eta } \end{equation}
\begin{equation} \label{eq:77} \frac{\partial}{\partial y} R(x,y; \xi,\eta) - a(\xi,y)R(x,y;\xi,\eta) = 0 \hspace{1mm} \; {\rm on} \; {\rm  x = \xi} \end{equation}

\item[(iii)]The kernel $R$ equals 1 at coinciding points, i.e.
\begin{equation} \label{eq:78} R(\xi,\eta;\xi,\eta)=1. \end{equation}
\end{description} 
Note that Eqs.($\ref{eq:76}$) and ($\ref{eq:77}$) reduce to ordinary differential equations for the kernel $R$ along the characteristics. The integration of Eq. ($\ref{eq:76}$) gives us 
\begin{equation} \label{eq:79} R(x,\eta;\xi,\eta)= u_M = u_A \exp \bigg{(}\int\limits_A^M b(\lambda,\eta) d\lambda \bigg{)} \end{equation}
\\
for every point $M$ along $AB$. In the same manner, we can integrate Eq. ($\ref{eq:77}$) and obtain
\begin{equation} \label{eq:80} R(\xi,y;\xi,\eta)= u_M = u_A \exp \bigg{(}\int\limits_A^M a(\lambda, \xi) d\lambda \bigg{)} \end{equation}
\\
for every point $M$ along $AC$. To fix the constant of integration $u_A$ to 1, we exploit Eq. ($\ref{eq:78}$). Therefore, we have

\begin{equation} \label{eq:81} R(x,\eta;\xi,\eta)= u_M = \exp \bigg{(}\int\limits_\xi^x b(\lambda, \eta) d \lambda \bigg{)}, \end{equation}
\begin{equation} \label{eq:82} R(\xi,y;\xi,\eta) = u_M = \exp \bigg{(}\int\limits_\eta^y a(\lambda, \xi) d \lambda \bigg{)}. \end{equation}
The formulae provide the value of $R$ along the characteristics passing through the point $A(\xi,\eta)$. The problem of finding a solution $R$ of Eq. ($\ref{eq:75}$) with data ($\ref{eq:81}$) and ($\ref{eq:82}$) is said to be a $\textit{characteristic initial-value problem}$. Riemann did not actually prove the existence of such a solution, but brought this novel perspective in mathematics, i.e. solving hyperbolic equations by finding kernel functions that obey characteristic initial-value problems. In the case under examination, the desired Riemann's representation formula can be written as
\begin{equation} \begin{split} \label{eq:83}   z_A = &\frac{z_B R(B;\xi,\eta) + z_C R(C; \xi, \eta)}{2} + \int\limits_{B}^C \bigg{(}\bigg{[} bRz + \frac{1}{2} \bigg{(} R \frac{\partial z}{\partial x} - \frac{\partial R}{\partial x} z \bigg{)}\bigg{]}dx \\
& - \bigg{[}aRz + \frac{1}{2} \bigg{(} R \frac{\partial z}{\partial y} - \frac{\partial R}{\partial y} z \bigg{)} \bigg{]}dy \bigg{)} =  \frac{z_B R(B;\xi,\eta) + z_C R(C; \xi, \eta)}{2} \\
&+ \int\limits_{B}^C (N dx - M dy). \end{split} \end{equation}
This is the fundamental result established by Riemann. He found the function $u$, which is the solution of an equation that is in fact $E(\beta,\beta')$ \cite{Friedlander:2010eqa}. Now, we want to make an observation about the previous results.\\
Let us suppose that the curve $BC$ reduces itself to two straight lines parallel to the axes $B'D$ and $DC'$, in fig. ($\ref{fig:2}$), and let $x_1$ and $y_1$ be the coordinates of a point $D$. We have 
\begin{equation} \label{eq:84} \int\limits_C^D (Ndx - M dy) = \int\limits_C^D N dx - \int\limits_D^B M dy. \end{equation}
\begin{figure}

\centering

\includegraphics{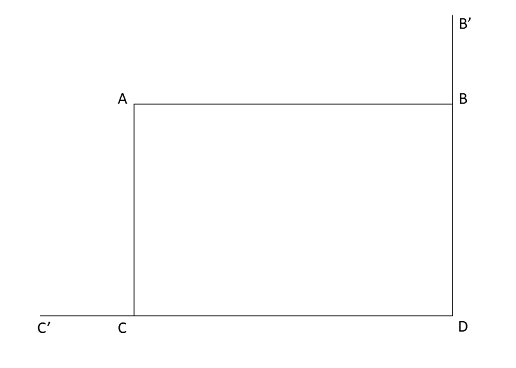}

\caption{}\label{fig:2}

\end{figure}
For this purpose, the right-hand side of Eq.($\ref{eq:84}$) can be replaced by
\begin{equation} \label{eq:85} \int\limits_C^D N dx = \int\limits_C^D \bigg{[}\frac{1}{2} \bigg{(} u \frac{\partial z}{\partial x} - z \frac{\partial u}{\partial x} \bigg{)} + buz \bigg{]} dx = \int\limits_C^D \bigg{[}- \frac{1}{2} \frac{\partial (uz)}{\partial x} + u \bigg{(} \frac{\partial z}{\partial x} + bz \bigg{)} \bigg{]} dx \end{equation} 
and then
\begin{equation} \label{eq:86} \int\limits_C^D N dx = \frac{(uz)_C - (uz)_D}{2} + \int\limits_C^D u \bigg{(} \frac{\partial z}{\partial x} + bz\bigg{)}dx \end{equation}
In the same manner we have
\begin{equation} \label{eq:87} - \int\limits_D^B M dy =  \frac{(uz)_B - (uz)_D}{2} + \int\limits_B^D u \bigg{(} \frac{\partial z}{\partial y} + az\bigg{)}dy \end{equation}
On inserting ($\ref{eq:86}$) and ($\ref{eq:87}$) inside Eq. ($\ref{eq:74}$) and Eq. ($\ref{eq:84}$), we then have
\begin{equation} \label{eq:88} z_A = (uz)_D - \int\limits_C^D u \bigg{(} \frac{\partial z}{\partial x} + bz \bigg{)} dx - \int\limits_B^D u \bigg{(} \frac{\partial z}{\partial y} + az \bigg{)} dy. \end{equation}
This formula holds for every solution $z$ of the proposed equation. It gives the most general analogy with Eq. ($\ref{eq:74}$), but it differs from it by an essential property. We can recognize immediately that it is still not necessary to specify one of the $z$ derivatives on the path $C'DB'$. In order to evaluate the two integrals inside Eq. ($\ref{eq:88}$) it is sufficient to know the values of the solution over the straight lines $C'D$ and $DB'$. We have to find the underlying reason for this result, in this case where the new contour consists of the characteristics of the proposed linear equation.\\
Let us suppose now to take as $z$ a particular solution $z(x,y;x_1,y_1)$ of the proposed equation which can be determined by the same conditions of $u(x,y;x_0,y_0)$ considered as the solution of the adjoint equation. When we pass from the equation to its adjoint, there is a sign change for the coefficients $a$ and $b$ and the solution becomes
\begin{equation} \label{eq:89} z = \exp\bigg{(}-\int\limits_{x_1}^x b d \lambda\bigg{)} \; {\rm for} \hspace{0.5cm} y=y_1; \hspace{1cm} z =  \exp \bigg{(}-\int\limits_{y_1}^y a d \lambda \bigg{)} \; {\rm for} \hspace{0.5cm} x=x_1 \end{equation}
and consequently $z=1$ when $x=x_1$ and $y=y_1$.

Hence we will have
\begin{equation} \label{eq:90} 
\left\{\begin{array} {l}
\frac{\partial z}{\partial x} + bz =0  \; \; {\rm along} \; {\rm every} \; {\rm point} \; {\rm of} \; {\rm CD}, \\
\frac{\partial z}{\partial y} + az=0 \; \; {\rm along} \; {\rm every} \; {\rm point} \; {\rm of} \; {\rm BD},\\
z=1 \; \; {\rm for} \; {\rm the} \; {\rm point} \; {\rm D}.
\end{array}\right. \end{equation}
Then the equation
\begin{equation} \label{eq:91} z_A = (uz)_D - \int\limits_C^D u \bigg{(} \frac{\partial z}{\partial x} + bz \bigg{)} dx - \int\limits_B^D u\bigg{(} \frac{\partial z}{\partial y} + az \bigg{)} dy \end{equation}
reduces to $z_A=u_D$, i.e. $z(x_0,y_0; x_1,y_1)=u(x_1,y_1;x_0,y_0)$. 

The solution $u(x,y;x_0,y_0)$ of the adjoint equation can be considered as a function of the parameters $x_0$, $y_0$; it is a solution of the originary equation, where we have replaced $x$, $y$, with $x_0$, $y_0$, and it has in relation to this equation and the variable $x_0$, $y_0$, the property for which it has been defined as solution of the adjoint equation and a function of $x$, $y$. In other words, the definition of $u$ is still the same if we replace the linear equation with its adjoint, provided we replace the variables $x$, $y$, with $x_0$, $y_0$. It follows that the integration of two linear equations, the proposed equation and its adjoint, is reduced to the determination of the function $u(x,y;x_0,y_0)$, i.e. of $R$. This function can be defined, both as solution of the proposed equation and as a solution of the adjoint equation, by the boundary conditions to which this function is subjected. \\
Let us apply this general proposition to the equation

\begin{equation} \label{eq:92} E(\beta,\beta')=  \frac{\partial^2 z}{\partial x \partial y} - \frac{\beta'}{(x-y)} \frac{\partial z}{\partial x} + \frac{\beta}{(x-y)}\frac{\partial z}{\partial y}=0 \end{equation}
\\
and we try to define the function $u(x,y;x_0,y_0)$ associated with this equation, considered as a solution of the adjoint equation subjected to the previously stated conditions. The adjoint equation to $E(\beta,\beta')$ reads as
\begin{equation} \label{eq:93} \frac{\partial^2 u}{\partial x \partial y} + \frac{\beta'}{(x-y)} \frac{\partial u}{\partial x} - \frac{\beta}{(x-y)} \frac{\partial u}{\partial y} - \frac{\beta + \beta'} {(x-y)^2} u =0. \end{equation}
If we define
\begin{equation} \label{eq:94} u= (x-y)^{\beta + \beta'}\nu, \end{equation}
the equation ($\ref{eq:93}$) becomes
\begin{equation} \label{eq:95} \frac{\partial^2 \nu}{\partial x \partial y} -  \frac{\beta}{(x-y)} \frac{\partial \nu}{\partial x} + \frac{\beta'}{(x-y)} \frac{\partial \nu}{\partial y}=0. \end{equation}
A solution to this equation can be represented as $Z(\beta',\beta)$. We then have
\begin{equation}\label{eq:96} u \equiv (x-y)^{\beta + \beta'} Z(\beta',\beta). \end{equation}
Among the particular solutions $Z$, there exist many general properties that can be derived from the solutions to the homogeneous equation. We have that
\begin{equation} \label{eq:97} x^\lambda F\bigg{(}-\lambda,\beta; 1- \beta - \lambda, \frac{y}{x}\bigg{)} \end{equation}
is a solution of the equation $E(\beta,\beta')$. If we interchange $\beta$ and $\beta'$, the expression $\nu= x^\lambda F(-\lambda, \beta; 1-\beta'-\lambda,\frac{y}{x})$ will be a particular solution of ($\ref{eq:95})$ which contains only a constant $\lambda$; but we can introduce two new ones. We can make on the variables $x$ and $y$ the linear substitution
\begin{equation} \label{eq:98} x \rightarrow \frac{x- y_0}{x - x_0}; \hspace{2cm} y \rightarrow \frac{y- y_0}{y - x_0} \end{equation}
provided we multiply by a factor $(x- x_0)^{-\beta'}(y-y_0)^{-\beta}$. Hence, we obtain the most general formula
\begin{equation} \label{eq:99} \nu = (y_0 - x)^\lambda (x-x_0)^{-\beta' - \lambda} (y-x_0)^{-\beta} F(-\lambda, \beta; 1-\beta'-\lambda,\sigma),\end{equation}
where $\sigma = \frac{(x-x_0) (y-y_0)}{(x-y_0)(y-x_0)}$. It is enough to multiply by $(y-x)^{\beta+\beta'}$ to find $u$ as
\begin{equation}\label{eq:100} u =(y_0 - x)^\lambda (x_0 - x)^{-\beta' - \lambda} (y-x)^{\beta+\beta'}(y-x_0)^{-\beta} F(-\lambda,\beta;1-\beta'-\lambda,x), \end{equation}
that leads to the expected result. If we set $ x=x_0$ in the previous result, $\sigma$ vanishes, the $F$ series reduces to unity, while $(x_0-x)^{-\lambda - \beta'}$ makes $u$ equal to zero or makes it infinite, unless $\lambda= - \beta'$. If $\lambda$ takes this value, $u$ reads as

\begin{equation}\label{eq:101} u= (y_0 - x)^{-\beta'} (y-x)^{\beta+ \beta'} (y-x_0)^{-\beta}F(\beta,\beta';1,\sigma); \end{equation}
if $x=x_0$, we have $u=\bigg{(} \frac{y-x_0}{y_0-x_0 } \bigg{)}^{\beta'}$; whereas if $y=y_0$, we obtain $\sigma=0$ and $u= \bigg{(} \frac{y_0-x}{y_0-x_0 } \bigg{)}^{\beta}$. Thus, Eq. ($\ref{eq:101}$) is the expected solution. In this manner, Riemann's method applied to $E(\beta,\beta')$ makes it possible to determine its integrals with the most general boundary conditions. It will be enough to insert the value of $u$ inside Eq. ($\ref{eq:74}$) or Eq. ($\ref{eq:83}$) to find the general integral of the equation. For example, if we replace $u$ in Eq. ($\ref{eq:83}$), $z$ reads as
\begin{equation} \label{eq:102} z_{x_0,y_0} = (uz)_{x_1,y_1} + \int\limits_{x_0}^{x_1} u_{x,y_1} f(x) dx + \int\limits_{y_0}^{y_1} u_{x_1,y}, \varphi(y) dy \end{equation} 
where $f$ and $\varphi$ are two arbitrary functions which depend on the boundary value of $z$ and $\Phi_{\alpha.\beta}$ represent the result of replacing $x$ and $y$ with $\alpha$ and $\beta$ inside $\Phi(x,y)$.

\section{Proof of the existence of Riemann's Kernel}
There is nothing left to do but to demonstrate that the Poisson and M. Appell formula \cite{darboux1894leccons}
\begin{equation}\begin{split} \label{eq:103} Z(\beta, \beta')=& \int_x^y \varphi(u) (u - x)^{- \beta} (y - u)^{- \beta'} du \\
&+ (y-x)^{1- \beta - \beta'} \int_x^y \varphi(u) (u-x)^{\beta' - 1}(y-u)^{\beta -1} du \end{split}\end{equation}
effectively provides all the integrals of the proposed equation. Following the work of Darboux \cite{darboux1894leccons}, we will first present the following observation on this integral.
The equation ($\ref{eq:92}$) has its coefficients finite and continuous as long as $x$ is different from $y$. If we bisect the angle $yox$, as shown in fig. ($\ref{fig:3}$), we can say that this line is a line of discontinuity for the previous equation, and hence that the coefficients of the equation remain always finite and continuous as long as we remain on the same side of this line. 
\begin{figure}

\centering

\includegraphics{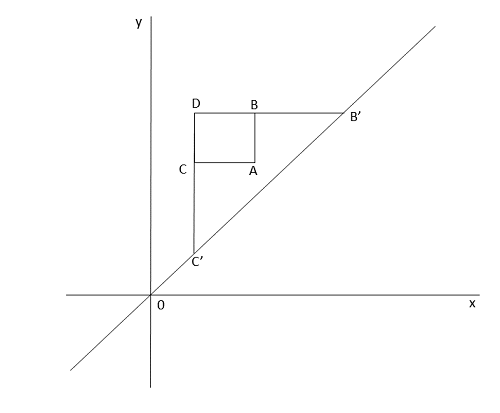}

\caption{}\label{fig:3}

\end{figure}

Let us see what happens to the Poisson integral when $y$ approaches $x$. If we revert to the form

\begin{equation} \label{eq:104} \frac{\partial^2 z}{\partial x \partial y} = \frac{M}{(1-x-y)} \frac{\partial z}{\partial x} + \frac{N}{(1-x-y)} \frac{\partial z}{\partial y} + \frac{P}{(1-x-y)^2}z \end{equation}
the first term on the right-hand side of ($\ref{eq:104}$) has principal part
\begin{equation} \label{eq:105}  (y-x)^{t- \beta-\beta'} \int\limits_0^1 \varphi(x) t^{-\beta}(1-t)^{-\beta'} dt = \frac{ \Gamma(1-\beta) \Gamma(1-\beta')}{\Gamma(2- \beta -\beta')} \varphi(x) (y-x)^{1-\beta-\beta'}. \end{equation}
This approximate value can be seen as the first term of the expansion in powers of $(y-x)$, the unwritten terms are of a higher degree. Likewise, the approximate expression of the second integral on the right-hand side of Eq. ($\ref{eq:104}$) will read as
\begin{equation} \label{eq:106} \psi(x) \int\limits_0^1 t^{\beta'-1} (1-t)^{\beta-1} dt = \frac{ \Gamma(\beta)\Gamma(\beta') }{\Gamma(\beta+\beta')} \psi(x). \end{equation}
It follows that, for any solution provided by the Poisson formula, the expansion according to powers of $(y-x)$ consists of two series of terms, one of integer degree and the other of degree $(1- \beta - \beta')$ increased by an integer; and limiting the expansion to the first term of each of the two series. The integral will have the approximate expression
\begin{equation} \label{eq:107} z = \frac{ \Gamma(1- \beta)\Gamma(1-\beta')}{\Gamma(2-\beta- \beta')} \varphi(x) (y-x)^{1-\beta-\beta'} + \frac{\Gamma(\beta)\Gamma(\beta')}{\Gamma(\beta+\beta')} \psi(x). \end{equation}
This formula will tell us the path that must be followed in order to verify that all solutions of the equation are given by the Poisson formula. We will first try to establish that, in the neighborhood of the discontinuity line, the solution sought is of the form

\begin{equation} \label{eq:108} \varphi_1(x) (y-x)^{1-\beta-\beta'} + \psi_1(x) \end{equation}
\\
The comparison of this form with the previous one will allow us to connect the functions $\varphi$ and $\psi$ that must appear in the Poisson formula; and all that will remain is to verify an equation or nothing will remain unknown.\\
Let us apply this method to the general integral as given by the formula
\begin{equation} \label{eq:109} z_{x_0,y_0} = (uz)_{x_1,y_1} + \int\limits_{x_0}^{x_1}u_{x,y_1} f(x) dx + \int\limits_{y_0}^{y_1} u_{x_1,y} \varphi(y) dy. \end{equation} 
The third term on the right-hand side of Eq. ($\ref{eq:109}$)  will be deduced from the second upon interchanging $x$ with $y$. We can verify that the first two terms on the right-hand side are given by the Poisson formula. To be clearer, let us suppose that we are on the discontinuity line, as in figure ($\ref{fig:3}$); $x_1$ and $y_1$ are the coordinates of $D$, while $x_0$, $y_0$ are those of $A$, and we have $x_1 <x_0<y_0$, where $x_0$ and $y_0$ are the independent variables. The first term of the previous integral is $u(x_1,y_1;x_0,y_0)$ multiplied by the constant $z_{x_1,y_1}$. We must, therefore, first check that the expression $u(x_1,y_1;x_0,y_0)$ considered as a function of the variables $x_0$, $y_0$ verifies the proposed equation and that it is given by the Poisson formula. In accordance with the general method that we are going to indicate, it will therefore be necessary to obtain first its approximation when $(y_0-x_0)$ becomes infinitely small.

If we refer to the expression of $u$
\begin{equation}\label{eq:110} u =(y_0-x)^{-\beta'} (y-x)^{\beta+\beta'}(y-x_0)^{-\beta}F(\beta,\beta';1,\sigma) \end{equation}
and to the definition of $\sigma$, we see that we will have
\begin{equation}\label{eq:111} 1 - \sigma = \frac{(y-x)(y_0-x_0)}{(y-x_0)(y_0-x)} \end{equation}
and then $(1-\sigma)$ is of the same order of $(y_0-x_0)$ and we are led to expand $F(\beta,\beta';1,\sigma)$ according to the powers of $(1-\sigma)$. For this we will borrow the theory of the hypergeometric series:
\begin{equation} \begin{split} \label{eq:112} &F(\beta,\beta',1,\sigma)= \frac{\Gamma(1-\beta-\beta')}{\Gamma(1-\beta)\Gamma(1-\beta')}F(\beta,\beta',\beta+\beta',1- \sigma) \\
&+ \frac{\Gamma(\beta+\beta'-1)}{\Gamma(\beta)\Gamma(\beta')} F(1-\beta,1-\beta',2-\beta-\beta',1-\sigma)(1-\sigma)^{1-\beta-\beta'}. \end{split} \end{equation}
If we bring this value of $F(\beta,\beta',1,\sigma)$ into the formula of $u$, we immediately deduce the approximate expression of $u$, when $y_0$ is approaching $x_0$. It is sufficient to bring the series $F$ back to the unit and we will find therefore for the first two terms of $u$
\begin{equation} \begin{split} \label{eq:113} u =& \frac{\Gamma(1-\beta-\beta')}{\Gamma(1-\beta)\Gamma(1-\beta')} (x_0-x)^{-\beta'} (y-x)^{\beta+\beta'}(y-x_0)^{-\beta} \\
 &+\frac{\Gamma(\beta+\beta'-1)}{\Gamma(\beta)\Gamma(\beta')}(x_0-x)^{\beta-1} (y-x)(y-x_0)^{\beta'-1}(y_0-x_0)^{1-\beta-\beta'}. \end{split}\end{equation}
From the comparison of this formula with the equation for $z$, where we have replaced $x$ and $y$ with $x_0$ and $y_0$, we immediately obtain the two functions that must occur in the Poisson formula. We then find

\begin{equation} \label{eq:114} 
\left\{\begin{array} {l}
\varphi(\alpha)=-A(\alpha - x)^{\beta-1}(y-x)(y-\alpha)^{\beta'-1}, \\
\psi(\alpha)=A(\alpha-x)^{-\beta'}(y-x)^{\beta+\beta'}(y-\alpha)^{-\beta},
\end{array}\right. \end{equation}
\\
where $A$ denotes the constant
\begin{equation} \label{eq:115} A = \frac{\Gamma(1-\beta-\beta')\Gamma(\beta+\beta')}{\Gamma(\beta)\Gamma(1-\beta)\Gamma(\beta')\Gamma(1-\beta')} = \frac{sin(\beta\pi)sin(\beta'\pi)}{\pi sin(\beta+\beta')\pi}. \end{equation}
If we replace these values in the Poisson formula we find the following result:
\begin{equation}\label{eq:116} \begin{split} u =& A(y-x)^{\beta+\beta'}(y_0-x_0)^{1-\beta-\beta'} \times \\
&\int\limits_{x_0}^{y_0} (\alpha-x)^{-\beta'}(y-\alpha)^{-\beta}(y_0-\alpha)^{\beta-1}( \alpha -x_0)^{\beta'-1}d \alpha  \\
&- A(y-x) \int\limits_{x_0}^{y_0}(\alpha - x)^{\beta-1}(y-\alpha)^{\beta'-1}(y_0-\alpha)^{-\beta'}(\alpha-x_0)^{-\beta}d\alpha; \end{split}\end{equation}
and we must at most verify the agreement of this expression with that given by the expression of $z$. We verify this as follows. 

We match the two expressions of $u$ and the equation to be verified will take the form
\begin{equation}\label{eq:117} \begin{split} F(\beta,\beta',1,\sigma)=& A(y_0-x_0)^{1-\beta-\beta'}(y_0-x)^{\beta'}(y-x_0)^{\beta}  \\
& \times \int\limits_{x_0}^{y_0}(\alpha-x)^{-\beta'}(y-\alpha)^{-\beta}(y_0-\alpha)^{\beta-1} (\alpha - x_0)^{\beta'-1}d\alpha \\
& -A(y-x)^{1-\beta-\beta'}(y_0 - x)^{\beta'}(y-x_0)^{\beta} \\
& \times \int\limits_{x_0}^{y_0} (\alpha - x)^{\beta -1} (y-\alpha)^{\beta'-1}(y_0 - \alpha)^{-\beta'}(\alpha- x_0)^{-\beta}d\alpha; \end{split}\end{equation}
and we note that the two terms on the right-hand side remain formally unaffected if we make the same linear substitution on $\alpha$, $x$, $ y$, $x_0$, and $y_0$. We choose the coefficients of this substitution in such a way that $x$, $ x_0$, $y_0$, reduce to infinity, 0 and 1, respectively. Then $y$ will reduce to $\frac{1}{1-\sigma}$, $\sigma$ is the one defined previously as a harmonic ratio. The right-hand side of the equation to verify becomes
\begin{equation} \label{eq:118} \begin{split} & A\int\limits_0^1 [1-\alpha(1-\sigma)]^{-\beta}(1-\alpha)^{\beta-1}\alpha^{\beta'-1}d\alpha \\
& - A(1-\sigma)^{1-\beta-\beta'}\int\limits_0^1[1-\alpha(1-\sigma)]^{\beta'-1}(1-\alpha)^{-\beta'}\alpha^{-\beta}d\alpha. \end{split}\end{equation}
From a well-known formula of Euler, the two previous integrals are expressed through the hypergeometric series and we find the two terms on the right-hand side of the identity $(\ref{eq:112})$ for $F(\beta,\beta',1,\sigma)$. Hence, the equality is verified.\\
Let us now consider the term
\begin{equation}\label{eq:119} \int\limits_{x_0}^{x_1} u_{x,y_1}f(x)dx \end{equation}
of the Riemann integral. We wrote the approximate expression of $u_{x,y_1}$ when $(x_0-y_0)$ approaches zero. If we replace it in the previous integral, we have the same approximate expression of the integral
\begin{equation}\label{eq:120} \begin{split}& \frac{\Gamma(1-\beta-\beta')}{\Gamma(1-\beta)\Gamma(1-\beta')}\int\limits_{x_0}^{x_1} (x_0-x)^{-\beta'}(y_1-x)^{\beta+\beta'}(y_1-x_0)^{-\beta}f(x)dx \\
&+ \frac{\Gamma(\beta+\beta'-1)}{\Gamma(\beta)\Gamma(\beta')}(y_0-x_0)^{1-\beta-\beta'}\int\limits_{x_0}^{x_1}(x_0-x)^{\beta-1}(y_1-x)(y_1-x_0)^{\beta'-1}f(x)dx. \end{split}\end{equation}
The comparison of $(\ref{eq:120})$ with the expression of $z$ gives the two functions that must occur in the Poisson formula. We find therefore
\begin{equation} \label{eq:121} 
\left\{\begin{array} {l}
\varphi(\alpha)=-A\int\limits_\alpha^{x_1}(\alpha - x)^{\beta-1}(y_1-x)(y_1-\alpha)^{\beta'-1}f(x)dx, \\
\psi(\alpha)=A\int\limits_\alpha^{x_1}(\alpha-x)^{-\beta'}(y_1-x)^{\beta+\beta'}(y_1-\alpha)^{-\beta}f(x)dx,
\end{array}\right. \end{equation}
where $A$ is the constant that we have previously defined. It is sufficient to verify that, by introducing these values into the Poisson integral, we find the term ($\ref{eq:119}$). The substitution of the values $(\ref{eq:121})$ gives two terms that are both of the form
\begin{equation}\label{eq:122} \int\limits_{x_0}^{y_0} d\alpha\int\limits_\alpha^{\alpha_1}P dx, \end{equation}
where $x_1<x_0<\alpha<y_0<y_1$. The integration variable $x$, which lies between $\alpha$ and $x_1$, can be either smaller or bigger than $x_0$. We can, therefore, decompose the previous integral
\begin{equation} \label{eq:123} \int\limits_{x_0}^{y_0}d\alpha \int\limits_{\alpha}^{x_0} P dx + \int\limits_{x_0}^{y_0}d\alpha\int\limits_{x_0}^{x_1}P dx. \end{equation}
For the first term, the order of magnitude of the variables will be defined by the inequalities $x_1<x_0<x<\alpha<y_0<y_1$. We can therefore invert the order of integration, that will give us $ - \int\limits_{x_0}^{y_0}dx\int\limits_{x}^{y_0}Pd\alpha$. For the second term, we have $x_1<x<x_0<\alpha<y_0<y_1$ and we can then write $\int\limits_{x_0}^{x_1}dx \int\limits_{x_0}^{y_0}P d\alpha$.
 
If we apply these transformations to the two terms that make up the Poisson integral, we have the following result
\begin{equation}\label{eq:124}\begin{split}& - A \int\limits_{x_0}^{x_1}f(x)(y_1-x)dx\int\limits_{x_0}^{y_0}(y_0-\alpha)^{-\beta'}(\alpha-x_0)^{-\beta}(\alpha-x)^{\beta-1}(y_1-\alpha)^{\beta'-1}d\alpha  \\
&-A \int\limits_{x_0}^{y_0}f(x)(y_1-x)dx\int\limits_{x}^{y_0}(y_0-\alpha)^{-\beta'}(\alpha-x_0)^{-\beta}(\alpha-x)^{\beta-1}(y_1-\alpha)^{\beta'-1}d\alpha  \\
&-A \int\limits_{x_0}^{x_1}f(x)(y_0-x_0)^{1-\beta-\beta'}(y_1-x)^{\beta+\beta'}dx \times \\
& \times \int\limits_{x_0}^{y_0}(y_0-\alpha)^{\beta-1}(\alpha-x_0)^{\beta'-1}(\alpha-x)^{-\beta'} (y_1-\alpha)^{-\beta}d\alpha \\
&-A \int\limits_{x_0}^{y_0}f(x)(y_0-x_0)^{1-\beta-\beta'}(y_1-x)^{\beta-\beta'}dx\int\limits_{x}^{x_0}(y_0-\alpha)^{\beta-1} \times\\
& \times (\alpha-x_0)^{\beta'-1}(\alpha-x)^{-\beta'}(y_1-\alpha)^{\beta}d\alpha \end{split}\end{equation}
The first and third terms in ($\ref{eq:124}$) represent the expression ($\ref{eq:119}$). In order to recognize them, it is enough to refer to the expression of $u$; while as far as the second and fourth terms are concerned, their sum vanishes by virtue of the equation
\begin{equation}\label{eq:125} \begin{split} &(y_1-x)^{1-\beta-\beta'}\int\limits_{x}^{y_0} (y_0-\alpha)^{-\beta'}(\alpha-x_0)^{-\beta}(\alpha - x)^{\beta-1}(y_1-\alpha)^{\beta'-1} d\alpha \\
&= (y_0-x_0)^{1-\beta-\beta'}\int\limits_{x}^{y_0}(y_0-\alpha)^{\beta -1}(\alpha -x_0)^{\beta' -1}(\alpha -x)^{-\beta'}(y_1-\alpha)^{-\beta}d\alpha, \end{split}\end{equation}
that we will verify as follows. We will perform on the variable of the first integral the linear substitution for which $y_0$, $x_0$, $y_1$ and $x$ are turned into $x$, $y_1$, $x_0$ and $y_0$ and will find the second integral. \\
The Poisson formula contains two arbitrary functions $\varphi(\alpha)$ and $\psi(\alpha)$.  Suppose that we know these functions only for $\alpha_0 < \alpha <\alpha_1$; the general integral can be determined only for the values of $x$ and $y$ lying between these values of $\alpha$. Let us assume, to fix the ideas, that $y$ is greater than $x$.  If we construct the $OC'B'$ bisector of the angle formed by the axes and the points $C'$, $B'$, of abscissa $\alpha_0$ and $\alpha_1$, the value of the integral will be known for all the points of the plane included within and on $DB'$ , $DC'$ of the $DB'C'$ triangle; but it will be impossible to determine the solution outside that triangle. We have assumed that the functions $\varphi$ and $\psi$ are determined only for $\alpha_0<\alpha<\alpha_1$. We can extend them beyond this interval in an infinite number of ways, preserving also the continuity of the derivatives up to any order, both for $\alpha=\alpha_0$, and for $\alpha=\alpha_1$. By adopting different extensions, we will have different integrals of the proposed equation that will have the same values within the $DB'C'$ triangle, whose derivatives will be the same up to any order for the points of each of the $DB'$, $DC'$ segments; but that will be different outside the triangle. Thus, an integral of the proposed equation that assumes given values on the segmets $DB'$ and $DC'$, of the $DB'C'$ triangle is well determined for the points located within the triangle. This is evident from the expression 
\begin{equation} \label{eq:126} z_{x_0,y_0} =(uz)_{x_1,y_1} + \int\limits_{x_0}^{x_1}u_{x,y_1}f(x)dx + \int\limits_{y_0}^{y_1}u_{x_1,y} \varphi(y)dy. \end{equation}
Conversely, it is not defined outside the triangle. It can take on the outside of the triangle an infinity of values that we can define in a very general way, making sure to respect the continuity of $z$ and its derivatives up to any order for all the points of $DB'$ and $DC'$. It is interesting that the lines $DB'$ and $DC'$ are $\textit{characteristics}$. The general formula of Riemann shows us in fact that, if on any other curve that is a line parallel to the axes, the function is given as well as its first derivatives, it is determined on both sides of the curve. \\
The study we have made of the Riemann method, in the particular case of the equation $E(\beta,\beta')$, allows us to return to the general theory and to eliminate an objection that can be made to this theory. 

The value of $z$ is given by
\begin{equation} \label{eq:127} (uz)_A= \frac{(uz)_B -(uz)_C}{2} -\int\limits_{C}^{B} (Ndx-Mdy) \end{equation}
and satisfies the partial differential equation 
\begin{equation*}  \frac{\partial^2 z}{\partial x \partial y} + a \frac{\partial z}{\partial x} + b \frac{\partial z}{\partial y} + cz =0. \end{equation*}
It also satisfies the boundary conditions that have been set a priori, but it can be objected that the existence of the function $u$, on which all our reasoning is based and that we have determined in the particular case of $E(\beta,\beta')$, is not established for the more general equations. It is possible to raise this objection at least for the specific case in which the coefficients $a$, $b$ and $c$ of the linear equation are finite and continuous functions, for the consequent series expansion. The function $u$, considered as a solution to the adjoint equation, must reduces, for $y=y_0$, to a function given by $x$, $exp \bigg{(}\int\limits_{x_0}^{x}b dx\bigg{)}$, and, for $x=x_0$, to a function given by $y$, $exp \bigg{(}\int\limits_{y_0}^{y}a dy\bigg{)}$. The functions $a$ and $b$ can be expanded in series according to the powers of $(x-x_0)$ and $(y-y_0)$. Thus it will be enough to admit the existence of the function $u$ to establish the following general proposition:

\begin{prop} Given the linear equation 
\begin{equation} \label{eq:128} \frac{\partial^2 z}{\partial x \partial y} + a \frac{\partial z}{\partial x} + b \frac{\partial z}{\partial y} + cz=0,\end{equation}
where the coefficients $a$, $b$, $c$ can be expanded in series ordered according to the integer and positive powers of $(x-x_0)$ and $(y-y_0)$, there exists a solution to the partial differential equation, which reduces, for $y=y_0$, to a given function $\varphi(x)$ of $x$, expandable in a series according to the powers of $(x-x_0)$ and, for $x=x_0$, to a given function $\psi(y)$ of $y$, expandable in a series according to the powers of $(y-y_0)$. \end{prop}

To prove this proposition, we perform the substitution $ x \rightarrow x_0 + \frac{x}{\rho}$, $y \rightarrow y_0 + \frac{y}{\sigma}$, where $\rho$ and $\sigma$ are  two constants that we will choose in such a way that the expansions of the functions $a$, $b$, $c$, $\varphi(x)$ and $\psi(y)$, that are ordered, by substitution, according to the powers of $x$ and $y$, are convergent for all the values of those variables whose modulus is less than or equal to one. Hence, we plan to determine all the derivatives of the function $z$ for $x=y=0$. 

Since $z$ must be reduced to $\varphi(x)$ for $y = 0$, this condition will determine all derivatives of $z$ in relation to the single variable $x$; in the same way, since $z$ must be reduced to $\psi(y)$ for $x = 0$, we will know all derivatives in relation to the single variable $y$; eventually, the partial differential equation will allow us to know all derivatives, depending on the previous ones, in relation to $x$ and $y$. One can, with all these derivatives, form the series expansion of the solution sought according to the powers of $x$ and $y$, and all is reduced to determining whether this expansion is convergent; because, in the affirmative case, it will satisfy both the boundary conditions and the proposed equation. 

Now, if the series for the functions $a$, $b$, $c$, $\varphi$ and $\psi$ converge in a circle of radius 1, it is always possible to find the constants $M$, $N$, $P$ and $H$ which are positive and such that the derivatives of any order of $a$, $b$, $c$, $\varphi$ and $\psi$ have modules smaller than the derivatives of the corresponding functions 
\begin{equation} \label{eq:129} \frac{M}{(1-x-y)}, \hspace{0.5cm}  \frac{N}{(1-x-y)}, \hspace{0.5cm}  \frac{P}{(1-x-y)^2}, \hspace{0.5cm}  \frac{H}{(1-x)}, \hspace{0.5cm}  \frac{H}{(1-y)}. \end{equation}
In fact, if we aim at determining the function that satisfies the equation

\begin{equation}\label{eq:130} \frac{\partial^2 z}{\partial x \partial y} =  \frac{M}{(1-x-y)}\frac{\partial z}{\partial x} + \frac{N}{(1-x-y)}\frac{\partial z}{\partial y} + \frac{P}{(1-x-y)^2}z \end{equation}
\\
which reduces to $\frac{H}{(1-x)}$, for $y=0$, and to $\frac{H}{(1-y)}$, for $x=0$. We will obtain for this function a series whose coefficients will be bigger than those related to the given equation. It will be sufficient to show that this new series is convergent for the values of $x$ and $y$ sufficiently close to zero.

The new problem to which we have arrived is already solved, in fact, if we replace in the equation $x$ with $(1-x)$, it assumes the same form of the equation $\varphi(x,y,y',...,y^n)$ and therefore it can be reduced to $E(\beta,\beta')$, for which the problem is solved. The result leads to a function that is actually expandable in series. It is therefore possible to determine a solution of the partial differential equation proposed with the boundary conditions that we have indicated and to establish the general theorem upon which relies the existence of the function $u$.

\chapter{Fundamental Solutions}
\epigraph{Natural science is the attempt to comprehend nature by precise concepts.}{Bernhard Riemann}

\section{Wavelike Propagation for a Generic Normal System}
Let us consider the two systems
\begin{equation} \label{eq:2.1} E_{\mu} = \sum\limits_{\nu =1}^{m} \sum\limits_{i=0}^{n} E^{i}_{\mu \nu} \frac{\partial \varphi_{\nu}}{\partial x^i} + \Phi_{\mu}(x|\varphi) =0, \hspace{2cm} \mu=1,2,...,m, \end{equation}
\begin{equation} \label{eq:2.2} E_{\mu}= \sum\limits_{\nu=1}^{m}\sum\limits_{i,k=0}^{n} E^{ik}_{\mu \nu} \frac{\partial^2 \varphi_\nu}{\partial x^i \partial x^k} + \Phi_{\mu}(x|\varphi|\chi) =0, \hspace{1cm} \mu=1,2,...,m, \end{equation}
following Levi-Civita \cite{levi1988caratteristiche}, we assume that, inside and outside the stratum determined by two hypersurfaces of equations
\begin{equation} \label{eq:2.3} z=c_1, \hspace{1cm} z=c_2, \end{equation}
they are satisfied by the $m$ functions $\varphi_1$, $\varphi_2$, ..., $\varphi_m$ and $\varphi^*_1$, $\varphi^*_2$, ..., $\varphi^*_m$, respectively. We assume that the stratum determined by Eq. ($\ref{eq:2.3}$) undergoes motion and possibly also bendings, and that through the hypersurfaces ($\ref{eq:2.3}$) the partial derivatives of first order for ($\ref{eq:2.1}$) and of second order for ($\ref{eq:2.2}$) undergo sudden variations (or jumps) and are therefore discontinuous therein.\\
The solutions $\varphi$ of Eq. ($\ref{eq:2.1}$) are taken to be continuous through the hypersurfaces ($\ref{eq:2.3}$), while the solutions $\varphi^*$ of Eq. ($\ref{eq:2.2}$) are taken to be continuous together with their first derivatives through the confining hypersurfaces. This describes a wavelike phenomenon, where the wave surfaces are those bounding the stratum. \\
For a system of maximal order $s$, the functions $\varphi$ and $\varphi^*$ should obey matching conditions through the wave surfaces of order less than $s$, whereas some discontinuities occur for the derivatives of order $s$. The wave surfaces turn out to be $\textit{characteristic manifolds}$, because out of them it is not possible to apply the theorem that guarantees uniqueness of the integrals.\\
Hereafter we merely assume the existence of the functions $\varphi$ and $\varphi^*$ with the associated wavelike propagation, and we describe some of their properties. If $z=c$ is a wave surface $\sigma$, the function $z$ must satisfy the equation
\begin{equation} \label{eq:2.4} \Omega(x|p)=0, \end{equation}
where the $p$ variables are given by
\begin{equation} \label{eq:2.5} p_i = \frac{\partial z}{\partial x^i}, \hspace{3cm} i=0,1,...,n. \end{equation}
The validity of Eq. ($\ref{eq:2.4}$) is indeed established only on $\sigma$, i.e. for $z=c$. However, the limitation $z=c$ is inessential, because $\Omega$ certainly vanishes whenever the $p_i$ are set equal to the derivatives of the function $z$. One therefore deals, with respect to $z$, with a partial differential equation. Such an equation can characterize $z$ by itself provided that the $E_\mu$ functions occurring in such systems depend only on the $x$ variables.\\
Now we aim at studying the velocity of progression of the wave surface $\sigma$ at a point $P$, by assuming that the space of variables $x^1$, $x^2$, ..., $x^n$ is endowed with an Euclidean metric, and that such variables are Cartesian coordinates. 

We suppose that
\begin{equation} \label{eq:2.6} z(t|x)=c, \hspace{1cm} z(t+dt|x)=c \end{equation}
are the equations of $\sigma$ at the instants of time $t$ and $t+dt$, respectively. The normal $N$ to $P$ at $\sigma$ intersects the second of Eq. ($\ref{eq:2.6}$) at a point $Q$. If $dN$ is the measure, with sign, of the segment $PQ$, counted positively towards the exterior of the stratum determined by $\sigma$ and  by the other wave surface pertaining to the instant $t+dt$, the ratio
\begin{equation} \label{eq:2.7} a \equiv \frac{dN}{dt} \end{equation}
is said to be the $\textit{progression velocity}$ of the wave surface at the point $P$ at the instant of time under consideration. \\
The directional cosines of the normal $N$ to $\sigma$ at $P$ are given by
\begin{equation} \label{eq:2.8} \alpha_i = \frac{p_i}{|\rho|}, \hspace{1cm} i=1,2,...,n, \end{equation}
where
\begin{equation} \label{eq:2.9} \rho^2\equiv \sum\limits_{i,j=1}^{n} \delta^{ij}p_ip_j=\sum\limits_{i=1}^n p_ip^i. \end{equation}
If the points $P$ and $Q$ have coordinates $x^i$ and $x^i+dx^i$, respectively, one has from ($\ref{eq:2.6}$)
\begin{equation} \label{eq:2.10} z(t|x)=c, \hspace{1cm} z(t+dt|x+dx)=c, \end{equation}
\\
and hence, by taking the difference,
\begin{equation} \label{eq:2.11}  dz=p_0dt + \sum\limits_{i=1}^n p_i dx^i=0. \end{equation}
Since the $dx^i$ are the components of the vector $Q-P$, one has also
\begin{equation} \label{eq:2.12} dx^i= \pm \sum\limits_{j=1}^n \delta^{ij} \alpha_j dN = \pm \alpha^i dN, \hspace{1cm} i=1,2,...,n. \end{equation}
The sign is $\pm$ depending on whether $z$ is positive or negative outside of the stratum. We do not need to fix it. By virtue of ($\ref{eq:2.8}$) and ($\ref{eq:2.12}$), Eq. ($\ref{eq:2.11}$) leads to (we set $\epsilon \equiv \pm 1$)

\begin{equation} \label{eq:2.13} \begin{split}& p_0dt + \sum\limits_{i=1}^n p_i \epsilon \alpha^i dN = p_0dt + \epsilon \sum\limits_{i=1}^{n} \frac{p_ip^i}{|\rho|}dN \\
&= p_0dt + \epsilon |\rho|dN=0, \end{split} \end{equation}
\\
from which

\begin{equation} \label{eq:2.14} |a|= \bigg{|} \frac{dN}{dt} \bigg{|} = \bigg{|} \frac{p_0}{\rho}\bigg{|}. \end{equation}
\\
This is the desired formula for the modulus of the velocity of progression. As the point $P$, the time parameter $t$ and the wave surface are varying. 

\section{Cauchy's Method for Integrating a First-Order Equation}
We have seen that the characteristic manifolds
\begin{equation} \label{eq:2.15} z(x^0,x^1,...,x^n)= {\rm const}. \end{equation}
of a normal system of equations in the $n+1$ independent variables $x^0, x^1$, ..., $x^n$ ensure the vanishing of a certain determinant
\begin{equation} \label{eq:2.16} \Omega(x|p)=0, \end{equation}
where the $p_i$ are obtained as
\begin{equation} \label{eq:2.17} p_i \equiv \frac{\partial z}{\partial x^i}, \hspace{2cm} i=0,1,2,...,n. \end{equation}
In the most general case, $\Omega$ depends not only on the $x$ and $p$, but also on the unknown functions $\varphi$ of the normal system under consideration. There exists however a particular set of normal systems, of order $s=1$ and $s=2$, where $\Omega$ depends only on $x$ and $p$ variables, provided that the coefficients $E^i_{\mu \nu}$ in Eq. ($\ref{eq:2.1} $) and $E^{ik}_{\mu \nu}$ in Eq. ($\ref{eq:2.2}$) depend only on the $x$ variables.
We are going to describe the Cauchy method for integrating a first-order partial differential equation, considering, in particular, Eq. ($\ref{eq:2.16}$), where the unknown function $z$ does not occur explicitly. We are therefore guaranteed that $\Omega$ contains at least one of the $p$ functions, e.g. $p_0$. If Eq. ($\ref{eq:2.16}$) can be solved with respect to $p_0$, one can write
\begin{equation} \label{eq:2.18} p_0 + H(t,x^1,...,x^n|p_1,...,p_n)=0. \end{equation}
Let us study first the linear case, i.e. when $H$ is a linear function of the $p$ variables. We are going to show that the task of integrating Eq. ($\ref{eq:2.18}$) is turned into the integration of a system of ordinary differential equations. Indeed, Eq. ($\ref{eq:2.18}$) is then of the type
\begin{equation} \label{eq:2.19} p_0+A_0+ \sum\limits_{i=1}^n A^ip_i=0, \end{equation}
where the $A$'s depend only on the variables $t$, $x^1$, ..., $x^n$. Let us consider the space $S_{n+2}$ of the $(n+2)$ variables $(t,x^1,...,x^n,z)$ and an integral hypersurface

\begin{equation} \label{eq:2.20} z= \varphi(t|x) \end{equation}
\\
of Eq. ($\ref{eq:2.19}$), that we shall denote by $\sigma$. Let $\Gamma$ be the section of $\sigma$ with the hyperplane $t=0$, i.e. the locus of points defined by the equation
\begin{equation} \label{eq:2.21} \Gamma: \hspace{1cm} z= \varphi(0|x)= \varphi_0(x). \end{equation}
The fundamental guiding principle adopted at this stage consists in viewing $\sigma$ as the locus of $\infty^n$ curves obtainable by integration of a suitable system of ordinary differential equations of the kind
\begin{equation} \label{eq:2.22} \frac{d}{dt} x^i = X^i(t|x), \hspace{2cm} i=1,2,...,n, \end{equation}
\begin{equation} \label{eq:2.23} \frac{d}{dt} z=Z(t|x), \end{equation}
\\
of rank $(n+1)$ in the unknown functions $x^1$, ..., $x^n$, $z$ of the variable $t$. Such a system involves $(n+1)$ arbitrary constants, but their number is reduced by 1 if one requires compatibility of the system with Eq. ($\ref{eq:2.20}$) for the integral surface $\sigma$. 

The basic assumption, which justifies the interest in the system ($\ref{eq:2.22}$) and ($\ref{eq:2.23}$), is that it is independent of the preliminary integration of Eq. ($\ref{eq:2.19}$). Once we have made this statement, we must express the condition that any integral curve of Eqs. ($\ref{eq:2.22}$) and ($\ref{eq:2.23}$) belongs to $\sigma$.

Upon viewing $z$ as a function of $t$ and $x$, Eqs. ($\ref{eq:2.22}$) and ($\ref{eq:2.23}$) lead to
\begin{equation} \label{eq:2.24} \frac{dz}{dt}= Z = p_0 + \sum\limits_{i=1}^n p_i \frac{dx^i}{dt} = p_0 + \sum\limits_{i=1}^n p_i X^i, \end{equation}
and, bearing in mind Eq. ($\ref{eq:2.19}$) to re-express $p_0$, one obtains
\begin{equation} \label{eq:2.25} Z= -A_0 + \sum\limits_{i=1}^n p_i (X^i-A^i). \end{equation}
Since we want to make sure that the differential system ($\ref{eq:2.22}$) and ($\ref{eq:2.23}$) is independent of the integration of Eq. ($\ref{eq:2.19}$), the coefficients of the $p_i$ must vanish, and hence
\begin{equation}\label{eq:2.26} X^i=A^i, \end{equation}
\\
from which if follows that
\\
\begin{equation}\label{eq:2.27} Z=-A_0. \end{equation}
\\
The desired differential system reads therefore
\begin{equation} \label{eq:2.28} \frac{dx^i}{dt} = A^i, \hspace{1cm} i = 1,2,...,n, \end{equation}
\begin{equation} \label{eq:2.29} \frac{dz}{dt}=-A_0, \end{equation}
\\
or also, with the notation used until the end of nineteenth century,
\\
\begin{equation}\label{eq:2.30} \frac{dx^i}{A^1} = \frac{dx^2}{A^2}=...=\frac{dx^n}{A^n} = - \frac{dz}{A_0}=dt, \end{equation}
\\
which is capable to determine the integral hypersurfaces $\sigma$ of Eq. ($\ref{eq:2.19}$). 

Indeed, in order to solve the Cauchy problem relative to a pre-assigned $\Gamma$ of the hyperplane $t=0$, it is enough to consider, in the first place, the whole set of integral curves, which are $\infty^n$, of the system of Eq. ($\ref{eq:2.28}$), in which the function $z$ does not occur. The integration of the residual differential equation ($\ref{eq:2.29}$), which is performed by a simple quadrature, once the system ($\ref{eq:2.28}$) has been integrated, completes the determination of the curves of the space $S_{n+2}$ (of the $t$, $x$ and $z$ variables) which are integral curves of the system ($\ref{eq:2.28}$) and ($\ref{eq:2.29}$). If one wants that these curves emanate from the points of $\Gamma$, it is necessary and sufficient that $z$ takes the value $\varphi_0(x)$ at $t=0$, the $x$ referring to the same zero value of $t$ and being therefore identifiable with the $n$ arbitrary constants introduced from the integration of the system ($\ref{eq:2.28}$). Thus, the total number of arbitrary constants is $n$, and every integral hypersurface $\sigma$ of Eq. ($\ref{eq:2.19}$) occurs as the locus of $\infty^n$ integral curves of Eq. ($\ref{eq:2.28}$) and Eq. ($\ref{eq:2.29}$), emanating from the points of $\Gamma$.

The concept of transforming the problem of the integration of a linear partial differential equation of first order into the problem of integrating a system of ordinary differential equations, originally developed by Lagrange, was generalized by Lagrange himself, Charpit, Cauchy and Jacobi to non-linear equations.  Hereafter, following Levi-Civita \cite{levi1988caratteristiche}, we describe the Cauchy method. 

For this purpose, let us revert to the general equation
\begin{equation} \label{eq:2.31} p_0+ H(t,x^1,...,x^n|p_1,...,p_n)=0, \end{equation}
and let us try to understand whether it is possible to determine a generic integral hypersurface (the one whose existence is guaranteed by virtue of the Cauchy theorem for given initial data) as a locus of integral curves of a suitable differential system. 

One can easily recognize that it is no longer possible, in general, to associate with Eq. ($\ref{eq:2.31}$) a congruence of curves of the space $S_{n+2}$ that holds for whichever integral hypersurface, but it is necessary to pass to an auxiliary higher-dimensional space. It will be useful to regard as arguments, besides the $x$ coordinates of a generic point $P$ of the integral hypersurface $\sigma$, also the $p_0$, $p_1$, ..., $p_n$ which, geometrically, define a facet for $P$. In order to give a concrete metrical meaning to such $p$ variables, we may regard $t$, the $x$ and $z$ as Cartesian coordinates of the space $S_{n+2}$. The variables $p_0$, $p_1$, ...., $p_n$, -1 are then proportional to the directional cosines of the normal to $\sigma$, with reference to the axes $t$, $x^1$, ..., $x^n$, $z$ respectively. 

Having made this choice, let us try to associate with Eq. ($\ref{eq:2.31}$) a differential system of the kind

\begin{equation} \label{eq:2.32} \frac{d}{dt}x^i = X^i (t,x|p), \hspace{0.5cm} \frac{d}{dt}p_i = P_i(t,x|p), \hspace{0.5cm} i=1,2,...,n, \end{equation}
\begin{equation} \label{eq:2.33} \frac{d}{dt} z= Z(t,x|p). \end{equation}
\\
Once the expressions of the $X^i$ have been determined in terms of the $t$, $x$, $p$ variables, one finds also the form of $Z$. Indeed, since $z$ is a function of $t$ by means of $x^0=t$ and of $x^1$, ..., $x^n$, one has

\begin{equation} \label{eq:2.34} \frac{dz}{dt} = p_0 + \sum\limits_{i=1}^n p_i \frac{dx^i}{dt}, \end{equation}
\\
and, by virtue of the first of Eq. ($\ref{eq:2.32}$),

\begin{equation} \label{eq:2.35} \frac{dz}{dt} = Z(t,x|p)=p_0 + \sum\limits_{i=1}^n p_i X^i. \end{equation}
Note that Eq. ($\ref{eq:2.33}$), with $Z$ given by Eq. ($\ref{eq:2.35}$), should be associated after having integrated the system Eq. ($\ref{eq:2.32}$), because then $z$ can be expressed in terms of $t$, by means of a quadrature. 

Hereafter we denote by $\Gamma$ a hypersurface in the hyperplane $t=0$, $M_0$ a point of $\Gamma$, $ \bar \omega_0$ the hyperplane tangent at $M_0$ to the integral hypersurface $\sigma$ of Eq. ($\ref{eq:2.32}$) that is passing through $\Gamma$. We aim at expressing the condition for the integral curve $C_0$ of the system ($\ref{eq:2.32}$) and ($\ref{eq:2.33}$), that emanates from $M_0$ and is tangent to $\bar \omega_0$, to belong to the integral hypersurface $\sigma$, while still fulfilling the equations

\begin{equation} \label{eq:2.36} p_i=\frac{\partial z}{\partial x^i}, \hspace{1cm} i=0,1, ..., n; \hspace{1cm} x^0=t, \end{equation}
\\
and this for whatever hypersurface $\Gamma$ passing through the point $M_0$. 

On passing from $t$ to $t+dt$, $p_i$ undergoes an infinitesimal change
\begin{equation}\label{eq:2.37} dp_i=P_i dt, \end{equation}
and on the other hand, for the  Eq. ($\ref{eq:2.36}$) to remain valid, one requires that

\begin{equation} \label{eq:2.38} dp_i = \sum\limits_{j=0}^n p_{ij}dx^j, \; \; \; \; i=0, 1, ..., n, \end{equation}
having defined

\begin{equation} \label{eq:2.39} p_{ij} \equiv \frac{\partial^2z}{\partial x^i \partial x^j}=p_{ji}, \; \; i,j=0,1,...,n. \end{equation}
\\
The formulae ($\ref{eq:2.37}$) and ($\ref{eq:2.38}$) for $dp_i$ should agree. Note that the quantities $p_{ij}$ when both indices are positive are arbitrary, because of the choice, arbitrary by hypothesis, of the hypersurface $\Gamma$. The components $p_{0i}$ satisfy instead relations that can be obtained by differentiation of Eq. ($\ref{eq:2.31}$). In other words, one has the $(n+1)$ equations

\begin{equation} \label{eq:2.40} p_{0i} + \sum\limits_{j=1}^n \frac{\partial H}{\partial p_j} p_{ji} + \frac{\partial H}{\partial x^i} =0, \; \; i=0,1,...,n. \end{equation}
Since the full number of $p_{ij}$ components is $\frac{1}{2}(n+1)(n+2)$, we are left with
\begin{equation*} \frac{1}{2} (n+1)(n+2) - (n+1) =\frac{1}{2} n(n+1) \end{equation*}
free components, while we have at our disposal the $2n$ quantities $x^1$, $x^2$, ..., $x^n$; $p_1$, $p_2$, ..., $p_n$. It would therefore seem impossible to determine $P_i$ in such a way that

\begin{equation} \label{eq:2.41} P_i dt= \sum\limits_{j=0}^n p_{ij}dx^j, \end{equation}
independently of the $p_{ij}$.

However, Cauchy's idea works because, by virtue of $p_i= \frac{\partial z}{\partial x^i}$, one finds, by differentiation with respect to $t$, 

\begin{equation} \label{eq:2.42} \frac{dp_i}{dt} = P_i = p_{i0} + \sum\limits_{j=1}^n p_{ij} \frac{dx^j}{dt} = p_{i0} + \sum\limits_{j=1}^n p_{ij}X^j. \end{equation}
Now we eliminate the $p_{i0}=p_{0i}$ by means of Eq. ($\ref{eq:2.40}$) and we exploit the symmetry of the $p_{ij}$. Hence we find that Eq. ($\ref{eq:2.42}$) is equivalent to

\begin{equation} \label{eq:2.43} P_i= - \frac{\partial H}{\partial x^i} + \sum\limits_{j=1}^n \bigg{(} X^j - \frac{\partial H}{\partial p_j} \bigg{)}p_{ij}, \hspace{1cm} i=1, 2, ..., n.\end{equation}
Such equations are satisfied independently of the $p_{ij}$ values provided that, for all $i$ ranging from 1 through $n$, the following equations hold:

\begin{equation} \label{eq:2.44} X^i= \frac{\partial H}{\partial p_i}, \end{equation}
\begin{equation} \label{eq:2.45} P_i=-\frac{\partial H}{\partial x^i}. \end{equation}

\section{The Bicharacteristics}
We have just shown that if, starting from a generic point $M_0$ of the integral hypersurface $\sigma$, one assigns to the $t$, $x$, $p$, $z$ variables some increments which obey the differential system ($\ref{eq:2.32}$) and ($\ref{eq:2.33}$), which is uniquely characterized in the form
\begin{equation} \label{eq:2.46} \frac{d}{dt} x^i = \frac{\partial H}{\partial p_i}, \hspace{1cm} \frac{d}{dt}p_i=- \frac{\partial H}{\partial x^i}, \hspace{1cm}  i=1, 2, ..., n, \end{equation}
\begin{equation} \label{eq:2.47} \frac{d}{dt}z= \sum\limits_{i=1}^n p_i \frac{\partial H}{\partial p_i} - H,  \end{equation}
one reaches an infinitely close point $M_1$ which belongs again to $\sigma$ and for which the $p_i + dp_i$ define the direction of the normal to the hypersurface $\sigma$ itself. 

The same considerations may be certainly repeated starting from the point $M_1$, and this for the essential reason that the system ($\ref{eq:2.32}$), ($\ref{eq:2.33}$) and hence ($\ref{eq:2.46}$), ($\ref{eq:2.47}$) has been built in such a way that it holds for all integral hypersurfaces $\sigma$ passing through $M_0$ with given orientation of the normal, i.e. with given values of $p$. As far as the integral hypersurface $\sigma$ is concerned, we are therefore at $M_1$ in the same conditions in which we found ourselves at $M_0$. Hence the whole curve $C$, defined uniquely from Eqs. ($\ref{eq:2.46}$) and ($\ref{eq:2.47}$) under the condition that the $x$, $p$, $z$ take at $t=0$ the values corresponding to $M_0$, belongs to the integral hypersurface under consideration, which is an integral hypersurface whatsoever among the many passing through $M_0$ and having therein the $\bar \omega_0$ as tangent hyperplane. 

Thus, we discover the geometric corollary according to which, if two integral manifolds meet each other at a point, they meet each other along the whole curve $C$ passing through that point. 

The curves $C$ are called $\textit{bicharacteristics}$ by Hadamard, whereas we call $\textit{characteristics}$ (of the space $S$ of the $t$, $x$ variables) the hypersurfaces having exceptional behaviour with respect to the Cauchy problem.

\section{Fundamental Solution and its relation to Riemann's Kernel}
Following the work of Hadamard, Chap. 3 of his famous book \cite{hadamard1952lectures}, he studies the fundamental solutions of partial differential operators, starting from the familiar form of the fundamental solution

\begin{equation} \label{eq:2.48} \mathcal{U}\; log \frac{1}{r} + \omega, \; r \equiv \sqrt{(x-x_0)^2 + (y-y_0)^2} \end{equation}
for the equation

\begin{equation} \label{eq:2.49} \bigg{(} \frac{\partial^2}{\partial x^2} + \frac{\partial^2}{\partial y^2} + C(x,y) \bigg{)} u=0. \end{equation}
\\
In the formula ($\ref{eq:2.48}$), $ \mathcal{U}$ and $\omega$ are properly chosen functions of $(x, y; x_0, y_0)$ which are regular in the neighbourhood of $x=x_0$, $y=y_0$. The function $\omega$ remains arbitrary to some extent, because any regular solution of Eq. ($\ref{eq:2.49}$) might be added to it. 

As a next step, Hadamard went on to consider the more general equation
\begin{equation} \label{eq:2.50} \bigg{(} \frac{\partial^2}{\partial x^2} + \frac{\partial^2}{\partial y^2} + D(x,y)\frac{\partial}{\partial x} + H(x,y)\frac{\partial}{\partial y} + E(x,y) \bigg{)} u=0, \end{equation}
where $D$, $H$, $E$ are taken to be analytic functions. In this analytic case, there is no essential distinction between Eq. ($\ref{eq:2.50}$) and the equation
\begin{equation} \label{eq:2.51} \mathcal{F}(u) = \bigg{(}\frac{\partial^2}{\partial x \partial y} + A(x,y) \frac{\partial}{\partial x} + B(x,y) \frac{\partial}{\partial y} + C(x,y) \bigg{)} u =0, \end{equation}
which can be obtained from Eq. ($\ref{eq:2.50}$) by changing $(x+iy) \rightarrow x$, $(x-iy) \rightarrow y$. This map has the effect of changing $r^2$ in Eq. ($\ref{eq:2.48}$) into $(x-x_0)(y-y_0)$. Thus, Hadamard wrote the fundamental solution of Eq. ($\ref{eq:2.51}$) in the form

\begin{equation} \label{eq:2.52} u= \mathcal{U} \; log \bigg{[} (x-x_0)(y-y_0) \bigg{]} + \omega. \end{equation}
For this to be a solution of Eq. ($\ref{eq:2.51}$) at all $x \neq x_0$, $y \neq y_0$, we have to require that
\begin{equation} \label{eq:2.53} \mathcal{F} \bigg{[} \mathcal{U} \; log(x-x_0)(y-y_0) \bigg{]} = \mathcal{M}, \end{equation}
where $\mathcal{M}$ is a regular function, while for $\omega$ we can take any regular solution of the equation
\begin{equation} \label{eq:2.54} \mathcal{F}(u) = - \mathcal{M}, \end{equation}
Indeed, by virtue of the definition ($\ref{eq:2.51}$) of the operator $\mathcal{F}$, one finds
\begin{equation} \label{eq:2.55} \begin{split} &\mathcal{F} \Biggl\{ \mathcal{U}\; log[(x-x_0)(y-y_0)] \Biggr\} = \mathcal{F}(\mathcal{U}) log[(x-x_0)(y-y_0)] \\ 
&+ \frac{1}{(x-x_0)} \bigg{(} \frac{\partial }{\partial y} + A(x,y) \bigg{)} \mathcal{U} + \frac{1}{(y-y_0)} \bigg{(} \frac{\partial}{\partial x} + B(x,y) \bigg{)} \mathcal{U}. \end{split} \end{equation}
This is found to be a regular function of $x$, $y$ near each of the lines $x=x_0$, $y=y_0$ if and only if the following conditions hold:

\begin{description} 
\item[(i)]The logarithmic term vanishes, so that $\mathcal{U}$ itself is a solution of Eq. ($\ref{eq:2.51}$).
\item[(ii)]The numerators of the two fractions on the second line of ($\ref{eq:2.55}$) vanish at same time as the denominators, i.e.
\begin{equation} \label{eq:2.56} \bigg{(} \frac{\partial}{\partial y} + A \bigg{)} \mathcal{U} =0 \; \; {\rm for}\; x=x_0, \end{equation}
\begin{equation} \label{eq:2.57} \bigg{(} \frac{\partial}{\partial x} + B \bigg{)} \mathcal{U} =0 \; \; {\rm for}\; y=y_0.\end{equation}
\end{description} 
Note now that these conditions, together with
\begin{equation} \label{eq:2.58} \mathcal{U}=1 \; \; {\rm at} \; x=x_0, y=y_0, \end{equation}
are precisely the conditions for the Riemann kernel. Thus, we have just proved that Riemann's kernel coincides with the coefficient of the logarithmic term in the fundamental solution of Eq. ($\ref{eq:2.51}$).

\section{The concept of Characteristic Conoid}
In general, the fundamental solution is singular not only at a point, e.g. the pole, but along a certain surface. What this surface must be was the content of an important theorem of Le Roux \cite{le1895integrales} and Delassus \cite{delassus1895equations, delassus1896equations}, who proved that $\textit{a\-ny \- sin\-gu\-lar \- sur\-fa\-ce \- of \-a \- solu\-tion \- of \-a lin\-ear dif\-fe\-ren\-tial e\-qua\-tion} $ \\
$\textit{ \- mu\-st \- be \- cha\-rac\-teris\-tic}$. Such a singular surface must therefore satisfy the first-order differential equation

\begin{equation} \label{eq:2.59} \Omega \bigg{(} \frac{\partial z}{\partial x^1}, \frac{\partial z}{\partial x^2}, ..., \frac{\partial z}{\partial x^m}; x^1, x^2, ..., x^m \bigg{)} = 0. \end{equation}
\\
Among the solutions of Eq. ($\ref{eq:2.59}$), one which was especially considered by Darboux \cite{darboux1883memoire} is the one which has a given point $a(a^1, a^2, ..., a^m)$ as a conic point, which is called, since Hadamard, the $\textit{characteristic conoid}$. It coincides with the characteristic cone itself when the coefficients of the equation, or at least the coefficients of the terms of second order, are constants. In general, however, it is a kind of cone with curved generatrices. A more precise definition of the $\textit{characteristic conoid}$ can be given if we introduce some basic concepts of pseudo-Riemannian geometry. 

A space-time $(M,g)$ is the following collection of mathematical entities:
\begin{description} 
\item[(i)] A connected, four-dimensional, Hausdorff (distinct points belong always to disjoint open neighbourhoods) $C^{\infty}$ manifold $M$;
\item[(ii)] A Lorentz metric $g$ on $M$, i.e. the assignment of a non-degenerate bilinear form $g|_p : T_pM \times T_pM \rightarrow \mathbb{R}$ with diagonal form $(-,+,+,+)$ to each tangent space. Thus, $g$ has signature +2 and is not positive-definite;
\item[(iii)] A time orientation, given by a globally defined timelike vector field $X : M \rightarrow TM$. A timelike or null tangent vector $v \in T_pM$ is said to be future-directed if $g(X(p),v) < 0$, or past-directed if $g(X(p),v)>0$. 
\end{description}
Some important remarks are now in order:
\begin{description}
\item[(a)] The condition (i) can be formulated for each number of space-time dimensions $\geq 2$;
\item[(b)] Also the convention $(+,-,-,-)$ for the diagonal form of the metric can be chosen. The definitions of timelike and spacelike will then become opposite to out definitions: $X$ is timelike if $g(X(p),X(p))>0$ for $p \in M$, and $X$ is spacelike if $g(X(p),X(p))<0$ for $p \in M$;
\item[(c)] The pair $(M,g)$ is only defined up to equivalence. Two pairs $(M,g)$ and $(M',g')$ are said to be equivalent if there exists a diffeomorphism $\alpha: M \rightarrow M'$ such that $\alpha_* g = g'$. Thus, we are really dealing with $\textit{an equivalence class of pairs}$ (M,g). 
\end{description} 
Now, if $M$ is a connected, four-dimensional, Hausdorff four-manifold of class $C^{\infty}$, a linear partial differential operator is a linear map
\begin{equation} \label{eq:2.60} L : u \in C^{\infty}(M) \rightarrow (Lu) \in C^k(M), \end{equation}
with coefficients $a^{i_1...i_m}$ given by functions of class $C^k$.  The $\textit{cha\-rac\-te\-ri\-stic}$ $\textit{po\-ly\-no\-mial}$ of the operator $L$ at a point $x \in M$ is 

\begin{equation} \label{eq:2.61} H(x, \xi) = \sum\limits_{i_1,...,i_m}a^{i_1...i_m} (x)\xi_{i_1}...\xi_{i_m}, \end{equation}
where $\xi_i$ is a cotangent vector at $x$. The cone in the cotangent space $T^*_x$ at $x$ defined by

\begin{equation} \label{eq:2.62} H(x,\xi)=0 \end{equation}
is called the $\textit{characteristic conoid}$. By construction, such a cone is independent of the choice of coordinates, because the terms of maximal order (also called leading or principal symbol) of $L$ transform into terms of the same order by a change of coordinates. 

The concept of $\textit{hyperbolicity}$ at $x$ of the operator $L$, requires the existence of a domain $\Gamma_x$, a convex open cone in $T^*_x$, such that every line through $\lambda \in \Gamma_x$ cuts the characteristic conoid in $m$ real distinct points. \\
In particular, second-order differental operators with higher-order terms 
\begin{equation*} (g^{-1})^{\alpha \beta} (x) \frac{\partial}{\partial x^{\alpha}}\frac{\partial}{\partial x^{\beta}} \end{equation*}
are hyperbolic at $x$ if and only if the cone defined by 

\begin{equation} \label{eq:2.63} H_2(x,\xi) \equiv \sum\limits_{\alpha, \beta=1}^{n} (g^{-1})^{\alpha \beta}(x) \xi_\alpha \xi_\beta =0 \end{equation}
is convex, i.e. if the quadratic form $H_2(x,\xi)$ has signature $(1, n-1)$.

\section{Fundamental Solutions with an Algebraic Singularity}
Following \cite{hadamard1952lectures}, we study in the first place the case of a surface without a singular point. We look for fundamental solutions of Eq. like ($\ref{eq:2.50}$), but in $m$ variables, having the form
\begin{equation} \label{eq:2.64} u = UG^p + \omega, \end{equation}
where $G=0$ is the equation of the desired regular surface, $p$ a given constant, $U$ and $\omega$ are regular functions. Since we assume for $u$ the homogeneous equation
\begin{equation} \label{eq:2.65} \mathcal{F}(u)= \bigg{(} \sum\limits_{i,k=1}^m A^{ik} \frac{\partial^2}{\partial x^i \partial x^k} + \sum\limits_{i=1}^m B^i \frac{\partial}{\partial x^i} + C \bigg{)} u =0, \end{equation}
the insertion of the factorized ansatz $u= UF(G)$ into Eq. ($\ref{eq:2.65}$) yields, upon defining $p_i \equiv \frac{\partial G}{\partial x^i}$, terms involving the first derivatives 
\begin{equation} \label{eq:2.66} \frac{\partial u}{\partial x^i} = U p_i F'(G) + \frac{\partial U}{\partial x^i} F(G), \end{equation}
and terms involving the second derivatives
\begin{equation} \label{eq:2.67} \begin{split} \frac{\partial^2 u}{\partial x^i \partial x^k}= &U p_i p_k F''(G) + \bigg{(} p_i \frac{\partial U}{\partial x^k} + p_k \frac{\partial U}{\partial x^i} + U \frac{\partial^2 G}{\partial x^i \partial x^k} \bigg{)} F'(G) \\
& + \frac{\partial^2 U}{\partial x^i \partial x^k}F(G). \end{split}\end{equation}
Now we have to multiply Eq. ($\ref{eq:2.66}$) for every $i$ by $B^i$, and Eq. ($\ref{eq:2.67}$) for every $i$, $k$ by $A^{ik}$, and add to $Cu=CUF$. In this combination one finds that \cite{hadamard1952lectures}:
\begin{description}
\item[(i)]The coefficient of $F''(G)$ is $A(p_1, ..., p_m)$;
\item[(ii)]In the coefficient of $F'(G)$, the terms in $\frac{\partial U}{\partial x^i}$ are
\begin{equation*} \frac{\partial U}{\partial x^i} \sum\limits_{k=1}^m 2A^{ik}p_k = \frac{\partial U}{\partial x^i}\frac{\partial A}{\partial p_i}. \end{equation*}
\end{description}
Thus, Eq. ($\ref{eq:2.65}$) becomes
\begin{equation} \label{eq:2.68} U F''(G)A(p_1, ..., p_m) + F'(G) \bigg{(} \sum\limits_{i=1}^m \frac{\partial U}{\partial x^i} \frac{\partial A}{\partial p_i} + MU \bigg{)} + F(G) \mathcal{F}(U) =0, \end{equation}
where $M$ denotes \\
\begin{equation} \label{eq:2.69} M \equiv \mathcal{F}(G) - CG. \end{equation}
\\
In particular, if $F(G)$ reduces to the $p$-th power of $G$, i.e. $F(G) =G^p$, one gets from Eq. ($\ref{eq:2.68}$) the equation

\begin{equation} \label{eq:2.70} p(p-1)G^{p-2}UA(p_1, ..., p_m) + pG^{p-1} \bigg{(} \sum\limits_{i=1}^m \frac{\partial U}{\partial x^i} \frac{\partial A}{\partial p_i} + MU \bigg{)} + G^p \mathcal{F}(U) =0. \end{equation}
If the cases $p=0, 1$ are ruled out, the left-hand side of Eq. ($\ref{eq:2.70}$) cannot vanish identically or even be a regular function if the coefficient $A(p_1, ..., p_m)$ does not vanish. In other words, $G=0$ is not a characteristic. The equation 

\begin{equation*} A \bigg{(} \frac{\partial G}{\partial x^1}, ..., \frac{\partial G}{\partial x^m } \bigg{)} = 0 \end{equation*}
must be either an identity or a consequence of $G=0$, hence there exists a function $A_1$, regular also for $G=0$, such that
\begin{equation} \label{eq:2.71} A \bigg{(} \frac{\partial G}{\partial x^1}, ..., \frac{\partial G}{\partial x^m } \bigg{)} = A_1 G. \end{equation}
The Delassus theorem is therefore proved. Hereafter we assume that Eq. ($\ref{eq:2.71}$) is satisfied, so that the term involving $G^{p-2}$ disappears from Eq. ($\ref{eq:2.70}$). More precisely, one finds, by virtue of Eq. ($\ref{eq:2.71}$), that Eq. ($\ref{eq:2.70}$) reads as
\begin{equation} \label{eq:2.72} p G^{p-1} \bigg{[} (p-1) A_1 U + MU + \sum\limits_{i=1}^m \frac{\partial U}{\partial x^i} \frac{ \partial A}{\partial p_i} \bigg{]} + G^{p} \mathcal{F}(U) =0. \end{equation}
At this stage, multiplication by $G^{1-p}$ and subsequent restriction to the surface $G=0$ imply that Eq. ($\ref{eq:2.72}$) leads to
\begin{equation} \label{eq:2.73} \sum\limits_{i=1}^m \frac{\partial U}{\partial x^i} \frac{\partial A}{\partial p_i} + \big{[} M + (p-1) A_1 \big{]} U=0. \end{equation}
This is a linear partial differential equation of first order in $U$, whose integration would lead to the introduction of the lines defined by the ordinary differential equations
\begin{equation} \label{eq:2.74} \frac{dx^1}{\frac{1}{2} \frac{\partial A}{\partial p_1}} = ... = \frac{dx^m}{\frac{1}{2} \frac{\partial A}{\partial p_m}} = ds. \end{equation}
In the denominators it is possible to recognize the direction cosines of the transversal to $G=0$; this is, in the case considered, tangent to that surface (since the latter is a characteristic; the transversal is the direction of the generatrix of contact between the plan $(p_1, ..., p_m)$ and the characteristic cone). Thus, a line satisfying Eq. ($\ref{eq:2.74}$) and issuing from a point of $G=0$ is lying entirely on that surface. These lines are in fact the bicharacteristics of Eq. ($\ref{eq:2.59}$), with $\Omega=A$ and $z=G$. If the function $A_1$ in Eq. ($\ref{eq:2.71}$) vanishes, so that the function $G$ satisfies identically the equation $A=0$, the theory of partial differential equations of first order shows that, besides Eq. ($\ref{eq:2.74}$), the bicharacteristics satisfy also the equations
\begin{equation} \label{eq:2.75} \frac{dp_1}{-\frac{1}{2} \frac{\partial A}{\partial x^1}} = ... = \frac{dp_m}{-\frac{1}{2} \frac{\partial A}{\partial p_m}} = ds, \end{equation}
and hence they can be determined without knowing the equation $G=0$ by integrating the system of ordinary differential equations ($\ref{eq:2.74}$) and ($\ref{eq:2.75}$).

\section{\-G\-e\-o\-d\-e\-s\-i\-c\- E\-q\-u\-a\-t\-i\-o\-n\-s\- W\-i\-t\-h\- a\-n\-d\- W\-i\-t\-h\-o\-u\-t\- R\-e\-p\-a\-r\-a\-m\-e\-t\-r\-i\-z\-a\-t\-i\-o\-n\- I\-n\-v\-a\-r\-i\-a\-n\-c\-e\-}
The characteristic conoid with any point $a(a^1, ..., a^m)$ as its vertex has that point for a singular point, and to study this new case one has to first form the equation of the characteristic conoid. That is $\textit{the locus of all bicharacteristics}$ $\textit{issuing from a}$. One has to take any set of quantities $p_1$, ..., $p_m$ fulfilling the equation
\begin{equation} \label{eq:2.76} A(p_1, ..., p_m; x^1, ..., x^m)=0 \end{equation}
and, with the initial conditions
\begin{equation} \label{eq:2.77} x^i(s=0)=a^i, \hspace{1cm} p_i(s=0)=p_{0i}, \end{equation}
integrate the Eqs. ($\ref{eq:2.74}$) and ($\ref{eq:2.75}$), here written concisely in Hamilton form
\begin{equation} \label{eq:2.78} \frac{dx^i}{ds} = \frac{1}{2} \frac{\partial A}{\partial p_i}, \hspace{1cm} \frac{dp_i}{ds} =- \frac{1}{2} \frac{\partial A}{\partial x^i}. \end{equation}
Since the ratios of the quantities $p_{0i}$, ..., $p_{0m}$ under consideration ($\ref{eq:2.76}$) depend on $(m-2)$ parameters, the locus of the line generated in such a way is a surface. Our task is to obtain a precise form for the equation of this surface. For this purpose, we construct every line issuing from the point $(a^1, ..., a^m)$ and satisfying the differential system ($\ref{eq:2.78}$), whether or not the initial values $p_{01}$, ..., $p_{0m}$ of the variables $p_i$ satisfy Eq. ($\ref{eq:2.76}$). Such lines are indeed the $\textit{geodesics}$ of a suitably chosen line element. Within this framework we recall the definition of a $\textit{geodesic}$.

If $T$ is a tensor field defined along a curve $\lambda$ of class $C^r$, and if $\bar T$ is an arbitrary tensor field of class $C^r$ which extends $T$ in an open neighbourhood of $\lambda$, the covariant derivative of $T$ along $\lambda(t)$ can be denoted by $\frac{DT}{\partial t}$ and is equal to
\begin{equation} \label{eq:2.79} \frac{DT}{\partial t} \equiv \nabla_{\frac{\partial}{\partial t}} \bar T, \end{equation}
where $\nabla$ is the connection of the Riemannian or pseudo-Riemannian manifold we are considering. The formula ($\ref{eq:2.79}$) describes a tensor field of class $C^{r-1}$, defined along the curve $\lambda$, and independent of the extension $\bar T$ \cite{hawking1973large}. In particular, if $\lambda$ has local coordinates $x^a(t)$, and $X^{a} = \frac{dx^a}{dt} $ are the components of its tangent vector, the expression in local coordinates of the covariant derivative of a vector $Y$ along a curve is 
\begin{equation} \label{eq:2.80} \frac{DY^a}{\partial t} = \frac{\partial Y^a}{\partial t} + \sum\limits_{b,c=1}^n \Gamma\{a, [b,c] \} \frac{dx^b}{dt}Y^c, \end{equation}
where the $\Gamma'$s are the Christoffel symbolds of second kind.

The tensor field $T$ (and also, in particular, the vector field corresponding to $Y$) is said to undergo $\textit{parallel transport along}$ $\lambda$ if 
\begin{equation} \label{eq:2.81} \frac{DT}{\partial t} =0. \end{equation}
In particular, one may consider the covariant derivative of the tangent vector itself along $\lambda$. The curve $\lambda$ is said to be a $\textit{geodesic}$ if 
\begin{equation*} \nabla_{X}X= \frac{D}{\partial t} \bigg{(}\frac{\partial}{\partial t} \bigg{)}_\lambda \end{equation*}
is parallel to the tangent vector $\big{(}\frac{\partial}{\partial t} \big{)}_\lambda$. This implies that there exists a smooth function on the manifold $M$ for which (the semicolon being used to denote the covariant derivative $\nabla_b$)
\begin{equation} \label{eq:2.82} \sum\limits_{b=1}^n X^a_{;b}X^b = fX^a. \end{equation}
The parameter $v(t)$ along the curve $\lambda$ such that
\begin{equation} \label{eq:2.83} \frac{D}{\partial v} \bigg{(} \frac{\partial}{\partial v} \bigg{)}_\lambda =0 \end{equation}
is said to be an $\textit{affine parameter}$. The corresponding tangent vector $U \equiv \big{(} \frac{\partial}{\partial v} \big{)}_\lambda$ obeys the equation
\begin{equation} \label{eq:2.84} \sum\limits_{b=1}^n U^a_{;b}U^b=0, \end{equation}
i.e., by virtue of ($\ref{eq:2.80}$),

\begin{equation} \label{eq:2.85} \frac{d^2 x^a}{d v^2} + \sum\limits_{b,c=1}^n \Gamma \{a, [b,c]\} \frac{dx^b}{dv} \frac{dx^c}{d v}=0. \end{equation}
\\
The affine parameter is determined up to a linear transformation
\begin{equation} \label{eq:2.86} v' = av+b. \end{equation}
We stress that our geodesics are $\textit{auto-parallel}$ curves \cite{hawking1973large}. We prefer auto-parallel curves because they involve the full connection. Given this definition of a geodesic, we have, in the case under consideration its alternative definition as extremal curve for the Lorentzian arc-length. In fact, if
\begin{equation} \label{eq:2.87} \mathbf{H} (dx^1, ..., dx^m; x^1, ..., x^m) = \sum\limits_{i,k=1}^m H_{ik}dx^i \otimes dx^k \end{equation}
is any non-singular quadratic form, the coefficients $H_{ik}$ being given functions of $x^1$, ..., $x^m$, then if the differentials $dx^l$ are viewed as differentials of the corresponding $x^l$, the $\mathbf{H}$ can be taken as the squared line element in a $m$-dimensional manifold. The integral
\begin{equation} \label{eq:2.88} \int \sqrt{ \mathbf{H}(dx^1, ..., dx^m)} = \int \sqrt{\mathbf{H}(x'^1, ..., x'^m)}dt, \end{equation}
where $x'^i \equiv \frac{d x^i}{dt}$, is therefore the arc-length of a smooth curve. The corresponding geodesics are the lines which make the variation of this functional vanish. Their differential equations are
\begin{equation} \label{eq:2.89} \frac{d}{dt} \bigg{(} \frac{\partial}{\partial x'^i} \sqrt{\mathbf{H}} \bigg{)} - \frac{\partial}{\partial x^i} \sqrt{\mathbf{H}}=0, \hspace{1cm} i=1, 2, ..., m. \end{equation}

On the other hand, Lagrangian dynamics leads to writing these differential equations in a different form, i.e.
\begin{equation} \label{eq:2.90} \frac{d}{ds} \bigg{(} \frac{\partial}{\partial x'^i} \mathbf{H} \bigg{)} - \frac{\partial}{\partial x^i} \mathbf{H} =0, \hspace{1cm} i=1, ..., m; \end{equation}
this being the law governing the motion of a system whose $\textit{vis viva}$ is $\mathbf{H}(x',x)$, and on which no forces act. The equations ($\ref{eq:2.89}$) and ($\ref{eq:2.90}$) are not exactly equivalent, but are $\textit{conditionally equivalent}$ \cite{hadamard1952lectures}. The former determines the required lines but not $t$, the time remaining an arbitrary parameter whose choice is immaterial. In other words, Eq. ($\ref{eq:2.89}$) are reparametrization-invariant, because they remain unchanged if $t$ gets replaced by any smooth function $\phi(t)$. 

However the latter equations, i.e. ($\ref{eq:2.90}$), define not only a line, but a motion on that line, and this motion is no longer arbitrary in time. It must satisfy the vis viva integral
\begin{equation} \label{eq:2.91} \mathbf{H}= {\rm constant}, \end{equation}
hence the representative point $(x^1, ..., x^m)$ must move on the curve with constant kinetic energy. But on taking into account Eq. ($\ref{eq:2.91}$), the systems ($\ref{eq:2.89}$) and ($\ref{eq:2.90}$) become in general equivalent. A simple way to see this is to point out that, if we choose $t$ in Eq. ($\ref{eq:2.89}$) in such a way that $\mathbf{H}$ is constant in time, then the denominator $2 \sqrt{\mathbf{H}}$  in the identity
\begin{equation} \label{eq:2.92} \frac{\partial}{\partial x'^i} \sqrt{\mathbf{H}}= \frac{1}{2 \sqrt{\mathbf{H}}}\frac{\partial}{\partial x'^i} \mathbf{H} \end{equation}
is not affected by the time derivative, and we obtain eventually Eq. ($\ref{eq:2.90}$). 

Conversely, if one wants to write Eq. ($\ref{eq:2.90}$) in such a way that the independent variable $t$ may become arbitrary, one has to note that, as a function of $t$, the variable $s$ can be easily evaluated from the vis viva integral ($\ref{eq:2.91}$) according to
\begin{equation} \label{eq:2.93} ds= \sqrt{\mathbf{H}} dt. \end{equation}
On replacing $ds$ by this value, and accordingly $x'^i$ by $\frac{x'^i}{\sqrt{\mathbf{H}}}$, one recovers ($\ref{eq:2.89}$) \cite{darboux1896leccons, hadamard1952lectures}. All these recipes no longer hold for bicharacteristics, for which $A=\mathbf{H} =0$. For them the system ($\ref{eq:2.89}$) becomes meaningless, whereas Eqs. ($\ref{eq:2.90}$) remain valid.

\chapter{How to Build the Fundamental Solution}
\epigraph{The game's afoot.}{William Shakespeare, King Henry V}

\section{Hamiltonian Form of Geodesic Equations}
Let us now try to see how it is possible to build the fundamental solution. For this purpose, we here consider the fundamental form of $n$-dimensional Euclidean space \cite{garabedian1998partial}
\begin{equation} \label{eq:3.1} g_{E} = \sum \limits_{i,j=1}^n A_{ij}dx^i \otimes dx^j, \end{equation}
where $A_{ij}=A_{ji}=A_{ij}(x^1, ..., x^n)$. Let $\Omega_{pq}$ be the set of piecewise smooth curves in the manifold $M$ from $p$ to $q$. Given the curve $c: [0,1] \rightarrow M$ and belonging to $\Omega_{pq}$, there is a finite partition of $[0,1]$ such that $c$ restricted to the sub-interval $[t_i,t_{i+1}]$ is smooth $\forall i$. If we consider the interval $[t_0,t_1]$, the arc-length of $c$ with respect to $g_E$ is defined by
\begin{equation} \label{eq:3.2} I \equiv \int_{t_0}^{t_1} \sqrt{\sum\limits_{i,j=1}^n A_{ij} \frac{dx^i}{dt} \frac{dx^j}{dt}} dt, \end{equation}
and at the ends of the integration interval we define 
\begin{equation} \label{eq:3.3} x^i(t_0)\equiv y^i, \; x^i(t_1)\equiv z^i, \hspace{1cm} i=1, 2, ..., n. \end{equation}
Let $a^{ij}$ be the controvariant components of the inverse metric, for which 
\begin{equation} \label{eq:3.4} \sum\limits_{k=1}^n a^{ik}A_{kj} = \delta^i_j. \end{equation}
To begin with the variational problem let us define 
\begin{equation} \label{eq:3.5} Q \equiv \sum\limits_{i,j=1}^n A_{ij} \frac{dx^i}{dt} \frac{dx^j}{dt}, \end{equation}
thus the Lagrangian related to this problem, defined as 
\begin{equation} \label{eq:3.6} L = \sqrt{Q}, \end{equation}
is a function homogeneous of degree 1 in the $\dot{x}^j = \frac{dx^j}{dt}$. Since the associated Hessian matrix is singular, i.e.
\begin{equation} \label{eq:3.7} det \bigg{(} \frac{\partial^2 L}{\partial \dot{x}^i \partial \dot{x}^j} \bigg{)} =0, \end{equation}
it is not possible to define the Legendre transform. However, it is possible to overcome this difficulty by writing the Euler-Lagrange equations, which in terms of Q are
\begin{equation} \label{eq:3.8}  \frac{d}{dt} \frac{1}{\sqrt{Q}} \frac{\partial Q}{\partial \dot{x}^i} - \frac{1}{\sqrt{Q}}\frac{\partial Q}{\partial x^i} =0, \hspace{1cm} i=1, 2, ..., n. \end{equation}
This suggests taking $t$, the parameter along the geodesics, as the arc-length measured from the point $y^1, ..., y^n$. 

The integral 
\begin{equation} \label{eq:3.9} J \equiv \int_0^s Q dt \end{equation}
is stationary. The terminal values $y^1, ..., y^n$ and $z^1, ..., z^n$ of $x^1, ..., x^n$ are fixed, but the upper limit of integration $s$ is allowed to vary. Hence, the extremals of $J$ come to depend on $s$ and on the variable of integration $t$ according to \cite{garabedian1998partial}
\begin{equation} \label{eq:3.10} x^i=x^i \bigg{(} \frac{I}{s} \; t \bigg{)} \end{equation}
for a change of scale, since $Q$ is a function homogeneous of degree 2 in the velocity variables $\dot{x}^1, ..., \dot{x}^n$. Thus, the constant value of $Q$ becomes $\frac{I^2}{s^2}$ along each such extremal curve, and
\begin{equation} \label{eq:3.11} J = \frac{I^2 (z^1, ..., z^n)}{s^2}(s - 0) = \frac{I^2(z^1, ..., z^n)}{s}. \end{equation}
Now we can apply the Hamilton-Jacobi theory to the equations of motion that we are studying. Since the corresponding momenta are 
\begin{equation} \label{eq:3.12} p_i = \frac{\partial L}{\partial \dot{x}^i} = \frac{\partial Q}{\partial \dot{x}^i} = 2 \sum\limits_{j=1}^n A_{ij} \frac{dx^j}{dt}, \end{equation}
we can re-express the velocity variables in the form $ \frac{dx^j}{dt} = \frac{1}{2} \sum\limits_{i=1}^n a^{ji}p_i$. Thus, it is possible to write
\begin{equation} \label{eq:3.13} \sum\limits_{i=1}^n p_i \frac{dx^i}{dt} = 2 \sum\limits_{i,j=1}^n A_{ij} \frac{dx^i}{dt } \frac{dx^j}{dt} = \frac{1}{2} \sum\limits_{i,j=1}^n a^{ij}p_i p_j. \end{equation}
The Hamiltonian reads as
\begin{equation} \label{eq:3.14} H=Q= \frac{1}{4} \sum\limits_{i,j=1}^n a^{ij}p_i p_j. \end{equation}
The functional $J$, previously defined, satisfies the Hamilton-Jacobi equation
\begin{equation} \label{eq:3.15} \frac{\partial J}{\partial s} + \frac{1}{4} \sum\limits_{i,i=1}^n a^{ij} (z^1, ..., z^n) \frac{\partial J}{\partial z^i}\frac{\partial J}{\partial z^j} =0. \end{equation}
In this equation we can insert the form $(\ref{eq:3.11})$ of $J$ and set eventually $s=1$. The non-vanishing factor $I^2$, common to both terms, drops therefore out, and the equation $(\ref{eq:3.15})$ reduces to 
\begin{equation} \label{eq:3.16} \sum\limits_{i,i=1}^n a^{ij} \frac{\partial I}{\partial z^i}\frac{\partial I}{\partial z^j} =1. \end{equation}
If we now define
\begin{equation} \label{eq:3.17} \Gamma \equiv I^2, \end{equation}
we obtain $I= \sqrt{\Gamma}$, and Eq. ($\ref{eq:3.16}$) takes the remarkable form
\begin{equation} \label{eq:3.18} \sum\limits_{i,j=1}^n a^{ij} \frac{\partial \Gamma}{\partial z^i}\frac{\partial \Gamma}{\partial z^j} = 4 \Gamma. \end{equation}
This equation coincides with Eq. ($\ref{eq:2.71}$) upon setting therein 
\begin{equation} \label{eq:3.19} G= \Gamma, \; A_1=4, \; A=a({\rm grad} \Gamma,{\rm grad}\Gamma). \end{equation}
The function $\Gamma$ is a conoidal solution of Eq. ($\ref{eq:3.18}$), generated by all bicharacteristics of this equation passing through $(y^1, ..., y^n)$ which are geodesic of the metric $g_E$.

The geodesics satisfy the equations of motion in Hamiltonian form
\begin{equation} \label{eq:3.20} \frac{dx^i}{ds} = \frac{1}{2} \sum\limits_{j=1}^n a^{ij}p_j; \hspace{1cm} \frac{dp_i}{ds} = - \frac{1}{4} \sum\limits_{j,k=1}^n \frac{\partial a^{jk}}{\partial x^i} p_j p_k,\end{equation}
together with the initial conditions
\begin{equation} \label{eq:3.21} x^i(0)=y^i; \hspace{1cm} p_i(0) = \gamma_i. \end{equation}
In a generic space-time manifold, the $a^{ij}$ of Eq. ($\ref{eq:3.18}$) will denote the contravariant components $(g^{-1})^{ij}$ of
\begin{equation} \label{eq:3.22} g^{-1}=\sum\limits_{i,j=1}^n a^{ij} \frac{\partial}{\partial x^i} \otimes \frac{\partial}{\partial x^j}, \end{equation}
the signature of $g$ being $(n-2)$. Equation ($\ref{eq:3.17}$) will then be interpreted by stating that $\Gamma$ is a two-point function, called the $\textit{world function}$ and equal to the square of the geodesic distance between the space-time points $x=(x^1, ..., x^n)$ and $y=(y^1, ..., y^n)$. This means that such a formalism can only be used $\textit{locally}$, when there exists a unique geodesic from $x$ to $y$. Such a space-time is said to be $\textit{geodesically convex}$.

\section{The Unique Real-Analytic World Function}
We aim now to demostrate, following Hadamard \cite{hadamard1952lectures}, that Eq. ($\ref{eq:3.18}$), or ($\ref{eq:2.71}$), is the fundamental equation in the theory of the characteristic conoid, in that any function real-analytic in the neighbourhood of the desired vertex $a$, vanishing on the conoid and satysfying Eq. ($\ref{eq:3.18}$), can only be the world function $\Gamma$ itself (besides this, there exist infinitely many non-analytic solutions of Eq. ($\ref{eq:3.18}$).

\proof The desired function should be of the form $\Gamma \Pi$ , where $\Pi$ is a real-analytic function. By insertion into Eq. ($\ref{eq:3.18}$), this yields
\begin{equation} \begin{split} \label{eq:3.23} 4\Gamma \Pi &= \sum\limits_{i,j=1}^n a^{ij} \bigg{(} \Pi \frac{\partial \Gamma}{\partial x^i} + \Gamma \frac{\partial \Pi}{\partial x^i} \bigg{)} \bigg{(} \Pi \frac{\partial \Gamma}{\partial x^j} + \Pi \frac{\partial \Pi}{\partial x^j} \bigg{)} \\
&= \Pi^2 \nabla_1 \Gamma + 2 \Pi \Gamma \nabla_1 (\Pi, \Gamma) + \Gamma^2 \nabla_1 \Pi, \end{split} \end{equation}
On the right-hand side of Eq. ($\ref{eq:3.23}$), the term involving the mixed differential parameter can be expressed, making use of the derivative of $\Pi$ along a geodesic and the symmetry of $a^{ij}$, as
\begin{equation} \begin{split} \label{eq:3.24}& \nabla_1(\Gamma,\Pi) = \sum\limits_{i,j=1}^n a^{ij} \frac{\partial \Gamma}{\partial x^i} \frac{\partial \Pi}{\partial x^j} = 2s \sum\limits_{i,j=1}^n a^{ij}p_i \frac{\partial \Pi}{\partial x^j} \\
&= s \sum\limits_{j=1}^n \frac{\partial \Pi}{\partial x^j} \frac{\partial A}{\partial p_j} = 2s \sum\limits_{j=1}^n \frac{\partial \Pi}{\partial x^j} \frac{dx^j}{ds} = 2s \frac{d\Pi}{ds}. \end{split} \end{equation}
Thus, Eq. ($\ref{eq:3.23}$) becomes
\begin{equation} \label{eq:3.25} \Pi^2 \nabla_1 \Gamma + 4s \Gamma \Pi \frac{d \Pi}{ds} + \Gamma^2 \nabla_1 \Pi = 4 \Gamma \Pi. \end{equation} 
In this equation, we can divide both sides by $4 \Gamma \Pi$, finding therefore
\begin{equation} \begin{split} \label{eq:3.26} &\Pi \frac{\nabla_1 \Gamma}{4 \Gamma} + s \frac{d \Pi}{ds} + \frac{\Gamma}{4 \Pi} \nabla_1 \Pi - 1 = \bigg{(} \Pi + s \frac{d \Pi}{ds} - 1 \bigg{)} + \frac{ \Gamma}{4 \Pi} \nabla_1 \Pi \\
&= \frac{d}{ds} [s (\Pi - 1)] + \frac{\Gamma}{4 \Pi} \nabla_1 \Gamma = 0. \end{split}\end{equation}
This equation shows that the function $\Pi$ equals 1 over the whole conoid, and hence we can write the general formula
\begin{equation} \label{eq:3.27} \Pi = 1 + \Gamma^l E, \end{equation}
where $l$ is a positive exponent, and $E$ is yet another real-analytic function, not vanishing over the whole surface of the conoid. But this leads to a contradiction, because the insertion of Eq. ($\ref{eq:3.27}$) for $\Pi$ into Eq. ($\ref{eq:3.26}$) yields
\begin{equation} \begin{split} \label{eq:3.28} &\frac{d}{ds} \big{(} s \Gamma^l E \big{)} + \frac{\Gamma}{4 \Pi} \nabla_1 (\Gamma^l E) = \Gamma^l E + s \frac{d \Gamma^l}{ds} E + s \Gamma^l \frac{dE}{ds} \\
& + \frac{\Gamma}{4 \Pi} \sum\limits_{i,j=1}^n a^{ij} \bigg{(} l \Gamma^{l-1} \frac{\partial \Gamma}{\partial x^i} E + \Gamma^l \frac{ \partial E}{\partial x^i} \bigg{)} \bigg{(} l \Gamma^{l - 1} \frac{\partial \Gamma}{\partial x^j} E + \Gamma^l \frac{\partial E}{\partial x^j} \bigg{)} \\
& = \Gamma^l \bigg{[} s \frac{d E}{ds} + (2l + 1) E \bigg{]} + \frac{\Gamma}{4 \Pi} \bigg{[} l^2 \Gamma^{2l - 2} E^2(\Delta_1 \Gamma ) \\
& + 2l \Gamma^{2l-1}E \Delta_1 (\Gamma, E) + \Gamma^{2l} \Delta_1 E \bigg{]}. \end{split} \end{equation}
This equation, when restricted to the characteristic conoid, reduces to
\begin{equation} \label{eq:3.29} s \frac{dE}{ds} + (2l + 1)E=0, \end{equation}
which is solved by
\begin{equation} \label{eq:3.30} E = E_0 \bigg{(} \frac{s}{s_0} \bigg{)}^{- ( 2l+1)}, \end{equation}
which can only be regular if $E_0=0$, that implies $E=0$. \endproof

\section{Examples of Fundamental Solutions}
Now we aim to study the linear partial differential equation 
\begin{equation} \label{eq:3.31} \mathcal{F}(u)= \bigg{(} \sum\limits_{i,k=1}^m A^{ik} \frac{\partial^2}{\partial x^i \partial x^k} + \sum\limits_{i=1}^m B^i \frac{\partial }{\partial x^i} + C \bigg{)} u =0, \end{equation}
with associated world function $\Gamma$, the square of the geodesic distance between two points, obeying Eq. ($\ref{eq:3.18}$), with coefficients $a^{ij}$ equal to the controvariant components $A^{ij}$ of the inverse metric. A fundamental solution of $ \mathcal{F}(u)=0$ is a two-point function $R(x, \xi)$, with $x=(x^1, ..., x^n)$ and $\xi=(\xi^1, ..., \xi^n)$, which solves Eq. ($\ref{eq:3.31}$) in its dependence on $x$ and possesses, at the parameter point $\xi$, a singularity characterized by the split reading as \cite{garabedian1998partial}
\begin{equation} \label{eq:3.32} R = \frac{U}{\Gamma^m} + V log (\Gamma) + W, \end{equation}
where $U$, $V$ and $W$ are taken to be smooth functions of $x$ in a neighbourhood of $\xi$, with $U \neq 0$ at $\xi$, and where the exponent $m$ is given by
\begin{equation} \label{eq:3.33} m = \frac{n}{2} - 1. \end{equation}
We are going to show that, when $n$ is odd, the coefficient $V$ of the logarithm vanishes, whereas the term $W$ is redundant for $n$ even. Thus, the dimension of Euclidean space affects in a non-trivial way the conceivable form of the fundamental solution.

\subsection{Odd Number of Variables}
Following Garabedian \cite{garabedian1998partial}, we consider first the odd values of $n$. We then put $V=W=0$ in Eq. ($\ref{eq:3.32}$), and look for a convergent series expressing the unknown function $U$, in the form
\begin{equation} \label{eq:3.34} U = U_{l} \Gamma^{l} = U_0 + U_1\Gamma + U_2 \Gamma^2 + O(\Gamma^3), \end{equation}
with regular coefficients $U_l$. By replacing $R= U_{l} \Gamma^{l-m}$ inside ($\ref{eq:3.31}$), where we recall that $u=R$, and exploiting the symmetry of the inverse metric $a^{ij}$, we have
\begin{equation} \begin{split} \label{eq:3.35} &\mathcal{F}[U_l\Gamma^{l-m}] =(l-m) (l-m-1)U_l \Gamma^{l-m-2} \sum\limits_{i,j=1}^n a^{ij} \frac{\partial \Gamma}{\partial x^i} \frac{\partial \Gamma}{\partial x^j} + \\
& + (l-m) \bigg{[} 2 \sum\limits_{i,j=1}^n a^{ij} \frac{\partial U}{\partial x^i} \frac{\partial \Gamma}{\partial x^j} + 4 D U_l \bigg{]} \Gamma^{l-m-1} + \mathcal{F}[U_l]\Gamma^{l-m}, \end{split} \end{equation}
where $D$ is the term
\begin{equation} \label{eq:3.36} D \equiv \frac{1}{4} \sum\limits_{i,j=1}^n a^{ij} \frac{\partial^2 \Gamma}{\partial x^i \partial x^j} + \frac{1}{4} \sum\limits_{i=1}^n B^i \frac{\partial \Gamma}{\partial x^i}. \end{equation}
One should stress that the possibility of eliminating the lowest power of $\Gamma$ on the right-hand side of Eq. ($\ref{eq:3.35}$) by means of the first order partial differential equation ($\ref{eq:3.18}$) now shows why the fundamental solution $R$ should be expanded in terms of this particular function, i.e. the world function $\Gamma$. 

It is now convenient to introduce again a parameter $s$ which is measured along the geodesics that generate $\Gamma$. We can then write
\begin{equation} \label{eq:3.37} \sum\limits_{i,j=1}^n a^{ij} \frac{\partial U_l}{\partial x^i} \frac{\partial \Gamma}{\partial x^j} = 2s \frac{d U_l}{ds}. \end{equation}
Hence we arrived at a simplified form of Eq. ($\ref{eq:3.35}$), i.e. \cite{garabedian1998partial}
\begin{equation} \label{eq:3.38} \mathcal{F}[U_l\Gamma^{l-m}] = 4 (l-m) \bigg\{ s \frac{dU_l}{ds} + (D+l-m-1)U_l \bigg\} \Gamma^{l-m-1} + \mathcal{F}[U_l] \Gamma^{l-m}. \end{equation}
At this stage, in order to solve, $\forall x \neq \xi$, the equation
\begin{equation} \label{eq:3.39}  \mathcal{F}[R] = \mathcal{F} \bigg{[} \sum\limits_{l=0}^\infty U_l \Gamma^{l-m} \bigg{]} =0, \end{equation}
we set to zero all coefficients of the various powers of $\Gamma$. This leads to the fundamental recursion formulae
\begin{equation} \label{eq:3.40} \bigg{[} s \frac{d}{ds} + (D- m-1) \bigg{]} U_0 =0, \end{equation}
\begin{equation} \label{eq:3.41} \bigg{[} s \frac{d}{ds} + (D+ l-m-1) \bigg{]} U_l = - \frac{1}{4(l-m)} \mathcal{F} [U_{l-1}], \hspace{1cm} l \geq  1, \end{equation}
for the evaluation of $U_0$, $U_1$, $U_2$, .... For odd values of $n$, the division by $(l-m)$ on the right-hand side of ($\ref{eq:3.41}$) is always legitimate by virtue of the expression of $m$, because $(l-m)$ never vanishes.

Note that, when Eq. ($\ref{eq:3.31}$) is hyperbolic, the fundamental solution $R$ becomes infinite along a two-sheeted conoid $\Gamma=0$ separating $n$-dimensional space into three parts. This conoid is indeed a characteristic surface for the second-order equation ($\ref{eq:3.31}$), since ($\ref{eq:3.18}$) reduces on the level surface $\Gamma=0$ to the first-order partial differential equation

\begin{equation} \label{eq:3.42} \sum\limits_{i,j=1}^n a^{ij} \frac{\partial \Gamma}{\partial x^i} \frac{\partial \Gamma}{\partial x^j} = 0 \end{equation}
for such a characteristic. The basic property involved is that any locus of singularities of a solution of a linear hyperbolic equation can be expected to form a characteristic surface \cite{garabedian1998partial}.

The geodesics that lie on the conoid $\Gamma=0$ are the bicharacteristics of the original equation ($\ref{eq:3.31}$). We have found that, along the characteristic conoid $\Gamma=0$, the ordinary differential operators occurring on the left in the transport equations ($\ref{eq:3.40}$) and ($\ref{eq:3.41}$) apply in the directions of the bicharacteristics. This happens because, within any of its characteristic surfaces, Eq. ($\ref{eq:3.31}$) reduces to an ordinary differential equation imposed on the Cauchy data along each bicharacteristic \cite{garabedian1998partial}.

To evaluate the functions $U_0$, $U_1$, ... it is convenient to work in a new space with coordinates $\theta_1$, ..., $\theta_n$ defined by \cite{ hadamard1952lectures, garabedian1998partial}
\begin{equation} \label{eq:3.43} \theta_i=s p_i(0). \end{equation}
It is possible to do so in a sufficiently small neighbourhood of the parameter point $\xi=(\xi^1, ..., \xi^n)$ because the relevant Jacobian does not vanish. In this new space the geodesics become rays emanating from the origin, while the parameter $s$ can be chosen to coincide with the distance from the origin along each such ray. Each coefficient $U_l$ in the expansion $U= U_l \Gamma^l$ can be written in the form of a series
\begin{equation} \label{eq:3.44} U_l = \sum\limits_{j=0}^\infty P_{lj} \end{equation}
of polynomials $P_{lj}$ homogeneous in the coordinates $ \theta_1, ..., \theta_n$ of degree equal to the index $j$. 

Note that the differential operator $s \frac{d}{ds}$ in Eqs. $(\ref{eq:3.40})$ and $(\ref{eq:3.41})$ does not alter the degree of any of the polynomials $P_{lj}$, with the exception that it reduces a polynomial of degree zero, i.e. a constant, to zero. Thus, unless the coefficient $(D-m-1)$ vanishes for $\theta_1 = ...= \theta_n = 0$, there does not exist a solution $U_0$ of Eq. $(\ref{eq:3.40})$ satisfying the requirement $P_{00} \neq 0$. However, we have chosen the exponent as in $(\ref{eq:3.33})$ precisely so that this will be the case, because our $D= \frac{n}{2}$ at the parameter point $x= \xi$. Thus, we can integrate Eq. $(\ref{eq:3.40})$ to find
\begin{equation} \label{eq:3.45} U_0 = P_{00} e^{- \int\limits_0^s (D-m-1) \frac{ d \tau}{\tau}}, \end{equation}
where $P_{00}$ is a constant as a function of $x$ that might vary with $\xi$.

Similarly, Eq. $(\ref{eq:3.41})$ may be solved by the recursion formula
\begin{equation} \label{eq:3.46} U_l = - \frac{U_0}{4(l-m)s^l} \int\limits_{0}^s \frac{\mathcal{F}[U_{l-1}]\tau^{l-1}}{U_0} d\tau, \hspace{1cm} l \geq 1. \end{equation}
The linear operator on the right turns any convergent series $U_{l-1}$ of the type $(\ref{eq:3.44})$ into another series of the same kind for $U_l$. At this stage, one has still to prove uniform convergence of the expansion of $U$ in powers of $\Gamma$, for sufficiently small values of $s$. This can be obtained by using the method of majorants.

For the purpose of proving convergence, it is sufficient to treat only the particular case $U_0= {\rm constant}$, because the substitution $u_1\equiv \frac{u}{U_0},$ with $U_0$ given by $(\ref{eq:3.45})$, reduces $(\ref{eq:3.31})$ to a new partial differential equation reading as
\begin{equation} \label{eq:3.47} \mathcal{F}_1[u_1] = \mathcal{F}[U_0u_1] =0, \end{equation}
for which such an assumption is verified.

Let $K$ and $\epsilon$ be positive numbers such that the geometric series
\begin{equation} \label{eq:3.48} \sum\limits_{j=0}^\infty \frac{K}{\epsilon^j} \bigg{(} |\theta_1| + ... + |\theta_n| \bigg{)}^j = \frac{ K \epsilon}{ \epsilon - |\theta_1 | - ... - | \theta_n |} \end{equation}
is a majorant for the Taylor expansions in powers of $\theta_1$, ..., $\theta_n$ of all the coefficients of $\mathcal{F}$, which is now a differential operator expressed in these new coordinates. Hence one finds that, if
\begin{equation} \label{eq:3.49} M \{U_l \} = \frac{ M_l}{ \bigg{(} 1 - \frac{|\theta_1| + ... + |\theta_n|}{\epsilon} \bigg{)}^{2l}} \end{equation}
denotes a majorant for $U_l$, with $M_l$ taken as a suitably large constant, then
\begin{equation} \label{eq:3.50} M \{ \mathcal{F} [U_l] \} = \frac{ 2l(2l +1) \bigg{[} 1 + \frac{n}{\epsilon} + \frac{n^2}{\epsilon^2} \bigg{]} K M_l}{ \bigg{(} 1- \frac{ |\theta_1| + ... + |\theta_n |}{\epsilon} \bigg{)}^{2l + 3}} \end{equation}
is a majorant for $\mathcal{F}[U_l]$. We now apply the recursion formula $(\ref{eq:3.46})$ to $(\ref{eq:3.50})$ in order to establish that when $l$ is replaced by $(l+1)$, and with
\begin{equation} \label{eq:3.51} M_{l+1} = \frac{ l(2l+1)}{2(l+1)(l-m+1)} \bigg{[} 1 + \frac{n}{\epsilon} + \frac{n^2}{\epsilon^2} \bigg{]} KM_l \end{equation}
the rule $(\ref{eq:3.49})$ also defines a majorant for $U_{l+1}$.

Since we have recognized that it is enough to consider the case $U_0= {\rm constant}$, the proof that we are interested in reduces to a verification that
\begin{equation} \label{eq:3.52} M \Biggl\{ s^{-l-1} \int\limits_{0}^s \frac{ \tau^l}{(1- \gamma \tau )^{2l+3} } d \tau \Biggr\} = \frac{1}{(l+1)} (1- \gamma s)^{-2l-2} \end{equation}
is a majorant for the integral inside curly brackets on the left. This can be proved with the help of the convenient choice
\begin{equation} \label{eq:3.53} M \Biggl\{ \frac{s^l}{(1- \gamma s)^{2l+3}} \Biggr\} = [ 1+ \gamma s ] \frac{ s^l}{(1- \gamma s)^{2l + 3}} = \frac{1}{(l+1)} \frac{d}{ds} \frac{ s^{l+1}}{(1- \gamma s)^{2l+2}} \end{equation}
of a majorant for the integrand. With this notation, see Garabedian \cite{garabedian1998partial}, $\gamma$ is specified by
\begin{equation} \label{eq:3.54} |\theta_1| + ... + |\theta_n| = \epsilon \gamma s. \end{equation}
By induction, we conclude that the majorants $(\ref{eq:3.49})$ are valid for all $l \geq 1$, provided that $M_1$ is sufficiently large and that $M_2$, $M_3$, ... are given by $(\ref{eq:3.51})$. Thus, the series for $U$ in powers of the world function $\Gamma$ converges in a neighbourhood specified by the upper bound
\begin{equation} \label{eq:3.55} | \Gamma | < \frac{ \bigg{(} 1 - \frac{ | \theta_1| + ... + |\theta_n| }{\epsilon} \bigg{)}^2 }{\bigg{[} 1 + \frac{n}{\epsilon} + \frac{ n^2}{\epsilon^2} \bigg{]} K } \end{equation}
of the parameter point $\xi = (\xi^1, ..., \xi^n)$.

To sum up, we obtain $\textit{locally}$ a fundamental solution $S$ of Eq. $(\ref{eq:3.31})$ having special form
\begin{equation} \label{eq:3.56} R= \frac{U}{\Gamma^m} = \sum\limits_{l=0}^{\infty} U_l \Gamma^{l-m}, \end{equation}
when the number $n$ of independent variables is odd. The addition of a regular term $W$ to the right-hand side of $(\ref{eq:3.56})$ is not mandatory.

\subsection{Even Number of Variables and Logarithmic Term}
When the number $n \geq 4$ of independent variables is even, the exponent $m$ defined in Eq. $(\ref{eq:3.33})$ is a positive integer and the previous construction of $U$ no longer holds, because the whole algorithm involves division by $(l-m)$, which vanishes when $l=m$. Only the functions $U_0$, $U_1$, ..., $U_{m-1}$ can then be obtained as previously seen. This is why a logarithmic term is needed in the formula $(\ref{eq:3.32})$ for the fundamental solution $R$ in a space with even number of dimensions. Hence we look for $R$ in the form
\begin{equation} \label{eq:3.57} R = \sum\limits_{l=0}^{m-1} U_{l}\Gamma^{l-m} + V log (\Gamma) + W. \end{equation}
If the formula $(\ref{eq:3.57})$ is inserted into the homogeneous equation $(\ref{eq:3.31})$, one finds
\begin{equation} \begin{split} \label{eq:3.58} &\sum\limits_{l=0}^{m-1} \mathcal{F} [U_l\Gamma^{l-m}] + \mathcal{F}[V log(\Gamma)] + \mathcal{F}[W] = \mathcal{F} [U_{m-1}] \frac{1}{\Gamma} - \frac{V}{\Gamma^2} \sum\limits_{i,j=1}^n a^{ij} \frac{ \partial \Gamma}{\partial x^i} \frac{\partial \Gamma}{\partial x^j} \\
& + \bigg{[} 2 \sum\limits_{i,j=1}^n a^{ij} \frac{ \partial V}{\partial x^i}\frac{\partial \Gamma}{\partial x^j} + 4DV \bigg{]} \frac{1}{\Gamma} + \mathcal{F} [V]log(\Gamma) + \mathcal{F} [W] =0, \end{split}\end{equation}
by virtue of equation $(\ref{eq:3.38})$ and of the transport equations $(\ref{eq:3.40})$ and $(\ref{eq:3.41})$. In Eq. $(\ref{eq:3.58})$ the term which is non-linear in the derivatives of $\Gamma$ is re-expressed from Eq. $(\ref{eq:3.18})$ (with $z^i$ therein written as $x^i$), and we arrive at
\begin{equation} \label{eq:3.59} \Bigg\{ \mathcal{F}[U_{m-1}] + 4 \bigg{[} s \frac{dV}{ds} + (D-1) V \bigg{]} \Biggr\} \frac{1}{\Gamma} + \mathcal{F}[V] log(\Gamma) + \mathcal{F}[W] =0. \end{equation}
We are now going to prove that this equation determines $V$ uniquely, whereas $W$ can be selected in a number of ways, in order to satisfy the requirements imposed on it.

We note that the ${\rm log}(\Gamma)$ is not balanced by other terms in Eq. $(\ref{eq:3.59})$, hence the function $V$ must solve the homogeneous equation
\begin{equation} \label{eq:3.60} \mathcal{F}[V] =0. \end{equation}
Moreover the coefficient of $\frac{1}{\Gamma}$ in Eq. $(\ref{eq:3.59})$ must vanish along the whole characteristic conoid $\Gamma=0$, since the remaining regular term $\mathcal{F}[W]$ cannot balance the effect of $\frac{1}{\Gamma}$. Thus, the function $V$ has to solve also the ordinary differential equation.
\begin{equation} \label{eq:3.61} \bigg{[} s \frac{d}{ds} + (D-1) \bigg{]} V = - \frac{1}{4} \mathcal{F} [U_{m-1}] \end{equation}
on each bicharacteristic that generates the conoid. From our study of the transport equations $(\ref{eq:3.40})$ and $(\ref{eq:3.41})$ we know that Eq. $(\ref{eq:3.61})$ determines the function $V$ uniquely on the characteristic surface $\Gamma=0$, and that $V$ must indeed coincide there with the function $V_0$ defined in a neighbourhood of the parameter point $\xi$ by the integral
\begin{equation} \label{eq:3.62} V_0 = - \frac{U_0}{4s^m} \int\limits_0^s \frac{ \mathcal{F}[U_{m-1}]\tau^{m-1}}{U_0} d\tau. \end{equation}
We have therefore formulated a characteristic initial-value problem for the partial differential equation $(\ref{eq:3.60})$, in which the unknown function $V$ is prescribed on the conoid $\Gamma=0$. This result agrees with our previous findings, according to which the coefficient of the logarithm is a Riemann kernel satisfying a characteristic initial-value problem.

From another point of view \cite{garabedian1998partial}, one can think of Eq. $(\ref{eq:3.62})$ as a substitute for the recursion formula $(\ref{eq:3.46})$ in the case $l=m$. This property suggests trying to find $V$ as a convergent power series
\begin{equation} \label{eq:3.63} V= \sum\limits_{l=0}^\infty V_l\Gamma^l. \end{equation}
Insertion of Eq. $(\ref{eq:3.63})$ into Eq. $(\ref{eq:3.60})$ leads to infinitely many powers of $\Gamma$ whose coefficients should all be set to zero. We do so, and integrate the resulting ordinary differential equations, finding therefore
\begin{equation} \label{eq:3.64} V_l =- \frac{U_0}{4ls^{l+m}} \int\limits_0^s \frac{ \mathcal{F}[V_{l-1}] \tau^{l+m-1}}{U_0} d\tau, \hspace{1cm} \forall l \geq 1. \end{equation}
The first term $V_0$ is instead obtained from Eq. $(\ref{eq:3.62})$. The method of majorants can be used to deduce estimates like $(\ref{eq:3.49})$ for the functions $V_l$ provided by $(\ref{eq:3.64})$. Thus, the series $(\ref{eq:3.63})$ converges uniformly in a region like $(\ref{eq:3.54})$ surrounding the parameter point $\xi$. Another result of this method consists on the fact that $(\ref{eq:3.59})$ becomes a partial differential equation for $W$ with a inhomogeneous term that is regular in the neighbourhood of $\xi$. This determines $W$ only up to the addition of an arbitrary solution of the homogeneous equation $(\ref{eq:3.31})$. A particular choice for $W$ that agrees with the method used so far demands that
\begin{equation} \label{eq:3.65} W = \sum\limits_{l=1}^\infty W_l \Gamma^l, \end{equation}
where the coefficient functions $W_1$, $W_2$, ... will be found from a recursive scheme like $(\ref{eq:3.64})$. By requiring that the series for $W$ should not include the term $W_0$ corresponding to the value $l=0$, one obtains a unique determination of the fundamental solution $(\ref{eq:3.57})$ such that the functions $U_0$, ..., $U_{m-1}$, $V$ and $W$ are all regular as functions of the parameter point $\xi$. 

The limitation of the Hadamard approach described so far is that it yields the fundamental solution only locally, i.e. in a sufficiently small neighbourhood of $\xi$. Furthermore, when the inverse metric components $a^{ij}$ are varying, also the world function $\Gamma$ is defined only in the small.

\subsection{Example of Fundamental Solution: Scalar Wave Equation with Smooth Initial Conditions }
Following Sobolev \cite{sobolev1963applications}, we study the wave operator
\begin{equation} \label{eq:3.66} \Box u = \bigg{(} \Delta - \frac{\partial^2 }{\partial t^2} \bigg{)} u \end{equation}
on the domain $\Omega$ of the $(n+1)$-dimensional space of coordinates $x_1$, $x_2$, ..., $x_n$, $t$ limited by a smooth surface $S$. Let $u(x_1, x_2, ..., x_n, t)$ and $v(x_1, x_2, ..., x_n, t)$ be twice differentiable in $\Omega$ with all their first derivatives continuous on the surface $S$. Thus, we consider
\begin{equation} \label{eq:3.67} \Box u = f; \hspace{3cm} \Box v = \varphi \end{equation}
and the integral
\begin{equation}\begin{split} \label{eq:3.68} J = & \underbrace{ \int \dots \int  }_S   \Biggl\{ - \sum\limits_{i=1}^n \bigg{(} \frac{\partial u}{\partial x_i}\frac{\partial v}{\partial t} + \frac{\partial v}{\partial x_i}\frac{\partial u}{\partial t} \bigg{)} cos( \vec{n} x_i )\\
& + \bigg{(} \frac{\partial u}{\partial t}\frac{\partial v}{\partial t} + \sum\limits_{i=1}^n \frac{\partial u}{\partial x_i}\frac{\partial v}{\partial x_i}\bigg{)} cos (\vec{n}t)\Bigg\}  dS \end{split} \end{equation}
where $\vec{n}$ is the inward-pointing normal to $S$.

A simple transformation leads to 
\begin{equation} \begin{split} \label{eq:3.69}& J= - \underbrace{ \int \dots \int }_\Omega \Biggl\{ \frac{\partial}{\partial t} \bigg{(} \frac{\partial u}{\partial t} \frac{\partial v}{\partial t} + \sum\limits_{i=1}^n \frac{\partial u}{\partial x_i} \frac{\partial v}{\partial x_i} \bigg{)} - \sum\limits_{i=1}^n \frac{\partial}{\partial x_i} \bigg{(} \frac{\partial u}{\partial x_i} \frac{\partial v}{\partial t} + \frac{\partial v}{\partial x_i}\frac{\partial u}{\partial t} \bigg{)} \Biggr\} d \Omega \\
& = -  \underbrace{ \int \dots \int }_\Omega \Biggl\{ \frac{\partial u}{\partial t} \bigg{[} \frac{\partial^2 v}{\partial t^2} - \sum\limits_{i=1}^n \frac{\partial^2 v}{\partial {x_i}^2} \bigg{]} + \frac{\partial v}{\partial t} \bigg{[} \frac{\partial^2 u}{\partial t^2} - \sum\limits_{i=1}^n \frac{\partial^2 u}{\partial t^2} \bigg{]} \Biggr\} d \Omega \\
&=  \underbrace{ \int \dots \int }_\Omega \Biggl\{ \frac{\partial u}{\partial t}\Box v + \frac{\partial v}{\partial t} \Box u \Biggr\} d\Omega  \end{split}\end{equation}
Replacing $\Box u$ and $\Box v$ with their values, we have
\begin{equation} \label{eq:3.70} J =  \underbrace{ \int \dots \int }_\Omega \bigg{(} \frac{\partial u}{\partial t} \varphi + \frac{\partial v}{\partial t} f \bigg{)} d \Omega. \end{equation}
Let us now consider the expression
\begin{equation} \begin{split} \label{eq:3.71} & \sum\limits_{i=1}^n \bigg{(} \frac{\partial u}{\partial x_i} cos( \vec{n} t) - \frac{\partial u}{\partial t} cos(\vec{n} x_i) \bigg{)} \bigg{(} \frac{\partial v}{\partial x_i} cos( \vec{n}t) - \frac{\partial v}{\partial t} cos( \vec{n}x_i) \bigg{)}\\
& = \frac{\partial u}{\partial t} \frac{\partial v}{\partial t} \sum\limits_{i=1}^n (cos(\vec{n} x_i))^2  + (cos(\vec{n} t))^2\sum\limits_{i=1}^n \frac{\partial u}{\partial x_i}\frac{\partial v}{\partial x_i} \\
& - \sum\limits_{i=1}^n \bigg{(} \frac{\partial u}{\partial t} \frac{\partial v}{\partial x_i} - \frac{\partial v}{\partial t} \frac{\partial u}{\partial x_i} \bigg{)} cos(\vec{n} x_i) cos(\vec{n}t) = cos(\vec{n}t) \bigg{[} \frac{\partial u}{\partial t}\frac{\partial v}{\partial t} cos(\vec{n}t) \\
& + \sum\limits_{i=1}^n \frac{\partial u}{\partial x_i}\frac{\partial v}{\partial x_i} cos(\vec{n} t)  - \sum\limits_{i=1}^n \bigg{(} \frac{\partial u}{\partial t}\frac{\partial v}{\partial x_i} + \frac{\partial v}{\partial t} \frac{\partial u}{\partial x_i} \bigg{)} cos(\vec{n} x_i) \bigg{]} \\
& + \frac{\partial u}{\partial t} \frac{\partial v}{\partial t} \Biggl\{ \sum\limits_{i=1}^n (cos(\vec{n}x_i))^2 - (cos(\vec{n}t))^2 \Biggr\}
 \end{split}\end{equation}
Everywhere, except at the points of the surface $S$ where $cos(\vec{n}t)=0$, we have
\begin{equation} \begin{split} \label{eq:3.72} & J= \underbrace{ \int \dots \int }_S \bigg{[} \frac{1}{cos(\vec{n} t)} \sum\limits_{i=1}^n \bigg{(} \frac{\partial u}{\partial x_i} cos( \vec{n} t) - \frac{\partial u}{\partial t} cos(\vec{n} x_i) \bigg{)} \\
& \times \bigg{(} \frac{\partial v}{\partial x_i} cos(\vec{n}t) - \frac{\partial v}{\partial t} cos(\vec{n}x_i) \bigg{)} - \frac{\partial u}{\partial t} \frac{\partial v}{\partial t} \frac{1 - 2(cos(\vec{n}t))^2 }{cos(\vec{n}t)} \bigg{]} dS \\
&= \underbrace{ \int \dots \int }_S \Phi dS \end{split}\end{equation}
where $\Phi$  is the integer of $J$. If $u=v$, then for all the points of $S$ where $|cos(\vec{n}t) | \geq \frac{1}{\sqrt{2}}$, the previous expression always assumes the same sign, that is the same sign of $cos(\vec{n}t)$.
\begin{equation} \label{eq:3.73} sign \Phi = sign (cos(\vec{n}t)). \end{equation}
Whereas, if $cos(\vec{n}t) =0$, $\Phi$ becomes
\begin{equation} \label{eq:3.74} \Phi = - \bigg{(} \frac{\partial u}{\partial t} \frac{\partial v}{\partial \vec{n}} + \frac{\partial v}{\partial t}\frac{\partial u}{\partial \vec{n}} \bigg{)} \end{equation}

Now, let us suppose that $u$ is the solution of the wave equation
\begin{equation} \label{eq:3.75} \Box u =0 \end{equation}
in an infinite homogeneous medium.

By taking $u=v$ and making use of Eq. ($\ref{eq:3.69}$), where $\Omega$ it is the truncated cone whose generators form an angle of $\frac{\pi}{4}$ degrees with the axis $0t$ (as shown in fig ($\ref{fig:4}$)). Thus
\begin{equation} \label{eq:3.76} cos(\vec{n} t) =
\left\{\begin{array} {l}
- \frac{1}{\sqrt{2}} \; \; \rm{in} \; S_1 \\
- 1 \; \; \; \rm{in} \; S_2 \\
+ 1 \; \; \; \rm{in} \; S_3
\end{array}\right.
\end{equation}

\begin{figure}
\centering

\includegraphics{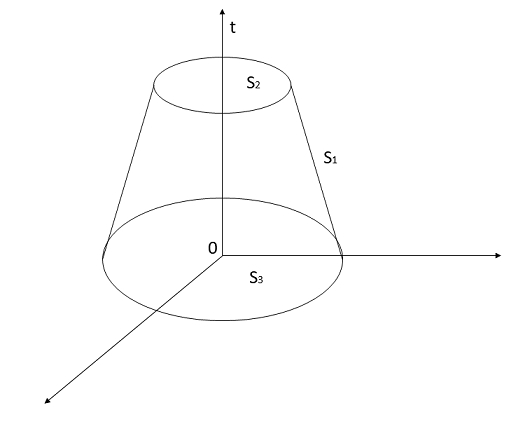}

\caption{}\label{fig:4}

\end{figure}

where $S_3$ is the lower base, $S_2$ is the upper base and $S_1$ is the lateral surface of the truncated cone. Let us suppose that $t$ on $S_2$ is equal to $t_0$. Making use of Eq. ($\ref{eq:3.69}$) and $\Box u =0$, we obtain
\begin{equation} \begin{split} \label{eq:3.77} J(u,u) =&   \underbrace{ \int \dots \int }_{S_1} \Phi dS +  \underbrace{ \int \dots \int }_{S_2} \Phi dS +  \underbrace{ \int \dots \int }_{S_3} \Biggl\{ - i \sum\limits_{i=1}^\infty \frac{\partial u}{\partial x_i} \frac{\partial u}{\partial t} cos(\vec{n} x_i) \\
& + \bigg{[} \bigg{(} \frac{\partial u}{\partial t} \bigg{)}^2 + \sum\limits_{i=1}^n \bigg{(} \frac{\partial u}{\partial x_i} \bigg{)}^2 \bigg{]} cos( \vec{n} t) \Biggr\} dS =0 \end{split} \end{equation}
By recalling the ($\ref{eq:3.73}$), we have
\begin{equation} \label{eq:3.78}    \underbrace{ \int \dots \int }_{S_1} \Phi dS < 0. \end{equation}
Therefore, the Eq. ($\ref{eq:3.77}$) reads as
\begin{equation} \label{eq:3.79} - \underbrace{ \int \dots \int }_{S_2} \Phi dS = \underbrace{ \int \dots \int }_{S_3} \Phi dS + \underbrace{ \int \dots \int }_{S_1} \Phi dS <  \underbrace{ \int \dots \int }_{S_3} \Phi dS. \end{equation}
Since on $S_2$ and $S_3$ we have $cos(\vec{n} x_i)=0$, it is necessary to evaluate
\begin{equation} \label{eq:3.80}  \underbrace{ \int \dots \int }_{S_2} \bigg{[} \sum\limits_{i=1}^n \bigg{(} \frac{\partial u}{\partial x_i} \bigg{)}^2 + \bigg{(} \frac{\partial u}{\partial t} \bigg{)}^2 \bigg{]} dS \leq  \underbrace{ \int \dots \int }_{S_3} \bigg{[} \sum\limits_{i=1}^n \bigg{(} \frac{\partial u}{\partial x_i} \bigg{)}^2 + \bigg{(} \frac{\partial u}{\partial t} \bigg{)}^2 \bigg{]} dS. \end{equation}
For this purpose, it is also necessary to evaluate
\begin{equation} \label{eq:3.81}  \underbrace{ \int \dots \int }_{S_2} u^2 dS. \end{equation}
From Eq. ($\ref{eq:3.80}$) it follows that $\underbrace{ \int \dots \int }_{S_2} \big{(} \frac{\partial u}{\partial t} \big{)}^2 dS$ is bounded. If we denote by $y(t)$ the quantity
\begin{equation} \label{eq:3.82} y(t) = \underbrace{ \int \dots \int }_{\Sigma_t} u^2 dS, \end{equation}
where $\Sigma_t$ is the surface on which the coordinates $x_1$, $x_2$, ..., $x_n$ assume the same values as above $S_2$, whereas $t$ goes from 0 to $t_0$, we have
\begin{equation} \label{eq:3.83} y'(t) = 2 \underbrace{ \int \dots \int }_{\Sigma_t} u(t) \frac{\partial u}{\partial t} dS. \end{equation}
Making use of Cauchy-Bunjakovsky inequality, it follows that 
\begin{equation} \label{eq:3.84} | y'(t)| \leq 2 \bigg{[}  \underbrace{ \int \dots \int }_{\Sigma_t} \bigg{(} \frac{\partial u}{\partial t} \bigg{)}^2 dS \bigg{]}^{\frac{1}{2}}  \bigg{[} \underbrace{ \int \dots \int }_{\Sigma_t} u^2(t) dS \bigg{]}^{\frac{1}{2}} \end{equation}
and by virtue of the inequality ($\ref{eq:3.80}$), we have
\begin{equation} \label{eq:3.85} |y'(t)| \leq 2 A (y(t))^{\frac{1}{2}}, \end{equation} 
where A is given by
\begin{equation} \label{eq:3.86} A = \bigg{[}  \underbrace{ \int \dots \int }_{S_2} \Biggl\{ \sum\limits_{i=1}^n \bigg{(} \frac{\partial u}{\partial x_i} \bigg{)}^2 +  \bigg{(} \frac{\partial u}{\partial t} \bigg{)}^2 \Biggr\} dS_3 \bigg{]}^{\frac{1}{2}}. \end{equation}
The inequality ($\ref{eq:3.85}$) implies that
\begin{equation} \label{eq:3.87} \frac{1}{2} \frac{dy}{\sqrt{y}} \leq A dt \; \rightarrow \; \frac{d}{dt} \sqrt{y} \leq A. \end{equation}
Similary we obtain $\sqrt{y_1} \leq \sqrt{y_0} + At$, from which eventually we have
\begin{equation} \label{eq:3.88} y \leq y_0 + 2 A \sqrt{y_0} t + A^2 t^2; \end{equation} 
if we set $ \underbrace{ \int \dots \int }_{S_3} u^2 dS = B^2 $ it follows that $y_0 \leq B^2$ and then
\begin{equation} \label{eq:3.89} y \leq (B + At)^2. \end {equation}
Now, by intersecting our truncated cone with the plane $t= {\rm const.}$, where $\Sigma_t$ is the $n$-dimensional space domain resulting from that intersection, and applying the same procedure that we have previously shown, we have
\begin{equation}\begin{split} \label{eq:3.90} & \overbrace{\underbrace{ \int \dots \int }_{\Sigma_t}}^n \Biggl\{ \sum\limits_{i=1}^n \bigg{(} \frac{\partial u}{\partial x_i} \bigg{)}^2 + \bigg{(} \frac{\partial u}{\partial t} \bigg{)}^2 \Biggr\} dS \\
& \leq  \underbrace{ \int \dots \int }_{S_3} \Biggl\{ \sum\limits_{i=1}^n \bigg{(} \frac{\partial u}{\partial x_i} \bigg{)}^2 + \bigg{(} \frac{\partial u}{\partial t} \bigg{)}^2 \Biggr\}dS. \end{split} \end{equation}
If we integrate Eq. ($\ref{eq:3.90}$) over $t$ from 0 to $t_0$, it reads as
\begin{equation} \begin{split} \label{eq:3.91} & \overbrace{\underbrace{ \int \dots \int }_{V}}^{n+1} \Biggl\{ \sum\limits_{i=1}^n \bigg{(} \frac{\partial u}{\partial x_i} \bigg{)}^2 + \bigg{(} \frac{\partial u}{\partial t} \bigg{)}^2 \Biggr\}dV \leq  t_0 \underbrace{ \int \dots \int }_{S_3} \Biggl\{ \sum\limits_{i=1}^n \bigg{(} \frac{\partial u}{\partial x_i} \bigg{)}^2 \\
&+ \bigg{(} \frac{\partial u}{\partial t} \bigg{)}^2 \Biggr\}dS \leq A^2 t_0, \end{split} \end{equation}
where $V$ is the truncated cone. In the same manner, making use of 
\begin{equation} \label{eq:3.92} y =   \underbrace{ \int \dots \int }_{\Sigma_t}u^2 dS \leq (B+ A^2 t)^2 \end{equation}
and by integration over $t$ from 0 to $t_0$, we have
\begin{equation} \label{eq:3.93}  \overbrace{ \underbrace{ \int \dots \int }_{V}}^{n+1} u^2 dV \leq \frac{1}{A} [ (B+ At_0)^3 - B^3 ] = 3B^2 t_0 + 3AB {t_0}^2+A^2 {t_0}^3. \end{equation}
The inequality ($\ref{eq:3.92}$) has two important corollaries.
\begin{cor}$\\$
Suppose that the initial values of $u$ and $\frac{\partial u}{\partial t}$ are on $S_3$. This implies that $A=B=0$ and as a consequence of $(\ref{eq:3.92}$) we have $y=0$, i.e. $u=0$ in $V$. Hence, if on the base of the truncated cone $u=0$ and $\frac{\partial u}{\partial t}=0$, then $u=0$ inside this cone.
\end{cor} 
\begin{cor} $\\$
The value of the function $u$, that is a solution of the given equation, at a point ${x_1}^0$, ..., ${x_n}^0$, $t^0$, is given by the initial values of $u$ and $\frac{\partial u}{\partial t}$ on the sphere $\eta = \bigg{(} \sum\limits_{i=1}^n (x_i - {x_i}^0)^2 \bigg{)}^{\frac{1}{2}} \leq t_0$, that is the intersection of the characteristic cone with vertex in every point with the plane $t=0$.
\end{cor}
In fact, if for every two solutions of the wave equation the initial data of $u$ and $\frac{\partial u}{\partial t}$ coincide on this domain, the data of their difference will vanish on this domain, and from the corollary 3.3.1, the difference will be null on the vertex of the cone. Thus, at the top of the cone, the two solutions will coincide.
\begin{thm}$\\$ \label{thm:2}
Let $u$ be a solution to the homogeneous wave equation. If the initial values $u|_{t=0}$ and ${\frac{\partial u}{\partial t}}|_{t=0}$ are infinitely differentiable on the whole space of the $x_1$, .., $x_n$, then the function $u$ itself has all its derivatives up to every order.
\end{thm} 

To begin with, we will estabilish this theorem in a more general situation. Hence we will demonstate the lemma
\begin{lem} $\\$
Let $u$ be the function that satisfies the equation
\begin{equation*} \Delta u - \frac{\partial^2 u}{\partial t^2} = 0 \end{equation*}
on the domain $-a \leq x \leq a$ , for $i=1, 2, ..., n$, where $a$ is a constant, and let us suppose that
\begin{equation} \label{eq:3.94} u|_{x_i} = \pm a =0, \end{equation}
i.e. $u$ vanishes on the boundary of this domain, and that, at $t=0$, we have
\begin{equation} \label{eq:3.95} u|_{t=0} = \varphi_0 \hspace{0.5cm} \rm{and} \hspace{0.5cm} {\frac{\partial u}{\partial t}}\bigg{|}_{t=0} = \varphi_1 \end{equation}
where the functions $\varphi_0$ and $\varphi_1$ have their derivatives continuous up to every order and they, together with their derivatives, vanish on the boundary of this domain. Hence, $u$ has continuous derivatives up to every order.
\end{lem}
\proof
In this case, the solution can be written in explicit form by making use of the Fourier series. Thus, we expand in Fourier series $\varphi_0$ and $\varphi_1$
\begin{equation} \label{eq:3.96} \varphi_0 = \sum\limits_{j_k =1}^\infty b_{j_1, j_2, ..., j_n} sin \bigg{(} j_1 \frac{(x_1 + a) \pi}{2a} \bigg{)} \dots \; sin \bigg{(} j_n \frac{ (x_n + a) \pi}{2a} \bigg{)}, \end{equation}
\begin{equation} \label{eq:3.97} \varphi_1 = \sum\limits_{j_k =1}^\infty g_{j_1, j_2, ..., j_n} sin \bigg{(} j_1 \frac{(x_1 + a) \pi}{2a} \bigg{)} \dots \; sin \bigg{(} j_n \frac{ (x_n + a) \pi}{2a} \bigg{)}. \end{equation}
The $\varphi_0$ and $\varphi_1$ functions are continuous with their derivatives and they can be extended periodically to the whole space preserving the continuity of all their derivatives. It follows that the Fourier series of all these functions will converge uniformly with all their derivatives of arbitrary order.

It is possible to consider the partial sum of these series, which reads as
\begin{equation} \label{eq:3.98} { \varphi_0}^{(N)} = \sum\limits_{j_k =1}^N b_{j_1, j_2, ..., j_N} sin \bigg{(} j_1 \frac{(x_1 + a) \pi}{2a} \bigg{)} \dots \; sin \bigg{(} j_n \frac{ (x_n + a) \pi}{2a} \bigg{)}, \end{equation}
\begin{equation} \label{eq:3.99} { \varphi_1}^{(N)} = \sum\limits_{j_k =1}^N g_{j_1, j_2, ..., j_N} sin \bigg{(} j_1 \frac{(x_1 + a) \pi}{2a} \bigg{)} \dots \; sin \bigg{(} j_n \frac{ (x_n + a) \pi}{2a} \bigg{)}. \end{equation}
If we replace in the initial data $\varphi_0$ and $\varphi_1$ with ${\varphi_0}^{(N)}$ and ${\varphi_1}^{(N)}$, we obtain as a solution of the wave equation with the previous initial data the function
\begin{equation} \label{eq:3.100} \begin{split} &u^{(N)} = \sum\limits_{j_k =1}^N \Biggl\{ b_{j_1, j_2, ..., j_N} cos \bigg{(}\sqrt{ {j_1}^2 + {j_2}^2 + ... {j_N}^2 }\bigg{)} + \frac{g_{j_1, j_2, ..., j_N}}{\sqrt{{j_1}^2 +{j_2}^2 + ... + {j_N}^2}} \\
& \times sin \bigg{(} \sqrt{{j_1}^2 +{j_2}^2 + ... + {j_N}^2 }\bigg{)} \Biggr\}  sin \bigg{(} j_1 \frac{(x_1 + a) \pi}{2a} \bigg{)} \dots \; sin \bigg{(} j_N \frac{ (x_N + a) \pi}{2a} \bigg{)} \end{split} \end{equation}
This solution is infinitely differentiable. We have to show that, with the increase of $N$, $u^{(N)}$ converges to every Sobolev space ${W_2}^{(l)}$, where $l$ is an arbitrary number of some function $u$ (see Appendix A). It follows from this that the limit function $u$ is a solution of the wave equation that satisfies the initial conditions $u|_{t=0} = \varphi_0$ and ${\frac{\partial u}{\partial t}}|_{t=0} =\varphi_1 $ and it is infinitely differentiable. Since this solution is unique the lemma is shown. It is left to prove the convergence of $u^{(N)}$.

Let us apply the ($\ref{eq:3.85}$) to the parallelepiped, in fig. ($\ref{fig:5}$), whose base $S_2$ is the domain $\Omega: |x_i| \leq a$, on the plane $t=0$, and its upper base $S_3$ lies on the plane $t=t_0$. Since on the lateral surface $S_1$ of this domain $(|x_i|=a)$ we have $ u^{(N)}= \frac{\partial {u^{(N)}}}{\partial t} =0 $, then
\begin{figure}
\centering

\includegraphics{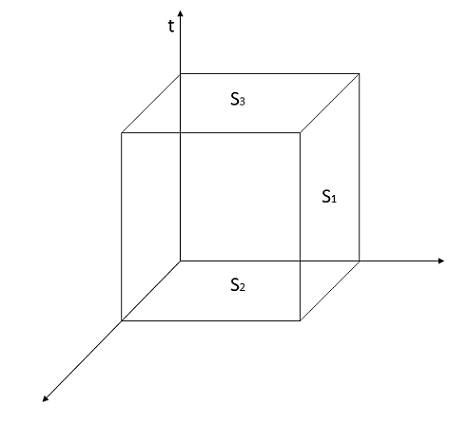}

\caption{}\label{fig:5}

\end{figure}
\begin{equation} \label{eq:301} \underbrace{\int \dots \int}_{S_1}  \frac{\partial u^{(N)}}{\partial t}\frac{\partial u^{(N)}}{\partial n} dS_1=0, \end{equation}
and hence
\begin{equation} \begin{split} & \label{eq:3.102} \underbrace{\int \dots \int}_{S_2} \bigg{[} \sum\limits_{i=1}^N \bigg{(} \frac{\partial u^{(N)}}{\partial x_i} \bigg{)}^2 + \bigg{(}\frac{\partial u^{(N)}}{\partial t}\bigg{)}^2 \bigg{]} dS_2 = \underbrace{\int \dots \int}_{S_3} \bigg{[} \sum\limits_{i=1}^N \bigg{(} \frac{\partial u^{(N)}}{\partial x_i} \bigg{)}^2  \\
& + \bigg{(}\frac{\partial u^{(N)}}{\partial t}\bigg{)}^2 \bigg{]} dS_3. \end{split} \end{equation}
At this stage, we consider the functions
\begin{equation} \label{eq:3.103} v^{(N)}_{a_1 ... a_n} = \frac{\partial^{\alpha} u^{(N)}}{\partial {x_1}^{\alpha_1}... \partial {x_N}^{\alpha_N}}, \end{equation}
which are solutions of the wave equation. 

We also note that on the boundary of the parallelepiped, $|x_i|=a$, the functions $ v^{(N)}_{a_1 ... a_n}$ satisfy the condition $v^{(N)}_{a_1 ... a_n} =0$, if $\alpha_j$ is even, or $\frac{\partial  v^{(N)}_{a_1 ... a_n}}{\partial n} =0$, if $\alpha_j$ is odd. If we apply ($\ref{eq:3.85}$) to these functions, we have
\begin{equation} \begin{split} \label{eq:3.104} & \underbrace{\int \dots \int}_{S_2} \bigg{[} \sum\limits_{i=1}^N \bigg{(} \frac{\partial v^{(N)}_{a_1 ... a_n} }{\partial x_i} \bigg{)}^2 + \bigg{(} \frac{ \partial v^{(N)}_{a_1 ... a_n}}{\partial t} \bigg{)}^2 \bigg{]} dS_2 =  \underbrace{\int \dots \int}_{S_3}  \bigg{[} \sum\limits_{i=1}^N \bigg{(} \frac{\partial v^{(N)}_{a_1 ... a_n} }{\partial x_i} \bigg{)}^2 \\
&+ \bigg{(} \frac{ \partial v^{(N)}_{a_1 ... a_n}}{\partial t} \bigg{)}^2 \bigg{]} dS_3 \end{split} \end{equation}
On the initial plane $S_2$, all integrals for given $\alpha_1$, $\alpha_2$, ..., $\alpha_N$, are bounded numbers that do not depend on $N$. 

We also consider the functions 
\begin{equation} \label{eq:3.105} {\omega^{(k,r)}}_{\alpha_1 ... \alpha_N} ={ v^{(k)}}_{\alpha_1 ... \alpha_N} - { v^{(r)}}_{\alpha_1 ... \alpha_N}. \end{equation}
For these functions we obtain, as before, that
\begin{equation} \begin{split} \label{eq:3.106} & \underbrace{\int \dots \int}_{S_3} \bigg{[} \sum\limits_{i=1}^N \bigg{(} \frac{\partial \omega^{(k,r)}_{a_1 ... a_n} }{\partial x_i} \bigg{)}^2 + \bigg{(} \frac{ \partial \omega^{(k,r)}_{a_1 ... a_n}}{\partial t} \bigg{)}^2 \bigg{]} dS_3 =  \underbrace{\int \dots \int}_{S_3}  \bigg{[} \sum\limits_{i=1}^N \bigg{(} \frac{\partial \omega^{(k,r)}_{a_1 ... a_n} }{\partial x_i} \bigg{)}^2 \\
&+ \bigg{(} \frac{ \partial \omega^{(k,r)}_{a_1 ... a_n}}{\partial t} \bigg{)}^2 \bigg{]} dS_2. \end{split} \end{equation}
For $k$ and $r$ sufficiently large, the integral on the right-hand side will be small enough. This immediately follows from the convergence with all its derivatives of the Fourier series for $\varphi_0$ and $\varphi_1$. This implies that the quantity on the left-hand side will be arbitrarily small.
\begin{equation} \label{eq:3.107} \underbrace{\int \dots \int}_{S_3} \bigg{[} \sum\limits_{i=1}^N \bigg{(} \frac{\partial \omega^{(k,r)}_{a_1 ... a_n} }{\partial x_i} \bigg{)}^2 + \bigg{(} \frac{ \partial \omega^{(k,r)}_{a_1 ... a_n}}{\partial t} \bigg{)}^2 \bigg{]} dS_3 < \epsilon.  \end{equation}
By Integration of this latter inequality with respect to the variable $t$ from 0 to $T$, it follows that
\begin{equation} \label{eq:3.108} \underbrace{\int \dots \int}_{\Omega} \bigg{[} \sum\limits_{i=1}^N \bigg{(} \frac{\partial \omega^{(k,r)}_{a_1 ... a_n} }{\partial x_i} \bigg{)}^2 + \bigg{(} \frac{ \partial \omega^{(k,r)}_{a_1 ... a_n}}{\partial t} \bigg{)}^2 \bigg{]} d\Omega < T\epsilon,  \end{equation}
where $\Omega$ is the domain $0 \leq t \leq T$; $0 \leq x_i \leq a$. Now, with a procedure analogous to that we have used to obtain ($\ref{eq:3.93}$), it is possible to prove that
\begin{equation} \label{eq:3.109} \underbrace{\int \dots \int}_{\Omega} \bigg{(} \omega^{(k,r)}_{a_1 ... a_n} \bigg{)}^2 d\Omega \leq (\epsilon_0 + \epsilon_1 T )^2 \leq \epsilon.  \end{equation}
By virtue of the completeness of the space ${W_2}^l$, it is possible to conclude that $v^{(N)}_{a_1 ... a_n}$, which satisfy the Cauchy convergence criterion, must converge in this space. The convergence of all derivatives of $u^{(N)}$ in ${W_2}^l$ implies the uniform convergence of all the derivatives of these functions. \endproof

Now, by making use of this lemma, we can prove the theorem.
\proof The values of the unknown functions within the piramid are
\begin{equation*} 0 \leq |x_i| + t \leq \frac{a}{2} \end{equation*}
and they depend only on the values of $\varphi_0$ and $\varphi_1$ inside the domain $0 \leq x_1 \leq \frac{a}{2}$, at $t=0$.
Thus, we build the functions
\begin{equation} \label{eq:3.110} {\varphi_0}^{(a)} = \varphi_0 \prod\limits_{i=1}^n \psi \bigg{(} \bigg{|} \frac{x_i}{a} \bigg{|} \bigg{)}; \end{equation}
\begin{equation} \label{eq:3.111} {\varphi_1}^{(a)} = \varphi_1 \prod\limits_{i=1}^n \psi \bigg{(} \bigg{|} \frac{x_i}{a} \bigg{|} \bigg{)}; \end{equation}
where $\psi(\xi)$ is a function equal to 1 when $\xi < \frac{1}{2}$ and to 0 when $\xi >1$ and it is infinitely differentiable. Our aim is to find a solution $u^a$ of the wave equation that satisfies
\begin{equation} \label{eq:3.112} {u^a}|_{t=0} = {\varphi_0}^a, \hspace{1cm} {\frac{\partial u^a}{\partial t}}\bigg{|}_{t=0} = {\varphi_1}^a. \end{equation}
We note that the previously demonstrated lemma shows that $u^a$ is infinitely differentiable, but, as we have seen before, on the piramid $0 \leq |x_i| + t \leq \frac{a}{2}$, this solution coincides with $u$. Hence, $u$ will be, in this case, infinitely differentiable. \endproof

At this stage, we aim to find a solution to the wave equation $\Box u = \Delta u - \frac{\partial^2 u}{\partial t^2} =0$, on the whole space that satifies the initial conditions
\begin{equation} \label{eq:3.113} {u}|_{t=0} = {u_0}, \hspace{1cm} {\frac{\partial u}{\partial t}}\bigg{|}_{t=0}= {u_1}. \end{equation}
It has been shown that if $u_0$ and $u_1$ have derivatives of every order, the solution of the problem exists and it is infinitely differentiable. There is no need for the infinite differentiation of data to obtain solutions, especially since the equation involves only second-order derivatives. To solve our problem we first propose to determine which conditions imposed on $u_0$ and $u_1$ ensure the existence of doubly differentiable solutions. 

Thus, let $u(t, x_1, ..., x_n)$ be a summable function on a domain $\Omega$ of the $(n+1)$-dimensional space. If there exists a summable function $f(x_1, ..., x_n, t)$ such that
\begin{equation} \label{eq:3.114} \overbrace{\underbrace{\int \dots \int}_{\Omega}}^{n+1} u \; \Box \psi dV = \underbrace{\int \dots \int}_{\Omega} \psi f dV \end{equation}
for each twice differentiable $\psi(t, x_1, ..., x_n)$ function that vanishes out of some closed subset of $\Omega$, then $f$ takes the name of $\textit{Generalized Wave Operator}$ of $u$ and we will write $\Box u = f$. A function that has a generalized wave operator equal to zero will take the name of a$\textit{ generalized solution}$ of the wave equation.

\begin{thm} $\\$
If $u_0$ has generalized derivatives up to the $ \big{(} \frac{n}{2} + 3 \big{)} $ -order square integrable on each bounded domain and $u_1$ has similar generalized derivatives up to the $\big{(} \frac{n}{2}+ 2 \big{)}$ -order, then the equation $\Delta u - \frac{\partial^2 u}{\partial t^2}=0$ has a doubly differentiable solution that satisfies the conditions
\begin{equation*} {u}|_{t=0} = {\varphi_0}; \hspace{1cm} {\frac{\partial u}{\partial t}}\bigg{|}_{t=0}= {\varphi_1}. \end{equation*}
\end{thm}
\proof We build the sequence of average functions $\{ u_{0h} \}$ and $\{u_{1h} \}$. Using the theorem for the solution of the wave equation with smooth initial conditions, there exist solutions of the equation $\Box u=0$ which satify the initial conditions
\begin{equation} \label{eq:3.115}  {u}|_{t=0} = {u_{0h}}; \hspace{1cm} {\frac{\partial u}{\partial t}}\bigg{|}_{t=0}= {u_{1h}} \end{equation}
and having derivatives of arbitrary order.

Let us consider the function $v_{p,q} = v_{hp} - v_{hq}$; $v_{p,q}$ is a solution of $\Box u=0$, which satisfies the conditions
\begin{equation*}  {v_{p,q}}|_{t=0} = {u_{0h_p}} - u_{0h_p}; \; \; \; {\frac{\partial v_{p,q}}{\partial t}}\bigg{|}_{t=0}= {u_{1h_p} - u_{1h_q}}. \end{equation*}
From inequality ($\ref{eq:3.80}$), which refers to fig. ($\ref{fig:4}$), we have
\begin{equation} \begin{split} \label{eq:3.116}& \underbrace{\int \dots \int}_{S_2} \bigg{[} \sum\limits_{i=1}^n \bigg{(} \frac{\partial v_{p,q}}{\partial x_i} \bigg{)}^2 + \bigg{(} \frac{\partial v_{p,q}}{\partial t} \bigg{)}^2 \bigg{]} dS \leq \underbrace{\int \dots \int}_{S_3} \bigg{[} \sum\limits_{i=1}^n \bigg{(} \frac{\partial v_{p,q}}{\partial x_i} \bigg{)}^2 \\
& + \bigg{(} \frac{\partial v_{p,q}}{\partial t} \bigg{)}^2 \bigg{]} dS, \end{split} \end{equation}
and similarly for each derivative of $v_{p,q}$ we have
\begin{equation} \begin{split} \label{eq:3.117}& \underbrace{\int \dots \int}_{S_2} \bigg{[} \sum\limits_{i=1}^n \bigg{(} \frac{\partial}{\partial x_i}\frac{\partial^\alpha v_{p,q}}{\partial {x_1}^{\alpha_1} \dots \partial {x_n}^{\alpha_n}} \bigg{)}^2 + \bigg{(} \frac{\partial}{\partial t}\frac{\partial^\alpha v_{p,q}}{\partial {x_1}^{\alpha_1}\dots \partial {x_n}^{\alpha_n}} \bigg{)}^2 \bigg{]} dS \leq \underbrace{\int \dots \int}_{S_3} \\
&\bigg{[} \sum\limits_{i=1}^n \bigg{(} \frac{\partial}{\partial x_i}\frac{\partial^\alpha v_{p,q}}{\partial {x_1}^{\alpha_1} \dots \partial {x_n}^{\alpha_n}} \bigg{)}^2 + \bigg{(} \frac{\partial}{\partial t}\frac{\partial^\alpha v_{p,q}}{\partial {x_1}^{\alpha_1} \dots \partial {x_n}^{\alpha_n}}\bigg{)}^2 \bigg{]} dS; \end{split} \end{equation}
whereas from ($\ref{eq:3.92}$) follows
\begin{equation} \begin{split} \label{eq:3.118} & \underbrace{\int \dots \int}_{S_2} (v_{p,q})^2 dS \leq \Biggl\{ \bigg{[} \underbrace{\int \dots \int}_{S_3} (v_{p,q})^2 dS \bigg{]}^{\frac{1}{2}} + \bigg{[} \underbrace{\int \dots \int}_{S_3} \sum\limits_{i=1}^n \bigg{(}\frac{\partial v_{p,q}}{\partial x_i} \bigg{)}^2 dS \\
&+ \underbrace{\int \dots \int}_{S_3} \bigg{(} \frac{\partial v_{p,q}}{\partial t} \bigg{)}^2 dS \bigg{]}^{\frac{1}{2}}t \Biggr\}^2, \end{split} \end{equation}
From one of the properties of the average functions it follows that $u_{0h} \rightarrow u_0$ in ${W_2}^{\frac{n}{2} + 3}$ and $u_{1h} \rightarrow u_{1}$ in ${W_2}^{\frac{n}{2} + 2}$ and consequently the right-hand side of the previous inequalities can be arbitrarily small for $h_p$ and $h_q$ sufficiently small and $\alpha \leq \frac{n}{2} + 3$, and then the left-hand side has the same behaviour, thus for an arbitrary domain of the plane $t= {\rm const.}$ the sequence $\{u_h \}$ strongly converges in the sense of ${W_2}^{\frac{n}{2} +3}$. \endproof
However, the convergence of the $u$ functions in $C^2$ follows from this inclusion theorem. With a similar estimate, we show that $\frac{\partial u}{\partial t} \in C^1$ and $\frac{\partial^2 u}{\partial t^2} \in C^0$, i.e. $u$ is twice continuosly differentiable in the $(n+1)$ -dimensional space and it is solution of the wave equation.

\section{Parametrix of Scalar Wave Equation in Curved Space-Time}
Let us recall that the solution of the wave equation
\begin{equation} \label{eq:3.119} \Box u =0 \end{equation}
in Minkowski space-time involves amplitude and phase functions, which characterize the integral representation
\begin{equation} \label{eq:3.120} u (t, x_1, x_2, x_3) = \int_{-\infty}^{\infty} d \xi_1  \int_{-\infty}^{\infty} d \xi_2  \int_{-\infty}^{\infty} d \xi_3 A(\xi_1, \xi_2, \xi_3, t) e^{i(\xi_1 x_1 + \xi_2 x_2 + \xi_3 x_3 )}. \end{equation}
This is completely specified once suitable Cauchy data 
\begin{equation} \label{eq:3.121} u(t, x)|_{t=0} \equiv u_{0}(x), \; \; \; {\frac{\partial u}{\partial t}}\bigg{|}_{t=0}  \equiv u_1 (x), \end{equation}
are assigned. 
However, when the wave operator refers to a curved-space time, Eq. ($\ref{eq:3.120}$) has to be generalized. This is possible, since we have seen that a theorem guarantees that the solution of the Cauchy problem for the system under examination can be expressed in the form \cite{treves1980introduction}
\begin{equation} \label{eq:3.122} u(x, t) = \sum\limits_{i=0}^1 E_{i}(t) u_{i}(x), \end{equation}
where, on denoting by $\hat{u_i}$ the Fourier transform of the Cauchy data, the operators $E_{i}(t)$ act according to 
\begin{equation} \label{eq:3.123} E_{i}(t) u_i (x) = \sum\limits_{k=1}^2 (2\pi)^{-3} \int e^{i \varphi_k(x, t, \xi)}\alpha_{ik}(x, t, \xi) \hat{u_i}(\xi)d^3 \xi + R_i(t) u_i(x), \end{equation}
where the $\varphi_k$ are real-valued phase functions which satisfy the initial condition
\begin{equation} \label{eq:3.124} \varphi_k(t, x, \xi) |_{t=0} = x \cdot \xi = \sum\limits_{s=1}^3 x^s \xi_s, \end{equation}
and $R_i(t)$ is a regularizing operator which smoothes out the singularities acted upon by it. In other words, the Cauchy problem is here solved by a pair of Fourier-Maslov integral operators \cite{treves1980introduction} of the form ($\ref{eq:3.123}$), and such a construction generalizes the monochromatic plane waves for the d'Alembert operator from Minkowski space-time to curved space-time. Strictly, we are dealing with the $\textit{parametrix}$ for the wave equation. In our case, since we know that ($\ref{eq:3.122}$) and ($\ref{eq:3.123}$) yield an exact solution of the Cauchy problem, we can insert them into Eq. ($\ref{eq:3.119}$) with $P= \Box$, finding that, for all $i=0, 1,$
\begin{equation} \label{eq:3.125} P[E_i(t)u_i(x)] \sim \sum\limits_{k=1}^2 (2\pi)^{-3} \int P [e^{i \varphi_k} \alpha_{ik} ] \hat{u_i}(\xi)d^3\xi, \end{equation}
where $PR_i(t)u_i(x)$ can be neglected with respect to the integral on the right-hand side of Eq. ($\ref{eq:3.123}$), because $R_i(t)$ is a regularizing operator. Next, we find from Eq. ($\ref{eq:3.119}$) that
\begin{equation} \label{eq:3.126} P[e^{i \varphi_k}\alpha_{ik}] = e^{i \varphi_k}(iA_{ik} + B_{ik}), \end{equation}
where, on considering the form of $P$ in Kasner space-time (see Appendix B), i.e.
\begin{equation} \label{eq:3.127} P = - \frac{ \partial^2}{\partial t^2} - \frac{1}{t} \frac{\partial}{\partial t} \sum \limits_{l=1}^3 t^{-2p_l} \frac{\partial^2}{\partial x^{l^2}}, \; \; \sum\limits_{k=1}^3 p_k = \sum\limits_{k=1}^3 (p_k)^2 =1, \end{equation}
one finds
\begin{equation} \label{eq:3.128} A_{ik} \equiv \frac{\partial^2 \varphi_k}{\partial t^2} \alpha_{ik} + 2 \frac{\partial \varphi_k}{\partial t} \frac{\partial \alpha_{ik}}{\partial t} + \frac{1}{t} \frac{\partial \varphi_k}{\partial t} \alpha_{ik} - \sum\limits_{l=1}^3 t^{-2p_l} \bigg{(} \frac{\partial^2 \varphi_k}{\partial {x_l}^2} \alpha_{ik} + 2 \frac{\partial \varphi_k}{\partial x_l} \frac{\partial \alpha_{ik}}{\partial x_l} \bigg{)}, \end{equation}
\begin{equation} \label{eq:3.129}B_{ik} \equiv \frac{\partial^2 \alpha_{ik}}{\partial t^2} - \bigg{(} \frac{\partial \varphi_k}{\partial t} \bigg{)}^2 \alpha_{ik} + \frac{1}{t} \frac{\partial \alpha_{ik}}{\partial t} - \sum\limits_{l=1}^3t^{-2p_l} \bigg{(} \frac{\partial^2 \alpha_{ik}}{\partial {x_l}^2} -  \bigg{(}\frac{\partial \varphi_k}{\partial x_l}\bigg{)}^2 \alpha_{ik}\bigg{)}. \end{equation}
Then, if the phase functions $\varphi_k$ are real-valued, since the exponentials $e^{i \varphi_k}$ can be taken to be linearly independent, we can fulfill Eq. ($\ref{eq:3.119}$), up to the negligible contributions resulting from $PR_{i}(t)u_i(x)$, by setting to zero in the integrand ($\ref{eq:3.125}$) both $A_{ik}$ and $B_{ik}$. This leads to a coupled system of partial differential equations. We want to remark that the choice of Kasner space-time is merely an useful example and it is not necessary for the validity of our argumentation. Our Cauchy problem is therefore equivalent to solving the equations
\begin{equation} \label{eq:3.130} A_{ik}=0, \hspace{1cm} B_{ik}=0. \end{equation}
This equation is the $\textit{dispersion relation}$ for the scalar wave equation in Kasner space-time. Such a dispersion relation takes a neater geometric form upon bearing mind the form ($\ref{eq:3.127}$) of the wave operator $P=\Box$ in Kasner coordinates, i.e.
\begin{equation} \label{eq:3.131} A_{ik}=0 \rightarrow \bigg{[} - \alpha_{ik} (\Box \varphi_k ) - 2 \sum\limits_{\beta, \gamma =1}^4 (g^{-1})^{\beta \gamma}(\varphi_k)_{, \beta}(\alpha_{ik})_{,\gamma} \bigg{]} =0, \end{equation}
\begin{equation} \label{eq:3.132} B_{ik}=0 \rightarrow \bigg{[} - \Box + \sum\limits_{\beta, \gamma =1}^4 (g^{-1})^{\beta \gamma}(\varphi_k)_{, \beta}(\varphi_{k})_{,\gamma} \bigg{]} \alpha_{ik} =0. \end{equation}
Let us bear in mind that the indices $i$ and $k$ count the number of functions contributing to the Fourier-Maslov integral operator. We can therefore exploit the four-dimensional concept of gradient of a function as the four-dimensional covariant vector defined by the differential of the function, i.e.
\begin{equation} \label{eq:3.133} df = \sum\limits_{\alpha =1}^4 \frac{\partial f}{\partial x^\alpha} d x^{\alpha} = \sum\limits_{\alpha =1}^4 (\nabla_\alpha f) dx^{\alpha} = \sum\limits_{\alpha=1}^4 (\rm{grad} f)_\alpha dx^\alpha,\end{equation}
where $\nabla$ is the Levi-Civita connection on four-dimensional space-time, and we exploit the identity $f_{,\alpha} = \nabla_\alpha f$, $\forall f \in C^{\infty}(M)$. The consideration of $\nabla_\alpha f$ is not mandatory at this stage, but it will be helpful in a moment, when we write in tensor language the equations expressing the dispersion relation.

We arrive therefore, upon multiplying Eq. ($\ref{eq:3.131}$) by $\alpha_{ik}$, while dividing Eq. ($\ref{eq:3.132}$)  by $\alpha_{ik}$, at the following geometric form of dispersion relation in Kasner space-time
\begin{equation} \label{eq:3.134} \sum\limits_{\beta,\gamma=1}^4 (g^{-1})^{\beta \gamma} \nabla_\beta \bigg{[} (\alpha_{ik})^2 \nabla_\gamma \varphi_k \bigg{]} = {\rm div} \big{[}(\alpha_{ik})^2 \rm{grad} \varphi_k \big{]} =0, \end{equation}
\begin{equation} \label{eq:3.135} \sum\limits_{\beta,\gamma=1}^4 (g^{-1})^{\beta \gamma} (\nabla_\beta  \varphi_{k})( \nabla_\gamma \varphi_k) = <\rm{grad}(\varphi_k), \rm{grad}\varphi_k> = \frac{(\Box \alpha_{ik})}{\alpha_{ik}}, \end{equation}
where the four-dimensional divergence operator acts according to
\begin{equation} \label{eq:3.136}{\rm div}F=\sum\limits_{\beta=1}^4 \nabla^\beta F_\beta = \sum\limits_{\alpha, \beta=1}^4 (g^{-1})^{\alpha \beta} \nabla_\beta F_\alpha. \end{equation}

\section{Tensor Generalization of the Ermakov-Pinney Equation}
Note that, if the ratio $\frac{(\Box \alpha_{ik})}{\alpha_{ik}}$ is much smaller than a suitable parameter having dimension $\rm{length}^{-2}$, Eq. ($\ref{eq:3.135}$) reduces to the eikonal equation and hence the phase functions reduces to the Hadamard-Ruse-Synge world function that we have defined in the course of studying the characteristic conoid. It is possible to solve exactly Eqs. ($\ref{eq:3.134}$) and ($\ref{eq:3.135}$). For this purpose we remark that, upon defining the covariant vector
\begin{equation} \label{eq:3.137} \psi_\gamma \equiv (\alpha_{ik})^2 \nabla_{\gamma} \varphi_k, \end{equation}
Eq. ($\ref{eq:3.134}$) is equivalent to solving the first-order partial differential equation expressing the vanishing divergence condition for $\psi_\gamma$, i.e.
\begin{equation} \label{eq:3.138} \sum\limits_{\gamma=1}^4 \nabla^\gamma \psi_\gamma = \rm{div}\psi =0. \end{equation}
This equation is not enough to determine the four components of $\psi_\gamma$, but there are cases where further progress can be made. After doing that, we can express the covariant derivative of the phase function from the definition ($\ref{eq:3.137}$), i.e.
\begin{equation} \label{eq:3.139} \nabla_\gamma \varphi_k = \partial_\gamma \varphi_k = (\alpha_{ik})^{-2} \psi_{\gamma}, \end{equation}
and the insertion of Eq. ($\ref{eq:3.139}$) into Eq. ($\ref{eq:3.135}$) yields
\begin{equation} \label{eq:3.140} (\alpha_{ik})^3 \Box \alpha_{ik} = g(\psi,\psi) =\sum\limits_{\beta, \gamma=1}^4 (g^{-1})^{\beta \gamma} \psi_\beta \psi_\gamma = \sum\limits_{\gamma=1}^4 \psi_\gamma \psi^{\gamma}. \end{equation}
This is a tensorial generalization of a famous non-linear ordinary differential equation, i.e. the $\textit{Ermakov-Pinney equation}$ \cite{pinney1950nonlinear}
\begin{equation} \label{eq:3.141} y'' + p y=qy^{-3}.\end{equation}
If $y''$ is replaced by $\Box y$, $p$ is set to zero and $q$ is promoted to a function of space-time location, Eq. ($\ref{eq:3.141}$) is mapped into Eq. ($\ref{eq:3.140}$). After solving this nonlinear equation for $\alpha_{ik}= \alpha_{ik}[g(\psi, \psi)]$, we have to find the phase function $\varphi_k$ by writing and solving the four components of Eq. ($\ref{eq:3.139}$). To sum up, we have proved the following result:

\begin{thm} $\\$
Fior any Lorentzian space-time manifold $(M,g)$, the amplitude functions $\alpha_{ik} \in C^2(T^*M)$ and phase functions $\varphi_k \in C^1(T^*M)$ in the parametrix ($\ref{eq:3.123}$) for the scalar wave equation can be obtained by solving, first, the linear conditions ($\ref{eq:3.138}$) of vanishing divergence for a covariant vector $\psi_\gamma$. All non-linearities of the coupled system are then mapped into solving the non-linear equation ($\ref{eq:3.140}$) for the amplitude function $\alpha_{ik}$. Eventually, the phase function $\varphi_k$ is found by solving the first-order linear equation ($\ref{eq:3.139}$).
\end{thm}

\chapter{Linear Systems of Normal Hyperbolic Form}
\epigraph{The most incomprehensible thing about the world is that it is comprehensible.}{Albert Einstein }
In the next chapters, our aim is to demonstrate, following the work by Fourès-Bruhat \cite{foures1952theoreme}, that it is possible to solve the Cauchy problem for Einstein field equations in vacuum. The great achievement of this work was a rigorous and constructive proof that the Cauchy problem for Einstein's theory is well posed and admits a unique solution also with non-analytic Cauchy data. 

Hence, in this Chapter, we first consider a system of $n$ second-order partial differential equations $[E]$, with $n$ unknown functions $u_s$ and four variables $x^\alpha$, hyperbolic and linear, of the following type:
\begin{equation*} E_r = \sum\limits_{\lambda,\mu=1}^4 A^{\lambda \mu} \frac{\partial^2 u_s}{\partial x^\lambda \partial x^\mu} + \sum\limits_{s=1}^n \sum\limits_{\mu=1}^4 {{B^s}_r}^\mu \frac{\partial u_s}{\partial x^\mu} + f_r =0, \hspace{3cm} [E] \end{equation*}
where the coefficients $A^{\lambda \mu}$, ${{B^s}_r}^\mu$ and $f_r$ are given functions of the four variables $x^\alpha$ and they are taken to satisfy some useful assumptions. We will consider some linear combinations of these equations  whose coefficients are some auxiliary functions which possess at $M$ the parametrix properties and then we will obtain, by integrating them over the characteristic conoid $\Sigma$ with vertex $M$, a system of integral equations of the type of Kirchhoff formulae. 

Thus, by adjoining to these Kirchhoff formulae the equations determining the characteristic conoid and the auxiliary functons we will find a system of integral equations that is the solution of the system $[E]$.

\section{Assumptions on the Coefficients and The Characteristic Conoid}
 Following the Foures-Bruhat work \cite{foures1952theoreme}, let us first consider a system of $n$ second-order partial differential equations $[E]$, with $n$ unknown functions $u_s$ and four variables $x^\alpha$, hyperbolic and linear, of the following type:
\begin{equation} \label{eq:4.1} E_r = \sum\limits_{\lambda,\mu=1}^4 A^{\lambda \mu} \frac{\partial^2 u_s}{\partial x^\lambda \partial x^\mu} + \sum\limits_{s=1}^n \sum\limits_{\mu=1}^4 {{B^s}_r}^\mu \frac{\partial u_s}{\partial x^\mu} + f_r =0, \; \; r=1, 2, ..., n. \end{equation}
The coefficients $A^{\lambda \mu}$, which are the same for all the $n$ equations, ${{B^s}_r}^\mu$ and $f_r$ are given functions of the four variables $x^\alpha$. These are taken to satisfy the following assumptions:

Within a domain defined by
\begin{equation} \label{eq:4.2} |x^i - \bar{x}^i | \leq d, \; \; |x^4| \leq \epsilon \hspace{1cm} (i=1, 2, 3), \end{equation}
where $\bar{x}^i$, $d$ and $\epsilon$ are some given numbers, it holds that

\begin{description} 
\item[(1)] The coefficients $A^{\lambda \mu}$ and ${{B^s}_r}^\mu$ possess continuous and bounded derivatives up to the orders four and two, respectively. The coefficients $f_s$ are continuous and bounded. 
\item[(2)] The quadratic form $\sum\limits_{\lambda, \mu =1}^4A^{\lambda \mu}x_\lambda x_\mu$ is of the normal hyperbolic type, i.e. it has one positive square and three negative squares. We will assume in addition that the variable $x^4$ is a temporal variable, the three variables $x^i$ being spatial, i.e.
\begin{equation} \label{eq:4.3} A^{44}>0  \; \rm{and} \; \rm{the} \; \rm{quadratic} \; \rm{form} \; \sum\limits_{i,j=1}^3 A^{ij}x_i x_j < 0. \end{equation}
\item[(3)] The partial derivatives of the $A^{\lambda \mu}$ and ${{B^s}_r}^\mu$ of order four and two, respectively, satisfy Lipschitz conditions with respect to all their arguments.

\end{description}
The characteristic surfaces of ($\ref{eq:4.1}$) are three-dimensional manifolds of the space of four variables $x^\alpha$, solutions of the differential system
\begin{equation} \label{eq:4.4} F = \sum\limits_{\lambda, \mu =1}^4 A^{\lambda \mu} y_\lambda y_\mu =0 \; \rm{with} \; \sum\limits_{\lambda=1}^4 y_\lambda dx^\lambda =0. \end{equation}
The four quantities $y_\lambda$ denote a system of directional parameters of the normal to the contact element, having support $x^\alpha$. If we take this system, which is only defined up to a proportionality factor, in such a way that $y_4=1$ and if we set $y_i=p_i$, the desired surfaces are solution of 
\begin{equation} \label{eq:4.5} F=A^{44} + 2 \sum\limits_{i=1}^3 A^{i4}p_i + \sum\limits_{i,j=1}^3 A^{ij}p_i p_j =0 \; \rm{with} \; dx^4 + \sum\limits_{i=1}^3 p_i dx^i = 0.\end{equation}
The characteristics of this differential system, bicharacteristics of Eq. ($\ref{eq:4.1}$), satisfy the following differential equations:
\begin{equation} \label{eq:4.6} \frac{ dx^i}{A^{i4} + \sum\limits_{j=1}^3A^{ij}p_j}= \frac{ dx^4}{A^{44} + \sum\limits_{i=1}^3A^{i4}p_i}= \frac{- dp_i}{\frac{1}{2} \bigg{(} \frac{\partial F}{\partial x^i} -p_i \frac{\partial F}{\partial x^4}\bigg{)}}= d \lambda_1, \end{equation}
where $\lambda_1$ is an auxiliary parameter.

The characteristic conoid $\Sigma_0$ with vertex $M_0(x^\alpha_0)$ is the characteristic surface generated from the bicharacteristics passing through $M_0$. Any such bicharacteristic satisfies the system of integral equations
\begin{equation} \label{eq:4.7} x^i = x^i_0 + \int\limits_{0}^{\lambda_1} \bigg{[} A^{i4} + \sum\limits_{j=1}^3 A^{ij}p_j \bigg{]} d \lambda=  x^i_0 + \int\limits_{0}^{\lambda_1} T^i d \lambda, \end{equation}
\begin{equation} \label{eq:4.8} x^4 = x^4_0 + \int\limits_{0}^{\lambda_1} \bigg{[} A^{44} + \sum\limits_{i=1}^3 A^{i4}p_i \bigg{]} d \lambda=  x^4_0 + \int\limits_{0}^{\lambda_1} T^4 d \lambda, \end{equation}
\begin{equation} \label{eq:4.9} p_i = p_i^0 - \int\limits_{0}^{\lambda_1} \frac{1}{2} \bigg{(} \frac{\partial F}{\partial x^i} -p_i \frac{\partial F}{\partial x^4}\bigg{)}d \lambda=  p_i^0 + \int\limits_{0}^{\lambda_1} R_i d \lambda, \end{equation}
where the $p_i^0$ verify the relation
\begin{equation} \label{eq:4.10} A_{0}^{44} + 2 \sum\limits_{i=1}^3 A^{i4}_0 p_i^0 + \sum\limits_{i,j=1}^3 A^{ij}_0 p_i^0p_j^0 =0, \end{equation}
whereas $A^{\lambda \mu}_0$ denotes the value of the coefficient $A^{\lambda \mu}$ at the vertex $M_0$ of the conoid $\Sigma_0$. We will assume that at the point $M_0$ the coefficients $A^{\lambda \mu}$ take the following values:
\begin{equation} \label{eq:4.11} A^{44}_0=1,\; \; A^{i4}_0=0,\; \; A^{ij}_0 = - \delta^{ij}. \end{equation}
Thus, Eq. ($\ref{eq:4.10}$) reads as
\begin{equation} \label{eq:4.12} \sum\limits_{i=1}^3 (p_i^0)^2=1. \end{equation}
We will introduce to define the points of the surface $\Sigma_0$, besides the parameter $\lambda_1$ which defines the position of a point on a given bicharacteristic, two new parameters $\lambda_2$ and $\lambda_3$, that vary with the bicharacteristic under consideration, by means of
\begin{equation} \label{eq:4.13} p_1^0 = sin (\lambda_2)cos(\lambda_3), \hspace{0.5cm} p_2^0=sin (\lambda_2)sin(\lambda_3), \hspace{0.5cm} p_3^0=cos(\lambda_2). \end{equation}
The assumptions ($\ref{eq:4.11}$) make it possible to prove that there exists a number $\epsilon_1$ defininig a variation domain $\Lambda$ of the parameters $\lambda_i$ by means of
\begin{equation} \label{eq:4.14} |\lambda_1| \leq \epsilon_1, \; \; 0 \leq \lambda_2 \leq \pi, \; \; 0 \leq \lambda_3 \leq 2 \pi, \end{equation}
such that the integral equations ($\ref{eq:4.7}$), ($\ref{eq:4.8}$) and ($\ref{eq:4.9}$) possess within ($\ref{eq:4.14}$) a unique solution, continuous and bounded 
\begin{equation} \label{eq:4.15} x^\alpha = x^{\alpha}(x^\alpha_0, \lambda_1, \lambda_2, \lambda_3), \hspace{1cm} p_i=p_i (x^\alpha_0, \lambda_1, \lambda_2, \lambda_3), \end{equation}
satisfying the inequalities
\begin{equation*} |x^i - \bar{x}^i | \leq d, \hspace{1cm} |x^4| \leq \epsilon \end{equation*}
and possessing partial derivatives, continuous and bounded, of the first three orders with respect to the overabundant variables $\lambda_1$, $p_i^0$ (hence with respect to the three variables $\lambda_i$).

The first four equations ($\ref{eq:4.15}$) define, as a function of the three parameters $\lambda_i$, varying within the domain $\Lambda$, a point of a domain $V$ of the characteristic conoid $\Sigma_0$. We shall be led, in the following part of this work, to consider other parametric representations of the domain $V$:

\begin{description}
\item[(1)] We shall take as independent parameters the three quantities $x^4$, $\lambda_2$, $\lambda_3$. The function $x^4(\lambda_1, \lambda_2, \lambda_3)$ satisfies the equation
\begin{equation} \label{eq:4.16} x^4= \int\limits_{0}^{\lambda_1} T^4 d \lambda + x^4_0. \end{equation}
Or it turns out from ($\ref{eq:4.10}$) that, on $\Sigma_0$, one has
\begin{equation} \label{eq:4.17} 2 \sum\limits_{i=1}^4 A^{i4}p_i = - \sum\limits_{i,j=1}^3 A^{ij}p_ip_j - A^{44} \geq - A^{44}, \end{equation}
from which $T^4 \geq \frac{A^{44}}{2} >0$; $x^4$ is thus a monotonic increasing function of $\lambda_1$, the correspondence between $(x^4, \lambda_2, \lambda_3)$ and $(\lambda_1, \lambda_2, \lambda_3)$ is bijective.
\item[(2)] We shall take as representative parameters of a point of $\Sigma_0$ his three spatial coordinates $x^i$. The elimination of $\lambda_1$, $\lambda_2$, $\lambda_3$ among the four equations yields $x^4$ as a function of the $x^i$.

From the relation 
\begin{equation*} dx^4 + \sum\limits_{i=1}^3 p_i dx^i =0, \end{equation*}
identically verified from the solutions of equations ($\ref{eq:4.7}$), ($\ref{eq:4.8}$) and ($\ref{eq:4.9}$) on the characteristic surface $\Sigma_0$, one infers that the partial derivatives of this function $x^4$ with respect to the $x^i$ verify the relation
\begin{equation*} \frac{\partial x^4}{\partial x^i} = - p_i. \end{equation*}
If we denote by $[\varphi]$ the value of a function $\varphi$ of four coordinates $x^\alpha$ on the surface of the characteristic conoid $\Sigma_0$, it can be expressed as a function of the three variables of a parametric representation of $\Sigma_0$, in particular of the three coordinates $x^i$. The partial derivatives of this function with respect to the $x^i$ fulfill therefore:
\begin{equation} \label{eq:4.18}\frac{\partial [\varphi]}{\partial x^i} = \bigg{[} \frac{\partial \varphi}{\partial x^i} \bigg{]} -  \bigg{[} \frac{\partial \varphi}{\partial x^4} \bigg{]}p_i. \end{equation}
In the same manner it is possible to evaluate the derivatives $\big{[} \frac{\partial^2 \varphi}{\partial x^i \partial x^j} \big{]}$ and $\big{[} \frac{\partial^2 \varphi}{\partial x^i \partial x^4} \big{]}$, which are
\begin{equation*} \bigg{[} \frac{\partial^2 \varphi}{\partial x^i \partial x^4} \bigg{]} = \frac{\partial}{\partial x^i} \bigg{[} \frac{\partial \varphi}{\partial x^4} \bigg{]} + \bigg{[} \frac{\partial^2 \varphi}{\partial (x^4)^2} \bigg{]}p_i, \end{equation*}
\begin{equation*} \begin{split} \bigg{[} \frac{\partial^2 \varphi}{\partial x^i \partial x^j} \bigg{]} = & \frac{\partial^2 [\varphi]}{\partial x^i \partial x^j} + \frac{\partial}{\partial x^i} \bigg{[} \frac{\partial \varphi}{\partial x^4} \bigg{]}p_j + \frac{\partial}{\partial x^j} \bigg{[} \frac{\partial \varphi}{\partial (x^4)} \bigg{]}p_i + \bigg{[} \frac{\partial \varphi}{\partial x^4} \bigg{]}\frac{\partial p_i}{\partial x^j} \\
& + \bigg{[} \frac{\partial^2 \varphi}{\partial (x^4)^2} \bigg{]}p_i p_j. \end{split} \end{equation*}
These identities make it possible to write the following relations satisfied by the unknown functions $u_s$ on the characteristic conoid:
\begin{equation}\begin{split} \label{eq:4.19}[E_r] = &\sum\limits_{i,j=1}^3 [A^{ij}] \frac{\partial^2 [u_r]}{\partial x^i \partial x^j} + \Biggl\{ \sum\limits_{i,j=1}^3 [A^{ij}]p_ip_j + 2 \sum\limits_{i=1}^3 [A^{i4}]p_i + [A^{44}] \Biggr\} \frac{\partial^2 u_r}{\partial (x^4)^2} \\
& + 2 \sum\limits_{i=1}^3 \Biggl\{ \sum\limits_{j=1}^3 [A^{ij}]p_j +  [A^{i4}] \Biggr\} \frac{\partial}{\partial x^i} \bigg{[}\frac{\partial u_r}{\partial x^4} \bigg{]} + \bigg{[}\frac{\partial u_r}{\partial x^4} \bigg{]}\sum\limits_{i,j=1}^3 [A^{ij}]\frac{\partial p_i}{\partial x^j} \\
& + \sum\limits_{s=1}^n \sum\limits_{\mu=1}^4 B^{s \mu}_r \bigg{[}\frac{\partial u_s}{\partial x^\mu} \bigg{]} + [f_r] =0. \end{split}\end{equation}
The coefficient of the term $\big{[} \frac{\partial^2 u_r}{\partial (x^4)^2} \big{]}$ is the value on the characteristic conoid of the first member of Eq. ($\ref{eq:4.5}$); it therefore vanishes. We might have expected on the other hand that the equations $[E_r]=0$ would not contain second derivatives of the functions $u_r$ but those obtained by derivation on the surface $\Sigma_0$, the assignment on a characteristic surface of the unknown functions $[u_r]$ and of their first derivatives $\big{[} \frac{\partial u_r}{\partial x^\alpha} \big{]}$ not being able to determine the set of second derivatives.

\end{description}

\section{Integral Equations for Derivatives of $x^i$ and $p_i$}
Let us now derive Eqs. ($\ref{eq:4.7}$), ($\ref{eq:4.8}$) and ($\ref{eq:4.9}$) under the summation sign with respect to the $p_i^0$, they will read as
\begin{equation} \label{eq:4.20} \frac{\partial x^i}{\partial p_j^0} = \int\limits_{0}^{\lambda_1} \frac{\partial T^i}{\partial p_j^0} d \lambda = \int \limits_{0}^{\lambda_1} \Biggl\{ \sum\limits_{h=1}^3 \Biggl\{ \frac{\partial}{\partial x^h} \sum\limits_{k=1}^3 [A^{ik}]p_k + \frac{\partial}{\partial x^h}[A^{i4}] \Biggr\} y_j^h + [A^{ih}]z_j^h \Biggr\} d\lambda, \end{equation}
\begin{equation} \label{eq:4.21} \frac{\partial p_i}{\partial p_j^0} = \int\limits_{0}^{\lambda_1} \frac{\partial R_i}{\partial p_j^0} d \lambda = \int \limits_{0}^{\lambda_1} \sum\limits_{k=1}^3 \bigg{(} \frac{\partial R_i}{\partial x_k} \frac{\partial x^k}{\partial p_j^0} + \frac{\partial R_i}{\partial p_k} \frac{\partial p_k}{\partial p_j^0} \bigg{)}d\lambda, \end{equation}
\begin{equation} \label{eq:4.22} \frac{\partial^2 x^i}{\partial p_j^0 \partial p_k^0} = \int\limits_{0}^{\lambda_1} \frac{\partial^2 T^i}{\partial p_j^0 \partial p_k^0} d \lambda = \int \limits_{0}^{\lambda_1} \Biggl\{ \sum\limits_{h=1}^3 \bigg{(} \frac{\partial T^i}{\partial x^h} \frac{\partial^2 x^h}{\partial p_j^0 \partial p_k^0} + \frac{\partial T^i}{\partial p_h} \frac{\partial^2 p_h}{\partial p_j^0 \partial p_k^0} \bigg{)} + \phi^i_{jk} \Biggr\} d\lambda, \end{equation}
\begin{equation} \label{eq:4.23} \frac{\partial^2 p_i}{\partial p_j^0 \partial p_k^0} = \int\limits_{0}^{\lambda_1} \frac{\partial^2 R^i}{\partial p_j^0 \partial p_k^0} d \lambda = \int \limits_{0}^{\lambda_1} \Biggl\{ \sum\limits_{h=1}^3 \bigg{(} \frac{\partial R_i}{\partial x^h} \frac{\partial^2 x^h}{\partial p_j^0 \partial p_k^0} + \frac{\partial R_i}{\partial p_h} \frac{\partial^2 p_h}{\partial p_j^0 \partial p_k^0} \bigg{)} + \psi^i_{jk} \Biggr\} d\lambda, \end{equation}
\begin{equation}\begin{split} \label{eq:4.24} \frac{\partial^3 x^i}{\partial p_j^0 \partial p_h^0 \partial p_k^0} = &\int\limits_{0}^{\lambda_1} \frac{\partial^3 T^i}{\partial p_j^0 \partial p_h^0 \partial p_k^0} d \lambda = \int \limits_{0}^{\lambda_1} \Biggl\{ \sum\limits_{l=1}^3 \bigg{(} \frac{\partial T^i}{\partial x^l} \frac{\partial^3 x^l}{\partial p_j^0  \partial p_h^0 \partial p_k^0} \\
&+ \frac{\partial T^i}{\partial p_l} \frac{\partial^3 p_l}{\partial p_j^0 \partial p_h^0 \partial p_k^0} \bigg{)} + \phi^i_{jhk} \Biggr\} d\lambda, \end{split}\end{equation}
\begin{equation} \begin{split} \label{eq:4.25} \frac{\partial^3 p_i}{\partial p_j^0 \partial p_h^0 \partial p_k^0} = &\int\limits_{0}^{\lambda_1} \frac{\partial^3 R^i}{\partial p_j^0 \partial p_h^0 \partial p_k^0} d \lambda = \int \limits_{0}^{\lambda_1} \Biggl\{ \sum\limits_{l=1}^3 \bigg{(} \frac{\partial R_i}{\partial x^l} \frac{\partial^3 x^l}{\partial p_j^0 \partial p_h^0 \partial p_k^0} \\
&+ \frac{\partial R_i}{\partial p_l} \frac{\partial^3 p_l}{\partial p_j^0 \partial p_h^0\partial p_k^0} \bigg{)} + \psi^i_{jhk} \Biggr\} d\lambda, \end{split}\end{equation}
where $\phi^i_{jk}$ and $\psi^i_{jk}$ are polynomials of the functions $p_i(\lambda)$, $\frac{\partial x^i}{\partial p_j^0}(\lambda)$, $\frac{\partial p_i}{\partial p_j^0}(\lambda)$, of the coefficients $A^{\lambda \mu}(x^\alpha)$ and of their partial derivatives with respect to the $x^\alpha$ up to the third order included. Whereas  $\phi^i_{jhk}$ and $\psi^i_{jhk}$ are polynomials of the functions $p_i(\lambda)$, $\frac{\partial x^i}{\partial p_j^0}(\lambda)$, $\frac{\partial p_i}{\partial p_j^0}(\lambda)$, $\frac{\partial^2 x^i}{\partial p_j^0 \partial p_h^0}(\lambda)$, $\frac{\partial^2 p_i}{\partial p_j^0 \partial p_h^0}(\lambda)$ as well as of the coefficients $A^{\lambda \mu}$ and of their partial derivatives up to the fourth order included. In these functions the $x^\alpha$ are replaced from the $x^\alpha(\lambda)$ given by ($\ref{eq:4.15}$).
If we set
\begin{equation*} \frac{\partial x^i}{\partial p_j^0} \equiv y^i_j, \hspace{0.5cm} \frac{\partial^2 x^i}{\partial p_j^0 \partial p_h^0} \equiv y^i_{jh},  \hspace{0.5cm} \frac{\partial^3 x^i}{\partial p_j^0 \partial p_h^0 \partial p_k^0} \equiv y^i_{jhk}, \end{equation*}
\begin{equation*} \frac{\partial p_i}{\partial p_j^0} \equiv z^i_j,  \hspace{0.5cm} \frac{\partial^2 p_i}{\partial p_j^0 \partial p_h^0} \equiv z^i_{jh},  \hspace{0.5cm} \frac{\partial^3 p_i}{\partial p_j^0 \partial p_h^0 \partial p_k^0} \equiv z^i_{jhk},\end{equation*}
\begin{equation*} T^i_j \equiv \frac{\partial T^i}{\partial p_j^0},  \hspace{0.5cm} T^i_{jk} \equiv \frac{\partial^2 T^i}{\partial p_j^0 \partial p_k^0},  \hspace{0.5cm} T^i_{jhk} \equiv \frac{\partial^3 T^i}{\partial p_j^0 \partial p_h^0 \partial p_k^0}, \end{equation*}
\begin{equation*} R^i_j \equiv\frac{\partial R_i}{\partial p_j^0},  \hspace{0.5cm} R^i_{jk} \equiv \frac{\partial^2 R_i}{\partial p_j^0 \partial p_k^0},  \hspace{0.5cm} R^i_{jhk} \equiv \frac{\partial^3 R_i}{\partial p_j^0 \partial p_h^0 \partial p_k^0};\end{equation*}
 Eqs. ($\ref{eq:4.20}$), ($\ref{eq:4.21}$), ($\ref{eq:4.22}$), ($\ref{eq:4.23}$), ($\ref{eq:4.24}$) and ($\ref{eq:4.25}$) read as
\begin{equation} \begin{split} \label{eq:4.26} & y^i_j = \int\limits_{0}^{\lambda_1} T^i_j d\lambda,  \hspace{0.5cm}  z^i_j = \int\limits_{0}^{\lambda_1} R^i_j d\lambda, \hspace{0.5cm} y^i_{jk} = \int\limits_{0}^{\lambda_1} T^i_{jk} d\lambda, \\
& z^i_{jk} = \int\limits_{0}^{\lambda_1} R^i_{jk} d\lambda,  \hspace{0.5cm}   y^i_{jhk} = \int\limits_{0}^{\lambda_1} T^i_{jhk} d\lambda,  \hspace{0.5cm} z^i_{jhk} = \int\limits_{0}^{\lambda_1} R^i_{jhk} d\lambda. \end{split} \end{equation}

\section{The Auxiliary Functions $\sigma^r_s$}
Let us now form $n^2$ linear combinations $\sum\limits_{r=1}^n \sigma^r_s [E_r]$ of the Eq. ($\ref{eq:4.19}$) verified by the unknown functions within the domain $V$ of $\Sigma_0$, the $\sigma^r_s$ denoting $n^2$ auxiliary functions which possess at $M_0$ a singularity.
If we set
\begin{equation*} M (\varphi)= \sum\limits_{i,j=1}^3 [A^{ij}] \frac{\partial^2 \varphi}{\partial x^i \partial x^j}, \end{equation*}
$\varphi$ denoting a function whatsoever of the three variables $x^i$, it is possible to perfom the stated $n^2$ linear combinations as
\begin{equation} \begin{split} \label{eq:4.27} \sum\limits_{r=1}^n \sigma^r_s [E_r] & =  \sum\limits_{r=1}^n \Biggl\{ M ([u_r]) + 2 \sum\limits_{i=1}^3 \bigg{(}  \sum\limits_{j=1}^3 [A^{ij}]p_j + [A^{i4}] \bigg{)} \frac{\partial}{\partial x^i} \bigg{[} \frac{\partial u_r}{\partial x^4} \bigg{]} \\
& + \bigg{[} \frac{\partial u_r}{\partial x^4} \bigg{]} \sum\limits_{i,j=1}^3 [A^{ij}] \frac{\partial p_i}{\partial x^j} +  \sum\limits_{t,\mu=1}^3 [{B^{t }_r}^\mu ] \bigg{[} \frac{\partial u_t}{\partial x^\mu} \bigg{]} + [f_r] \Biggr\} \sigma^r_s=0. \end{split} \end{equation}
We will transform these equations in such a way that a divergence occurs therein, whose volume integral will get transformed into a surface integral, while the remaining terms will contain only $[u_r]$ and $\big{[} \frac{\partial u_r}{\partial x^4} \big{]}$. We will use for that purpose the following identity, verified by two functions whatsoever $\varphi$ and $\psi$ of the three variables $x^i$:
\begin{equation*}  \psi M(\varphi) = \sum\limits_{i,j=1}^3 \frac{\partial}{\partial x^i} \bigg{(} [A^{ij}]\psi \frac{\partial \varphi}{\partial x^j} \bigg{)} -\sum\limits_{i,j=1}^3 \frac{\partial \varphi}{\partial x^j} \frac{\partial}{\partial x^i} ([A^{ij}]\psi ) \end{equation*}
or 
\begin{equation*} \psi M(\varphi) =\sum\limits_{i,j=1}^3 \frac{\partial}{\partial x^i} \bigg{(} [A^{ij}]\psi \frac{\partial \varphi}{\partial x^j} \bigg{)} - \varphi \frac{\partial}{\partial x^j} ([A^{ij}] \varphi) \bigg{)} + \varphi \bar{M}(\psi), \end{equation*}
where $\bar{M}$ is the adjoint operator of $M$, i.e.
\begin{equation*} \bar{M}(\psi)= \sum\limits_{i,j=1}^3 \frac{\partial^2 ([A^{ij}]\psi)}{\partial x^i \partial x^j}, \end{equation*}
and the identity ($\ref{eq:4.18}$) yields here
\begin{equation*} \bigg{[} \frac{\partial u_r}{\partial x^i} \bigg{]} = \frac{\partial [u_r]}{\partial x^i} + p_i \bigg{[} \frac{\partial u_r}{\partial x^4} \bigg{]}. \end{equation*}
Thus, $\sum\limits_{r=1}^4 \sigma^r_s[E_r]$ take the form
\begin{equation*} \sum\limits_{r=1}^4 \sigma^r_s[E_r] = \sum\limits_{i=1}^3 \frac{\partial}{\partial x^i}E_s^i + \sum\limits_{r=1}^n [u_r]L_s^r +\sum\limits_{r=1}^n \sigma^r_s [f_r] - \sum\limits_{r=1}^n \bigg{[} \frac{\partial u_r}{\partial x^4}\bigg{]} D_s^r, \end{equation*}
where we have defined 
\begin{equation} \begin{split} \label{eq:4.28} E_s^i =&  \sum\limits_{j=1}^3 \sum\limits_{r=1}^n \bigg{(} [A^{ij}]\sigma_s^r \frac{\partial [u_r]}{\partial x^j} - [u_r]\frac{\partial}{\partial x^j}([A^{ij}]\sigma_s^r) \bigg{)} + 2 \sum\limits_{r=1}^n \sigma_s^r \Biggl\{ \sum\limits_{j=1}^3[A^{ij}]p_j \\
&+ [A^{i4}] \Biggr\} \bigg{[}\frac{\partial u_r}{\partial x^4} \bigg{]}  + \sum\limits_{r,t=1}^n[B^t_{ri}][u_t]\sigma_s^r, \end{split} \end{equation}
\begin{equation} \label{eq:4.29} L_s^r = \bar{M}(\sigma_s^r) - \sum\limits_{i=1}^3\sum\limits_{t=1}^n \frac{\partial}{\partial x^i} ([B_t^
{ri}] \sigma_s^t ), \end{equation}
\begin{equation} \begin{split} \label{eq:4.30} D_s^r =&  \sigma_s^r \Biggl\{ 2 \sum\limits_{i=1}^3 \frac{\partial}{\partial x^i} \bigg{(}\sum\limits_{j=1}^3 [A^{ij}]p_j + [A^{i4}] \bigg{)} - \sum\limits_{i,j=1}^3[A^{ij}] \frac{\partial p_j}{\partial x^i} + 2 \sum\limits_{i=1}^3 \bigg{(} \sum\limits_{j=1}^3 [A^{ij}]p_j  \\
& + [A^{i4}] \bigg{)} \frac{\partial \sigma_s^r}{\partial x^i} - \sum\limits_{t=1}^n ([B_t^{r4}] + \sum\limits_{i=1}^3 [B_t^{ri}]p_i) \sigma_s^t. \end{split} \end{equation}
We will choose the auxiliary functions $\sigma_s^r$ in such a way that, in every equation, the coefficient of $\big{[} \frac{\partial u_r}{\partial x^4} \big{]}$ vanishes. These functions will therefore fulfill $n^2$ partial differential equations of first order
\begin{equation} \label{eq:4.31} D_s^r=0. \end{equation}
If we look for its solution of the form $\sigma_s^r = \sigma \omega_s^r$,  where $\sigma$ is infinite at the point $M_0$ and the $\omega_s^r$ are bounded,  Eq. ($\ref{eq:4.31}$) reads as
\begin{equation} \begin{split} \label{eq:4.32} & \sigma_s^r \Biggl\{ \sum\limits_{i=1}^3 \frac{\partial}{\partial x^i} \bigg{(} \sum\limits_{j=1}^3[A^{ij}]p_j + [A^{i4}] \bigg{)} + \sum\limits_{i,j=1}^3 p_j \frac{\partial}{\partial x^i}[A^{ij}] + \sum\limits_{i=1}^3 \frac{\partial }{\partial x^i}[A^{i4}] \Biggr\} \\
& - \sum\limits_{t=1}^n \bigg{(}[B_t^{r4}] + \sum\limits_{i=1}^3[B_t^{ri}]p_i \bigg{)}\sigma_s^t + 2 \sum\limits_{i=1}^3 \bigg{(}\sum\limits_{j=1}^3[A^{ij}]p_j + [A^{i4}]\bigg{)} \frac{\partial \sigma_s^r}{\partial x^i} = 0. \end{split} \end{equation}
The coefficients $A^{\lambda \mu}$, $B_s^{t \lambda}$ , the first derivatives of the $A^{\lambda \mu}$ and the functions $p_i$ are bounded within the domain $V$, the coefficients of the linear first-order partial differential equations are therefore a sum of bounded terms, perhaps with exception of the terms
\begin{equation*} \sum\limits_{i=1}^3 \frac{\partial }{\partial x^i} \Biggl\{ \sum\limits_{j=1}^3 [A^{ij}]p_j + [A^{i4}] \Biggr\} \end{equation*}
We will therefore choose the $\omega_s^r$, that we want to be bounded, as satisfying the equation
\begin{equation} \begin{split} \label{eq:4.33} & \omega_s^r \bigg{(} \sum\limits_{i,j=1}^3p_j \frac{\partial}{\partial x^i}[A^{ij}] + \sum\limits_{i=1}^3 \frac{\partial}{\partial x^i} [A^{i4}]\bigg{)} - \sum\limits_{t=1}^n \omega_s^t \Biggl\{[B_t^{r4}] +\sum\limits_{i=1}^3[B_t^{ri}]p_i \Biggr\} \\
& + 2\sum\limits_{i=1}^3 \Bigg\{\sum\limits_{j=1}^3 [A^{ij}]p_j + [A^{i4}] \Biggr\} \frac{\partial \omega_s^r}{\partial x^i}=0, \end{split}\end{equation}
fulfilling in turn
\begin{equation}\label{eq:4.34} \sigma \sum\limits_{i=1}^3 \frac{\partial}{\partial x^i} \bigg{(}\sum\limits_{j=1}^3 [A^{ij}]p_j + [A^{i4}] \bigg{)} + 2\sum\limits_{i=1}^3 \bigg{(}\sum\limits_{j=1}^3 [A^{ij}]p_j + [A^{i4}] \bigg{)} \frac{\partial \sigma}{\partial x^i} =0. \end{equation}
Our task is to evaluate $\omega_s^r$ and then $\sigma_s^r$.

\section{Evaluation of  the $\omega_s^r$ and $\sigma$}
Let us consider the equation ($\ref{eq:4.33}$), it can be written in form of integral equations analogous to the Eqs. ($\ref{eq:4.7}$), ($\ref{eq:4.8}$) and ($\ref{eq:4.9}$) obtained in the search for the conoid $\Sigma_0$. We have indeed, on $\Sigma_0$:
\begin{equation*} \sum\limits_{j=1}^3[A^{ij}]p_j + [A^{i4}] = T^i = \frac{\partial x^i}{\partial \lambda_1}, \end{equation*}
from which, for an arbitrary function $\varphi$ defined on $\Sigma_0$, 
\begin{equation*}\sum\limits_{i=1}^3T^i \frac{\partial \varphi}{\partial x^i}= \frac{\partial \varphi}{\partial \lambda_1}. \end{equation*}
Let us impose upon the $\omega_s^r$ the limiting conditions $\omega_s^r = \delta_s^r$ for $\lambda_1=0$. These quantities satisfy therefore the integral equations
\begin{equation} \label{eq:4.35} \omega_s^r = \int\limits_{0}^{\lambda_1} \bigg{(}\sum\limits_{t=1}^n \mathcal{Q}_t^r \omega_s^t + \mathcal{Q}\omega_s^r \bigg{)} d \lambda + \delta_s^r \end{equation}
where we have defined
\begin{equation*} \mathcal{Q}_t^r = \frac{1}{2} ([B_t^{r4}] + \sum\limits_{i=1}^3[B_t^{ri}]p_i) \; \rm{and} \; \mathcal{Q}= - \frac{1}{2} \bigg{(}\sum\limits_{i,j=1}^3 p_j \frac{\partial}{\partial x^i} [A^{ij}] + \sum\limits_{i=1}^3 \frac{\partial}{\partial x^i}[A^{i4}] \bigg{)}, \end{equation*}
the assumptions made upon the coefficients $A^{\lambda \mu}$ and $B_s^{r \lambda}$ and the results obtained on the functions $x^i$, $p_i$ enabling moreover to prove that, for a convenient choice of $\epsilon_1$, these equations have a unique, continuous, bounded solution which has partial derivatives of the first two orders with respect to the $p_i^0$, continuous and bounded within the domain $\Lambda$. We will denote these derivatives by $\omega^r_{si}$ and $\omega^r_{sij}$.

Once we have found $\omega_s^r$, let us consider the Eq. ($\ref{eq:4.34}$) verified by $\sigma$. We know that
\begin{equation*} \sum\limits_{i=1}^3 \bigg{(} \sum\limits_{j=1}^3 [A^{ij}]p_j + [A^{i4}] \bigg{)} \frac{\partial \sigma}{\partial x^i} = \frac{\partial \sigma}{\partial \lambda_1}, \end{equation*}
and we are going to evaluate the coefficient of $\sigma$,
\begin{equation*}  \sum\limits_{i=1}^3 \frac{\partial}{\partial x^i} \bigg{(} \sum\limits_{j=1}^3 [A^{ij}]p_j + [A^{i4}]\bigg{)}, \end{equation*}
by relating it very simply to the determinant
\begin{equation*} \frac{D(x^1, x^2, x^3)}{D(\lambda_1, \lambda_2, \lambda_3)} \equiv J_{x \lambda}. \end{equation*}
This Jacobian $J_{x \lambda}$ of the change of variables $x^i = x^i (\lambda_j)$ on the conoid $\Sigma_0$, has for elements
\begin{equation} \label{eq:4.36}  \frac{\partial x^i}{\partial \lambda_1}=T^i,  \hspace{0.5cm} \frac{\partial x^i}{\partial \lambda_2}= \sum\limits_{j=1}^3 {y^i}_j \frac{\partial p_j^0}{\partial \lambda_2},  \hspace{0.5cm} \frac{\partial x^i}{\partial \lambda_3}= \sum\limits_{j=1}^3 {y^i}_j \frac{\partial p_j^0}{\partial \lambda_3}. \end{equation}
Let us denote by $J^i_j$ the minor relative to the element $\frac{\partial x^i}{\partial \lambda_j}$ of the determinant $J_{x \lambda}$.
A function whatsoever $\varphi$, defined on $\Sigma_0$, verifies the identities
\begin{equation*} \frac{\partial \varphi}{\partial x^i} = \sum\limits_{j=1}^3 \frac{J^j_i}{J_{x \lambda}} \frac{\partial \varphi}{\partial \lambda_j}. \end{equation*}
Let us apply this formula to the function $\frac{\partial x^i}{\partial \lambda_1} = T^i$:
\begin{equation*} \sum\limits_{i=1}^3 \frac{\partial}{\partial x^i}T^i = \sum\limits_{i,j=1}^3 \frac{J^i_j}{J_{x \lambda}} \frac{\partial}{\partial \lambda_j} T^i = \sum\limits_{i,j=1}^3 \frac{J^i_j}{J_{x \lambda}} \frac{\partial}{\partial \lambda_1} \bigg{(} \frac{ \partial x^i}{\partial \lambda_j} \bigg{)}, \end{equation*}
$J^i_j$ being the minor relative to the element $\frac{\partial x^i}{\partial \lambda_j}$ of the determinant $J_{x \lambda}$ we have
\begin{equation*} \sum\limits_{i=1}^3 \frac{\partial}{\partial x^i}T^i  = \frac{1}{J_{x \lambda}} \frac{ \partial J_{x \lambda}}{\partial \lambda_1}. \end{equation*}
Thus, the function $\sigma$ verifies the relation
\begin{equation*} \sigma \frac{\partial J_{x \lambda}}{ \partial \lambda_1} + 2 J_{x \lambda} \frac{\partial \sigma}{\partial \lambda_1} = 0, \end{equation*}
whose general solution is
\begin{equation*} \sigma = \frac{ f (\lambda_2, \lambda_3)}{ {|J_{x \lambda}|}^{\frac{1}{2}}}, \end{equation*}
where $f$ denotes an arbitrary function.

For $\lambda_1=0$ the determinant $J_{x \lambda}$ vanishes, because the $y^i_j$ are vanishing; the function $\sigma$ is therefore infinite. The coefficients $A^{\lambda \mu}$ and their first and second partial derivatives with respect to the $x^\alpha$ being continuous and bounded within the domain $V$ of $\Sigma_0$, as well as the functions $x^i$, $y^i_j$, $z^i_j$, we have
\begin{equation} \label{eq:4.37} \lim_{\lambda \to 0} \frac{y^i_j}{\lambda_1} = [A^{ij}]_{\lambda_1 =0} = - \delta^i_j. \end{equation}
By dividing the second and third line of $J_{x \lambda}$ by $\lambda_1$ we obtain a determinant equal to $\frac{J_{x \lambda}}{(\lambda_1)^2}$; we deduce from the formulas ($\ref{eq:4.36}$) and ($\ref{eq:4.37}$)
\begin{equation*} \begin{split} \lim_{\lambda \to 0} \frac{y^i_j}{(\lambda_1)^2}& = det   
\left ( {\begin{array}{ccc}
- sin (\lambda_2)cos(\lambda_3) & - sin (\lambda_2)sin(\lambda_3) & - cos(\lambda_2) \\
- cos (\lambda_2)cos(\lambda_3)& - cos (\lambda_2)sin(\lambda_3) &  sin (\lambda_2) \\
+ sin (\lambda_2)sin(\lambda_3) & - sin (\lambda_2)cos(\lambda_3) & 0 \\
\end{array} } \right )
 \\
& = - sin(\lambda_2).
\end{split}
\end{equation*}
As a matter of fact:
\begin{equation*} \lim_{\lambda_1 \to 0} T^i = - \sum\limits_{j=1}^3 \delta^j_ip_j^0 = - p_i^0, \end{equation*}
\begin{equation*} \lim_{\lambda_1 \to 0}  \frac{1}{\lambda_1} \frac{\partial x^i}{\partial \lambda_u} = \lim_{\lambda_1 \to 0} \sum\limits_{j=1}^3 \frac{y^i_j}{\lambda_1} \frac{\partial p_j^0}{\partial \lambda_u} = - \sum\limits_{j=1}^3 \delta^i_j \frac{\partial p_j^0}{\partial \lambda_u }. \end{equation*}
We will take for auxiliary function $\sigma$ the function
\begin{equation*} \sigma= {\bigg{|} \frac{ sin(\lambda_2)}{J_{x \lambda}} \bigg{|}}^\frac{1}{2}. \end{equation*}
We will then have $\lim_{\lambda_1 \to 0} \sigma \lambda_1=1$.

\section{Derivatives of the Functions $\sigma^r_s$}
Let us now consider
\begin{equation*} D^r_s=0. \end{equation*}
These equations possess a solution having at $M_0$ the desired singularity. If the auxiliary functions $\sigma^r_s$ verify these $n^2$ relations, the equations, verified by the unknown functions $u_r$ on the characteristic conoid $\Sigma_0$, take the simple form
\begin{equation} \label{eq:4.38} \sum\limits_{r=1}^n ([u_r]L^r_s + \sigma^r_s [f_r] ) + \sum\limits_{i=1}^3 \frac{\partial}{\partial x^i}E^i_s=0. \end{equation}
We will integrate these equations with respect to the three variables $x^i$ on a portion $V_\eta$ of hypersurface of the characteristic conoid $\Sigma_0$, limited by the hypersurfaces $x^4=0$ and $x^4 = x^4_0 - \eta$. This domain $V_\eta$ is defined to be simply connected and internal to the domain $V$ if the coordinate $x^4_0$ is sufficiently small. As a matter of fact
\begin{equation*} |x^4_0| < \epsilon \; \rm{implies}  \; \rm{within}  \; V_\eta \; |x^4-x^4_0| < \epsilon_0. \end{equation*}
The formula ($\ref{eq:4.16}$) shows in such a case that, for a suitable choice of $\epsilon_0$, we will have $\lambda_1 \leq \epsilon_1$. Since the boundary of $V_\eta$ consists of two-dimesional domains $S_0$ and $S_\eta$ cut over $\Sigma_0$ from the hypersurfaces $x^4=0$, $x^4=x^4_0 - \eta$ we will have, upon integrating Eq. ($\ref{eq:4.38}$) within $V_\eta$, the following fundamental relations:
\begin{equation} \begin{split}\label{eq:4.39} & \underbrace{ \int \int \int}_{V_\eta} \sum\limits_{r=1}^n \{ [u_r]L^r_s + \sigma^r_s[f_r] \} dV + \underbrace{\int \int}_{S_\eta}  \sum\limits_{i=1}^3 E^i_s cos(n,x^i)dS \\
& - \underbrace{ \int \int}_{S_0} \sum\limits_{i=1}^3 E^i_s cos(n,x^i)dS=0, \end{split}\end{equation}
where $dV$, $dS$ and $cos(n,x^i)$ denote, in the space of three variables $x^i$, the volume element, the area element of a surface $x^4={\rm const.}$ and the directional cosines of the outward-pointing normal to one of such surfaces, respectively.

Equation ($\ref{eq:4.38}$) contains, on the one hand the values on $\Sigma_0$ of the unknown functions $u_r$, of their partial derivatives as well as the functions $p_i$, $y$ and $z$, on the other hand the functions $\sigma^r_s$ and their first and second partial derivatives.

Let us study therefore the partial derivatives of the first two orders of the functions $\sigma$ and $\omega^r_s$.
Since we have seen that $\sigma= {\big{|} \frac{sin(\lambda_2)}{J_{x \lambda}} \big{|}}^\frac{1}{2}$, thus it is a function of the trigonometric lines of $\lambda_u$ (with $u=2, 3$), of the functions $x^\alpha$ (through the intermediate effect of the $A^{\lambda \mu}$) and of the functions $p_i$, $y^i_j$. The first and second partial derivatives of $\sigma$ with respect to the $x^i$ will be therefore expressed with the help of the functions listed and of their first and second partial derivatives.

$\textbf{First derivatives of $\sigma$:}$ We have seen that the partial derivatives with respect to the $x^i$ of a  function whatsoever $\varphi$, defined on $\Sigma_0$, satisfy the identity
\begin{equation}  \label{eq:4.40} \frac{\partial \varphi}{\partial x^i} = \sum\limits_{j=1}^3 \frac{ J^j_i}{J_{x \lambda}} \frac{\partial \varphi}{\partial \lambda_j}, \end{equation}
where $\frac{ J^j_i}{J_{x \lambda}}$ is a given function of $cos(\lambda_u)$, $ sin(\lambda_u)$, $x^\alpha$, $p_i$, $y^j_i$; the partial derivatives with respect to $\lambda_1$ of the functions $x^i$, $p_i$, $y^j_i$, are the quantities $T^i$, $R_i$, $T^i_j$ which are expressed through these functions themselves and through $z^j_i$; the partial derivatives with respect to $\lambda_u$ of these functions $x^i$, $p_i$, $y^j_i$ being expressible by means of their derivatives with respect to the overabundant parameters $p_h^0$, denoted here by $y^i_h$, $z^i_h$, $y^i_{jh}$, and by means of $cos(\lambda_u)$, $sin(\lambda_u)$.

The function $\sigma$ admits therefore within $V$, under the assumpstions made, of first partial derivatives with respect to the $x^i$ which are expressible by means of the functions $x^\alpha$ ( with the intermediate help of the $[A^{\lambda \mu}]$ and of the $\big{[} \frac{\partial A^{\lambda \mu}}{\partial x^\alpha} \big{]}$ and of the functions $p_i$, $y^i_j$, $z^i_j$, $y^i_{jh}$ and of $cos(\lambda_u)$, $sin(\lambda_u)$).

$\textbf{Second derivatives of $\sigma$:}$ Another application of the formula ($\ref{eq:4.40}$) shows, in analogous fashion, that $\sigma$ admits within $V$ of second partial derivatives, which are expressible by means of the functions $x^\alpha$ (with the intermediate action of the $A^{\lambda \mu}$ and their first and second partial derivatives) and of the functions $p_i$, $y^i_j$, $z^i_j$, $y^i_{jh}$, $z^i_{jh}$, $y^i_{jhk}$ and of $cos(\lambda_u)$, $sin(\lambda_u)$.

$\textbf{Derivatives of  the $\omega^r_s$:}$ The identity ($\ref{eq:4.40}$) makes it possible moreover to state that the functions $\omega^r_s$ admit within $V$ of first and second partial derivatives with respect to the variables $x^i$ if these functions admit, within $V$, of first and second partial derivatives with respect to the variables $\lambda_u$; it suffices for that purpose that they admit of first and second partial derivatives with respect to the overabundant variables $p_i^0$.

We shall set 
\begin{equation*} \frac{\partial \omega^r_s}{\partial p_i^0} = \omega^r_{si},  \hspace{1cm} \frac{\partial \omega^r_s}{\partial p_i^0 \partial p_j^0} = \omega^r_{sij}. \end{equation*}
If these functions are continuous and bounded within $V$ they satisfy, under the assumptions made, the integral equations obtained by derivation under the summation symbol of Eq. ($\ref{eq:4.35}$) with respect to the $p_i^0$. Let us define
\begin{equation} \label{eq:4.41} \omega^r_{si} = \int\limits_{0}^{\lambda_1} \bigg{(} \sum\limits_{t=1}^n \mathcal{Q}^r_t \omega^t_{si} + \mathcal{Q} \omega^r_{si} + \Omega^r_{si} \bigg{)} d \lambda, \end{equation}
where
\begin{equation*} \Omega^r_{si}= \sum\limits_{t=1}^n \frac{\partial \mathcal{Q}^r_t}{\partial p_i^0} \omega^t_s + \frac{\partial \mathcal{Q}}{\partial p_i^0} \omega^r_s \end{equation*}
is a polynomial of the functions $\omega^r_s$, $p_i$, $y^i_j$, $z^i_j$ as well as of the values on $\Sigma_0$ of the coefficients $A^{\lambda \mu}$, $B^{r \lambda}_s$ of the equations ($\ref{eq:4.1}$) and their partial derivatives with respect to the $x^\alpha$ up to the orders two and one, respectively (quantities that are themselves functions of the functions $x^\alpha (\lambda_j)$).
\begin{equation} \label{eq:4.42} \omega^r_{sij} = \int\limits_{0}^{\lambda_1} \bigg{(} \sum\limits_{t=1}^n \mathcal{Q}^r_t \omega^t_{sij} + \mathcal{Q} \omega^r_{sij} + \Omega^r_{sij} \bigg{)} d \lambda, \end{equation}
where
\begin{equation*} \Omega^r_{sij} = \sum\limits_{t=1}^n \frac{\partial \mathcal{Q}^r_t}{\partial p_j^0} \omega^t_{si} + \frac{ \partial \mathcal{Q}}{\partial p_j^0} \omega^r_{si} + \frac{\partial \Omega^r_{si}}{\partial p_j^0}, \end{equation*}
is a polynomial of the functions $\omega^r_s$, $\omega^r_{si}$, $p_i$, $y^i_j$, $z^i_j$, $y^i_{jh}$, $z^i_{jh}$ as well as of the values on $\Sigma_0$ of the coefficients $A^{\lambda \mu}$, $B^{r \lambda}_s$ and of their partial derivatives with respect to the $x^\alpha$ up to the orders three and two, respectively.

The first and second partial derivatives of the $\omega^r_s$  with respect to the variables $x^i$ are expressed by means of the functions $x^\alpha$ (with the help of the coefficients $A^{\lambda \mu}$ and of their first partial derivatives), $p_i$, $y^i_j$, $z^i_j$, $y^i_{jh}$, $z^i_{jh}$, $\omega^r_s$, $\omega^r_{si}$ and $\omega^r_{sij}$.

Then, the functions $\omega^r_s$ and their first and second derivatives with respect to the $x^i$ are expressed only through some functions $X$ and $\Omega$, $X$ denoting any whatsoever of the functions $x^\alpha$, $p_i$, $y^i_j$, $z^i_j$, $y^i_{jh}$, $z^i_{jh}$, $y^i_{jhk}$, $z^i_{jhk}$ and $\Omega$ any whatsoever among the functions $\omega^r_s$, $\omega^r_{si}$, $\omega^r_{sij}$.

The functions $X$ and $\Omega$ satisfy integral equations of the form
\begin{equation*} X = \int\limits_{0}^{\lambda_1} E(X) d\lambda + X_0, \end{equation*}
\begin{equation*} \Omega = \int\limits_{0}^{\lambda_1} F(X, \Omega) d\lambda + \Omega_0, \end{equation*}
where $X_0$ and $\Omega_0$ denote the given values of the functions $X$ and $\Omega$ for $\lambda_1=0$.

$E(X)$ is a polynomial of the functions $X$ and of the values on $\Sigma_0$ of the coefficients $A^{\lambda \mu}$ and of their partial derivatives up to the fourth order (functions of the functions $x^\alpha$).

$F(X, \Omega)$ is a polynomial of the functions $X$ and $\Omega$, and of the values on $\Sigma_0$ of the coefficients $A^{\lambda \mu}$, $B^{r \lambda}_s$ and of their partial derivatives up to the orders three and two, respectively.

\section{Behaviour in the Neighbourhood of the Vertex}
We are going to study the quantities occurring in the integrals of the fundamental relations $(\ref{eq:4.39}$), and for this purpose we will look in a more precise way for the expression of the partial derivatives of the functions $\sigma$ and $\omega^r_s$ with respect to the variables $x^i$ by means of the functions $X$ and $\Omega$. The behaviour of these functions in the neighbourhood of $\lambda_1=0$, that is the vertex of the characteristic conoid $\Sigma_0$, will make it possible for us to look for the limit of Eq. ($\ref{eq:4.39}$) for $\eta=0$: the function $x^4(\lambda_1, \lambda_2, \lambda_3)$ being, within the domain $\Lambda$, a continuous function of the three variables $\lambda_i$, $\eta= x^4 - x^4_0$ approaches actually zero with $\lambda_1$. 
First of all, we have already seen that the quantity $\frac{J_{x \lambda}}{(\lambda_1)^2}$ is a polynomial of the functions $X$, that is $p_i$ in this case,  $\tilde{X}$, that is $\frac{ y^i_j}{\lambda_1}$, of the coefficients $A^{\lambda \mu}$ and of the $sin(\lambda_u)$, $cos(\lambda_u)$. It is therefore a continuous bounded function of $\lambda_1$, $\lambda_2$ and $\lambda_3$ within $V$. We have seen that the value of this function for $\lambda_1=0$ is 
\begin{equation*} \lim_{\lambda_1 \to 0}  \frac{J_{x \lambda}}{(\lambda_1)^2} = - sin (\lambda_2). \end{equation*}
In the neighbourhood of $\lambda_1=0$ the function $\frac{J_{x \lambda}}{(\lambda_1)^2} \neq 0$ but for $\lambda_2=0$ or $\lambda_2= \pi$.
To remove this difficulty we will show that the polynomial $J_{x \lambda}$ is divisible by $sin(\lambda_2)$ and we will make sure that the function $D= \frac{J_{x \lambda}}{(\lambda_1)^2sin(\lambda_2)}$ appears in the denominators we consider.

Let us therefore consider on the conoid $\Sigma_0$ the following change of variables: $\mu_i \equiv \lambda_i p_i^0$.
We set 
\begin{equation*} d \equiv \frac{D(\mu_1, \mu_2, \mu_3)}{D(\lambda_1, \lambda_2, \lambda_3)}= {\rm det}
\left ( {\begin{array}{ccc}
p_1^0 & p_2^0 & p_3^0 \\
\lambda_1 \frac{\partial p_1^0}{\partial \lambda_2} & \lambda_1 \frac{\partial p_2^0}{\partial \lambda_2}&  \lambda_1 \frac{\partial p_3^0}{\partial \lambda_2} \\
\lambda_1 \frac{\partial p_1^0}{\partial \lambda_3} & \lambda_1 \frac{\partial p_2^0}{\partial \lambda_3} & \lambda_1 \frac{\partial p_3^0}{\partial \lambda_3} \\
\end{array} } \right )
 =(\lambda_1)^2  sin(\lambda_2), 
\end{equation*}
and 
\begin{equation*} J_{x \lambda} \equiv \frac{D(x^1, x^2, x^3)}{D(\lambda_1, \lambda_2, \lambda_3)}. \end{equation*}
Since
\begin{equation*}\frac{D(x^1, x^2, x^3)}{D(\lambda_1, \lambda_2, \lambda_3)} = \frac{D(x^1, x^2, x^3)}{D(\mu_1, \mu_2, \mu_3)} \frac{D(\mu_1, \mu_2, \mu_3)}{D(\lambda_1, \lambda_2, \lambda_3)} \end{equation*}
we have
\begin{equation} \label{eq:4.43} J_{x \lambda} = D(\lambda_1)^2 sin(\lambda_2), \end{equation}
where the determinant $D$ has elements
\begin{equation*} \frac{\partial x^i}{\partial \mu_j} = \frac{\partial x^i}{\partial \lambda_1} \frac{\partial \lambda_1 }{\partial \mu_j} + \sum\limits_{h=1}^3 \sum\limits_{u=2}^3 \frac{\partial x^i}{\partial p_h^0}\frac{\partial p_h^0}{\partial \lambda_u}\frac{\partial \lambda_u}{\partial \mu_j}. \end{equation*}
It results directly from $\mu_i =\lambda_i p_i^0$ and from the identity
\begin{equation*} \sum\limits_{i=1}^3 (\mu_i)^2 = (\lambda_1)^2 \end{equation*}
that
\begin{equation*} \frac{\partial \lambda_1}{\partial \mu_j} = p_j^0  \hspace{0.5cm} \rm{and}  \hspace{0.5cm} \frac{\partial p_h^0}{\partial \lambda_u} = \frac{1}{\lambda_1}\frac{\partial \mu_h}{\partial \lambda_u}. \end{equation*}
On the other hand, we have
\begin{equation*} \frac{\partial \lambda_1}{\partial \mu_j} \frac{\partial \mu_h}{\partial \lambda_1} + \sum\limits_{u=2}^3  \frac{\partial \lambda_u}{\partial \mu_j} \frac{\partial \mu_h}{\partial \lambda_u} = \delta^h_j. \end{equation*}
The elements of $D$ are therefore
\begin{equation*}  \frac{\partial x^i}{\partial \mu_j} = T^i p_j^0 + \sum\limits_{h=1}^3 \frac{y^i_h}{\lambda_1} (\delta^h_j - p_j^0 p_h^0). \end{equation*}

The polynomial $\frac{J_{x \lambda}}{(\lambda_1)^2}$ is therefore divisible by $sin(\lambda_2)$, the quotient $D$ being a polynomial of the same functions $X$, $\tilde{X}$ as $\frac{J_{x \lambda}}{(\lambda_1)^2}$ is of $sin(\lambda_u)$, $cos(\lambda_u)$. 

$D$ is a continuous bounded function of $\lambda_1$, $\lambda_2$, $\lambda_3$ within $V$ whose value for $\lambda_1=0$ is $\lim_{\lambda_1 \to 0} D = -1$. As a matter of fact:
\begin{equation*} \lim_{\lambda_1 \to 0} \frac{\partial x^i}{\partial \mu_j} = - p_i^0 p_j^0 - \delta^i_j + p_i^0p_j^0 = - \delta^i_j. \end{equation*}
$\frac{J_{x \lambda}}{(\lambda_1)^2}$ being a homogeneous polynomial of second degree of the functions $\frac{y^j_i}{\lambda_1}$, the same is true of the polynomial $D$, and the quantity $(\lambda_1)^2D$ is a polynomial of the functions $X$, of the coefficients $A^{\lambda \mu}$ and of the three $p_i^0$, homogeneous of the second degree with respect to the $y^i_j$.

$D$ is actually a continuous and bounded function of $\lambda_1$ in the domain $\Lambda$ (where $\lambda_2$ and $\lambda_3$ vary over a compact set) and takes the value -1 for $\lambda_1=0$. There exists therefore a number $\epsilon_2$ such that, in the domain $\Lambda_2$, neighbourhood $\lambda_1=0$ of the domain $\Lambda$, defined by
\begin{equation*} |\lambda_1| \leq \epsilon_2, \; 0 \leq \lambda_2 \leq \pi, \; 0 \leq \lambda_2 \leq 2 \pi, \end{equation*}
one has for example 
\begin{equation*} |D + 1| \leq \frac{1}{2}  \hspace{0.5cm} \rm{therefore}  \hspace{0.5cm} |D| \geq \frac{1}{2}. \end{equation*}
We will denote by $W$ the domain of $\Sigma_0$ corresponding to the domain $\Lambda_2$. Hereafter, the behaviour of the minors of $J_{x \lambda}$ is studied.

$\textbf{Minors relative to elements of the first line of $J_{x \lambda}$:}$ $J_i^1$ is, as $J_{x \lambda}$ itself, a homogeneous polynomial of second degree with respect to the functions $y^j_i$, and $\frac{J_i^1}{(\lambda_1)^2}$ is a polynomial of the functions $X(p_i)$, $\tilde{X}\big{(}\frac{y^j_i}{\lambda_1}\big{)}$, of the coefficients $[A^{\lambda \mu}]$ and of $sin(\lambda_u)$, $cos(\lambda_u)$; it is therefore a continuous and bounded function of $\lambda_1$, $\lambda_2$, $\lambda_3$ in $V$.

In order to study the quantity $\frac{J_i^1}{J_{x \lambda}}= \frac{\partial \lambda_1}{\partial x^i}$ we shall put it in the form of a rational fraction with denominator $D$, which differs from 0 in $W$. 

We have
\begin{equation} \label{eq:4.44} \frac{J_i^1}{J_{x \lambda}}= \frac{\partial \lambda_1}{\partial x^i} = \sum\limits_{j=1}^3 \frac{\partial \lambda_1}{\partial \mu_j}\frac{\partial \mu_j}{\partial x^i} = \sum\limits_{j=1}^3 p_j^0 \frac{D^j_i}{D} \end{equation}
where $D^j_i$ is the minor relative to the element $\frac{\partial x^i}{\partial \mu_j}$ of the determinant $D$.

The quantity $\frac{J_i^1}{J_{x \lambda}}$ is therefore a continuous and bounded function of the three variables $\lambda_1$, $\lambda_2$, $\lambda_3$ in $W$. When we compute the value of this function for $\lambda_1=0$, we find
\begin{equation*} \lim_{\lambda_1 \to 0}\frac{J_i^1}{J_{x \lambda} }= - p_i^0. \end{equation*}
Indeed:
\begin{equation*} \lim_{\lambda_1 \to 0} \frac{\partial \lambda_1}{\partial x^i} = \lim_{\lambda_1 \to 0} \frac{\partial x^4}{\partial x^i}, \end{equation*}
or one has constantly, over $\Sigma_0$, $\frac{\partial x^4}{\partial x^i}=-p_i$. One deduces from the formulas $(\ref{eq:4.43})$ and $(\ref{eq:4.44}$) that
\begin{equation*} J_i^1 = \sum\limits_{j=1}^3 (\lambda_1)^2 sin(\lambda_2) p_j^0 D^j_i. \end{equation*}
One then sees that the quantity $(\lambda_1)^2 \sum\limits_{j=1}^3 p_j^0D^j_i$ is a polynomial of the functions $p_i$, $y^j_i$, of the coefficients $[A^{\lambda \mu}]$ and of the three $p_h^0$, homogeneous of second degree with respect  to the functions $y^j_i$.

$\textbf{Minors relative to the second and third line of $J_{x \lambda}$:}$ $J_i^u$ is a polynomial of the functions $X(p_i, y^j_i)$, $ [A^{\lambda \mu}]$ and of $sin(\lambda_u)$, $cos(\lambda_u)$, homogeneous of first degree with respect to the functions $y^j_i$.

$\frac{J_i^u}{\lambda_1}$ is a continuous and bounded function of $\lambda_1$, $\lambda_2$, $\lambda_3$ in $V$.

Let us study the quantity $\sum\limits_{u=2}^3 \frac{\partial p_h^0}{\partial \lambda_u} \frac{J_i^u}{J_{x \lambda}}$. One has
\begin{equation*} \sum\limits_{u=2}^3 \frac{\partial p_h^0}{\partial \lambda_u} \frac{J_i^u}{J_{x \lambda}}= \sum\limits_{u=2}^3 \frac{\partial p_h^0}{\partial \lambda_u}\frac{\partial \lambda_u}{\partial x^i} = \sum\limits_{u=2}^3 \sum\limits_{j=1}^3 \frac{1}{\lambda_1}\frac{\partial \mu_h}{\partial \lambda_u} \frac{\partial \lambda_u}{\partial \mu_j} \frac{\partial \mu_j}{\partial x^i} = \sum\limits_{j=1}^3 \frac{1}{\lambda_1} (\delta^h_j - p_j^0 p_h^0 ) \frac{D_i^j}{D}. \end{equation*}
The quantity $\lambda_1 \sum\limits_{u=2}^3 \frac{\partial  p_h^0}{\partial \lambda_u}\frac{J_i^u}{J_ {x \lambda}}$ is a rational fraction with nonvanishing denominator in the domain $W$ of the functions $X(p_i)$, $\hat{X}\big{(} \frac{y^j_i}{\lambda_1} \big{)}$, $[A^{\lambda \mu}]$ and of the three $p_i^0$. It is therefore a continuous and bounded function of $\lambda_1$, $\lambda_2$, $\lambda_3$ in the domain $W$; the value of this function for $\lambda_1=0$ is computed as follows. One has on one hand
\begin{equation*} \frac{\partial x^h}{\partial \lambda_u} = \sum\limits_{j=1}^3 \frac{\partial x^h}{\partial p_j^0}\frac{\partial p_j^0}{\partial \lambda_u} = \sum\limits_{j=1}^3 y^h_j \frac{\partial p_j^0}{\partial \lambda_u}, \end{equation*}
from which
\begin{equation*} \lim_{\lambda_1 \to 0} \frac{1}{\lambda_1}\frac{\partial x^h}{\partial \lambda_u} = - \sum\limits_{j=1}^3 \delta^h_j \frac{\partial p_j^0}{\partial \lambda_u} = - \frac{\partial p_h^0}{\partial \lambda_u}. \end{equation*}
One knows on the other hand that
\begin{equation*} \frac{J_i^u}{J_{x \lambda}}= \frac{\partial \lambda_u}{\partial x^i}, \end{equation*}
from which
\begin{equation*}\lim_{\lambda_1 \to 0} \lambda_1 \sum\limits_{u=2}^3 \frac{\partial p_h^0}{\partial \lambda_u} \frac{J^u_i}{J_{x \lambda}} = - \lim_{\lambda_1 \to 0} \sum\limits_{u=2}^3 \frac{\partial x^h}{\partial \lambda_u} \frac{\partial \lambda_u}{\partial x^i} = - \delta^h_i + \lim_{\lambda_1 \to 0} \frac{\partial x^h}{\partial \lambda_1} \frac{\partial \lambda_1}{\partial x^i}, \end{equation*}
from which eventually
\begin{equation*} \lim_{\lambda_1 \to 0} \lambda_1 \sum\limits_{u=2}^3 \frac{\partial p_h^0}{\partial \lambda_u} \frac{J^u_i}{J_{x \lambda}} = - \delta^h_i + p_i^0 p_h^0. \end{equation*}
By a reasoning analogous to the one of previous remarks, one sees that the quantity $\lambda_1 \sum\limits_{j=1}^3 (\delta^h_j - p_j^0 p_h^0)D_i^j$ is a polynomial homogeneous of first degree with respect to the $y^j_i$, of the functions $X(p_i, y^j_i)$, $[A^{\lambda \mu}]$, $p_i^0$.

\section{The First Derivatives}
The first partial derivatives of an arbitrary function $\varphi$ satisfy, in light of the identity $(\ref{eq:4.40})$ and of the previous results, the relation
\begin{equation*} \frac{\partial \varphi}{\partial x^i} = \frac{\partial \varphi}{\partial \lambda_1} \sum\limits_{j=1}^3 \frac{p_j^0 D_i^j}{D} + \frac{1}{\lambda_1} \sum\limits_{h,j=1}^3 \frac{\partial \varphi}{\partial p_h^0} (\delta_j^h - p_j^0 p_h^0) \frac{D_i^j}{D}. \end{equation*}
Let us apply this formula to the functions $p_h^0$ and $X$:
\begin{equation} \label{eq:4.45}
\left \{ \begin{array} {l}
\frac{\partial p_h^0}{\partial x^i} = \frac{1}{\lambda_1} \sum\limits_{j=1}^3 (\delta_j^h - p_j^0 p_h^0 ) \frac{D_i^j}{D},\\
\frac{\partial p_h}{\partial x^i} = R_h \sum\limits_{j=1}^3 \frac{p_j^0 D_i^j}{D} + \sum\limits_{k,j=1}^3 \delta^h_k \frac{1}{\lambda_1}(\delta^k_j - p_j^0 p_k^0) \frac{D_i^j}{D}, \\
\frac{\partial y^k_h}{\partial x^i} = \sum\limits_{j=1}^3 T^k_h \frac{p_j^0 D_i^j}{E} + \frac{1}{\lambda_1} \sum\limits_{l=1}^3 y^k_{hl} (\delta^l_j - p_j^0 p_l^0) \frac{D_i^j}{D}, \\
\frac{\partial z^k_h}{\partial x^i} = \sum\limits_{j=1}^3 R^k_h \frac{p_j^0 D_i^j}{D} + \frac{1}{\lambda_1} \sum\limits_{j,l=1}^3 z^k_{hl} (\delta^l_j - p_j^0 p_l^0) \frac{D_i^j}{D}.
\end{array}\right.
\end{equation}
These equations and the analogous equations verified by $\frac{\partial y^k_h}{\partial x^i}$, $\frac{\partial z^k_h}{\partial x^i}$, $\frac{\partial \omega^r_s}{\partial x^i}$, $\frac{\partial \omega^r_{si}}{\partial x^i}$ show that the quantities $\lambda_1 \frac{\partial p_h^0}{\partial x^i}$, $\lambda_1 \frac{\partial p_h}{\partial x^i}$, $\lambda_1 \frac{\partial z^k_h}{\partial x^i}$, $\lambda_1 \frac{\partial z^k_{hl}}{\partial x^i}$, $\frac{\partial y^k_h}{\partial x^i}$, $\frac{\partial y^k_{hl}}{\partial x^i}$, $\frac{\partial \omega^r_s}{\partial x^i}$, $\frac{\partial \omega^r_{si}}{\partial x^i}$ are rational fractions with denominator $D$ of the functions $X$, $\tilde{X}$, $\Omega$, $\tilde{\Omega}$, $[A^{\lambda \mu}]$, $\big{[} \frac{\partial A^{\lambda \mu}}{\partial x^\alpha} \big{]}$, $\big{[} \frac{\partial^2 A^{\lambda \mu}}{\partial x^\alpha \partial x^\beta} \big{]}$, $p_i^0$. These are bounded and continuous functions, within $W$, of the three variables $\lambda_1$, $\lambda_2$, $\lambda_3$.

\section{Study of $\sigma$ and its derivatives}
To begin the study of $\sigma$ and its derivatives we need first to revert to the functions $\sigma^r_s$ and to study their partial derivatives with respect to $x^i$. The previous results and the identity ($\ref{eq:4.40}$) show then that the quantities $\frac{1}{\lambda_1} \frac{\partial}{\partial x^i} (\lambda_1^2 D)$, $\frac{1}{\lambda_1} \frac{\partial}{\partial x^i} \sum\limits_{j=1}^3 ((\lambda_1)^2p_j^0 D^j_i)$, $\sum\limits_{i,j=1}^3 \frac{\partial}{\partial x^i} (\lambda_1 (\delta^h_j - p_j^0 p_h^0)D_i^j)$ are rational fractions with denominator $D$ of the functions  $X(p_i, y^j_i, z^j_i)$, $\tilde{X}\big{(} \frac{y^j_i}{\lambda_1}, \frac{y^j_{ih}}{\lambda_1}\big{)}$, $[A^{\lambda \mu}]$, $\big{[} \frac{\partial A^{\lambda \mu}}{\partial x^\alpha} \big{]}$, $p_i^0$. They are therefore continuous and bounded functions of $\lambda_1$, $\lambda_2$, $\lambda_3$ in $W$.

In the study of second partial derivatives of the function $\sigma$ with respect to the $x^i$ we will use the second partial derivatives $\frac{\partial^2 ((\lambda_1)^2 D)}{\partial x^i \partial x^j}$. The first-order partial derivatives of $(\lambda_1)^2D$ can be written
\begin{equation*} \frac{\partial ((\lambda_1)^2D)}{\partial x^i}= \frac{ P_1}{(\lambda_1)^2D}, \end{equation*}
where $P_1$ is a polynomial of the functions $X(p_i, y^j_i, z^j_i, y^j_{ih})$, $[A^{\lambda \mu}]$, $\big{[} \frac{\partial A^{\lambda \mu}}{\partial x^\alpha}\big{]}$, $p_i^0$ whose terms are of the third degree with respect to the set of functions $y^j_i$, $y^j_{ih}$. As a matter of fact, the partial derivatives $\frac{\partial p_h}{\partial x^i}$ and $\frac{\partial p_h^0}{\partial x^i}$ can be put in form of rational fractions, by multiplying denominator and numerator of the right-hand side of the equations by $(\lambda_1)^2$, with denominator $(\lambda_1)^2D$ and whose numerators are polynomials of the functions $X(p_i, y^j_i, z^j_i)$, $[A^{\lambda \mu}]$, $p_i^0$ whose terms are of first degree with respect to the $y^j_i$, and the partial derivatives $\frac{\partial p_h^k}{\partial x^i}$ can be put in form of rational fractions with denominator $(\lambda_1)^2D$ and whose numerators are polynomials of the functions $X(p_i, y^j_i, z^j_i, y^j_{hk})$, $[A^{\lambda \mu}]$, $\big{[} \frac{\partial A^{\lambda \mu}}{\partial x^\alpha}\big{]}$, $p_i^0$ homogeneous of second degree with respect to the set of functions $y^0_i$, $y^i_{hk}$. The polynomial $(\lambda_1)^2D$ being homogeneous of first degree with respect to the $y^j_i$, its first partial derivatives have for sure the desired form.

Let us then consider the second partial derivatives:
\begin{equation*} \frac{\partial^2((\lambda_1)^2D)}{\partial x^i \partial x^j} = \frac{1}{(\lambda_1)^2D}\frac{\partial p_1}{\partial x^i} = \frac{P_1}{((\lambda_1)^2D)^2} \frac{\partial ((\lambda_1)^2D)}{\partial x^i}. \end{equation*}

It turns out from the form of the polynomial $P_1$ and from the previous results that:
\begin{description}
\item[(1)] $\frac{P_1}{(\lambda_1)^3}$ is a polynomial of the functions $X(p_i, y^j_i, z^j_i, y^j_{ih})$, $X \big{(} \frac{y^j_i}{\lambda_1}, \frac{y^j_{ih}}{\lambda_1}\big{)}$, $[A^{\lambda \mu}]$, $\big{[} \frac{\partial A^{\lambda \mu}}{\partial x^\alpha} \big{]}$, $p_i^0$.
\item[(2)]$\frac{1}{(\lambda_1)^2}\frac{\partial P_1}{\partial x^i}$ is a rational fraction with denominator $D$ of the functions $X(p_i, y^j_i, z^j_i, \\ y^j_{ih}, z^j_{ih}, y^j_{ihk})$, $\tilde{X} \big{(} \frac{y^j_i}{\lambda_1}, \frac{y^j_{ih}}{\lambda_1}, \frac{\partial y^j_{ihk}}{\lambda_1}\big{)}$, $[A^{\lambda \mu}]$, $\big{[} \frac{\partial A^{\lambda \mu}}{\partial x^\alpha} \big{]}$,  $\big{[} \frac{\partial^2 A^{\lambda \mu}}{\partial x^\alpha \partial x^\beta} \big{]}$ $p_i^0$.
\end{description}
The derivatives $\frac{\partial^2 ((\lambda_1)^2 D)}{\partial x^i \partial x^j}$ are therefore rational functions with denominator $D^3$ of the functions we have just listed.

To pursue our aim, let us proceed with the study of $\sigma$ and its derivatives. The auxiliary functions $\sigma$ has been defined by $\sigma= \bigg{|}{\frac{sin(\lambda_2)}{J_{x \lambda}} \bigg{|}}^\frac{1}{2}$. Since $J_{x \lambda}=D(\lambda_1)^2 sin(\lambda_2)$, we have
\begin{equation*} \sigma= \frac{1}{{|(\lambda_1)^2D|}^\frac{1}{2}}. \end{equation*}
Thus we deduce that, in the domain $W$, the function $\sigma \lambda=\frac{1}{{|D|}^\frac{1}{2}}$ is the square root of a rational fraction, bounded and non-vanishing, of the function $X$, $\tilde{X}$, $[A^{\lambda \mu}]$, $p_i^0$; it is a continuous and bounded function of the three variables $\lambda_i$, whose value for $\lambda_1=D$ is $\lim_{\lambda_1 \to 0} \sigma \lambda_1=1$. 

The first partial derivatives of $\sigma$ with respect to the $x^i$ are
\begin{equation*} \frac{\partial \sigma}{\partial x^i} = \frac{\sigma}{2} \frac{1}{(\lambda_1)^2D} \frac{\partial ((\lambda_1)^2D)}{\partial x^i}. \end{equation*}
Thus we can conclude that, in the domain $W$, the function
\begin{equation*} (\lambda_1)^2\frac{\partial \sigma}{\partial x^i}= - \frac{\sigma}{2}\frac{\lambda_1}{D} \frac{1}{\lambda_1}\frac{\partial ((\lambda_1)^2D)}{\partial x^i} \end{equation*}
is the product of the square root of a non-vanishing bounded rational fraction with a bounded rational fraction of the functions $X$, $\tilde{X}$, $[A^{\lambda \mu}]$, $\big{[} \frac{\partial A^{\lambda \mu}}{\partial x^\alpha} \big{]}$, $p_i^0$. It is a continuous and bounded function of $\lambda_1$, $\lambda_2$, $\lambda_3$ of which we are going to compute the value for $\lambda_1=0$.

The identities $\frac{\partial \sigma}{\partial \lambda_1}=\sum\limits_{i=1}^3T^i \frac{\partial \sigma}{\partial x^i}$ and $\frac{\partial \sigma}{\partial p_h^0}= \sum\limits_{i=1}^3 \frac{\partial \sigma}{\partial x^i}y^i_h$ show that the functions $(\lambda_1)^2 \frac{\partial \sigma}{\partial \lambda_1}$ and $\lambda_1 \frac{\partial \sigma}{\partial p_h^0}$ are continuous and bounded in $W$. We can therefore differentiate $\lim_{\lambda_1 \to 0} \sigma \lambda_1 = 1$ with respect to $p_h^0$, and we find $\lim_{\lambda_1 \to 0} \lambda_1 \frac{\partial \sigma}{\partial p_h^0} = 0$. Furthermore, we can write
\begin{equation*} \frac{\partial (\sigma (\lambda_1)^2)}{\partial \lambda_1} = 2 \lambda_1 \sigma + (\lambda_1)^2 \frac{\partial \sigma}{\partial \lambda_1} \end{equation*}
and $\lim_{\lambda_1 \to 0} \frac{\partial (\sigma (\lambda_1)^2)}{\partial \lambda_1} =\lim_{\lambda_1 \to 0} \lambda_1 \sigma$, from which
\begin{equation*} \lim_{\lambda_1 \to 0}(\lambda_1)^2 \frac{\partial \sigma}{\partial \lambda_1}= - \lim_{\lambda_1 \to 0} \lambda_1 \sigma = -1. \end{equation*}
In order to compute the value for $\lambda_1=0$ of the function $(\lambda_1)^2 \frac{\partial \sigma}{\partial x^i}$ we shall use the identity
\begin{equation*} (\lambda_1)^2 \frac{\partial \sigma}{\partial x^i} = (\lambda_1)^2 \frac{\partial \sigma}{\partial \lambda_1} \frac{J_1^i}{J_{x \lambda}} + \lambda_1 \sum\limits_{h,u=2}^3 \frac{\partial \sigma}{\partial p_h^0} \lambda_1 \frac{\partial p_h^0}{\partial \lambda_u} \frac{J_u^i}{J_{x \lambda}}, \end{equation*}
from which we have $\lim_{\lambda_1 \to 0} (\lambda_1)^2 \frac{\partial \sigma}{\partial x^i} = p_i^0$.

The second partial derivatives of $\sigma$ with respect to the $x^i$ are
\begin{equation*}\begin{split} \frac{\partial^2 \sigma}{\partial x^i \partial x^j}=& - \frac{\sigma}{2} \frac{1}{((\lambda_1)^2 D)} \frac{\partial^2( (\lambda_1)^2D)}{\partial x^i \partial x^j} + \frac{\sigma}{2((\lambda_1)^2 D)^2} \frac{\partial( (\lambda_1)^2D)}{\partial x^i} \frac{\partial((\lambda_1)^2 D)}{\partial x^j} \\
&- \frac{1}{2(\lambda_1)^2 D} \frac{\partial \sigma}{\partial x^j} \frac{\partial ((\lambda_1)^2D)}{\partial x^i}. \end{split}\end{equation*}
It is easily seen that in the domain $W$ the function $(\lambda_1)^3 \frac{\partial^2 \sigma}{\partial x^i \partial x^j}$ is the product of the square root of a non-vanishing bounded rational fraction with a bounded rational fraction, having denominator $D^4$, of the functions $X$, $\tilde{X}$, $[A^{\lambda \mu}]$, $\big{[} \frac{\partial A^{\lambda \mu}}{\partial x^\alpha} \big{]}$, $ \big{[} \frac{\partial^2 A^{\lambda \mu}}{\partial x^\alpha \partial x^\beta} \big{]}$, $p_i^0$. It is a continuous and bounded function of the three variables $\lambda_i$. We are going to compute the value for $\lambda_1=0$ of the function $(\lambda_1)^3 \sum\limits_{i=0}^3\frac{\partial^2 \sigma}{\partial {x^i}^2}$ which, only, we will need: the second partial derivatives of $\sigma$ do not occur actually in the fundamental equations except for the quantity $\sum\limits_{i,j=1}^3 [A^{ij}] \frac{\partial^2 \sigma}{\partial x^i \partial x^j}$, and one has
\begin{equation*} \lim_{\lambda_1 \to 0} \sum\limits_{i,j=1}^3 [A^{ij}](\lambda_1)^3 \frac{\partial^2 \sigma}{\partial x^i \partial x^j} = \lim_{\lambda_1 \to 0}(\lambda_1)^3 \sum\limits_{i=1}^3 \frac{\partial^2 \sigma}{\partial ({x^i})^2}. \end{equation*}
Furthermore, by differentiating $\lim_{\lambda_1 \to 0} (\lambda_1)^2 \frac{\partial \sigma}{\partial x^i}=p_i^0$ with respect to $p_h^0$ we have
\begin{equation*} \lim_{\lambda_1 \to 0} (\lambda_1)^2 \frac{\partial}{\partial p_h^0} \bigg{(}\frac{\partial \sigma}{\partial x^i}\bigg{)}=\delta^i_h; \end{equation*}
but, on the other hand, we also have
\begin{equation*} \frac{\partial}{\partial \lambda_1} \bigg{[}(\lambda_1)^3 \frac{\partial \sigma}{\partial x^i} \bigg{]} = 3 (\lambda_1)^2 \frac{\partial \sigma}{\partial x^i} + (\lambda_1)^3 \frac{\partial}{\partial \lambda_1} \bigg{(}\frac{\partial \sigma}{\partial x^i}\bigg{)}, \end{equation*}
from which
\begin{equation*} \lim_{\lambda_1 \to 0} (\lambda_1)^3 \frac{\partial}{\partial \lambda_1} \bigg{(} \frac{\partial \sigma}{\partial x^i} \bigg{)} = \lim_{\lambda_1 \to 0} \bigg{(} - 2 (\lambda_1)^2 \frac{\partial \sigma}{\partial x^i} \bigg{)} = - p_i^0. \end{equation*}
We find therefore, by using the identity
\begin{equation*} \sum\limits_{i=1}^3 (\lambda_1)^3 \frac{\partial^2 \sigma}{\partial (x^i)^2} = (\lambda_1)^3 \sum\limits_{i=1}^3 	\frac{\partial}{\partial \lambda_1}\bigg{(} \frac{\partial \sigma}{\partial x^i} \bigg{)} \frac{J_i^1}{J_{x \lambda}} + (\lambda_1)^3 \sum\limits_{h=1}^3 \sum\limits_{u=2}^3 \frac{\partial}{\partial p_h^0} \bigg{(}\frac{\partial \sigma}{\partial x^i} \bigg{)} \frac{\partial p_h^0}{\partial \lambda_u} \frac{J_i^u}{J_{x \lambda}} \end{equation*}
and the previous results, that
\begin{equation*} \lim_{\lambda_1 \to 0} \sum\limits_{i=1}^3 (\lambda_1)^3 \frac{\partial^2 \sigma}{\partial (x^i)^2} =0. \end{equation*}
Let us show that the function
\begin{equation*} (\lambda_1)^2 \sum\limits_{i,j=1}^3 [A^{ij}]\frac{\partial^2 \sigma}{\partial x^i \partial x^j} \end{equation*}
is a continuous and bounded function of the three variables $\lambda_i$, in the neighbourhood of $\lambda_1=0$.

We have seen that $(\lambda_1)^3\sum\limits_{i,j=1}^3 \frac{\partial^2 \sigma}{\partial x^i \partial x^j}[A^{ij}]$ is the product of a square root of a non-vanishing bounded rational fraction $\big{(} \frac{1}{D} \big{)}$ with a rational fraction having denominator $D^4$, whose numerator, polynomial of the functions $X$, $\tilde{X}$, $[A^{\lambda \mu}]$, $\big{[} \frac{\partial A^{\lambda \mu}}{\partial x^\alpha} \big{]}$, $ \big{[} \frac{\partial^2 A^{\lambda \mu}}{\partial x^\alpha \partial x^\beta} \big{]}$, $p_i^0$, vanishes for the values of these functions corresponding to $\lambda_1=0$. We have
\begin{equation*} \sum\limits_{i,j=1}^3 (\lambda_1)^3 [A^{ij}] \frac{\partial^2 \sigma}{\partial x^i \partial x^j} = \frac{P \bigg{(} X, \tilde{X}, [A^{\lambda \mu}], \big{[} \frac{\partial A^{\lambda \mu}}{\partial x^\alpha} \big{]},  \big{[} \frac{\partial^2 A^{\lambda \mu}}{\partial x^\alpha \partial x^\beta} \big{]}, p_i^0 \bigg{)}}{D^4} \frac{1}{{|D|}^\frac{1}{2}} \end{equation*}
with
\begin{equation*} P_0 = P \bigg{(} X_0, \tilde{X}_0, \pm \delta^\mu_\lambda, \bigg{[} \frac{\partial A^{\lambda \mu}}{\partial x^\alpha} \bigg{]}_0,  \bigg{[} \frac{\partial^2 A^{\lambda \mu}}{\partial x^\alpha \partial x^\beta} \bigg{]}_0, p_i^0 \bigg{)}=0. \end{equation*}
We then write:
\begin{equation} \label{eq:4.46} (\lambda_1)^3 \sum\limits_{i,j=1}^3 [A^{ij}] \frac{\partial^2 \sigma}{\partial x^i \partial x^j}= \frac{(P - P_0)}{D^4} \frac{1}{{|D|}^\frac{1}{2}}. \end{equation}
By applying the Taylor formula for $P$ one sees that the quantity ($\ref{eq:4.46}$) is a polynomial of the functions $X- X_0$, $\tilde{X} - \tilde{X}_0$, $A^{\lambda \mu} \pm \delta^\mu_\lambda$, ..., whose terms are of first degree with respect to the set of these functions.

To show that $(\lambda_1)^2 \sum\limits_{i,j=1}^3[A^{ij}]\frac{\partial^2 \sigma}{\partial x^i \partial x^j}$ is a continuous and bounded function of $\lambda_1$, $\lambda_2$, $\lambda_3$ in the domain $W$, it is enough to show that the same holds for the functions
\begin{equation*} \frac{(X- X_0)}{\lambda_1}, \; \frac{(\tilde{X}- \tilde{X}_0)}{\lambda_1}, \; \frac{[A^{\lambda \mu}] - \delta^\lambda_\mu}{\lambda_1}, \; ..., \; \frac{ \big{[} \frac{\partial^2 A^{\lambda \mu}}{\partial x^\alpha \partial x^\beta} \big{]} - \big{[} \frac{\partial^2 A^{\lambda \mu}}{\partial x^\alpha \partial x^\beta} \big{]}_0}{\lambda_1}. \end{equation*}
The functions $X$ verify
\begin{equation*} X= X_0 + \int\limits_0^{\lambda_1} E(X) d \lambda, \end{equation*}
$\frac{(X- X_0)}{\lambda_1}$ is therefore a continuous and bounded function of the $\lambda_i$ in $V$: 
$$|X-X_0| \leq \lambda_1 M.$$
The coefficients $A^{\lambda \mu}$ possessing in $D$ partial derivatives continuous and bounded up to the fourth order with respect to the $x^\alpha$, the $x^\alpha$ fulfilling the previous inequalities, we see that
\begin{equation} \label{eq:4.47} [A^{\lambda \mu}] - \delta^\lambda_\mu \leq \lambda_1 A, \; ..., \; \bigg{[} \frac{\partial^3 A^{\lambda \mu}}{\partial x^\alpha \partial x^\beta \partial x^\gamma} \bigg{]} - \bigg{[} \frac{\partial^3 A^{\lambda \mu}}{\partial x^\alpha \partial x^\beta \partial x^\gamma} \bigg{]}_0 \leq \lambda_1 A. \end{equation}
Let us consider $\frac{(X- X_0)}{\lambda_1}$. The corresponding $X$ functions are $y^j_i$, $y^j_{ih}$, $y^j_{ihk}$ which verify the equation $X= \int\limits_{0}^{\lambda_1} E(X) d \lambda$, $E(X)$ being a polynomial of the functions $X$, of the $A^{\lambda \mu}$ and of their partial derivatives up to the third order $\big{[}\frac{\partial A^{\lambda \mu}}{\partial x^\alpha} \big{]}$, ..., $\big{[} \frac{\partial^3 A^{\lambda \mu}}{\partial x^\alpha \partial x^\beta \partial x^\gamma} \big{]}$. We have
\begin{equation*} \tilde{X} - \tilde{X}_0 = \frac{ \int\limits_{0}^{\lambda_1} (E(X) - E(X)_0) d \lambda}{(\lambda_1)^2}. \end{equation*}
The Taylor formula applied to the polynomial $E$ shows that $E(X) - E(X)_0$ is a polynomial of the functions $X_0$, $\delta^\mu_\lambda$, ..., $\big{[}\frac{\partial^3 A^{\lambda \mu}}{\partial x^\alpha \partial x^\beta \partial x^\gamma} \big{]}_0$ and of the functions $X - X_0$, $[A^{\lambda \mu}] - \delta^\mu_\lambda$, ..., $\big{(}\big{[} \frac{\partial^3 A^{\lambda \mu}}{\partial x^\alpha \partial x^\beta \partial x^\gamma} \big{]} - \big{[} \frac{\partial^3 A^{\lambda \mu}}{\partial x^\alpha \partial x^\beta \partial x^\gamma} \big{]}_0\big{)}$ whose terms are of first degree with respect to this last set of terms.

All these functions being bounded in $V$ and satisfying 
\begin{equation*} |X - X_0| \leq \lambda_1 M, \; [A^{\lambda \mu}] - \delta^\lambda_\mu \leq \lambda_1 A, \; ..., \; \bigg{[} \frac{\partial^3 A^{\lambda \mu}}{\partial x^\alpha \partial x^\beta \partial x^\gamma} \bigg{]} - \bigg{[} \frac{\partial^3 A^{\lambda \mu}}{\partial x^\alpha \partial x^\beta \partial x^\gamma} \bigg{]}_0 \leq \lambda_1 A, \end{equation*}
we see easily that $\frac{(\tilde{X} - \tilde{X}_0)}{\lambda_1}$ is continuous and bounded in $V$. The function $(\lambda_1)^2 \sum\limits_{i,j=1}^3[A^{ij}]\frac{\partial^2 \sigma}{\partial x^i \partial x^j}$ is therefore continuous and bounded in $W$.

\section{Derivatives of the $\omega^r_s$}
Let us now study the first and second partial derivatives of the $\omega^r_s$ with respect to the $x^i$.
Our aim is to prove that the first and second partial derivatives of the $\omega^r_s$ with respect to the $x^i$ are, as $\sigma$ and its partial derivatives, simple algebraic functions of the functions $X$ and $\Omega$, $\tilde{X}$ and $\tilde{\Omega}$, and of the values on the conoid $\Sigma_0$ of the coefficients of the given equations and of their partial derivatives.

The first partial derivatives of the $\omega^r_s$ with respect to the $x^i$ are expressed as functions of their partial derivatives with respect to the $\lambda_i$
\begin{equation*} \frac{\partial \omega^r_s}{\partial x^i} = \sum\limits_{j=1}^3 \frac{\partial \omega^r_s}{\partial \lambda_j} \frac{J_i^j}{J_{x \lambda}}, \end{equation*}
therefore
 \begin{equation} \label{eq:4.48} \frac{\partial \omega^r_s}{\partial x^i} = \bigg{(} \sum\limits_{t=1}^n \mathcal{Q}^r_t \omega^t_s + \mathcal{Q}\omega^r_s \bigg{)} \sum\limits_{j=1}^3 \frac{P_j^0 D^j_i}{D} + \sum\limits_{h,j=1}^3 \frac{\omega^r_{sh}}{\lambda_1} \frac{(\delta^h_j - P_j^0 P_h^0)D^j_i}{D}. \end{equation}
The first partial derivatives of the $\omega^r_s$ with respect to the $x^i$ are therefore rational fraction with denominator $D$ of the functions $X(P_i, y^j_i )$, $\Omega(\omega^r_s)$, $\tilde{X}\big{(} \frac{y^j_i}{\lambda_1} \big{)}$, $\tilde{\Omega} \big{(}\frac{\omega^r_{sh}}{\lambda_1} \big{)}$, $[A^{\lambda \mu}]$, $\big{[} \frac{\partial A^{\lambda \mu}}{\partial x^\alpha} \big{]}$, $[B^{s \lambda}]$ and $P_i^0$. These are continuous and bounded functions in $W$.

The second partial derivatives of the $\omega^r_s$ with respect to the $x^i$ can be evaluated by writing $\frac{\partial \omega^r_s}{\partial x^i}$ in the form $\frac{\partial \omega^r_s}{\partial x^i} = \frac{P_2}{(\lambda_1)^2D}$. The equality ($\ref{eq:4.48}$) and the previous remarks show that $P_2$ is a homogeneous polynomial of second degree with respect to the set of functions $y^j_i$, $\omega^r_s$. By differentiating the previous equality, we have
\begin{equation*} \frac{\partial^2 \omega^r_s}{\partial x^i \partial x^j} = \frac{1}{(\lambda_1)^2D} \frac{\partial P_2}{\partial x^j} - \frac{P_2}{((\lambda_1)^2D)^2} \frac{\partial ((\lambda_1)^2D)}{\partial x^i}. \end{equation*}
These functions $\lambda_1 \frac{\partial^2 \omega^r_s}{\partial x^i \partial x^j} $ are rational fractions with denominator $D^3$ of the functions $X$, $\Omega$, $\tilde{X}$, $\tilde{\Omega}$, $[A^{\lambda \mu}]$, $\big{[}\frac{\partial A^{\lambda \mu}}{\partial x^\alpha} \big{]}$, $\big{[}\frac{\partial^2 A^{\lambda \mu}}{\partial x^\alpha \partial x^\beta} \big{]}$, $[B^{s \lambda}_r]$, $\big{[} \frac{\partial B^{s \lambda}_r}{\partial x^\alpha} \big{]}$. These are therefore continuous and bounded functions in $W$.

\section{Kirchhoff Formulae}
We can now study in a more precise way the fundamental equations 
\begin{equation*} \begin{split}& \underbrace{\int \int \int}_{V_\eta} \sum\limits_{r=1}^n \{ [u_r]L^r_s + \sigma^r_s[f_r] \} dV + \underbrace{\int \int}_{S_\eta} \sum\limits_{i=1}^3 E^i_s cos(n,x^i)dS \\
& - \underbrace{\int \int}_{S_0} \sum\limits_{i=1}^3 E^i_s cos(n,x^i)dS=0, \end{split}\end{equation*}
and look for their limit as $\eta$ approaches zero. We have seen that the functional determinant $D= \frac{D(x^i)}{D(\lambda_j)}$ is equal to -1 for $\lambda_1=0$. The correspondence between the parameters $x^i$ and $\lambda_j$ is therefore surjective in a neighbourhood  of the vertex $M_0$ of $\Sigma_0$. One derives from this that the correspondence between the parameters $x^i$ and $\lambda_j$ is one-to-one in a domain $\Lambda_\eta$ defined by
\begin{equation*} \eta \leq \lambda_1 \leq \epsilon_3, \hspace{0.5cm} 0 \leq \lambda_2 \leq \pi, \hspace{0.5cm} 0 \leq \lambda_3 \leq 2 \pi, \end{equation*}
where $\epsilon_3$ is a given number and where $\eta$ is arbitrarily small.

To the domain $\Lambda_\eta$ of variations of the $\lambda_i$ parameters there corresponds, in a one-to-one way, a domain $W_\eta$ of $\Sigma_0$, because the correspondence between $(x^4, \lambda_2, \lambda_3)$ and $(\lambda_1, \lambda_2, \lambda_3)$ is one-to-one. We shall then assume that the coordinate $x^4_0$ of the vertex $M_0$ of $\Sigma_0$ is sufficiently small to ensure that the domain $V_\eta \subset V$, previously considered, is interior to the domains $W$ and $W_\eta$. We can, under these conditions, compute the integrals by means of the parameters $\lambda_i$, the integrals that we are going to obtain being convergent. For this purpose, let us evaluate the Area and Volume elements. 

We have $dV= dx^1\; dx^2\; dx^3 = d\lambda_1 \; d\lambda_2 \;d\lambda_3$. Then, we begin by computing the element of Area $dS$. The surfaces $S_0$ and $S_\eta$ are $x^4={\rm const.}$ surfaces on the characteristic conoid $\Sigma_0$. Thus, they satisfy the differentiation relation
\begin{equation*} \sum\limits_{i=1}^3 p_i dx^i =0, \end{equation*}
from which we have
\begin{equation*} cos(n, x^i)= \frac{p_i}{{\bigg{(} \sum\limits_{j=1}^3 (p_j)^2} \bigg{)}^\frac{1}{2}}. \end{equation*}
In order to evaluate $dS$ we shall write an alternative expression of the Volume element $dV$ in which the surfaces $S(x^4= cost.)$ and the bicharacteristics are involved
\begin{equation*}dV = cos (\nu) {|T|}^\frac{1}{2} d\lambda_1 dS, \end{equation*}
where ${|T|}^\frac{1}{2} d\lambda_1$ denotes the length element of the bicharacteristic, and $\nu$ is the angle formed by the bicharacteristic with the normal to the surface $S$ at the point considered. 

A system of directional parameters of the tangent to the bicharacteristic being
\begin{equation*} T^h = \sum\limits_{j=1}^3 [A^{hj}]p_j + [A^{h4}], \end{equation*}
we have
\begin{equation*} cos(\nu) {|T|}^\frac{1}{2} = \sum\limits_{h=1}^3 \Biggl\{ \sum\limits_{j=1}^3 [A^{hj}]p_j + [A^{h4}] \Biggr\} cos (n, x^h), \end{equation*}
from which, by comparing the two expressions of $dV$,
\begin{equation*} cos(n, x^i)dS= \frac{J_{x \lambda}p_i d\lambda_2 d\lambda_3}{\sum\limits_{h,j=1}^3[A^{hj}]p_j p_h + \sum\limits_{h=1}^3[A^{h4}]p_h} = \frac{- J_{x \lambda} p_i}{[A^{44}] + \sum\limits_{j=1}^3 [A^{j4}]p_j} d\lambda_2 d\lambda_3. \end{equation*}
Hence the integral relations read in terms of the $\lambda_i$ as
\begin{equation} \begin{split} \label{eq:4.49} & \int \int \int_{V_\eta} \sum\limits_{r=1}^n ([u_r]L^r_s + \sigma^r_s[f_r])d\lambda_1 d\lambda_2 d\lambda_3 - \int\limits_{0}^{2 \pi} \int\limits_{0}^{\pi} \frac{\sum\limits_{i=1}^3 E^i_s J_{x \lambda} p_i}{[A^{44}] + \sum\limits_{i=1}^3 [A^{i4}]p_i}d\lambda_{2|_{x^4=0}}d\lambda_3 \\
&= -  \int\limits_{0}^{2 \pi} \int\limits_{0}^{\pi} \Biggl\{ \frac{\sum\limits_{i=1}^3 E^i_s J_{x \lambda} p_i}{[A^{44}] + \sum\limits_{i=1}^3 [A^{i4}]p_i}\Biggr\}_{x^4 = x^4_0 - \eta} d\lambda_2 d\lambda_3. \end{split} \end{equation}
The previous results prove that the quantities to be integrated are continuous and bounded functions of the variables $\lambda_i$. They read actually as:
\begin{equation*} (\lambda_1)^2 \sum\limits_{r=1}^n \{ [u_r]L^r_s + \sigma^r_s[f_r] \} \frac{J_{x \lambda}}{(\lambda_1)^2} \; {\rm and} \; (\lambda_1)^2 \sum\limits_{i=1}^n E^i_s \frac{J_{x \lambda}}{(\lambda_1)^2} \frac{p_i}{T^4}. \end{equation*}
$E^i_s$ and $L^r_s$ being given by the equalities $(\ref{eq:4.28})$ and ($\ref{eq:4.29}$), the quantities considered are continuous and bounded in $W$ if the functions $u_r$ and $\frac{\partial u_r}{\partial x^\alpha}$ are continuous and bounded in $D$. Thus, when $\eta$ approaches zero, the two sides of $(\ref{eq:4.49})$ tend towards a finite limit. In particular, the triple integral tends to the value of this integral taken over the portion $V_0$ of hypersurface of the conoid $\Sigma_0$ in between the vertex $M_0$ and the initial surface $x^4=0$. Let us evaluate the limit of the double integral on the right-hand side. All terms of the quantity $(\lambda_1)^2E^i_s$ approach uniformly zero with $\lambda_1$ but 
\begin{equation*} - (\lambda_1)^2 \sum\limits_{r=1}^n \sum\limits_{j=1}^3 [u_r][A^{ij}]\omega^r_s \frac{\partial \sigma}{\partial x_j}, \end{equation*}
tends, when $\lambda_1$ approaches zero, to
\begin{equation*} \sum\limits_{r=1}^n \sum\limits_{j=1}^3[u_r(x^\alpha_0)]\delta^j_i \delta^r_s p_j^0 = u_s (x^\alpha_0)p_i^0. \end{equation*}
Hence we obtain
\begin{equation*} \lim_{\lambda_1 \to 0} \frac{\sum\limits_{i=1}^3 E^i_s J_{x \lambda} p_i}{[A^{44}] + \sum\limits_{i=1}^3 [A^{i4}]p_i}= - u_s(x^\alpha_0) sin(\lambda_2). \end{equation*}
The right-hand side of Eq. ($\ref{eq:4.49}$), when $\eta$ approaches zero, tends to
\begin{equation*} \int\limits_{0}^{2 \pi} \int\limits_{0}^{\pi} u_s(x^\alpha_0) sin(\lambda_2) d\lambda_2 d\lambda_3 = 4 \pi u_s(x^\alpha_0). \end{equation*}
Eventually, under the limit for $\eta \rightarrow 0$, the Eqs. ($\ref{eq:4.49}$) become
\begin{equation} \begin{split} \label{eq:4.50} 4 \pi u_s(x^\alpha_0) &= \int \int \int_{V_\eta} \sum\limits_{r=1}^n ([u_r]L^r_s + \sigma^r_s[f_r]) J_{x \lambda} d\lambda_1 d\lambda_2 d\lambda_3 \\
& +\int\limits_{0}^{2 \pi} \int\limits_{0}^{\pi} \Biggl\{ \frac{\sum\limits_{i=1}^3 E^i_s J_{x \lambda}p_i}{T^4} \Biggr\}_{x^4=0} d\lambda_2 d\lambda_3, \end{split} \end{equation}
known as the $\textit{Kirchhoff formulae}$. In order to compute its right-hand side, it will be convenient to take for parameters, on the hypersurface of the conoid $\Sigma_0$, the three independent variables $x^4$, $\lambda_2$, $\lambda_3$. Thus, the previous formulae read as
\begin{equation} \begin{split} \label{eq:4.51} 4 \pi u_s(x_j) &= \int_{x^4_0}^{0} \int_{0}^{2 \pi} \int_{0}^\pi \sum\limits_{r=1}^n ([u_r]L^r_s + \sigma^r_s[f_r]) \frac{J_{x \lambda}}{T^4} dx^4 d\lambda_2 d\lambda_3 \\
& +\int\limits_{0}^{2 \pi} \int\limits_{0}^{\pi} \Biggl\{ \frac{\sum\limits_{i=1}^3 E^i_s J_{x \lambda}p_i}{T^4} \Biggr\}_{x^4=0} d\lambda_2 d\lambda_3. \end{split} \end{equation}
The quantity under the sign of triple integral is expressed by means of the functions $[u]$ and of the functions $X(\lambda_1, \lambda_2, \lambda_3)$ and $\Omega(\lambda_1, \lambda_2, \lambda_3)$, solutions of the integral equations ($\ref{eq:4.7}$), ($\ref{eq:4.8}$), ($\ref{eq:4.9}$) and ($\ref{eq:4.35}$).

We shall obtain the expression of the $X$ and $\Omega$ as functions of the new variables $x^4$, $\lambda_2$, $\lambda_3$ by replacing $\lambda_1$ with its value defined by Eq. ($\ref{eq:4.16}$), function of the $x^4$, $\lambda_2$, $\lambda_3$.

These functions satisfy the integral equations
\begin{equation*} X(x^4, \lambda_2, \lambda_3)= X_0(x^4_0, \lambda_2, \lambda_3) + \int_{x^4_0}^{x^4} \frac{E(X)}{T^4} d \omega^4, \end{equation*}
\begin{equation*} \Omega(x^4, \lambda_2, \lambda_3)= \Omega_0(x^4_0, \lambda_2, \lambda_3) + \int_{x^4_0}^{x^4} \frac{F(X,\Omega)}{T^4} d \omega^4.\end{equation*}
The quantity under sign of double integral is expressed by means of the values for $x^4=0$ of the Cauchy data, $[u]$ and $\big{[} \frac{\partial u}{\partial x^\alpha}\big{]}$, and of the values for $x^4=0$ of the functions $X$ and $\Omega$. 
Thus, it is possible to conclude that:

$\mathbf{Conclusion}$ Every solution of the equations
\begin{equation*}E_r = \sum\limits_{\lambda,\mu=1}^4 A^{\lambda \mu} \frac{\partial^2 u_s}{\partial x^\lambda \partial x^\mu} + \sum\limits_{s=1}^n \sum\limits_{\mu=1}^4 {{B^s}_r}^\mu \frac{\partial u_s}{\partial x^\mu} + f_r =0, \hspace{1cm} r=1, 2, ..., n. \end{equation*}
continuous, bounded and with first partial derivatives continuous and bounded in $D$ verifies the integral relations ($\ref{eq:4.51}$) if the coordinates $x^\alpha_0$ of $M_0$ satisfy the inequalities of the form
\begin{equation*} |x^4_0| \leq \epsilon_0, \hspace{1cm} |x^i_0 - \tilde{x}^i_0| \leq d, \end{equation*}
defining a domain $D_0 \subset D$.

\section{Application of the Results}
We are going to estabilish formulae analogous to ($\ref{eq:4.51}$), verified by the solutions of the given equations $[E]$ at every point of a domain $D_0$ of space-time, where the values of coefficients will be restricted uniquely by the requirement of having to verify some conditions of normal hyperbolicity and differentiability.

Let us consider the system $[E]$ of equations
\begin{equation*} E_r = \sum\limits_{\lambda,\mu=1}^4 A^{\lambda \mu} \frac{\partial^2 u_s}{\partial x^\lambda \partial x^\mu} + \sum\limits_{s=1}^n \sum\limits_{\mu=1}^4 {B^{r \lambda}_s} \frac{\partial u_r}{\partial x^\lambda} + f_s =0, \hspace{0.5cm} s=1, 2, ..., n. \end{equation*} 
We assume that in the space-time domain $D$, defined by
\begin{equation*} |x^4| \leq \epsilon, \hspace{1cm} |x^i - \tilde{x}^i| \leq d, \end{equation*}
where the three $\tilde{x}^i$ are given numbers, the equations $[E]$ are of the normal hyperbolic type, i.e.
\begin{equation*} A^{44}>0, \, \rm{the} \; \rm{quadratic} \; \rm{form} \; \sum\limits_{i,j=1}^3A^{ij}X_iX_j \; \rm{is} \; \rm{negative-definite}. \end{equation*}
At every point $M_0(x_j)$ of the domain $D$ we can associate to the values $A^{\lambda \mu}(x^\alpha_0)$ of the coefficients $A$ a system of real numbers $a^{\alpha \beta}_0$, algebraic functions, defined and indefinitely differentiable of the $A^{\lambda \mu}_0= A^{\lambda \mu}(x^\alpha_0)$, satisfying the identity
\begin{equation*}\sum\limits_{i,j=1}^4 A^{\lambda \mu}_0 X_\lambda X_\mu = \bigg{(}\sum\limits_{\alpha=1}^4 a^{4 \alpha}_0 X_\alpha \bigg{)}^2 - \bigg{(}\sum\limits_{\alpha=1}^4 a^{i \alpha}_0 X_\alpha \bigg{)}^2. \end{equation*}
We shall denote by $a^0_{\alpha \beta}$ the quotient by the determinant $a_0$ of elements $a^{\alpha \beta}_0$ of the minor relative to the element $a^{\alpha \beta}_0$ of this determinant. The quantities $a^0_{\alpha \beta}$ are algebraic functions defined and indefinitely differentiable of the $A^{\lambda \mu}_0$ in $D$. The square of the determinant $a_0$, being equal to the absolute value $A$ of the determinant having elements $A^{\lambda \mu}$, $a_0$, is different from zero in $D$.

Let us perform the linear change of variables
\begin{equation*} y_\alpha \equiv \sum\limits_{\beta=1}^4 a^0_{\alpha \beta} x^\beta. \end{equation*}
The partial derivatives of the unknown functions $u_s$ are covariant under such a change of variables, hence the equations $[E]$ read as 
\begin{equation} \label{eq:4.52}  \sum\limits_{\alpha ,\beta=1}^4 A^{*\alpha \beta} \frac{\partial^2 u_s}{\partial y^\alpha \partial y^\beta} + \sum\limits_{r=1}^n \sum\limits_{\alpha=1}^4 {{B_s}^{*r \alpha}} \frac{\partial u_s}{\partial y^\alpha} + f_r =0, \end{equation} 
with 
\begin{equation} \label{eq:4.53} A^{* \alpha \beta}= \sum\limits_{\lambda, \mu=1}^4 A^{\lambda \mu} a^0_{\alpha \lambda} a^0_{\beta \mu}, \end{equation}
\begin{equation} \label{eq:4.54} B_s^{*r\alpha} = \sum\limits_{\lambda=1}^4 B_s^{r \lambda}a^0_{\alpha \lambda}. \end{equation}
The coefficients of Eq. ($\ref{eq:4.52}$) take at the point $M_0$ the values ($\ref{eq:4.11}$). As a matter of fact:
\begin{equation*} \begin{split} A_0^{* \alpha \beta} &= \sum\limits_{\lambda, \mu=1}^4 A^{\lambda \mu} a^0_{\alpha \lambda} a^0_{\beta \mu}= - \sum\limits_{\lambda, \mu, \gamma=1}^4 a_0^{\gamma \lambda } a_0^{\gamma \mu}   a^0_{\alpha \lambda} a^0_{\beta \mu} + 2 \sum\limits_{\lambda, \mu=1}^4 a_0^{4 \lambda} a_0^{4 \mu} a^0_{\alpha \lambda} a^0_{\beta \mu} \\
&= - \delta^\beta_\alpha + 2 \delta^4_\alpha \delta^4_\beta,  \end{split}\end{equation*}
hence one has 
\begin{equation*} A^{*44}=1, \hspace{0.5cm} A^{*i4}= 0, \hspace{0.5cm} A^{*ij}= - \delta^{ij}. \end{equation*}
We can apply to the equations $[E]$, written in the form $(\ref{eq:4.52})$, in the variables $y^\alpha$ and for the corresponding point $M_0$, the results that we obtained before. The integration parameters so introduced will be $y^4$, $\lambda_2$, $\lambda_3$ but, the surface carrying the Cauchy data being always $x^4\equiv a_0^{\alpha 4} y^4=0$, the integration domains will be determined from $M_0$ and the intersection of this surface with the characteristic conoid with vertex $M_0$. We see that it will be convenient, in order to evaluate these integrals, to choose the variables $y^\alpha$ relative to a point $M_0$ whatsoever in such a way that the initial space section, $x^4=0$, is a hypersurface $y^4=0$. It will be enough for that purpose to choose the coefficients $a_0^{\alpha \beta}$ in such a way that $a_0^{i4}=0$. We shall then have
\begin{equation*} a^0_{4i}=0, \; a^0_{44}=\frac{1}{a_0^{44}}= (A^{44}_0)^{-\frac{1}{2}} \; \rm{and} \; y_4 = a^0_{44} x^4, \end{equation*}
where $a^0_{44}$ is a bounded positive number.

The application of the results proves then the existence of a domain $D_0 \subset D$, defined by $|x^4_0| \leq \epsilon$, which implies at every point $M_0 \in D_0$, $|y^4_0| \leq \eta$, such that one can write at every point $M_0$ of $D_0$ a Kirchhoff formula whose first member is the value at $M_0$ of the unknown $u_s$, in terms of the quantities $y^\alpha_0 = \sum\limits_{\beta=1}^4 a^0_{\alpha \beta} x_0^\beta$, and whose right-hand side consists of a triple integral and of a double integral. The quantities to be integrated are expressed by means of the functions $X(y^4, \lambda_2, \lambda_3, y^\alpha_0)$ representing $(y^\alpha, p_i, y^j_i, z^j_i, ..., z^j_{ihk})$ and $\Omega(y^4, \lambda_2, \lambda_3)(\omega^r_s, ..., \omega^r_{sij})$, solutions of an equation of the kind
\begin{equation} \label{eq:4.55} X = X_0 + \int\limits_{y^4_0}^{y^4} E^*(X) dY^4, \hspace{0.5cm}  \Omega = \Omega_0 + \int\limits_{y^4_0}^{y^4} F^*(X, \Omega) dY^4, \end{equation}
where the functions $E^*$ and $F^*$ are the functions $E$ and $F$ considered before, but evaluated starting from the coefficients ($\ref{eq:4.53}$) and ($\ref{eq:4.54}$) and from their partial derivatives with respect to the $y^\alpha$, and where $\Omega_0$, $X_0$ denote the values for $y^4=y^4_0$ of the corresponding functions $\Omega$, $X$.

In order to obtain, under a simpler form, some integral equations holding in the whole domain $D_0$, we will take as integration parameter, in place of $y^4$, $x^4$. Also, we shall replace those of the auxiliary unknown functions $X$ which are the values (in terms of the three parameters) of the coordinates $y^\alpha$ of a point of the conoid $\Sigma_0$ of vertex $M_0$, with the values of the original coordinates $x^\alpha$ of a point of this conoid.

We shall replace, for that purpose, those of the integral equations which have on the left-hand side $y^\alpha$ with their linear combinations of coefficients $a^{\alpha \beta}_0$, i.e. with the equations of the same kind 
\begin{equation*} \sum\limits_{\beta=1}^4 a^{\alpha \beta}_0 y^\beta = x^\alpha = x^\alpha_0 + \int\limits_{x^4_0}^{x^4} \sum\limits_{\beta=1}^4 a^{\alpha \beta}_0 \frac{T^{*\alpha \beta}}{T^{*4}} a^0_{44} d\omega^4, \end{equation*}
and we will replace the quantities under integration signs of all our equations in terms of the $x^\alpha$ in place of the $y^\beta$ by replacing in these equations the $y^\beta$ with the linear combinations $\sum\limits_{\alpha=1}^4 a^0_{\alpha \beta}x^\alpha$.

The system of integral equations obtained in such a way has, for every point $M_0$ of the domain $D$, solutions which are of the form $X(x^\alpha_0, x^4, \lambda_2, \lambda_3)$. 

At this stage of our argumentation, we are able to consider a more complex case which is the study of non-linear systems of partial differential equations. Since this is the aim of the next Chapter, we can state the results of our study of linear systems of partial differential equations which will be useful for that purpose.

$\mathbf{Conclusion.}$ Every solution of Eqs. $[E]$, possessing in $D$ first partial derivatives with respect to the $x^\alpha$ continuous and bounded, verifies, if $x^\alpha$ are the coordinates of a point $M_0$ of the domain $D_0$ defined by
\begin{equation*} |x^4_0| \leq \epsilon_0 \leq \epsilon; \hspace{2cm} |x^i_0 - \tilde{x}^i | \leq d_0 \leq d, \end{equation*}
some Kirchhoff formulae whose left-hand side are the values at the point $M_0$ of the unknown functions $u_s$ and whose right-hand side consists of a triple integral in the parameters $x^4$, $\lambda_2$ and $\lambda_3$, and of a double integral in the parameters $\lambda_2$ and $\lambda_3$. The quantities to be integrated are expressed by means of the functions $X(x^\alpha_0, x^4, \lambda_2, \lambda_3)$ and $\Omega(x^\alpha_0, x^4, \lambda_2, \lambda_3)$, themselves solutions of given integral equations ($\ref{eq:4.55}$), and of the unknown functions $[u_s]$; the quantity under the sign of double integral, which is taken for the zero value of the $x^4$ parameter, contains, besides the previous functions, the first partial derivatives of the unknown functions $\big{[}\frac{\partial u_s}{\partial x^\alpha} \big{]}$ (value over $\Sigma_0$ of the Cauchy data). We obtain in such a way a system of integral equations verified in $D_0$ from the solutions of Eqs. $[E]$. We write this system in the following reduced form \cite{foures1952theoreme}:
\begin{equation*} X= X_0 + \int\limits_{x^4_0}^{x^4} E(X) d\omega^4, \end{equation*}
\begin{equation*} 4 \pi U= \int\limits_{x^4_0}^0 \int\limits_{0}^{2 \pi} \int\limits_{0}^{\pi} H dx^4 d\lambda_2 d \lambda_3 + \int\limits_{0}^{2 \pi} \int\limits_{0}^{\pi} I d\lambda_2 d\lambda_3. \end{equation*}

\chapter{Linear System from a Non-linear Hyperbolic System}
\chaptermark{Non-linear hyperbolic systems}
\epigraph{Curiouser and curiouser.}{Lewis Carroll, Alice's Adventures in Wonderland and Through the Looking-Glass}
At this point of our analysis, we focus the attention on the non-linear hyperbolic systems of partial differential equations. We will prove that it is possible to begin with a non-linear system and turn it into a linear system for which the results obtained in the previous chapter hold.
In particular, we consider a system $[F]$ of $n$ second-order partial differential equations, with $n$ unknown functions and four variables, $\textit{non linear}$ of the following type:
$$ \sum\limits_{\lambda, \mu=1}^4 A^{\lambda \mu} \frac{\partial^2 W_s}{\partial x^\lambda \partial x^\mu} + f_s =0, \hspace{3.5cm} s=1, 2, ..., n. \hspace{2cm}  [F]$$
The coefficients $A^{\lambda \mu}$, which are the same for the $n$ equations, and $f_s$ are given functions of the four variables $x^\alpha$, the unknown functions $W_s$, and of their first derivatives $\frac{\partial W_s}{\partial x^\alpha}$. 
The calculations made in the previous chapter for the linear equations $[E]$ are valid for the non-linear equations $[F]$: it suffices to consider in these calculations the functions $W_s$ as functions of the four variables $x^\alpha$; the coefficients $A^{\lambda \mu}$ and $f_s$ are then functions of these four variables  and the previous calculations are valid, subject to considering, in all formulae where there is occurrence of partial derivatives of the coefficients with respect to $x^\alpha$, these derivations as having been performed. 

Furthermore, we do not apply directly to the equations $[F]$ the results of previous chapters; but we are going to show that, by differentiating five times with respect to the variables $x^\alpha$ the given equations $[F]$, and by applying to the obtained equations the result of Chapter 4, one obtains a system of integral equations whose left-hand side are the unknown functions $W_s$, their partial derivatives with respect to the $x^\alpha$ up to the fifth order and some auxiliary functions $X$, $\Omega$, and whose right-hand sides contain only these functions and the integration parameters.

Then, in order to solve the Cauchy problem for the nonlinear equations $[F]$ we will try to solve, independently of these equations, the system of integral equations verified by the solutions. Unfortunately, some difficulties arise for this solution: we have seen in the previous chapter that the quantities occurring under the integral sign are continuous and bounded, upon assuming differentiability of the coefficients $A^{\lambda \mu}$, viewed as given functions of the variables $x^\alpha$, these conditions not being realized when the functions $W_s$, $W_{s \alpha}$, ..., $U_S$ are independent; the quantity $[A^{ij}] \frac{\partial^2 \sigma}{\partial x^i \partial x^j}J_{x \lambda}$ will then fail to be bounded and continuous. 

Moreover, to pursue our purpose, we will have to pass through the intermediate stage of approximate equations $[F_1]$, where the coefficients $A^{\lambda \mu}$ will be some functions of the $x^\alpha$. We will then be in a position to solve the integral equations and show that their solution is a solution of the equations $[F_1]$ and to show which are the partial solution of $W_s$; but we will see that the obtained solution $W_s$ will be only five times differentiable and the method we are going to use is therefore applicable only if the $A^{\lambda \mu}$ depend uniquely on the $W_s$ and not on the $W_{s \alpha}$: it will be then enough to assume the approximation function five times differentiable. 

Eventually, we will solve the Cauchy problem for the system $[G]$
$$ \sum\limits_{\lambda, \mu=1}^4 A^{\lambda \mu} \frac{\partial^2 W_s}{\partial x^\lambda \partial x^\mu} + f_s =0, \hspace{3.5cm} s=1, 2, ..., n. \hspace{2cm}  [G] $$
where the coefficients $A^{\lambda \mu}$ do not depend on first partial derivatives of the unknown functions. It will be enough to apply the results of Chapter 4 to the equations $[G']$ deduced from the equations $[G]$ by four differentiations with respect to the variables $x^\alpha$ in order to obtain a system of integral equations whose right-hand sides do not contain other functions than those occurring on the left-hand sides. 

The integral equations $[J]$, verified by the bounded solutions and with bounded first derivatives of equations $[G']$, will only involve the coefficients $A^{\lambda \mu}$ and $B^{T \lambda}_S$ and their partial derivatives up to the orders four and two, respectively, as well as of the coefficients $F_S$. We would face clearly, in order to solve the system of integral equations $[J]$ directly, the same difficulty as in the general case:  the quantity under the sign of triple integral is not bounded in general if $W_s$, $W_{s \alpha}$, ..., $U_S$ are independent functions. We shall be able however, in the case in which the $A^{\lambda \mu}$ depend only on the first derivatives of the $W_s$, to solve the Cauchy problem by using the results obtained on the system of integral equations verified in a certain domain, from the solutions of the given equations $[G]$, by considering a system $[G_1]$, which is the approximate version of $[G]$. This system is obtained by substitution in $A^{\lambda \mu}$ of the $W_s$ with their approximate values $W_s^{(1)}$. 

We will prove that the system of integral equations $[J_1]$, verified by the solutions of the Cauchy problem assigned with respect to the equations $[G_1]$, admits of a unique, continuous and bounded solution in a domain $D$. 

Then, we will prove that the solutions of $[J_1]$ are solutions of the Cauchy problem given for the equations $[G_1]$ in the whole domain $D$, and that the functions $W_s$ obtained admit of partial derivatives up to the fourth order equal to $W_{s \alpha}$, ..., $U_S$. 

Eventually, since the solution of the Cauchy problem given for the equations $[G_1]$ defines a representation of the space of the functions $W_s^{(1)}$ into itself, we will prove that this representation admits a fixed point, belonging to the space. The corresponding $W_s$ are solutions of the given equations $[G]$. This solution is unique and possesses partial derivatives continuous and bounded up to the fourth order.

\section{The Equations $[F]$}
Let us consider the system of $n$ second-order partial differential equations, with $n$ unknown functions and four variables 
$$ \sum\limits_{\lambda, \mu=1}^4 A^{\lambda \mu} \frac{\partial^2 W_s}{\partial x^\lambda \partial x^\mu} + f_s =0, \hspace{3.5cm} s=1, 2, ..., n. \hspace{2cm}  [F]$$
We assume that in a space-domain $D$, centered at the point $\bar{M}$ with coordinates $x^i$, 0 and defined by
\begin{equation*} |x^i - \bar{x}^i | \leq d, \hspace{2cm} |x^4| \leq \epsilon \end{equation*}
and for values of the unknown functions $W_s$ and their first partial derivatives satisfying
\begin{equation} \label{eq:5.1} |W_s - \bar{W}_s | \leq l, \hspace{2cm} \bigg{|} \frac{\partial W_s}{\partial x^\alpha} -\frac{\overline{\partial W_s}}{\partial x^\alpha} \bigg{|} \leq l, \end{equation}
where $\bar{W}_s$ and $\bar{\frac{\partial W_s}{\partial x^\alpha}}$ are the values of the functions $W_s$ and $\frac{\partial W_s}{\partial x^\alpha}$ at the point $\bar{M}$, the coefficients $A^{\lambda \mu}$ and $f_s$ admit of partial derivatives with respect to all their arguments up to the fifth order.

We shall then obtain, by differentiating five times the equations $[F]$ with respect to the variables $x^\alpha$, a system of $N$ equations, where $N$ is the product by $n$ of the number of derivatives of order five of a function of four variables, verified in the domain $D$ by the solutions of the equations $[F]$ which satisfy the inequalities $(\ref{eq:5.1})$ and possess derivatives with respect to the $x^\alpha$ up to the seventh order.

Let us write this system of $N$ equations. We set
\begin{equation*} \frac{\partial W_s}{\partial x^\alpha}=W_{s \alpha}, \hspace{2cm} \frac{\partial^2 W_s}{\partial x^\alpha \partial x^\beta}= W_{s \alpha \beta} \end{equation*}
and we denote by $U_S$ the partial derivatives of order five of $W_s$
\begin{equation*} \frac{\partial^5 W_s}{\partial x^\alpha \partial x^\beta \partial x^\gamma \partial x^\delta \partial x^\epsilon} = W_{s \alpha \beta \gamma \delta \epsilon} =U_S, \hspace{2cm} s=1, 2, ..., N. \end{equation*}
Let us differentiate the given equations $[F]$ with respect to any whatsoever of the variables $x^\alpha$; we obtain $n$ equations of the form
\begin{equation*} \begin{split} & A^{\lambda \mu} \frac{\partial^2 W_{s \alpha}}{\partial x^\lambda \partial x^\mu} + \Biggl\{\frac{\partial A^{\lambda \mu}}{\partial W_r}W_{r \alpha} + \frac{\partial A^{\lambda \mu}}{\partial W_{r \nu}}\frac{\partial W_{r \nu}}{\partial x^\alpha} + \frac{\partial A^{\lambda \mu}}{\partial x^\alpha} \Biggr\} \frac{\partial W_{s \mu}}{\partial x^\lambda} + \frac{\partial f_s}{\partial W_r}W_{r \alpha} \\
&+ \frac{\partial f_s}{\partial W_{r \nu}}\frac{\partial}{\partial x^\alpha}W_{r \nu} + \frac{\partial f_s}{\partial x^\alpha}=0. \end{split} \end{equation*}
If we differentiate the previous equations four times, we obtain the following system of $N$ equations:
\begin{equation} \begin{split} \label{eq:5.2}& A^{\lambda \mu} \frac{\partial^2 W_{s \alpha \beta \gamma \delta \epsilon}}{\partial x^\lambda \partial x^\mu} + \Biggl\{ \frac{\partial A^{\lambda \mu}}{\partial W_r}W_{r \alpha} + \frac{\partial A^{\lambda \mu}}{\partial W_{r \nu}}W_{r \nu \alpha} + \frac{A^{\lambda \mu}}{\partial x^\alpha} \Biggr\} \frac{\partial}{\partial x^\lambda} W_{s \beta \gamma \delta \epsilon \mu} \\
& + \Biggl\{ \frac{\partial A^{\lambda \mu}}{\partial W_r}W_{r \beta} + \frac{\partial A^{\lambda \mu}}{\partial W_{r \nu}}W_{r \nu \beta} +\frac{A^{\lambda \mu}}{\partial x^\beta}  \Biggr\} \frac{\partial}{\partial x^\lambda} W_{s \alpha \gamma \delta \epsilon \mu} \dots \\
&+ \Biggl\{ \frac{\partial A^{\lambda \mu}}{\partial W_r}W_{r \epsilon} + \frac{\partial A^{\lambda \mu}}{\partial W_{r \nu}}W_{r \nu \epsilon} +\frac{A^{\lambda \mu}}{\partial x^\epsilon}  \Biggr\} \frac{\partial}{\partial x^\lambda} W_{s \alpha \beta \gamma \delta \mu} + \frac{\partial A^{\lambda \mu}}{\partial W_{r \nu}} \frac{\partial W_{r \nu \alpha \beta \gamma \delta}}{\partial x^\epsilon} \\
& + \frac{\partial f_s}{\partial W_{r \nu}}\frac{\partial W_{r \nu \alpha \beta \gamma \delta}}{\partial x^\epsilon} + F_S=0, \end{split} \end{equation}
where $F_S$ is a function of the variables $x^\alpha$, of the unknown functions $W_s$ and of their partial derivatives up to the fifth order included, but not of the derivatives of higher order.

The fifth derivatives $U_S$ of the functions $W_s$ satisfy therefore, in the domain $D$ and under the conditions specified, a system of $N$ equations $[F']$ of the following type:
\begin{equation} \label{eq:5.3} A^{\lambda \mu} \frac{\partial^2U_S}{\partial x^\lambda \partial x^\mu} + B_S^{T \lambda} \frac{\partial U_T}{\partial x^\lambda} + F_S=0. \end{equation}
The coefficients $A^{\lambda \mu}$, $B_S^{T \lambda}$ and $F_S$ of these equations are polynomials of the coefficients $A^{\lambda \mu}$ and $f_s$, of the given equations $[F]$ and of their partial derivatives with respect to all arguments up to the fifth order, as well as of the unknown functions $W_s$ and of their partial derivatives with respect to the $x^\alpha$ up to the fifth order. The coefficients $A^{\lambda \mu}$ depend only on the variables $x^\alpha$, the unknown functions $W_s$ and their first partial derivatives $W_{s \alpha}$, the coefficients $B_S^{T \lambda}$ depend only on the variables $x^\alpha$, on the unknown functions $W_s$ and their first and second partial derivatives $W_{s \alpha}$ and $W_{s \alpha \beta}$.

Thus, we apply to Eqs. ($\ref{eq:5.3}$), which is a system of $N$ linear equations of second order, with the unknown functions $U_S$, the result of the previous chapter. We obtain a system of integral equations whose left-hand sides will be some auxiliary functions $\Omega$, $X$, and the unknown functions $U_S$ whereas, their right-hand sides have, under the sign of integral, quantities expressed by means of the auxiliary functions $X$, of the unknown functions $U_S$ and of the value for $x^4=0$ of their first partial derivatives $\frac{\partial U_S}{\partial x^\alpha}$, of the integration parameters, as well as of the coefficients $A^{\lambda \mu}$, $B^{T \lambda}_S$ and $F_S$ and of their partial derivatives up to the orders four, three and zero.

$A^{\lambda \mu}$, $B^{T \lambda}_S$ and $F_S$ not involving the partial derivatives of the functions $W_s$ except for the orders up to one, two and five, respectively, the right-hand sides of the integral equations considered do not contain, besides the auxiliary functions $X$, $\Omega$, the functions $U_S$ and the value for $x^4=0$ of their first derivatives, and the integration parameters, nothing but the unknown functions $W_s$ and their partial derivatives up to the fifth order included.

If the functions $W_s$ and their partial derivatives up to the fifth order $W_{s\alpha}$, $W_{s\alpha\beta}$, ..., $W_{s \alpha \beta \gamma \delta \epsilon}=U_S$ are continuous and bounded in a spacetime domain $D$ of equations $|x^i - \bar{x^i}| \leq d$, $|x^4| \leq \epsilon$, they verify in this domain the integral relations
\begin{equation} \begin{split}\label{eq:5.4}
&W_s(x^\alpha)= \int_0^{x^4} W_{s 4} (x^i, t) dt + W_s(x^i, 0), \\
& \dots \\
&W_{s \alpha \beta \gamma \delta \epsilon} (x^\alpha)= \int_0^{x^4} W_{s \alpha \beta \gamma \delta 4} (x^i, t) dt + W_{s\alpha \beta \gamma \delta \epsilon}(x^i, 0).
\end{split}
\end{equation}
By adjoining this system to the system of integral equations, we are able to obtain a system of integral equations, verified, under certain assumptions, by the solutions of the given equations $[F]$, whose right-hand sides contain the functions occurring on the left-hand sides.

We search for solutions $W_s$ of the equations $[F]$ which take, as well as their first partial derivatives, some values given in a domain $(d)$ of the initial hypersurface $x^4=0$:
\begin{equation*} W_s(x^i, 0)= \varphi_s(x^i), \hspace{2cm} \frac{\partial W_s}{\partial x^4}(x^i, 0)= \psi_s(x^i); \end{equation*}
where $\varphi_s$ and $\psi_s$ are given functions of the three variables $x^i$ in the domain $(d)$. We will prove that the data $\varphi_s$ and $\psi_s$ determine the values in $(d)$ of the partial derivatives up to the sixth order of the solution $W_s$ of the equations $[F]$.

$\mathbf{Assumptions}$
\begin{description}
\item[(1)] In the domain $(d)$, defined by
\begin{equation*} |x^i - \bar{x}^i| \leq d, \end{equation*}
the functions $\varphi_s$ and $\psi_s$ admit of partial derivatives continuous and bounded with respect to the three variables $x^i$ and satisfy the inequalities 
\begin{equation} \label{eq:5.5} |\varphi_s - \bar{\varphi}_s | \leq l_0 \leq l, \; \; | \psi_s - \bar{\varphi}_s| \leq l_0 \leq l, \; \; \bigg{|} \frac{\partial \varphi_s}{\partial x^i} - \frac{\overline{\partial \varphi_s}}{\partial x^i}\bigg{|} \leq l_0 \leq l. \end{equation}
\item[(2)] In the domain $(d)$ and for values of the functions
\begin{equation*} W_s= \varphi_s, \hspace{0.5cm} \frac{\partial W_s}{\partial x^4}= \psi_s , \hspace{0.5cm}  \frac{\partial W_s}{\partial x^i}= \frac{\partial \varphi_s}{\partial x^i}, \end{equation*}
satisfying the inequalities ($\ref{eq:5.5}$), the coefficients $A^{\lambda \mu}$ and $f_s$ have partial derivatives continuous and bounded with respect to all their arguments, up to the fifth order.
\item[(3)] In the domain $(d)$ and for the functions $\varphi_s$ and $\psi_s$ considered, the coefficient $A^{44}$ is different from zero.
\end{description}
It follows, from the assumption $\textbf{(1)}$, that the values in $(d)$ of partial derivatives up to the sixth order, corresponding to a differentiation at most with respect to $x^4$, of the solutions $W_s$ of the assigned Cauchy problem are equal to the corresponding partial derivatives of the functions $\varphi_s$ and $\psi_s$, and they are continuous and bounded in $(d)$.

The values in $(d)$ of partial derivatives up to the sixth order of the functions $W_s$, corresponding to more than one derivative with respect to $x^4$, are expressed in terms of the previous ones, of the coefficients $A^{\lambda \mu}$ and $f_s$ of the equations $[F]$ and of their partial derivatives up to the fourth order. 

Moreover, from the assumption $\textbf{(3)}$, it follows that the equations $[F]$ make it possible to evaluate (being given within $(d)$ the values of the functions $W_s$, $W_{s \alpha}$, $W_{s \alpha i}$) the value in $(d)$ of $W_{s 44}$, from which one will deduce by differentiation the value in $(d)$ of the partial derivatives corresponding to two differentiations with respect to $x^4$. 

The equations that are derivatives of the equations $[F]$ with respect to the variables $x^\alpha$, up to the fourth order, make it possible to evaluate in $(d)$ the values of partial derivatives up to the sixth order of the functions $W_s$. 

It turns out from the three previous assumptions that all functions obtained are continuous and bounded in $(d)$.

We shall set
\begin{equation*} W_{s j}(x^i,0)=\varphi_{s j}(x^i), \hspace{0.5cm}  U_S(x^i,0)=\Phi_S(x^i),  \hspace{0.5cm} \frac{\partial U_S}{\partial x^4}(x^i,0)=\Psi_s(x^i). \end{equation*}
At this stage, it is useful to make a sum up of the assumptions made and the results obtained.

$\mathbf{Assumptions}$
\begin{description}
\item[(a)] In the domain $D$ defined by $|x^i - \bar{x}^i | \leq d$, $|x^4| \leq \epsilon$ and for values of the unknown functions satisfying 
\begin{equation*}|W_s - \bar{\varphi}_s | \leq l,  \hspace{0.5cm} \bigg{|} \frac{\partial W_s}{\partial x^i} - \frac{\overline{\partial \varphi_s}}{\partial x^i} \bigg{|} \leq l,  \hspace{0.5cm} \bigg{|} \frac{\partial W_s}{\partial x^4} - \bar{\psi}_s \bigg{|} \leq l;\end{equation*}
\item[(b)] In the domain of the initial surface $x^4=0$, defined by $|x^i - \bar{x}^i | \leq d$, the Cauchy data $\varphi_s$ and $\psi_s$ admit of partial derivatives continuous and bounded up to the orders six and five.
\end{description}

It follows from the assumption $\textbf{(a)}$ that the coefficients $A^{\lambda \mu}$ and $f_s$ have partial derivatives with respect to all their arguments up to the fifth order continuous and bounded, the derivatives of order five satisfying some Lipschitz conditions. 

Furthermore, the quadratic form $A^{\lambda \mu}X_\lambda X_\mu$ is of normal hyperbolic form, i.e. $A^{44} >0$ and the form $A^{ij}X_iX_j$ is negative definite.

In conclusion, we have seen that if we consider a solution $W_s$ seven times differentiable of the assigned Cauchy problem, possessing partial derivatives with respect to the $x^\alpha$ up to the sixth order, continuous and bounded and satisfying the inequalities $(\ref{eq:5.1})$ in $D$, it satisfies in this domain the equations $(\ref{eq:5.3}$), which are linear equations in the unknown functions $U_S$. 

These equations satisfy the assumptions of Chapter 4 and therefore there exists a domain $D_0 \subset D$ in which the functions $W_s$ verify the following system of integral equations.

This system consists of equations having the form
\begin{description}
\item[(1)]
\begin{equation*} X= X_0 + \int_{x^4_0}^{x^4} E(X) d\omega^4 \end{equation*}
where $X$ is a function of the three parameters $x^4$, $\lambda_2$ and $\lambda_3$, representatives of a point of the characteristic conoid of vertex $M_0(x_0)$, and of the four coordinates $x^\alpha_0$ of a point $M_0 \in D_0$. These functions $X$ are the functions $x^i$, $p_i$, $y^j_i$, $z^j_i$, $y^j_{ih}$, $z^j_{ih}$, $y^j_{ihk}$, $z^j_{ihk}$, whereas $X_0$ is the value of $X$ for $x^4=x^4_0$ and it is a given function of $x^\alpha_0$, $\lambda_2$, $\lambda_3$.
\item[(2)] 
\begin{equation*} \Omega = \Omega_0 + \int_{x^4_0}^{x^4} F(X, \Omega) d\omega^4, \end{equation*}
where $\Omega$ is a function of $x^\alpha_0$, $x^4_0$, $\lambda_2$ and $\lambda_3$. These functions are $\omega^r_s$, $\omega^r_{si}$ and $\omega^r_{sij}$, whereas $\Omega_0$ is the value of $\Omega$ for $x^4=x^4_0$ and it is a function of $x^4_0$, $\lambda_2$ and $\lambda_3$.
\item[(3)]
\begin{equation*} W=  W_0 + \int_{0}^{x^4} G(W, U) d\omega^4, \end{equation*}
where $W$ is a function of the four coordinates $x^\alpha$ of a point $M \in D$. These functions are $W_s$, $W_{s \alpha}$, $W_{s \alpha \beta}$, $W_{s \alpha \beta \gamma}$ and $W_{s \alpha \beta \gamma \delta}$, whereas $W_0$ is the value of $W$ for $x^4=0$ and it is a given function of the three variables $x^i$.
\item[(4)]
\begin{equation*} U = \int_{x^4_0}^{0} \int_{0}^{2 \pi}\int_{0}^{\pi} H d\omega^4 d\lambda_2 d\lambda_3 + \int_{0}^{2 \pi}\int_{0}^{\pi} I d\lambda_2 d\lambda_3, \end{equation*}
the Kirchhoff formulae, where $U$ is a function of the four coordinates $x^\alpha_0$ of a point $M_0 \in D_0$. These functions are the functions $U_S$.
\end{description}
The quantities $E$, $F$, $G$, $H$ and $I$ are formally identical to the corresponding quantities evaluated for the equations $[E]$. The quantity $G$ is a function of $W$ or $U$. All these quantities are therefore expressed by means of the functions $X$, $\Omega$, $W$ and $U$, occurring on the left-hand sides of the integral equations considered, and involve the partial derivatives of the $A^{\lambda \mu}$ and $f_s$ with respect to all their arguments, up to the fifth order, and the partial derivatives of the Cauchy data $\varphi_s$ and $\psi_s$ up to the orders six and five, in the quantity $I$ and by means of $W_0$.

Let us now try to solve the system of integral equations verified by the solutions of the non-linear equations $[F]$. We have seen in Chapter 4 that the quantities occurring under the integral sign, in particular $H$, are continuous and bounded, upon assuming differentiability of the coefficients $A^{\lambda \mu}$, viewed as given functions of the variables $x^\alpha$, these conditions not being realized when the functions $W_s$, $W_{s \alpha}$, ..., $U_S$ are independent; the quantity $[A^{ij}]\frac{\partial^2 \sigma}{\partial x^i \partial x^j} J_{x \lambda}$ will then fail to be bounded and continuous. Thus, it is possible to overcome this difficulty on the way towards solving the Cauchy problem by passing through the intermediate stage of approximate equations $[F_1]$, where the coefficients $A^{\lambda \mu}$ are some given functions of the $x^\alpha$, obtained by replacing $W_s$ with a given function $W_s^{(1)}$. The quantities occurring under the integration signs of the integral equations verified by the solutions will be continuous and bounded if the same holds for the functions $W_s$, ..., $U_S$ considered as independent. 

We will then be in a position to solve the integral equations and show that their solution $W_s$, ..., $U_S$ is a solution of the equations $[F_1]$, and that $W_{s \alpha}$, ..., $U_S$ are the partial derivatives  of $W_s$; but we need for that purpose to take as a function $W_s^{(1)}$ a function six times differentiable because the integral equations involve fifth derivatives of the $A^{\lambda \mu}$. Since the obtained solution $W_s$ is merely five times differentiable, it will be impossible for us to iterate the procedure. 

The method described will be therefore applicable only if the $A^{\lambda \mu}$ depend uniquely on the $W_s$ and not on its first derivatives with respect to the $x^\alpha$. Hence, from now on, it will be enough for us to assume that the approximation function is five times differentiable. 

In the general case, where $A^{\lambda \mu}$ are functions of $W_s$ and $W_{s \alpha}$, it is possible to solve the Cauchy problem by passing through the intermediate step of approximate forms not of the equations $[F]$ themselves, but of equations previously differentiated with respect to the $x^\alpha$ and viewed as integro-differential equations in the unknown functions $W_{s \alpha}$.

\section{Solution of the Cauchy problem for the system $[G]$ in which the coefficients $A^{\lambda \mu}$ do not depend on first partial derivatives of the unknown functions}
Following Bruhat \cite{foures1952theoreme}, we will now proceed by showing the solution of the Cauchy problem for the system $[F]$ when the coefficients $A^{\lambda \mu}$ depend only on the variables $x^\alpha$, on the functions $W_s$ but not on their first derivatives with respect to the $x^\alpha$, i.e. $W_{s \alpha}$.

Let us consider the system $[G]$ of $n$ partial differential equations of second order with $n$ unknown functions and four variables
$$ A^{\lambda \mu} \frac{\partial^2 W_s}{\partial x^\lambda \partial x^\mu} + f_s =0, \hspace{8cm} [G] $$
where the coefficients $A^{\lambda \mu}$ depend only on the variables $x^\alpha$ and on the unknown functions $W_s$, and not on the first partial derivatives $W_{s \alpha}$ of these functions. The coefficients $f_s$ are functions of the variables $x^\alpha$, of the unknown functions $W_s$ and of their first partial derivatives $W_{s \alpha}$. 

We shall obtain a system of integral equations verified by the solutions of the equations $[G]$ by applying the methods used for the equations $[F]$. 

Since the $A^{\lambda \mu}$ do not contain $W_{s \alpha}$, it will be enough to apply the results of Chapter 4 to the equations $[G']$ deduced from the $[G]$ by four differentiation with respect to the variables $x^\alpha$ in order to obtain a system of integral equations whose right-hand sides do not contain other functions than those which occur on the left-hand sides. If we denote by $U_S$ any whatsoever of the fourth derivatives of the unknown functions $W_s$, the equations obtained with the previous calculations read as
$$A^{\lambda \mu} \frac{\partial^2 U_S}{\partial x^\lambda \partial x^\mu} + B^{T \lambda}_S \frac{\partial U_T}{\partial x^\lambda} + F_S =0. \hspace{6cm} [G']$$

$A^{\lambda \mu}$ depend only on the variables $x^\alpha$ and the functions $W_{s}$. 

$B^{T \lambda}_S$ are a sum of first partial derivatives of the functions $A^{\lambda \mu}$, viewed as functions of the variables $x^\alpha$ and of first partial derivatives of a function $f_s$ with respect to the first partial derivatives $W_{s \alpha}$ of the unknown functions, depend on the variables $x^\alpha$, on the unknown functions $W_s$ and on their first partial derivatives $W_{s \alpha}$. 

Eventually, $F_S$ is a polynomial of the coefficients $A^{\lambda \mu}$, of $f_s$ and of their partial derivatives with respect to all their arguments up to the fourth order, as well as of the functions $W_s$ and of their partial derivatives with respect to the variables $x^\alpha$ up to the fourth order.

In order to solve the Cauchy problem we proceed as follows.

$\mathbf{Assumptions}$
\begin{description}
\item[(1)] In the domain $D$, defined by $|x^i - \bar{x}^i | \leq d$, $|x^4| \leq \epsilon$, and for values of the unknown functions satisfying 
\begin{equation} \label{eq:5.6} |W_s - \bar{\varphi}_s| \leq l,  \hspace{0.5cm}\bigg{|} \frac{\partial W_s}{\partial x^i} - \frac{\overline{\partial \varphi_s}}{\partial x^i} \bigg{|} \leq l,  \hspace{0.5cm} \bigg{|} \frac{\partial W_s}{\partial x^4} - \bar{\psi}_s \bigg{|} \leq l, \end{equation}
one has that
\begin{description}
\item[(a)] The coefficients $A^{\lambda \mu}$ and $f_s$ admit partial derivatives with respect to all their arguments up to the fourth order, continuous, bounded and satisfying Lipschitz conditions.
\item[(b)] The quadratic form $A^{\lambda \mu}X_\lambda X_\mu$ is of normal hyperbolic type, i.e. $A^{44} >0$ and $\sum_{i, j=1}^3 A^{ij}X_i X_j$ is negative-definite.
\end{description}
\item[(2)] In the domain $(d)$ of the initial surface, defined by $|x^i - \bar{x}^i | \leq d$, the Cauchy data $\varphi_s$ and $\psi_s$ possess partial derivatives continuous and bounded up to the orders five and four, respectively, satisfying some Lipschitz conditions.
\end{description}

The integral equations $[J]$, verified by the bounded solutions and with bounded first derivatives of equations $[G']$, only involve the coefficients $A^{\lambda \mu}$ and $B^{T \lambda}_S$ and their partial derivatives up to the orders four and two respectively, as well as of the coefficients $F_S$. These equations $[J]$ contain only partial derivatives of the functions $W_s$ of order higher than four.

When we try to solve the system of integral equations $[J]$ directly, we see that the quantity $H$ under the sign of triple integral is not bounded in general if $W_S$, $W_{s \alpha}$, ... $U_S$ are independent functions. We shall be able however, in the case in which the $A^{\lambda \mu}$ depend only on $W_{s \alpha}$, to solve the Cauchy problem by using the results obtained on the system of integral equations verified in a certain domain, from the solutions of the given equations $[G]$. 

We shall consider a system $[G_1]$, which is the approximation of $[G]$, obtained by replacing in $A^{\lambda \mu}$ and not in $f_s$ the unknown $W_s$ with some approximate values $W_s^{(1)}$ which admit of partial derivatives continuous and bounded up to the fourth order, $W_{s \alpha}$, ..., $U_S$, in the domain $D$ and satisfy the inequalities 
\begin{equation*} |W_s^{(1)} - \bar{\varphi}_s| \leq l,  \hspace{0.5cm} \bigg{|} \frac{\partial W_s^{(1)}}{\partial x^i} - \frac{\overline{\partial \varphi_s}}{\partial x^i} \bigg{|} \leq l,  \hspace{0.5cm} \bigg{|} \frac{\partial W_s^{(1)}}{\partial x^4} - \bar{\psi}_s \bigg{|} \leq l. \end{equation*}
This approximated system reads as
$$ {A^{\lambda \mu}}^{(1)} \frac{\partial^2 W_s}{\partial x^\lambda \partial x^\mu} + f_s =0. \hspace{7cm} [G_1] $$
A solution $W_1$, six times differentiable and satisfying the inequalities ($\ref{eq:5.6}$), of the equations $[G_1]$ verifies therefore, in $D$, the following equations:  
$$ {A^{\lambda \mu}}^{(1)} \frac{\partial^2 U_S}{\partial x^\lambda \partial x^\mu} + {B^{T \lambda}_S}^{(1)} \frac{\partial U_T}{\partial x^\lambda} + {F_S}^{(1)} =0. \hspace{4cm} [{G'}_1] $$

${A^{\lambda \mu}}^{(1)}$ is a function of the variables $x^\alpha$ and of the unknown functions ${W_s}^{(1)}$.

${B^{T \lambda}_S}^{(1)}$ is a sum of the first partial derivatives of the ${A^{\lambda \mu}}^{(1)}$, viewed as functions of the variables $x^\alpha$ (hence as functions of the variables $x^\alpha$ and of the functions ${W_s}^{(1)}$ and ${W_{s \alpha}}^{(1)}$) and of the first partial derivatives of a function $f_s$ with respect to the functions $W_{r \nu}$ (therefore of the functions of $x^\alpha$, ${W_s}$ and $W_{s \alpha}$). 

Eventually, $F^{(1)}_S$ is a polynomial of the coefficients ${A^{\lambda \mu}}^{(1)}$, of $f_s$ and of their partial derivatives with respect to all their arguments up to the fourth order, as the functions ${W_s}^{(1)}$ and ${W_{s \alpha}}^{(1)}$ and of their partial derivatives with respect to the $x^\alpha$ up to the fourth order.

All these coefficients of equations $[G'_1]$, viewed as linear equation of type $[E]$ in the unknown functions $U_S$, satisfy in the domain $D$ the assumptions made in Chapter 4. 

Thus, there exists a domain $D_0 \subset D$ in which the fifth derivatives $U_S$ of a solution $W_s$ of the given Cauchy problem, which possess partial derivatives continuous and bounded up to the sixth order and satisfy the inequalities ($\ref{eq:5.6}$), verify some Kirchhoff formulae, whose left-hand sides are the values at the point $M_0 \in D_0$ of those functions $U_S$. 

These equations, together with the integral equations having on the left-hand side some auxiliary functions $X$ and $\Omega$, and with some integral equations analogous to the previous ones, form a system of integral equations that we denote by $[J_1]$. 

\subsection{The integral equations $[J_1]$}
Let us consider the set of integral equations $[J_1]$ as a system of integral equations with four groups of unknown functions $X$, $\Omega$, $W$ and $U$. The system consists of the following four group of equations:
\begin{description}
\item[(1)] Equations having on the left-hand side a function $X$ of the four coordinates $x^\alpha_0$ and of three parameters $x^4$, $\lambda_2$ and $\lambda_3$. These functions $X$ are $x^i$, $p_i$, $y^j_i$, $z^j_i$, ..., $z^j_{ihk}$ which define the characteristic conoids. These equations are of the form
$$ X = X_0 + \int_{x^4_0}^{x^4} E(X) d\omega^4, \hspace{6cm} [1] $$
where $X_0$ is the value of $X$ for $x^4=x^4_0$, whereas $E$ is a rational function, with denominator
\begin{equation*} T^{*4}= {A^{*44}}^{(1)} + {A^{*i4}}^{(1)}p_i, \end{equation*}
of the following quantities:
\begin{description}
\item[(a)] The coefficients ${A^{\lambda \mu}}^{(1)}$ and their partial derivatives with respect to all their arguments up to the fourth order (which are functions of ${W_s}^{(1)}(x^\alpha)$ and $x^\alpha$ where $x^i$ is replaced by the corresponding $X$ function), functions ${W_s}^{(1)}$ and partial derivatives up to the fourth order;
\item[(b)] The functions $X$;
\item[(c)] The quantities ${a^{\alpha \beta}_0}^{(1)}$ and  ${a_{\alpha \beta}^0}^{(1)}$, which are algebraic functions of the values of the coefficients ${A^{\lambda \mu}}^{(1)}$ for the values $x^\alpha_0$ and ${W_s}^{(1)}(x^\alpha_0)$ of their arguments.
\end{description}
\item[(2)] Equations having on the left-hand side a function $\Omega$ of the $x^\alpha_0$ and of the parameters $x^4$, $\lambda_2$, $\lambda_3$. These functions $\Omega$ correspond to $\omega^r_s$, $\omega^r_{si}$ and $\omega^r_{sij}$. These equations are of the form
$$ \Omega = \Omega_0 + \int_{x^4_0}^{x^4} F(X, \Omega) d\omega^4, \hspace{6cm} [2] $$
where $\Omega_0$ is the value of $\Omega$ for $x^4=x^4_0$, whereas $F$ is a rational fraction, with denominator
\begin{equation*} T^{*4} = {A^{*44}}^{(1)} + {A^{*i4}}^{(1)}p_i, \end{equation*}
of the following quantities:
\begin{description}
\item[(a)] The coefficients ${A^{\lambda \mu}}^{(1)}$ and ${B^{T \lambda}_S}^{(1)}$ and their partial derivatives with respect to all their arguments up to the orders three and two, respectively, i.e. coefficients ${A^{\lambda \mu}}^{(1)}$, $f_s$ and their partial derivatives up to the third order;
\item[(b)] The functions ${W_s}^{(1)}(x^\alpha)$ and their partial derivatives up to the third order and functions $W_{s \alpha}(x^\alpha)$, $W_{s \alpha \beta}(x^\alpha)$ and $W_{s \alpha \beta \gamma}(x^\alpha)$. The $x^i$ are always replaced by the corresponding functions $X$;
\item[(c)] The functions $X$ and $\Omega$;
\item[(d)] The quantities ${a^{\alpha \beta}_0}^{(1)}$ and  ${a_{\alpha \beta}^0}^{(1)}$.
\end{description}
\item[(3)] Equations having on the left-hand side a function $W$ of the four coordinates $x^\alpha$. These equations are of the form
$$ W(x^\alpha) = W_0 (x^i) + \int_{0}^{x^4} G  d\omega^4, \hspace{6cm} [3] $$
where $W_0$ is the value of $W$ for $x^4=0$, whereas $G$ is a function $W$ or a function $U$.
\item[(4)] Equations having on the left-hand side a function $U$ of the four coordinates $x^\alpha_0$, known as Kirchhoff formulae, of the form
$$ 4 \pi U (x^\alpha_0) = \int_{x^4_0}^{0} \int_0^{2 \pi} \int_0^\pi H d\omega^4 d \lambda_2 d \lambda_3 + \int_0^{2 \pi} \int_0^\pi I d \lambda_2 d\lambda_3, \hspace{1cm} [4] $$
where $H$ is the product of the square root of a rational fraction with denominator $D^*$, which is a polynomial of ${A^{\lambda \mu}}^{(1)}$, $X$, $\tilde{X}$ and $p_i^0$, and numerator 1, with the sum of the two following rational fractions:
\begin{description}
\item[(A)] A rational fraction $H_a$ with denominator $(D^*)^3(x^4_0 - x^4)T^{*4}$, which results only from those terms of the operator $L^r_s$ which contain the second partial derivatives of the function $\sigma$, whereas its numerator is a polynomial of the functions:

${A^{\lambda \mu}}^{(1)}$ and their first and second partial derivatives with respect to all their arguments, i.e. functions of ${W_s}^{(1)}(x^\alpha)$ and $x^\alpha$, where $x^i$ are replaced by the corresponding $X$ functions; 

${W_{s}}^{(1)}(x^\alpha)$, ${W_{s \alpha}}^{(1)}(x^\alpha)$ and ${W_{s \alpha \beta}}^{(1)}(x^\alpha)$;
 
$X$ and $\tilde{X}$, where $\tilde{X}$ is the quotient by $(x^4_0 - x^4)$ of the functions $X$ for which $X_0=0$; 

$U(x^\alpha)$ and $\Omega$, which only occur in the product $[U_r]\omega^r_s$ in the polynomial considered. 

This polynomial, which is a function of $x^\alpha_0$, $x^4$, $\lambda_2$ and $\lambda_3$, vanishes for $x^4=x^4_0$.
\item[(B)] A rational fraction $H_{1b}$ with denominator $(D^*)^2T^{*4}$ of the following quantities:

The coefficients ${A^{\lambda \mu}}^{(1)}$, ${B^{T \lambda}_s}^{(1)}$ and ${F_S}^{(1)}$, and their partial derivatives of up to the orders two and one, respectively, with respect to the $x^\alpha$. More precisely, the quantities involved are:
 
the coefficients ${A^{\lambda \mu}}^{(1)}$ and $f_s$ and their partial derivatives with respect to all their arguments up to the fourth order;

the functions ${W_s}^{(1)}(x^\alpha)$, $\dots$, ${U_S}^{(1)}(x^\alpha)$, ${W_s}(x^\alpha)$, $\dots$, ${U_S}(x^\alpha)$; 

the functions $X$ and $\tilde{X}$;

the functions $\Omega$, $\tilde{\Omega}$, where $\tilde{\Omega}$ is the quotient by $(x^4_0 - x^4)$ of the functions $\Omega$ for which $\Omega_0=0$.
\end{description}

Eventually, $I$ is the value for $x^4=0$ of the product of the square root of a rational fraction with denominator $D^*$ and numerator 1, with a rational fraction having denominator $(D^*)^2{A^{*44}}^{(1)}T^{*4}$ of the following functions:

${A^{\lambda \mu}}^{(1)}$ and their first partial derivatives with respect to all their arguments; 

the first partial derivatives of $f_s$ with respect to $W_{r \nu}$, which contribute through ${B^{T \lambda}_S}^{(1)}$, functions of $W_s(x^\alpha)$, $W_{s \alpha}(x^\alpha)$ and $X^{\alpha}$; 

${W_s}^{(1)}(x^\alpha)$ and ${W_{s \alpha}}^{(1)}(x^\alpha)$, $X$ and $\tilde{X}$, $\Omega$ and $\tilde{\Omega}$; 

the Cauchy data $\varphi_s (x^i)$ and $\psi_s (x^i)$ and their partial derivatives with respect to the $x^i$ up to the orders five and four, respectively.
\end{description}
Since the equations $[1]$ do not contain other unknown functions besides the functions $X$, we shall solve them first. 

Furthermore, the $H_{a}$ is a known function when the $X$ are known. We shall then be in a position to restrict the quantity $H$ without making assumptions on the derivatives of the functions $U$ and $W$, viewed as independent, and to solve the remaining equations $[2]$, $[3]$ and $[4]$.

Hence, we are going to prove that the system of integral equations $J_1$ admits a unique solution, by making use of the assumptions made on the coefficients $A^{\lambda \mu}$ and $f_s$ and of the assumptions on the functions ${W_s}^{(1)}$.

\subsection{Assumptions on the coefficients $A^{\lambda \mu}$, $f_s$ and on the functions ${W_s}^{(1)}$}
$\textbf{Assumptions B}$
\begin{description}
\item[($B_1$)] In the domain $D$ defined by $ |x^i - \bar{x}^i| \leq d$, $|x^4| \leq \epsilon$ and for the values of the functions $W_s$ and $W_{s \alpha}$ satisfying:
\begin{equation} \label{eq:5.7} | W_s - \varphi_s | \leq l, \hspace{0.5cm} |W_{si} - \varphi_{si}| \leq l, \hspace{0.5cm} |W_{s 4} - \psi_s | \leq l: \end{equation}
\begin{description}
\item[(a)] The coefficients $A^{\lambda \mu}$ and $f_s$ admit partial derivatives with respect to all their arguments up to the fourth order, continuous and bounded by a given number.
\item[(b)] The quadratic form $ \sum_{\lambda, \mu=1}^4 A^{\lambda \mu}X_\lambda X_\mu$ is of normal hyperbolic type. The coefficient $A^{44}$ is bigger than a given positive number.
\end{description}
The coefficients ${a^{\alpha \beta}_0}$ and  ${a_{\alpha \beta}^0}$ relative to the values of the coefficients $A^{\lambda \mu}$ at a point of the previous domain are bounded by a given number.
\item[($B_2$)] The approximating functions ${W_s}^{(1)}$ admit in the domain $D$ of partial derivatives up to the fourth order continuous, bounded and satisfying the inequalities
\begin{equation*} | {W_s}^{(1)} - \varphi_s | \leq l, \hspace{0.5cm}|{W_{si}}^{(1)} - \varphi_{si}| \leq l, \hspace{0.5cm} |{W_{s 4}}^{(1)} - \psi_s | \leq l \end{equation*}
and analogous identities 
$$| {W}^{(1)} - W_0| \leq l \hspace{0.5cm} {\rm up} \; {\rm to} \hspace{0.5cm} | {U_S}^{(1)} - \Phi_s | \leq l.$$
\item[($B_3$)] In the domain $(d)$ defined by $ |x^i - \bar{x}^i| \leq d$, the Cauchy data $\varphi_s(x^i)$ and $\psi_s(x^i)$ possess partial derivatives continuous and bounded with respect to the variables $x^i$ up to the orders five and four, respectively.
\end{description}
$\textbf{Assumptions B'}$
\begin{description}
\item[($B'_1$)] In the domain $D$ and for the values of the functions $W_s$ and $W_{s \alpha}$ satisfying the inequalities $(\ref{eq:5.7}$), the partial derivatives of order four of the coefficients $A^{\lambda \mu}$ and $f_s$ satisfy a Lipschitz condition assigned with respect to all their arguments.
\item[($B'_2$)] The assumptions B imply that, in the domain $D$ and for the values of the functions $W_s$ satisfying $(\ref{eq:5.7}$), the coefficients ${a^{\alpha \beta}_0}$ and ${a_{\alpha \beta}^0}$, as long as their partial derivatives up to the fourth order, verify a Lipschitz condition given with respect to their arguments $x^\alpha_0$, $W_s(x^\alpha_0)$. 
\item[($B'_3$)] The partial derivatives of order four of the functions $W_s$ satisfy a Lipschitz condition with respect to the three arguments $x^i$.
From the assumption $(B_3)$ one obtains the inequality
\begin{equation*} \big{|} {W_s}^{(1)}(x'^{\alpha}) - {W_s}^{(1)}(x^\alpha) \big{|} \leq l' \sum | x'^{\alpha} - x^{\alpha} | \end{equation*}
and the analogous inequalities for the partial derivatives of the ${W_s}^{(1)}$ up to the third order.
We shall have in addition:
\begin{equation*} \big{|} {U_S}^{(1)}(x'^{i}, x^4) - {U_S}^{(1)}(x^i, x^4) \big{|} \leq l \sum |x'^{i} - x^i |. \end{equation*}
\item[($B'_4$)] In the domain $(d)$ the partial derivatives of Cauchy data $\varphi_s$ and $\psi_s$ of orders five and four, respectively, satisfy a Lipschitz condition with respect to the variables $x^i$.

From the assumptions B, one finds the inequality
\begin{equation*} \big{|} {\varphi_s}(x'^{i}) - {\varphi_s}(x^i) \big{|} \leq l'_0 \sum | x'^{ i} - x^{i} | \end{equation*}
and the analogous inequalities for the functions $\psi_s$ and the partial derivatives of $\psi_s$ and $\varphi_s$ up to the orders three and four.

We have in addition:
\begin{equation*} \big{|} {\phi_{sj}}(x'^{i}) - {\phi_{sj}}(x^i) \big{|} \leq l' \sum | x'^{i} - x^{i} |, \end{equation*}
\begin{equation*} \big{|} {\psi_{s}}(x'^{i}) - {\psi_{s}}(x^i) \big{|} \leq l'_0 \sum | x'^{ i} - x^{i} |, \end{equation*}
where $l'$ and $l'_0$ are given numbers which satisfy $l' >l'_0$.
\end{description}
We are now able to proceed with the calculation of the solution of equations $[1]$.

\subsection{Solution of equations $[1]$}
We shall solve first the equations $[1]$ defining the characteristic conoid
$$ X= X_0 + \int_{x^4_0}^{x^4} E(X) d\omega^4. \hspace{7cm} [1] $$
These non-linear integral equations, having on the left-hand side a function $X$, do not contain other unknown functions besides the functions $X$. 

We shall solve this equations by considering a functional space $\Upsilon$, the $m$ coordinates of a point of $\Upsilon$ (where $m$ is the number of functions $X$) being some functions $X_1$ continuous and bounded of $x^\alpha_0$, $x^4$, $\lambda_2$ and $\lambda_3$ in the domain $\Lambda$ defined by
\begin{equation*}\begin{split}& |x^i_0 - \tilde{x}^i| \leq d, \hspace{0.5cm} |x^4_0| \leq \Upsilon(x^i_0), \\
& 0 \leq x^4 \leq x^4_0, \hspace{0.5cm} 0 \leq \lambda_2 \leq \pi, \hspace{0.5cm} 0 \leq \lambda_3 \leq 2 \pi, \end{split}\end{equation*}
with $\Upsilon(x^i_0) \leq \epsilon$.

The functions $X_1$ take for $x^4=x^4_0$ the assigned values $X_0$. We denote by $\bar{M}_0$ the point of $\Upsilon$ having coordinates $\tilde{X}_0$, which are the values of the functions $X_0$ for $x^i_0=\bar{x}^i$ and $x^4_0=0$, and we assume that the functions $X_1$ satisfy the inequalities
\begin{equation} \label{eq:5.8} |X_1 - \tilde{X}_0| \leq d \; {\rm and} \; |X_1 - X_0| \leq M |x^4_0 - x^4 |, \end{equation}
where $M$ is a given number. 

We shall define in the space $\Upsilon$ the distance of two points $M_1$ and $M_1'$ by the maximum in the domain $\Lambda$ of the sum of absolute values of the differences of their coordinates:
\begin{equation*} d(M_1, M_1')= Max_{\Lambda} \sum |X_1' - X_1 |. \end{equation*}
The norm introduced in such a way endows the space $\Upsilon$ of the topology of uniform convergence, and then $\Upsilon$ is a normed, complete and compact space.

To the point $M_1$ of $\Upsilon$ having coordinates $X_1$ we associate a point $M_2$ whose coordinates $X_2$ are defined by
\begin{equation} \label{eq:5.9} X_2 = X_0 + \int_{x^4_0}^{x^4} E_1 d\omega^4, \end{equation}
where $E_1$ denotes the quantity $E$ occurring in the equations $[1]$ and the functions $X$ are replaced by the corresponding coordinates $X_1$ of $M_1$.

Since this representation is a representation of $\Upsilon$ in itself, the $X_2$ are continuous and bounded functions of their arguments, they take for $x^4=x^4_0$ the values $X_0$ and satisfy the same inequalities ($\ref{eq:5.8}$) fulfilled by $X_1$, if $\epsilon(x^i_0)$, which defines the domain of variation of the argument $x^4_0$ of $X_1$ is suitably chosen.

The $E_1$ are indeed expressed rationally by means of the ${W_{s1}}^{(1)}$, ${A^{\lambda \mu}_1}^{(1)}$, of their partial derivatives up to the fourth order and $x^i$ are replaced in all its functions by the corresponding $X_1$ function: $X_1$, ${a^{\alpha \beta}_0}^{(1)}$, ${a_{\alpha \beta}^0}^{(1)}$.

All these functions are, by virtue of the assumptions B and of the assumptions made upon the $X_1$, functions continuous and bounded of $x^\alpha_0$, $x^4$, $\lambda_2$ and $\lambda_3$. On the other hand the denominator of the functions $E_1$ is
\begin{equation*} {T^{*4}_1}^{(1)}= \bigg{(} {A^{*44}}^{(1)} + {A^{*i4}}^{(1)}p_i \bigg{)}_1 \end{equation*}
and takes the value 1 for $x^4=x^4_0$, $X_1=X_0$. If follows from the assumptions B and B' and from the inequalities verified by the $X_1$ that ${T^{*4}}^{(1)}$ satisfies some Lipschitz conditions
\begin{equation*} \big{|} {T^{*4}_1}^{(1)} - 1 \big{|} \leq T' \Biggl\{ \sum \big{|} X_1 - X_0 \big{|} + |x^4 - x^4_0 | \Biggr\} \leq T' (m M + 1) |x^4_0 - x^4|, \end{equation*}
where $T'$ depends only on the bounds B and B'.

Therefore, we shall be in a position to choose $\epsilon (x^i_0)$ sufficiently small so that the denominator considered differs from zero in $\Lambda$.
The quantities $E_1$ are then continuous functions of their seven arguments in the domain $\Lambda$, and are bounded by a number $M$ which depends only on the bounds B, $E_1 \leq M$. This implies that the functions $X_2$ are continuous and bounded in their seven arguments. They fulfill the inequalities
\begin{equation} \label{eq:5.10} |X_2 - X_0| \leq M |x^4_0 - x^4 |, \end{equation}
where $M$ has been chosen in such a way that the functions $X_2$ verify the same inequality as the functions $X_1$.
It will be therefore enough to take $\epsilon(x^i_0)$ in such a way that
\begin{equation} \label{eq:5.11} \epsilon (x^i_0) \leq \frac{ d - |x^i_0 - \bar{x}^i_0|}{M}, \end{equation}
in order to obtain $|X_2 - \bar{X}_0| \leq d$.

The point $M_2$ will be therefore a point of $\Upsilon$ if $\epsilon(x^i_0)$ verifies the inequality $(\ref{eq:5.11})$.

Let us now show that the distance of two points $M_2$, $M_2'$ is less than the distance of the initial points $M_1$, $M_1'$ if $\epsilon(x^i_0)$ is suitably chosen. From the equations $(\ref{eq:5.9})$ there follows the inequality
\begin{equation} \label{eq:5.12} |X_2' - X_2 | \leq |x^4_0 - x^4 | \cdot Max|E_1' - E_1|, \end{equation}
where $E_1$ are rational fractions with non-vanishing denominators of boundend functions verifying Lipschitz conditions with respect to the $X_1$. We have on the other hand 
\begin{equation*} |E_1' - E_1| \leq M' \cdot \sum |X_1' - X_1|, \end{equation*}
where $M'$ is a number which depends only on the bounds B and B'. From which
\begin{equation*} d(M_2, M_2') \leq mM' \cdot Max_{\Lambda} \epsilon (x^i_0) \cdot d(M_1, M_1'). \end{equation*}
In order for the representation ($\ref{eq:5.9}$) of the space $\Upsilon$ into itself to reduce the distance it will be then enough that 
\begin{equation} \label{eq:5.13} \epsilon(x^i_0) < \frac{1}{m M'}. \end{equation}
We shall therefore choose $\epsilon(x^i_0)$ as satisfying the inequalities ($\ref{eq:5.11})$ and ($\ref{eq:5.13}$). The representation ($\ref{eq:5.9}$) of the space $\Upsilon$ normed, complete and compact into itself, reducing the distances, will then admit a unique fixed point belonging to this space.

$\textbf{Conclusion.}$ In the domain
\begin{equation} \begin{split} \label{eq:5.14} &|x^i_0 - \bar{x}^i | \leq d, \hspace{0.5cm}|x^4_0| < \epsilon(x^i_0), \\
& 0 \leq x^4 \leq x^4_0, \hspace{0.5cm} 0 \leq \lambda_2 \leq \pi, \hspace{0.5cm} 0 \leq \lambda_3 \leq 2\pi \end{split}\end{equation}
the system of integral equations $[1]$ admits a solution that is unique, continuous and bounded and verifying the inequalities 
\begin{equation} \label{eq:5.15} |X - \bar{X}_0| \leq d, \end{equation}
where the three functions $X$ corresponding to the $x^i$ define, with the variable $x^4$, a point belonging to the domain $D$.

Hence, having shown that there exists a unique solution of the equations $[1]$, and recalling that the quantities $E$ which are involved on the right-hand side of $[1]$ are only ${A^{\lambda \mu}}^{(1)}$ and their partial derivatives, possessing the same properties of Chapter 4, it is possible to apply the same method and to see that:

\begin{description}
\item[(1)] The functions $\frac{X- X_0}{x^4_0 - x^4}$ are continuous and bounded in $\Lambda$. The functions $\tilde{X}$, quotients by $x^4_0 - x^4$ of the $X$ which vanish for $x^4_0=x^4$, are continuous and bounded in $\Lambda$:
\begin{equation*} |X - X_0| < M|x^4_0 - x^4|, \hspace{0.5cm} |\tilde{X}| \leq M. \end{equation*}
\item[(2)] The functions 
\begin{equation*} \frac{ \tilde{X} - \tilde{X}_0}{x^4_0 - x^4} = \frac{\int_{x^4_0}^{x^4} (E - E_0)d\omega^4}{x^4_0 - x^4}, \end{equation*}
where $\tilde{X}_0$ and $E_0$ denote the values for $x^4_0=x^4$ of $\tilde{X}$ and $E$, are continuous and bounded in $\Lambda$. The bound on these functions is deduced from the Lipschitz conditions, verified by $E$ with respect to the $X$ and $x^4$:
\begin{equation*} |E - E_0| \leq M'' \Biggl\{ \sum |X- X_0| + |x^4 - x^4_0| \Biggr\}, \end{equation*}
where $M''$ depends only on the bounds B and B'. Thus, we have
\begin{equation} \label{eq:5.16} | \tilde{X} - \tilde{X}_0| \leq \frac{M}{2}(Mm+1) |x^4 - x^4_0|. \end{equation}
\item[(3)] The functions $X$ verify Lipschitz conditions with respect to the $x^i_0$.

It is sufficient, in order to prove it, to impose on the space $\Upsilon$ the following supplementary assumptions:

The functions $X_1$ verify a Lipschitz condition with respect to the $x^i_0$
\begin{equation} \label{eq:5.17} \big{|} X_1(x'^{i}_0, x^4_0, ...) - X_1(x^i_0, x^4_0, ...)\big{|} \leq d' \sum |x'^{i}_0 - x^i_0|, \end{equation}
where $d'$ is a given number. We have
\begin{equation*} X_2(x'^{i}_0, ...) - X_2(x^i_0, ...) = \int_{x^4_0}^{x^4} \bigg{(} E_1(x'^{i}_0, ...) - E_1(x^i_0, ...) \bigg{)} d\omega^4, \end{equation*}
where $ E_1(x'^{i}_0, ...)$ and $ E_1(x^i_0, ...)$ are evaluated with the help of the functions $X_1(x'^{i}_0, ...)$, in particular $x^i_1(x'^{i}_0, ...)$ and  $X_1(x^{i}_0, ...)$, respectively. Since the quantities $E_1$ verify a Lipschitz condition with respect to the $X_1$, one deduces from $(\ref{eq:5.17}$):
\begin{equation*} |X_2(x'^{i}_0, ...) - X_2(x^i_0, ...) | \leq  |x^4_0 - x^4| M'd' |x'^{i}_0 - x^i_0|, \end{equation*}
from which, for $\epsilon(x^i_0) \leq \frac{1}{M'}$, one has
\begin{equation*} |X_2(x'^{i}_0, ...) - X_2(x^i_0, ...) | \leq  d' \sum |x'^{i}_0 - x^i_0|. \end{equation*}
The point $M_2$, representative of $M_1$ by virtue of $(\ref{eq:5.9})$, with the supplementary assumption made, is still a point of $\Upsilon$, and the fixed point has coordinates verifying
\begin{equation*} |X(x'^{i}_0, ...) - X(x^i_0, ...) | \leq  d' \sum |x'^{i}_0 - x^i_0|, \end{equation*}
and
\begin{equation*} |X(x'^{i}_0, ...) - X(x^i_0, ...) | \leq  |x^4_0 - x^4| M'd' |x'^{i}_0 - x^i_0|, \end{equation*}
from which we have
\begin{equation*} |\tilde{X}(x^{i}_0, ...) - \tilde{X}(x^i_0, ...) | \leq M' d' \sum |x'^{i}_0 - x^i_0|. \end{equation*}
\end{description}

\subsection{Solution of equations $[2]$, $[3]$ and $[4]$}
Let us now consider the system of integral equations with the unknown functions $\Omega$, $W$ and $U$, obtained by replacing in the equations $[2]$, $[3]$ and $[4]$ the functions $X$ with the solutions found of equations $[1]$:
$$ \Omega = \Omega_0 + \int_{x^4_0}^{x^4} F(X, \Omega) d\omega^4, \hspace{7cm} [2] $$
$$ W = W_0 + \int_{0}^{x^4} G d\omega^4, \hspace{8cm} [3] $$
$$ 4 \pi U (x^\alpha_0) = \int_{x^4_0}^{0} \int_0^{2 \pi} \int_0^\pi H d\omega^4 d \lambda_2 d \lambda_3 + \int_0^{2 \pi} \int_0^\pi I d \lambda_2 d\lambda_3. \hspace{2cm} [4] $$
We shall solve these equations by considering a functional space $\mathcal{F}$, the coordinates of a point of $\mathcal{F}$ being defined by:
\begin{description}
\item[(1)] $m_1$ of these coordinates, that is the number of functions $\Omega$, are functions $\Omega_1$ continuous and bounded of $x^\alpha_0$, $x^4$, $\lambda_2$ and $\lambda_3$ in the domain $\Lambda$:
\begin{equation*} \begin{split} & |x^i_0 - \bar{x}^i | \leq d, \hspace{0.5cm} |x^4_0| \leq \epsilon(x^i_0), \\
& 0 \leq x^4 \leq x^4_0, \hspace{0.5cm} 0 \leq \lambda_2 \leq \pi, \hspace{0.5cm} 0 \leq \lambda_3 \leq 2 \pi. \end{split} \end{equation*}
These functions take for $x^4_0=x^4$ the given values $\Omega_0$ and satisfy the inequalities 
\begin{equation} \label{eq:5.18} |\Omega_1 - \Omega_0| \leq h, \end{equation}
where $h$ is a given number. We shall suppose in addiction
\begin{equation*} |\Omega_1 - \Omega_0| \leq N |x^4 - x^4_0|, \end{equation*}
where $N$ is a given number. The functions $\tilde{\Omega}_1$, quotients by $x^4 - x^4_0$ of the functions $\Omega_1$ that vanish identically for $x^4=x^4_0$, are then bounded in the domain $\Lambda$:
\begin{equation} \label{eq:5.19} |\Omega_1| \leq N. \end{equation} 
The functions $\Omega_1$ will be assumed continuous in $\Lambda$.
\item[(2)] $m_2$ of these coordinates, that is the number of functions $W$ and $U$, are functions $W_1$ and $U_1$ continuous and bounded of the four variables $x^\alpha$ in the domain $D$: $|x^i - \bar{x}^i| \leq d$, $|x^4| \leq \epsilon(x^i_0)$. 

These functions take for $x^4=0$ the values $W_0$ and $U_0$ defined by the Cauchy data and satisfy the inequalities
\begin{equation} \label{eq:5.20} |W_1 - W_0| \leq l, \hspace{0.5cm} |U_1 - U_0| \leq l, \end{equation}
where $l$ is the same number occurring in the assumptions B. The functions $\Omega_0$, $W_0$ and $U_0$ define a point $M_0 \in \mathcal{F}$.
\end{description}
Let us now define in the space $\mathcal{F}$ the distance of two points $M_1$ and $M_1'$ by the sum of the upper bounds of the absolute values of differences of their coordinates:
\begin{equation*} d(M_1, M_1')=Max \Biggl\{ \sum |\Omega_1' - \Omega_1| + \sum |W_1' - W_1 | + \sum |U_1' - U_1| \Biggr\}. \end{equation*}
The space $\mathcal{F}$ is then a normed, complete and compact space.

To the point $M_1$ of the space $\mathcal{F}$ we associate a point $M_2$ whose coordinates $\Omega_2$, $W_2$, $U_2$ are defined by
\begin{equation} \begin{split} \label{eq:5.21} &\Omega_2 = \Omega_0 + \int_{x^4_0}^{x^4} F_1 d\omega^4, \\
& W_2 = W_0 + \int_{0}^{x^4} G_1 d\omega^4, \\
& 4 \pi U_2 (x^\alpha_0) = \int_{x^4_0}^{0} \int_0^{2 \pi} \int_0^\pi H_1 d\omega^4 d \lambda_2 d \lambda_3 + \int_0^{2 \pi} \int_0^\pi I_1 d \lambda_2 d\lambda_3. \end{split}\end{equation}
where $F_1$, $G_1$, $H_1$ and $I_1$ denote the quantities $F$, $G$, $H$ and $I$ occurring in the equations $[2]$, $[3]$ and $[4]$, evaluated with the help of the functions $X$, solutions of the equations $[1]$, and by replacing the unknown functions $\Omega$, $W$ and $U$ with the coordinates $\Omega_1$, $W_1$ and $U_1$ of the point $M_1$.

Let us now prove that the representation $(\ref{eq:5.21}$) is a representation of the space $\mathcal{F}$ into itself if $\epsilon(x^i_0)$ is suitably chosen.

\begin{description}
\item[(1)] $F_1$ is expressed rationally by means of ${A^{\lambda \mu}}^{(1)}$, $f_s$, ${W_s}^{(1)}$ and of their partial derivatives up to the third order as long as of the ${a^{\alpha \beta}_0}$, ${a_{\alpha \beta}^0}$, and of $\Omega_1$. All these functions are continuous and bounded functions of $x^\alpha_0$, $x^4$, $\lambda_2$ and $\lambda_3$. The denominator ${T^{*4}}^{(1)}$ of these fractions $F_1$ being nonvanishing, the $F_1$ are continuous and bounded functions of $x^\alpha_0$, $x^4$, $\lambda_2$ and $\lambda_3$: $|F_1| \leq N$, where $N$ depends only on the bounds B and on $h$.

Hence, $\Omega_2$ and $\tilde{\Omega}_2$ are continuous and bounded functions of their arguments and verify
\begin{equation} \label{eq:5.22} |\Omega_2 - \Omega_0| \leq N |x^4_0 - x^4|, \hspace{0.5cm} \tilde{\Omega}_2 \leq N. \end{equation}
If $\epsilon(x^i_0)\leq \frac{h}{N}$, we shall have $|\Omega_2 - \Omega_0| \leq h$.

Then $\Omega_2$ satisfies the same conditions as $\Omega_1$ and the number $N$, which is the upper bound of the $F_1$ in $\Lambda$, occurring in the inequality $(\ref{eq:5.19}$), have been chosen for this purpose.
\item[(2)] $G_1$ being an $U_1$ or a $W_1$, the $W_2$ are continuous and bounded in $D$ by a number $P$ which depends only on the bounds B: 
$$|W_2 - W_0| \leq |x^4P|,$$
from which, for $\epsilon(x^i_0) \leq \frac{l}{p}$, we have
$$|W_2 - W_0| \leq l.$$
\item[(3)] Let us show that the functions $H_1$ are bounded by a number which only depends on the bounds B, B' and on $h$.

\begin{description}
\item[(a)] Let us consider the quantity ${D^*}^{(1)}$ occurring in the denominator: It is a polynomial of the functions ${A^{*\lambda \mu}}^{(1)}$, $X$, $\tilde{X}$ and $p_i^0$ which takes the value -1 for $x^4=x^4_0$ and $X=X_0$. By virtue of the inequalities $(\ref{eq:5.17})$ and ($\ref{eq:5.14}$), verified by the functions $x^i$ and the variable $x^4$ in the domain $\Lambda$, ${A^{*\lambda \mu}}^{(1)}$ verifies Lipschitz conditions with respect to the $x^\alpha$ in $\Lambda$. 

Hence, we obtain some inequalities verified by the functions $X$ and $\tilde{X}$ and some assumptions B stating that
\begin{equation*} \bigg{|} {D^*}^{(1)} + 1 \bigg{|} \leq D' \Biggl\{ \sum |X - X_0| + |x^4 - x^4_0| \Biggr\} \leq D'(mM+1)\epsilon(x^i_0), \end{equation*}
where $D'$ is a number which depends only on the bounds B and B'. Thus, we are able to choose $\epsilon(x^i_0)$ sufficiently small so that ${D^{*}}^{(1)}$ does not vanish. 
\item[(b)] Let us consider the rational fraction $H_{1a}$ with denominator $\big{(}{D^{*}}^{(1)}\big{)}^3 \\ \times (x^4_0 - x^4){T^{*4}}^{(1)}$. Its numerator is the product by $\big{(}[U_R]_1 \omega^{R}_{s_1} \big{)}$ of a polynomial of the functions $X$, $\tilde{X}$ and $p_i^0$: quantities that are all known, possessing the same properties as Chapter 4. Thus, the quotient $x^4_0 - x^4$ of the polynomial $p$ is a function continuous and bounded in $\Lambda$. The bound on this function is deduced from the Lipschitz conditions verified by $p$: 
$$p\leq P' \Bigl\{ \sum |X - X_0| + |x^4 - x^4_0| \Bigr\},$$
where $P'$ is a number which depends only on the bounds B and B'.

Thus, we have
\begin{equation*} \frac{p}{x^4_0 - x^4} \leq P' (mM + 1). \end{equation*}
The $H_{1 a}$ can be therefore put in the form of fractions with numerator 
\begin{equation*} [U_R]_1 \omega^R_{s_1} \frac{p}{x^4_0 - x^4} \end{equation*}
continuous and bounded in $\Lambda$, with denominator ${D^{*}}^{(1)}{T^{*4}}^{(1)}$ continuous and bounded in $\Lambda$. Hence, the $H_{1a}$ are continuous and bounded in $\Lambda$ and their bound depends only on the bounds B, B' and $h$.

\item[(c)] The $H_{1b}$, which are rational fractions with nonvanishing denominator of the functions continuous and bounded in $\Lambda$, are continuous and bounded in $\Lambda$. Eventually, we see that $H_1$ are continuous and bounded in $\Lambda$: 
$$ |H_1| \leq Q,$$
where $Q$ depends only on B, B' and on $h$.
\end{description}
\item[(4)] Let us consider $I_1$. Since
\begin{equation} \label{eq:5.23} I = \Biggl\{ E^{*i}_S \frac{D^*p_i}{T^{*4}}(x^4_0 - x^4)^2 sin (\lambda_2) \Biggr\}_{x^4=0}, \end{equation}
where the $E^{*i}_S$ involve the partial derivatives of the $\sigma^R_S$ with respect to the $x^i$ of first order only and linearly. Moreover, if we apply the results of Chapter 4, we see that $E^{i*}_{S_1}(x^4_0 - x^4)^2$ are continuous and bounded in $\Lambda$ because $X$, $\tilde{X}$, ${D}^{(1)}$ and ${D^*}^{(1)}$ and their partial derivatives possess the same properties as in Chapter 4, and that the $\Omega_1$ and $\tilde{\Omega}_1$ are continuous and bounded. 

Furthermore, the products of all terms of $(E^{*i}_S)_1$ by $x^4_0 - x^4$ are bounded by a number $R_1$ depending only on the bounds B, B' and on $h$, with the exception of the term
\begin{equation} \label{eq:5.24} - [U_R]_1 \omega^R_{S_1} \big{[}{A^{ij}}^{(1)} \big{]} \frac{\partial {\sigma}^{(1)}}{\partial x^j}. \end{equation}
Thus, we have
\begin{equation} \label{eq:5.25} I_1 \leq R_1 |x^4_0| + \Phi_R \big{(}\omega^R_{S_1} \big{)}_{x^4=0} \Biggl\{ {A^{ij}}^{(1)}\frac{\partial {\sigma}^{(1)}}{\partial x^j}p_i \frac{{D^*}^{(1)}}{{T^{*4}}^{(1)}}(x^4_0 - x^4)^2 \Biggr\}_{x^4=0} sin(\lambda_2), \end{equation}
where $J= \Bigl\{ {A^{ij}}^{(1)}\frac{\partial {\sigma}^{(1)}}{\partial x^j}p_i \frac{{D^*}^{(1)}}{{T^{*4}}^{(1)}}(x^4_0 - x^4)^2 \Bigr\}$ is a known quantity, which verifies a Lipschitz condition with respect to the functions $X$, $\tilde{X}$ and the variable $x^4$ and which takes the value 1 for $x^4=x^4_0$.

Therefore we have in $\Lambda$:
\begin{equation} \label{eq:5.26} |J - 1|\leq R_2 |x^4 - x^4_0|\hspace{0.5cm} {\rm and} \hspace{0.5cm} |(J)_{x^4=0} - 1 | \leq R_2 |x^4_0|, \end{equation}
where $R_2$ is a number that depends only on the bounds B, B' and on $h$. Furthermore, from the inequality ($\ref{eq:5.18}$), it follows that
\begin{equation} \label{eq:5.27} \bigg{|} \big{(} \omega^R_{S_1} \big{)}_{x^4=0} - \delta^R_S \bigg{|} \leq N |x^4_0|, \end{equation}
and the inequalities ($\ref{eq:5.25}$),  ($\ref{eq:5.26}$) and ($\ref{eq:5.27}$) imply that
\begin{equation*} |I_1 - \Phi_s sin(\lambda_2) | \leq R_3 |x^4_0|, \end{equation*}
where $R_3$ is a number which depends only on the bounds B, B' and on $h$.

The previous inequality is verified at every point $x^i(x^i_0, 0, \lambda_2, \lambda_3)$ of the domain $(d)$. We have assumed in B' that the $\Phi_S$ were verifying some Lipschitz conditions with respect to the $x^i$:
\begin{equation*} |\Phi_S(x^i) - \Phi_S(x^i_0)| \leq l'_0 |x^i - x^i_0|. \end{equation*}
The $x^i$ verify 
$$|x^i - x^i_0| \leq M_1 |x^4_0 - x^4|$$
and, having taken here for value $x^4=0$, we have
\begin{equation} \label{eq:5.28} |\Phi_S(x^i) - \Phi_S(x^i_0)| \leq l'_0 M|x^4_0|. \end{equation}
Eventually, we see that there exists a number $R$, which depends only on the bounds B, B' and on $h$, such that
\begin{equation*} |I_1 - \Phi_S(x^i_0)sin(\lambda_2)| \leq R|x^4_0|. \end{equation*}
\end{description}
The functions 
\begin{equation*} U_2 = \frac{1}{4 \pi} \int_{x^4_0}^{0} \int_0^{2 \pi} \int_{0}^{\pi} H_1 d\omega^4 d\lambda_2 d \lambda_3 + \frac{1}{4 \pi}  \int_0^{2 \pi} \int_{0}^{\pi} I_1 d\lambda_2 d \lambda_3 \end{equation*}
are hence continuous and bounded functions of the $x^\alpha_0$ and verify, by denoting $\Phi_S(x^i_0)$ by $U_0$, the inequality
\begin{equation*} |U_2 - U_0| \leq |x^4_0| \frac{\pi}{2} (Q + R), \end{equation*}
from which, for
\begin{equation} \label{eq:5.29} \epsilon(x^i_0)\leq \frac{2l}{\pi(Q+R)} \end{equation}
we shall have 
$$|U_2 - U_0| \leq l.$$
The functions $\Omega_2$, $W_2$ and $U_2$ possess then the same properties as $\Omega_1$, $W_1$ and $U_1$. Thus, the point $M_2$ is a point of $\mathcal{F}$ if $x^i_0$ verifies, besides the inequalities that were imposed upon it in the solution of the equations $[1]$, the inequalities ($\ref{eq:5.28}$), ($\ref{eq:5.29}$) and ($\ref{eq:5.22}$).

At this stage, it is possible to evaluate the distance of the points $M_2$ and $M_2'$ representative of $M_1$ and $M_1'$. From the Eqs. ($\ref{eq:5.20}$), defining the representation, we have that in the domain $\Omega$
$$ \Omega_2' - \Omega_2 \leq |x^4_0 - x^4| Max_{\Lambda} |F'_1 - F_1|. \hspace{5cm} (1)$$
It turns out from the expression $F_1$, from the assumptions B and the assumptions made on $\Omega_1$ and $W_1$, that $F_1$ verifies a Lipschitz condition with respect to the functions $\Omega_1$ and $W_1$ whose $N'$ coefficient depends only on the bounds B and on $h$. It implies the inequality
\begin{equation} \label{eq:5.30} |\Omega_2' - \Omega_2| \leq N'|x^4_0 - x^4| Max \Biggl\{ \sum |\Omega_1' - \Omega_1| + \sum |W_1' - W_1| \Biggr\}, \end{equation}
$$ |W_2' - W_2| \leq |x^4|Max_D |G_1' - G_1|, \hspace{5cm} (2)$$
where $G_1$ is a function $W_1$ or $U_1$; we have
\begin{equation*} |W_2' - W_2| \leq |x^4|Max_D \Biggl\{  \sum |W_1' - W_1| + \sum |U_1' - U_1| \Biggr\}, \end{equation*}
$$|U_2' - U_2| \leq \frac{\pi}{2} |x^4_0| Max_D |H_1' - H_1| + \frac{\pi}{2} Max_d (I_1' - I_1). \hspace{1.5cm} (3)$$
\begin{description}
\item[(a)] It turns out, from the fact that the polynomial $p$ occurring in the numerator of the function $H_{a}$ is independent of the point $M_1$ of $\mathcal{F}$ that we consider, from the assumptions B and from the previous inequalities, that $H_1$ verifies a Lipschitz condition with respect to the functions $\Omega$, $\tilde{\Omega_1}$, $W_1$ and $U_1$ whose $R'_1$ coefficient depends only on the bounds $B$, $B'$ and on $h$:
\begin{equation*}\begin{split} |H_1' - H_1| \leq R_1' \Biggl\{ & \sum |\Omega_1' - \Omega_1| + \sum |{\tilde{\Omega}}_1' - {\tilde{\Omega}}_1|\\
& + \sum |W_1' - W_1| + \sum |U_1' - U_1| \Biggr\}. \end{split}\end{equation*}
\item[(b)] Let us consider the quantity $I_1$, given by ($\ref{eq:5.21}$), where the only unknown functions are the functions $(\Omega_1)_{x^4=0}$. The expression of the $E^i_S$, the results of the Chapter 4 and those obtained from the solution of the equations $[1]$, the assumptions B and those made upon $\Omega_1$, show that the product 
$$\Bigl\{ E^i_{S_1}(x^4_0-x^4)^2 \Bigr\}_{x^4=0}$$ 
verifies a Lipschitz condition with respect to the functions $(\Omega_1)_{x^4=0}$ whose $R_2'$ coefficient depends only on the bounds B, B' and $h$:
\begin{equation*} |I_1' - I_1| \leq R_2' \sum |\Omega_1' - \Omega_1|_{x^4=0}. \end{equation*}
Therefore, we have
\begin{equation} \begin{split} \label{eq:5.31} & |U_2' - U_2| \leq R_2' |x^4_0| Max_D \Biggl\{ \sum|\Omega_1' - \Omega_1| + \sum |{\tilde{\Omega}}_1' - {\tilde{\Omega}}_1| \\
& + \sum |W_1' - W_1| + \sum |U_1' - U_1| \Biggr\} + \frac{\pi}{2}R_2' Max_d \sum |\Omega_1' - \Omega_1|_{x^4=0}. \end{split}\end{equation}
Let us then consider the point $M_3$ representative of the point $M_2$. The transformation mapping $M_1$ into $M_3$ is a representation of the space $\mathcal{F}$ into itself. Let us compute the distance of two representative points.

We shall deduce from the inequality ($\ref{eq:5.30}$) that
\begin{equation} \label{eq:5.32} |{\tilde{\Omega}}_2' - {\tilde{\Omega}}_2| \leq N' Max_{\Lambda} \Biggl\{ \sum |{\Omega}_1' - {\Omega}_1| + \sum |W_1' - W_1| \Biggr\}. \end{equation}
The inequalities $(\ref{eq:5.30}$), ($\ref{eq:5.31}$) and ($\ref{eq:5.32}$), written for the representations $M_1 \rightarrow M_2$ and $M_2 \rightarrow M_3$, show that there exists a number $\alpha$ non vanishing, depending only on the bounds B, B' and on $h$, such that, for 
$$\epsilon(x^i_0) < \alpha,$$
one has
\begin{equation*} d(M_3, M_3') < kd(M_1, M_1'), \end{equation*}
where $k$ is a given number less than 1.

Hence, the representation of the space $\mathcal{F}$ into itself which leads from $M_1$ to $M_3$ admits a unique fixed point, and the same holds for the representation ($\ref{eq:5.21})$ originally given.
\end{description}
$\textbf{Conclusion.}$ The exists a number $\epsilon(x^i_0)$, which depends only on the bounds $B$, $B'$ and on $h$ and nonvanishing, such that, in the representative domains:
$$|x^i_0 - \bar{x}^i| \leq d, \;|x^4_0| \leq \epsilon(x^i_0), \; 0 \leq x^4 \leq x^4_0, \; 0 \leq \lambda_2 \leq  \pi, \; 0 \leq \lambda_3 \leq 2 \pi; \hspace{0.5cm} (1) $$
$$  |x^i - \bar{x}^i| \leq d, \; |x^4| \leq \epsilon(x^i_0).  \hspace{3cm} (2) $$
The equations $[2]$, $[3]$ and $[4]$ have a unique solution, continuous and bounded $\Omega(x^\alpha_0, x^4, \lambda_2, \lambda_3)$ and $W(x^\alpha)$, $U(x^\alpha)$ verifying the inequalities 
\begin{equation*} |\Omega - \Omega_2| \leq l, \hspace{0.5cm}|W-W_0| \leq l, \hspace{0.5cm} |U-U_0| \leq l. \end{equation*}

We shall prove in addition that the functions $W$ and $U$ obtained satisfy, as ${W^{(1)}}$ and $U^{(1)}$, some Lipschitz conditions with respect to the variables $x^i$. 

In order to prove it, it is enough to make on the functional space $\mathcal{F}$ the following assumptions:

$\mathbf{Assumptions}$
\begin{description}
\item[(1)] The functions $\Omega_1$ and $\tilde{\Omega}_1$ satisfy Lipschitz conditions with respect to the three arguments $x^i_0$
\begin{equation} \label{eq:5.33} \bigg{|} \Omega_1(x^i_0, x^4_0, x^4, \lambda_2, \lambda_3) - \Omega_1(x'^{i}_0, x^4_0,x^4, \lambda_2, \lambda_3) \bigg{|} \leq h' \sum_{i=1}^3 | x'^{i}_0 - x^i_0| \end{equation}
with $h' \leq |x^4_0 - x^4| N'$; in particular
\begin{equation} \label{eq:5.34} \bigg{|} \tilde{\Omega}_1(x^i_0, ...) - \tilde{\Omega}_1(x'^{i}_0, ...) \bigg{|} \leq N' \sum_{i=1}^3 |x'^{i}_0 - x^i_0|, \end{equation}
where $h'$ is an arbitrary given number, $N'$ is a function of the previous bounds.
\item[(2)] The functions $W_1$ and $U_1$ satisfy Lipschitz conditions with respect to the $x^i$:
\begin{equation} \begin{split} \label{eq:5.35} &|W_1(x'^{i},x^4) - W_1(x^i,x^4)| \leq l \sum_{i=1}^3 |x'^{i} - x^i|,\\
&|U_1(x'^{i},x^4) - U_1(x^i, x^4)| \leq l \sum_{i=1}^3 |x'^{i} - x^i|. \end{split} \end{equation}
\end{description}
Hence, $\mathcal{F}$ endowed with the previous norm, is still a normed, complete and compact space. Then, let us now show that the representative points $M_2$ of the points $M_1 \in \mathcal{F}$ are still points of $\mathcal{F}$ if $\epsilon(x^i_0)$ is suitably chosen.
\begin{equation} \label{eq:5.36} \Omega_2(x'^{i}_0, ...) - \Omega_2(x^i_0, ...) = \int_{x^4_0}^{x^4} \bigg{(} F_1(x'^{i}_0, ...) - F_1(x^i_0, ...) \bigg{)} d \omega^4, \end{equation}
where the quantities $F_1(x'^{i}_0,...)$ and $F_1(x^i_0, ...)$ are evaluated with the help of the functions $X(x'^{i}_0, ...)$, more precisely of $x^i(x'^{i}_0, ...)$, $\Omega_1((x'^{i}_0, ...)$ and $x^i(x^i_0, ...)$, $\Omega_1(x^i_0, ...)$, respectively.

It turns out from the expression $F_1$ and from the inequalities $(\ref{eq:5.17}$) and $(\ref{eq:5.22}$) that
\begin{equation} \begin{split} \label{eq:5.37} &|F_1(x'^{i}_0, ...) - F_1(x^i_0, ...)| \leq N' \sum_{i=1}^3 |x'^{i}_0 - x^i_0|, \\
&|\Omega_2(x'^{i}_0, ...) - \Omega_2(x^i_0, ...)| \leq |x^4_0 - x^4| N'\sum_{i=1}^3 |x'^{i}_0 - x^i_0|, \end{split} \end{equation}
and hence, if 
$$\epsilon(x^i_0) \leq \frac{h'}{N'},$$
we will have
\begin{equation} \label{eq:5.38}|\Omega_2(x'^{i}_0, ...) - \Omega_2(x^i_0, ...)| \leq h' \sum_{i=1}^3|x'^{i}_0 - x^i_0|. \end{equation}
If $N'$ denotes the number, that depends only on the bounds B, B' and $h$, occurring in Eq. ($\ref{eq:5.34}$), we will have
\begin{equation*} |{\tilde{\Omega}}_2(x'^{i}_0, ...) - {\tilde{\Omega}}_2(x^i_0, ...)| \leq N'\sum_{i=1}^3 |x'^{i}_0 - x^i_0|. \end{equation*} 
Furthermore, if we consider
\begin{equation} \begin{split} \label{eq:5.39} |W_2(x'^{i},x^4) - W_2(x^i, x^4)| =& \int_{x^4_0}^{x^4} \bigg{(} G_1(x'^{i}, x^4) - G_1(x^i, x^4) \bigg{)} dt \\
& + W_0( x'^{i}) - W_0(x^i), \end{split} \end{equation}
where $G_1$ is a function $W_1$ or $U_1$, the Eq. ($\ref{eq:5.35}$) shows that, under the assumptions B' on the Cauchy data, we have
\begin{equation*} |W_2(x'^{i},x^4) - W_2(x^i, x^4)| \leq |x^4| l \sum_{i=1}^3 |x'^{i} - x^i|  + l_0 \sum_{i=1}^3 | x'^{i} - x^i|. \end{equation*}
Hence, we see that
 $$\epsilon(x^i_0) \leq \frac{ l - l_0}{l}$$
implies 
\begin{equation*} \begin{split} & |W_2(x'^{i},x^4) - W_2(x^i, x^4)| \leq l \sum_{i=1}^3 |x'^{i} - x^i|, \\
& U_2(x'^{i},x^4) - U_2(x^i, x^4) = \int_{x^4_0}^0 \int_0^{2 \pi} \int_0^{\pi} [H_1(x'^{i}_0, ...) - H_1(x^i_0, ...)] d\omega^4 d\lambda_2 d \lambda_3 \\
&+ \int_0^{2 \pi} \int_0^{\pi} [I_1(x'^{i}_0, ...) - I_1(x^i_0, ...)] d\lambda_2 d \lambda_3. \end{split} \end{equation*}
The quantities $H_1(x'^{i}_0)$, $I_1(x'^{i}_0)$ and $H_1(x^{i}_0)$, $I_1(x^{i}_0)$ are evaluated by means of the functions $X(x'^{i}_0, ...)$, $\Omega_1(x'^{i}_0, ...)$ and $X(x^{i}_0, ...)$, $\Omega_1(x^{i}_0, ...)$, respectively.

$\mathbf{Quantity \; H_1}$
\begin{description}
\item[(a)] Let us consider the polynomial $p$ occurring in the denominator of $H_{1a}$, $p$ is a polynomial of the functions $[{A^{\lambda \mu}}^{(1)}]$, ${W_s}^{(1)}(x^\alpha)$, of their first and second partial derivatives, of the functions $X$, $\tilde{X}$ and $p_i^0$. 

The Taylor series expansion of this polynomial, starting from the values 
\begin{equation*} \begin{split}& [{A^{\lambda \mu}}^{(1)}]_0= \delta^\mu_\lambda, \; \bigg{[} \frac{\partial {A^{\lambda \mu}}^{(1)}}{\partial x^\alpha} \bigg{]} = \bigg{[} \frac{\partial {A^{\lambda \mu}}^{(1)}}{\partial x^\alpha} \bigg{]}_0,\; \dots,\;\\
& {W_s}^{(1)}(x^\alpha) ={W_s}^{(1)}(x^\alpha_0), \;{W_{s \alpha}}^{(1)}(x^\alpha) ={W_{s \alpha}}^{(1)}(x^\alpha_0), \; \dots, \\
&\; X=X_0,\; \tilde{X} = \tilde{X}_0  \end{split}\end{equation*}
for which the polynomial $p$ vanishes, shows that $p$ is a polynomial of the functions already listed, and of the functions $[{A^{\lambda \mu}}^{(1)}]- \delta^\mu_\lambda$,$\dots$, ${W_s}^{(1)}(x^\alpha) -{W_s}^{(1)}(x^\alpha_0)$, ..., $\tilde{X} - \tilde{X}_0$, $X-X_0$ whose terms are at least of first degree with respect to the set of these last functions.

The quantity $\frac{p}{x^4_0 - x^4}$ is therefore a polynomial of the functions 
$${A^{\lambda \mu}}^{(1)}, \; \dots, \; {W_s}^{(1)}(x^\alpha), \; \dots, \; X, \; \tilde{X}, \; p_i^0$$
and of the functions 
$$\frac{\big{[} {A^{*\lambda \mu}}^{(1)}\big{]} - \delta^{\mu}_\lambda}{x^4_0 -x^4}, \; \frac{{W_s}^{(1)}(x^\alpha) - {W_s}^{(1)}(x^\alpha_0)}{x^4_0 - x^4}, \; \dots, \;\frac{X - X_0}{x^4_0 - x^4}, \; \frac{\tilde{X} - \tilde{X}_0}{x^4_0 - x^4}.$$
Since the coefficients ${A^{*\lambda \mu}}^{(1)}$ and the functions ${W_s}^{(1)}$ admit bounded derivatives with respect to the $x^\alpha$ up to the fourth order, whereas the functions considered involve only derivatives of the first two orders, it turns out from the assumptions B and the inequalities ($\ref{eq:5.15}$) and ($\ref{eq:5.16}$) that all the listed functions are bounded in $\Lambda$ by a number which only depends on the bounds B and B'.

Thus, the polynomial $\frac{p}{x^4_0 - x^4}$ verifies a Lipschitz condition with respect to each of these functions, whose coefficient depends only on the bounds B and B'. Then, we are going to prove that these functions themselves verify Lipschitz conditions with respect to the $x^i_0$. It will be enough for us, by virtue of the assumptions B and the previous inequalities to prove this result for:
\begin{description}
\item[(1)] the functions $\frac{\big{[} {A^{*\lambda \mu}}^{(1)}\big{]} - \delta^{\mu}_\lambda}{x^4_0 -x^4}$ and $\frac{{W_s}^{(1)}(x^\alpha) - {W_s}^{(1)}(x^\alpha_0)}{x^4_0 - x^4}$ and the analogous functions written with first and second partial derivatives of ${A^{*\lambda \mu}}^{(1)}$ and ${W_s}^{(1)}$ with respect to the $x^\alpha$;
\item[(2)] The functions $\frac{X - X_0}{x^4_0 - x^4}$.
\end{description}
Let us begin with $\mathbf{(1)}$ by setting 
\begin{equation*} F(x^i_0, x^4_0, x^4, \lambda_2, \lambda_3) = \frac{{A^{*\lambda \mu}}^{(1)} - \delta^\mu_\lambda}{x^4_0 - x^4}, \end{equation*}
where ${A^{*\lambda \mu}}^{(1)} - \delta^\mu_\lambda = {A^{*\lambda \mu}}({W_s}^{(1)}(x^i, x^4), {W_s}^{(1)}(x^i_0, x^4_0), x^i, x^4, x^i_0, x^4_0)$\\ $- {A^{*\lambda \mu}}^{(1)}({W_s}^{(1)}(x^i, x^4), {W_s}^{(1)}(x^i_0, x^4_0), x^i_0, x^4_0, x^i_0, x^4_0)$, with 
$$x^i = x^i(x^i_0, x^4_0, x^4, \lambda_2, \lambda_3).$$
Let us consider the quantity $F(x'^{i}_0, ...)- F(x^i_0, ...)$. The function occurring in the numerator vanishes for $x^4=x^4_0$, because the two functions $F(x'^{i}_0, ...)$ and $F(x^i_0, ...)$ vanish, and it admits a derivative with respect to $x^4$ continuous and bounded in the domain $\Lambda$. Thus, we have
 \begin{eqnarray}
\;&\;& F(x'^{i}_0, ...)- F(x^i_0, ...) \nonumber \\
&\;& = \Biggl\{ \frac{\partial}{\partial x^4} \bigg{[} ({A^{*\lambda \mu}}^{(1)}(x'^{i}_0, ...) - \delta^\mu_\lambda) - ({A^{*\lambda \mu}}^{(1)}(x^{i}_0, ...) - \delta^\mu_\lambda) \bigg{]}\Biggr\}_{x^4=x^4_0 - \theta(x^4 - x^4_0)} \end{eqnarray}
where $0 \leq \theta \leq 1$.

Since the derivative of the function ${A^{*\lambda \mu}}^{(1)}(x'^{i}_0, ...)$ with respect to the parameter $x^4$ verifies a Lipschitz condition with respect to the $x^i_0$, whose coefficient depends only on the bounds B and B', we see eventually that
\begin{equation*} F(x'^{i}_0, ...)- F(x^i_0, ...) \leq L_1 \sum_{i=1}^3 |x'^{i}_0 - x^i_0| \end{equation*}
where $L_1$ depends only on the bounds B and B'.

The same proof holds for the function $\frac{{W_s}^{(1)}(x^\alpha) - {W_s}^{(1)}(x^\alpha_0)}{x^4_0 - x^4}$ and for the functions built with the partial derivatives of the ${A^{*\lambda \mu}}^{(1)}$ or ${W_s}^{(1)}$ up to the third order included.

Eventually, we can prove the same result for $\mathbf{(2)}$. We have
\begin{equation*} \tilde{X} - \tilde{X}_0 = \frac{ \int_{x^4_0}^{x^4} (E- E_0)d\omega^4}{x^4_0 - x^4}, \end{equation*}
from which
\begin{equation*} (\tilde{X} - \tilde{X}_0)_{x'^{i}_0} - (\tilde{X} - \tilde{X}_0)_{x^i_0} = \frac{\int_{x^4_0}^{x^4} [ (E- E_0)_{x'^{i}_0} - (E- E_0)_{x^i_0} ] d \omega^4}{x^4_0 - x^4}, \end{equation*}
where $E$ is a rational fraction with denominator ${T^{*4}}^{(1)}$ of the coefficients ${A^{*\lambda \mu}}^{(1)}$ and of their partial derivatives up to the third order and of the functions $X$. We can write $E-E_0$ in the form of a rational fraction with denominator ${T^{*4}}^{(1)}$, because ${T^{*4}}^{(1)}=1$ for $x^4=x^4_0$, of the previous functions and of the functions $X-X_0$, ${A^{*\lambda \mu}}^{(1)} -{ {\delta}^\mu}_\lambda$, $\dots$, whose denominator has all its terms of first degree at least with respect to the set of these functions. Then, we can write
$$E-E_0=(x^4_0 - x^4)F,$$ 
where $F$ is a rational fraction with denominator ${T^{*4}}^{(1)}$ of the previous functions and of the functions 
$$\frac{X - X_0}{x^4_0 - x^4}, \; \frac{{A^{*\lambda \mu}}^{(1)} - \delta^\mu_\lambda}{x^4_0 - x^4}, \; \dots .$$
Since all these functions verify Lipschitz conditions with respect to the $x^i_0$, we have
\begin{equation*} |(E- E_0)_{x'^{i}_0} -(E- E_0)_{x^{i}_0} | \leq L_2 |x^4_0 - x^4| \sum_{i=1}^3 |x'^{i}_0 - x^i_0|, \end{equation*}
from which
\begin{equation*} |(X- X_0)_{x'^{i}_0} - (X- X_0)_{x^{i}_0}| \leq \frac{L_1}{2} |x^4_0 - x^4| \end{equation*}
and
\begin{equation*} \bigg{|} \bigg{(} \frac{X - X_0}{x^4_0 - x^4} \bigg{)}_{x'^{i}_0} -  \bigg{(} \frac{X - X_0}{x^4_0 - x^4} \bigg{)}_{x^{i}_0} \bigg{|} \leq \frac{L_2}{2}. \end{equation*}
Thus, we have proven that the quantity $\frac{p}{x^4_0 - x^4}$ verifies a Lipschitz condition with respect to the $x^i_0$, whose coefficient depends only on the bounds B and B'.
\item[(b)] There remains to prove that the quantity $H_1$, that is the product of the square root of a rational fraction with numerator 1 and non-vanishing denominator with a rational fraction with non-vanishing denominator of the bounded functions verifying all Lipschitz conditions with respect to the $x^i_0$, verifies in $\Lambda$ a Lipschitz condition with respect to the $x^i_0$ whose coefficients $Q'$ depends only on the bounds B, B', $h$ and on $h'$. Hence we have
\begin{equation*} |H'_1 - H_1| \leq Q' \sum_{i=1}^3 |x'^{i}_0 - x^i_0|. \end{equation*}
\end{description}
$\mathbf{Quantity\; I_1}$

By considering the expression of $I_1$ and the previous inequalities, we can prove that all terms of $I_1$, with the exception of the term ($\ref{eq:5.23}$), verify Lipschitz conditions with respect to the $x^i_0$ whose coefficient is of the form $R'_1|x^4_0|$, where $R'_1$ is a number that depends only on the bounds B and B'.

Let us consider $(\ref{eq:5.23}$). We find that $\frac{J(x^i_0) - 1}{x^4_0 - x^4}$ verifies a Lipschitz condition with respect to the variables $x^i_0$, from which
\begin{equation*} |J(x'^{i}_0) -J(x^{i}_0)|_{x^4=0} \leq R'_1|x^4_0| \sum_{i=1}^3 |x'^{i}_0 - x^i_0|, \end{equation*}
from which, by using the inequality $(\ref{eq:5.33}$) and the inequalities on $I$, we have
\begin{equation*} |I_0(x'^{i}_0) - I_0(x^{i}_0)| \leq R''_0|x^4_0| \sum_{i=1}^3 |x'^{i}_0 - x^i_0| + |U_0(x'^{i}_0) - U_0(x^i_0)|(1+ R''_2|x^4_0|). \end{equation*}
Then, we obtain Lipschitz conditions verified by $U_0$
\begin{equation*} |I_1(x'^{i}_0) -I_1(x^{i}_0)| \leq (R'|x^4_0| + l_0) \sum_{i=1}^3 |x'^{i}_0 - x^i_0|, \end{equation*}
where $R'$ is a number that depends only on the bounds B and B'.

Eventually, we shall deduce from the Lipschitz conditions verified by $H_1$ and $I_1$
\begin{equation*} |U_2(x'^{i}_0, x^4) - U_2(x^{i}_0, x^4)| \leq \frac{\pi}{2} [(Q' + R')|x^4_0| + l_0]\sum_{i=1}^3|x'^{i}_0 - x^i_0| \end{equation*}
hence that inequality 
$$\epsilon(x^i_0) \leq \frac{l - l_0}{Q' + R'} \frac{2}{\pi}$$
implies
\begin{equation*} |U_2(x'^{i}_0, x^4) - U_2(x^{i}, x^4)| \leq l \sum_{i=1}^3 |x'^{i}_0 - x^i_0|. \end{equation*}

$\mathbf{Conclusion.}$ The previous inequalities prove that, if $\epsilon(x^i_0)$ satisfies the corresponding inequalities, the point $M_2$ is, under the assumptions made, a point of $\mathcal{F}$. The application of the fixed-point theorem shows that, in the domain $D$, the functions $W$ and $U$ satisfy Lipschitz conditions with respect to the $x^i$ with coefficient $l$.

The functions $W$ and $U$, solutions of the integral equations $[J_1]$, satisfy therefore, in $D$, the same inequalities holding for the functions $W^{(1)}$, $\dots$, $U^{(1)}$.

\subsection{Solution of the equations $G_1$}
We will now prove that the functions $W_s$, which are solutions of the equations $I_1$, solve the equations $G_1$, and that the functions $W_{s \alpha}$, $\dots$, $U_S$, which are solutions of the equations $I_1$, are the partial derivatives up to the fourth order of the $W_s$ in a domain $D$ that depends only on the bounds B and B'. We shall use for the proof the approximation of continuous functions by means of analytic functions.

Let us consider some equations $[G_1]$:
\begin{equation*} {A^{\lambda \mu}}^{(1)} \frac{\partial^2 W_s}{\partial x^\lambda \partial x^\mu} + f_s =0, \hspace{7cm} [G_1] \end{equation*}
where the coefficients $A^{\lambda \mu}$, $f_s$, ${W_s}^{(1)}$ and the Cauchy data $\psi_s$, $\varphi_s$ are analytic functions of their arguments. The Cauchy problem for the equations $[G_1]$ admits an analytic solution in a neighbourhood $V$ of the domain $(d)$ of the surface $x^4=0$ carrying the initial data. If the coefficients and the Cauchy data satisfy the assumptions made for the system $[F_1]$, there exists a neighbourhood $V$ of $(d)$ where this solution satisfies the integral equations $[I_1]$.

Furthermore, let us consider, independently of equations $[G_1]$, the integral equations $[I_1]$. We shall prove that they admit, within a domain $D$ that depends only on the bounds B and B', a unique analytic solution which coincides therefore, in the part shared by the domains $V'$ and $D^*$, with the solution of equations $[G_1]$. This principle of analytic continuation shows then that this solution of equations $[I_1]$ is solution of equations $[G_1]$ in the whole of $D$.

Let us prove for example the analyticity in $D$ of the solution of equations $[1]$
\begin{equation*} X= \int_{x^4_0}^{x^4} E d\omega^4 + X_0, \end{equation*}
when $E$ is an analytic function of the quantities $X$, $x^\alpha_0$ and $x^4$, bounded by $M$ in the domain
\begin{equation*} R: \; |X - \bar{X}_0| \leq d, \hspace{0.5cm} |x^i_0 - \bar{x}^i| \leq d, \hspace{0.5cm} |x^4| \leq |x^4_0| \leq \epsilon(x^i_0) \end{equation*}
of variation of its real arguments and it is expandable in an absolutely convergent series in the neighbourhood of every point of $R$. Thus, we can extend the definition of $E$ to a domain of variation of the complex arguments $Z= x +iy$, $z^\alpha_0 =x^\alpha_0 + i y^\alpha_0$, $z^4=x^4+iy^4$ by expressing it in the form of a convergent series, hence holomorphic in the $m$ cylinders $V$, centered at a point whatsoever of $V$ and defined by
\begin{equation*} |Z' - X| \leq a_X, \hspace{0.5cm} |z^{' \alpha}_0 - x^\alpha_0|\leq b_{x^\alpha_0} \hspace{0.5cm} |z^4 - x^4| \leq C_{x^4}. \end{equation*}
The partial derivatives $\frac{\partial E}{\partial X_1}$ being bounded by $M'$ in $R$ one can choose the bounds $a_X$, $b_{x^\alpha_0}$ and $C_{x^4}$ in such a way that in $v$ one has
\begin{equation*}\bigg{|} \frac{\partial E}{\partial Z_1} \bigg{|} \leq M' + \alpha', \end{equation*}
where $\alpha'$ is an arbitrarily small number. One can also choose the bounds $b_{x^\alpha_0}$ and $C_{x^4}$ so that in $v$ one has
\begin{equation*} | I\; E(X_1, z^\alpha_0, z^4)| \leq \beta, \hspace{0.5cm} |R\;E(X_1, z^\alpha_0, z^4)| \leq M + \beta, \end{equation*}
where $\beta$ is an arbitrarily small number. One can build on the other hand a cover of the domain $R$ by means of a finite number of projections in $R$ of the $m$ previous cylinders. The corresponding $m$ cylinders determine a domain $\bar{R}$ of the space of complex arguments $Z$, $z_0$, $z^4$, which fulfill the inequalities
\begin{equation*} \begin{split} &|X - \bar{X}_0| \leq d, \hspace{0.5cm} |x^i_0 - \bar{x}^i| \leq d, \hspace{0.5cm} |x^4| \leq |x^4_0| \leq \epsilon(x^i_0); \\
& |Y| \leq a, \hspace{0.5cm} |y^\alpha_0| \leq b, \hspace{0.5cm} |y^4| \leq c, \end{split} \end{equation*}
where $a$, $b$ and $c$ are non-vanishing numbers, and in which the complex function $E$ is defined and analytic.

Let us write:
\begin{equation*} E(Z_1, z^\alpha_0, z^4) = E(Z_1, z^\alpha_0, z^4) - E(X_1, z^\alpha_0, z^4) +E(X_1, z^\alpha_0, z^4), \end{equation*}
from which
\begin{equation*} \begin{split} & | I\; E(Z_1, z^\alpha_0, z^4)| \leq \beta + m(M' + \alpha')a, \\
& |R\;E(Z_1, z^\alpha_0, z^4)| \leq M + \beta + m(M' + \alpha')a. \end{split} \end{equation*}
Now, let us consider the equations $[1]$ extended to the complex domain $\bar{R}$
\begin{equation*} Z = \int_{z^4_0}^{z^4} E(Z, z^\alpha_0, z^4) d\omega^4 + Z_0. \hspace{6cm} [\bar{1}] \end{equation*}
In order to solve it, we consider, as in the real case, a functional space $\Upsilon$ defined by the functions of complex variables $Z_1(z^\alpha_0, z^4)$, real for $z^\alpha_0$ and $z^4$ real, analytic in the domain $\bar{D}$ defined by
\begin{equation*} |x^i_0 - \bar{x}^i| \leq d, \hspace{0.5cm} |x^4| \leq |x^4_0| \leq \epsilon(x^i_0), \hspace{0.5cm} |y^\alpha_0| \leq b, \hspace{0.5cm} |y^4| \leq c, \end{equation*}
and satisfying $|X_1 - X_0| \leq d$ and $|y_1| \leq a$.

\begin{description} 
\item[(a)] The representation
\begin{equation*} Z_2 = Z-0 + \int_{z^4_0}^{z^4} E(Z_1, z^\alpha_0, z^4)d \omega^4 \end{equation*}
is a representation of the space into itself if $\epsilon(x^i_0)$, $b$ and $c$ are suitably chosen. As a matter of fact:
\begin{description}
\item[(1)] $Z_1$ is an analytic function of $z^\alpha_0$, $z^4$ because this holds for $E$, real for $z^\alpha_0$ and $z^4$ real.
\item[(2)] The equality
\begin{equation*} Z_2 = - \int_{x^4_0}^{x^4_0 + i y^4_0}E d\omega^4 + \int_{x^4_0}^{x^4}E d\omega^4 + \int_{x^4}^{x^4+ i y^4}E d\omega^4 + Z_0\end{equation*} 
implies that
\begin{equation*}\begin{split}  |X_2 - X_0| \leq &(b+c)[m(M' + \alpha')a + \beta] \\
         &+ |x^4_0 - x^4|[m(M' + \alpha')a + \beta + M],\\
 |Y_2| \leq &(b+c) [m(M'+ \alpha')a + \beta + M] \\
&+ |x^4_0 - x^4|[m(M' + \alpha') a + b].\end{split}\end{equation*}
Thus, if $\epsilon(x^i_0) \leq \frac{d-(b+c)[m(M' + \alpha')a + \beta]}{M+ m(M'+\alpha')a+ \beta}$, we have
\begin{equation*} |X_2 - \bar{X}_0| \leq d \end{equation*}
and if $b+c \leq \frac {a[1- mM'(x^4_0 -x^4)]- (m \alpha' a+ \beta)(x^4_0 - x^4)}{M + m(M' + \alpha')a + \beta}$, we have
\begin{equation*} |Y_2|\leq a. \end{equation*}
Let us recall that the number 
\begin{equation*}\epsilon(x^i_0) \leq \frac{1}{mM}. \hspace{6cm} (A) \end{equation*}
Therefore, we have
\begin{equation*}1 - mM'(x^4_0 - x^4) >0. \hspace{5cm} (B) \end{equation*}
Thus, we have to choose $\epsilon(x^i_0)$ as satisfying $(A)$ and the inequality $(B)$ shows that one can find, without supplementary assumptions upon $\epsilon(x^i_0)$, the numbers $b$ and $c$ defining $\bar{D}$, so that $M_2$ is a point of $\mathcal{F}$. The domain $\bar{D}$ has for real part a domain as close as one wants to $D$.
\end{description}
\item[(b)] Let us prove that the representation reduces the distances. We have seen that, in $\bar{R}$, one has $\big{|}\frac{\partial E}{\partial Z_1} \big{|} \leq M' + \alpha'$, from which
\begin{equation*}|E(Z_1', z^\alpha_0, z^4) - E(Z_1, z^\alpha_0, z^4) | \leq |Z'_1 - Z_1|(M' + \alpha'). \end{equation*}
Thus, we shall have
\begin{equation*} d(M_2, M_2') \leq m(M' + \alpha')|z^4_0 - z^4| d(M_1, M_1'), \end{equation*}
from which, if $|z^4_0 - z^4| < \frac{1}{m M' + \alpha}$, $\epsilon(x^i_0) \leq \frac{1}{mM' + \alpha'} - \eta $ and $b+c < \eta$, we have  
\begin{equation*} d(M_2, M_2') \leq d(M_1, M_1'), \end{equation*}
where $\eta$ is an arbitrary small number. 

Thus, the real part of the domain $\bar{D}$ is as close as one wants to $D$.
\end{description}
We can conclude, as in the real case, that the representation $I$ admits a unique fixed point. The corresponding $Z$ functions are solutions of equations $[1]$, and analytic in the domain $\bar{D}$. The functions $X$, values of these functions $Z$ for real arguments $x^4_0$ and $x^4$ are analytic functions, solutions in a domain as close as one wants to $D$ of equations $[1]$.

Analogous results can be proved in the same way for equations $[2]$, $[3]$ and $[4]$.

\subsection{Coefficients and Cauchy data satisfying only the assumptions $B$ and $B'$}
If the coefficients $A^{\lambda \mu}$, $f_s$, the given functions ${W_s}^{(1)}$ and the Cauchy data satisfy only the assumptions B and B', we shall approach uniformly these quantities and at the same time their partial derivatives up to the fourth order, by means of analytic functions $A^{\lambda \mu}_{(n)}$, $f_{s(n)}$, ${W_{s(n)}}^{(1)}$, $\varphi_{s(n)}$ and $\psi_{s(n)}$ verifying the assumptions B and B'. We shall build in this way a family of functions $W_{s(n)}$, ..., $U_{S(n)}$, which are solutions in $D$ of equations $I_{1(n)}$ and solutions in $D$ of the Cauchy problem, relatively to the equations $[G_{1(n)}]$:
\begin{equation*} {A^{\lambda \mu}}_{(n)} \frac{ \partial^2 W_{s(n)}}{\partial x^\lambda \partial x^\mu} + f_{s(n)} =0. \hspace{6cm} [G_{1(n)}] \end{equation*}
These functions $W_{s(n)}$ possess partial derivatives up to the fourth order and satisfy the assumptions B and B'. 

We want to prove that the functions $W_{s(n)}$, ..., $U_{S(n)}$ converge uniformly to some functions $W_s$, ..., $U_S$, when the functions $A^{\lambda \mu}_{(n)}$, ${W_{s(n)}}^{(1)}$, $\varphi_{s(n)}$, $\psi_{s(n)}$ and their partial derivatives converge uniformly to the given functions $A^{\lambda \mu}$, ${W_s}^{(1)}$, $\varphi_s$, $\psi_s$. This is possible by applying the same method we used before and the fact that the functions $W_{(n)}$ and $U_{(n)}$ verify a Lipschitz condition with respect to the $x$ variables (that one has to replace by $X_{(n)}$ in the integral equations $[I_{1(n)}]$ verified by these functions). Thus, we will have 
\begin{equation} \begin{split} \label{eq:5.40}& |X_{(n)} - X_{(m)}| \leq Max_{\Lambda} \Biggl\{ \alpha \bigg{(}\sum |A^{\lambda \mu}_{(n)} - A^{\lambda \mu}_{(m)} | + \dots \\
&+ \sum |{W_{s(n)}}^{(1)} - {W_{s(m)}}^{(1)}| + \dots \bigg{)} + M' \sum |X_{(n)} - X_{(m)}|\Biggr\}|x^4_0 - x^4|,\\
&|\Omega_{(n)} - \Omega_{(m)}| \leq Max_{\Lambda} \Biggl\{ \beta \bigg{(} \sum |A^{\lambda \mu}_{(n)} - A^{\lambda \mu}_{(m)}| + \dots \\
& + \sum |{W_{(n)}}^{(1)} - {W_{(m)}}^{(1)}| + \sum |X_{(n)} - X_{(m)}|\bigg{)} \\
&+ N' \bigg{(}|\Omega_{(n)} - \Omega_{(m)}| + \sum |W_{(n)} - W_{(m)}|\bigg{)} \Biggr\} |x^4_0 - x^4|, 
 \end{split} \end{equation}
and
\begin{equation} \begin{split} \label{eq:5.41} & |W_{(n)} - W_{(m)}| \leq Max \Biggl\{ \sum |W_{(n)}- W_{(m)}| + \sum |U_{(n)} - U_{(m)}| \Biggr\} |x^4| \\
&+ |W_{0(n)}- W_{0(m)}|, \\
& |U_{(n)} - U_{(m)}| \leq Max \Biggl\{ \gamma \bigg{(} \sum |A^{\lambda \mu}_{(n)} - A^{\lambda \mu}_{(m)}| + \dots + \sum |f_{s(n)} - f_{s(m)}|\\
& + \dots + \sum|{W_{(n)}}^{(1)} - {W_{(m)}}^{(1)}| + \sum |X_{(n)} - X_{(m)}| \bigg{)} \\
&+ R'_1 \bigg{(} \sum |U_{(n)} - U_{(m)}| + \sum |W_{(n)} - W_{(m)}| + \sum |\Omega_{(n)} - \Omega_{(m)}|\\
& +\sum |\tilde{\Omega}_{(n)} - \tilde{ \Omega}_{(m)}| \bigg{)} \Biggr\} |x^4_0| + Max \Biggl\{  \delta \bigg{(} \sum |X_{(n)} - X_{(m)}| + \sum |A^{\lambda \mu}_{(n)} - A^{\lambda \mu}_{(m)}| \\
& + \dots + \sum | \Phi_{s(n)} - \Phi_{s(m)}|  \bigg{)} + R'_2 \sum |\Omega_{(n)} - \Omega_{(m)}| \Biggr\}_{x^4=0}, \end{split}\end{equation}
where $\alpha$, $\beta$, $\gamma$, $\delta$ are bounded numbers which only depend on the bounds B, B' and on $h$ and $h'$. The previous inequalities show that the functions $X_{(n)}$, $\Omega_{(n)}$, $W_{(n)}$ and $U_{(n)}$ converge uniformly towards functions $X$, $\Omega$ and $W$, $U$ in their respective domains of definition, $\Lambda$ and $D$, when the approximating functions converge uniformly towards the given functions.

These functions $W$, $U$, uniform limit of the functions $W_{s(n)}$, $U_{(n)}$ satisfy the following properties.
\begin{description}
\item[(p.1)] The functions $W_{s \alpha}$, ..., $U_S$ are partial derivatives up to the fourth order of the functions $W_s$, and all these functions satisfy the same assumptions B and B' as the functions ${W_s}^{(1)}$ in $D$.
\item[(p.2)] The functions $W_s$ verify the partial differential equations $[G_1]$ in the domain $D$.
\end{description} 

\subsection{Solution of the equations $[G]$}
We consider the functional space $W$ defined by the functions ${W_s}^{(1)}$ and satisfying the assumptions B and B' in the domain $D$. We have just proved that the solution evaluated of the Cauchy problem for the equations $[G_1]$ defines a representation of this space into itself. Let us denote by ${W_s}^{(1)}$ this solution.

The space $W$ is a normed, complete and compact space if one defines the distance of two of its points by
\begin{equation*} d(M_1, M_1')=Max_{D} \bigg{(} \sum |{W_s}^{(1)}- {W'_s}^{(1)}| + ... + |{U_S}^{(1)} - {U'_S}^{(1)}| \bigg{)}. \end{equation*}
The distance of two representative points $M_2$, $M_2'$ from $M_1$, $M_1'$ will be compared to the distance of these points with the help of inequalities analogous to ($\ref{eq:5.40}$) and $(\ref{eq:5.41}$). 

Then, there exists a number $\eta$ bounded, non-vanishing and such that if
$$\epsilon(x^i_0)< \eta,$$ 
the distance of two representative points 
$$\big{(} {W'_s}^{(2)}, \dots, {U'_S}^{(2)} \big{)} \hspace{0.5cm} {\rm and} \hspace{0.5cm} \big{(} {W_s}^{(2)}, \dots, {U_S}^{(2)} \big{)}$$
is less than the distance of the initial points.

The representation considered admits then a unique fixed point $(W_s, ..., U_S)$ which belongs to the space. The functions $W_s$ corresponding to this fixed point are solutions of the Cauchy problem associated to the equations $[G]$, in the domain $D$. They possess partial derivatives up to the fourth order, continuous, bounded and satisfying some Lipschitz conditions with respect to the variables $x^i$.

Furthermore, the Cauchy problem relative to the system of non-linear partial differential equations $[G]$, admits in the domain $D$, under the assumptions $H$, a solution possessing partial derivatives up to the fourth order, continuous, bounded and satisfying Lipschitz conditions with respect to the variables $x^i$. This concerns the existence of the solution. 

Another implication of our argumentation is the uniqueness of this solution. As a matter of fact, if we consider the system of integral equations verified by the solutions of the given equations $[G]$, it has only one solution $W_s$, $W_{s \alpha}$, ..., $U_S$ where the $W_{s\alpha}$, ..., $U_S$ are partial derivatives of the $W_s$. In this case it is possible to write inequalities analogous to the inequalities for $[G_{1(n)}]$, where $W_{(n)}$, ..., $U_{(n)}$; ${W_{(n)}}^{(1)}$, ..., ${U_{(n)}}^{(1)}$ and $W_{(m)}$, ..., $U_{(m)}$; ${W_{(m)}}^{(1)}$, ..., ${U_{(m)}}^{(1)}$ are replaced by two solutions of equations $[G]$, respectively. From these inequalities one derives the coincidence of these two solutions.

$\mathbf{Conclusion.}$ We consider a system of non-linear, second-order, hyperbolic partial differential equations with $n$ unknown functions $W_s$ and four variables $x^\alpha$, of the form 

$$A^{\lambda \mu} \frac{\partial^2 W_S}{\partial x^\lambda \partial x^\mu} + f_s=0, \hspace{2cm} \lambda,\mu=1, ...,4, \; s=1, 2, ..., n. \hspace{1cm} [G] $$
\\
The $f_s$ are given functions of the unknown $W_s$, $W_{s \alpha}$ and of the variables $x^\alpha$. 

The $A^{\lambda \mu}$ are given functions of the $W_s$ and of the $x^\alpha$. 

The Cauchy data are, on the initial surface $x^4=0$, 
$$W_s(x^i, 0)= \varphi_s(x^i), \hspace{1cm} W_{s 4}(x^i, 0)=\psi_s(x^i).$$

On the system $[G]$ and the Cauchy data we make the following assumptions:
\begin{description}
\item[(1)] In the domain $(d)$, defined by $|x^i - \bar{x}^i| \leq d$, $\varphi_s$ and $\psi_s$ possess partial derivatives up to the orders five and four, continuous, bounded and satisfying Lipschitz conditions.
\item[(2)] For the values of the $W_s$ satisfying
\begin{equation*} |W_s - \varphi_s| \leq l,\hspace{0.5cm} |W_{si} - \varphi_{si}| \leq l, \hspace{0.5cm} |W_{s4} - \psi_s|\leq l \end{equation*}
\end{description}
and in the domain $D$, defined by 
$$|x^i - \bar{x}^i| \leq d, \hspace{1cm}|x^4| \leq \epsilon:$$
\begin{description}
\item[(a)] $A^{\lambda \mu}$ and $f_s$ possess partial derivatives up to the fourth order, continuous, bounded and satisfying Lipschitz conditions.
\item[(b)] The quadratic form $A^{\lambda \mu}X_\lambda X_\mu$ is of the normal hyperbolic type, i.e. $A^{44}>0$, $A^{ij}X_i X_j$ is negative-definite.
\end{description}
Then the Cauchy problem admits a unique solution, possessing partial derivatives continuous and bounded up to the fourth order, in relations with equations $[G]$ in a domain $\Delta$, which is a tronc of cone with base $d$, defined by
\begin{equation*} |x^i - \bar{x}^i| \leq d, \hspace{1cm} |x^4| \leq \eta(x^i). \end{equation*}

Once we have proved the existence and uniqueness of the solution of the Cauchy problem for non-linear, second-order, hyperbolic partial differential equations we are able now to apply these results to General Relativity. 

In the next Chapter, we will show the solution of the Cauchy problem for the field equations, which are ten partial differential equations of second-order that are linear in the second derivatives of the gravitational potentials and non-linear in their first derivatives.

\chapter{General Relativity and the Causal structure of Space-Time}
\chaptermark{General Relativity and the Causal structure }

\epigraph{There are more things in Heaven and Earth, Horatio, than are dreamt of in your philosophy.}{William Shakespeare, Hamlet }
Once we have argued about the existence and uniqueness of the solution of the Cauchy problem for systems of linear and non-linear equations, we are ready to discuss the applications to General Relativity. This will be the object of the discussion of the first part of this Chapter. More precisely, we will discuss how is it possible to use the results obtained in the previous chapters to solve the Cauchy problem for the field equations. 

The gravitation potentials, in a domain without matter and in absence of electromagnetic filed, must verify ten partial differential equations of second-order of the exterior case $R_{\alpha \beta}=0$, that are not independent because of the Bianchi identities. We will formulate the Cauchy problem relative to this system of equations and with initial data on a hypersurface $S$. 

The study of the values on $S$ of the consecutive partial derivatives of the potentials shows that, if $S$ is nowhere tangent to the characteristic manifold, and if the Cauchy data satisfy four given conditions, the Cauchy problem admits, with respect to the system of equations $R_{\alpha \beta}=0$, in the analytic case, a solution and this solution is unique. 

Thus, if there exist two solutions, they coincide up to a change of coordinates, conserving $S$ point-wise and the values on $S$ of the Cauchy data. 

Hence, by making use of $\textit{isothermal coordinates}$, we will solve the Cauchy problem for the equations $\mathcal{G}_{\alpha \beta}=0$. After that we have seen under which assumptions this is possible, we will define, in the second part, the causal structure of space-time. 

We will give the definition of $\textit{strong causality}$, and, since this is not enough to ensure that space-time is not just about to violate causality, we will define $\textit{stable causality}$. 

Eventually, we will deal with $\textit{global hyperbolicity}$ and its meaning in relation to Cauchy surfaces.

\section{Cauchy Problem for General Relativity}
The ten potentials, which are the metric components, $g_{\alpha \beta}$ of an Einstein universe satisfy, in the domains without matter and in absence of electromagnetic field, ten partial differential equations of second-order of the exterior case
\begin{equation*} \begin{split} R_{\alpha \beta} \equiv &\sum_{\lambda=1}^4 \Biggl\{ \partial_\lambda \Gamma \{ \lambda, [\alpha, \beta] \} -  \partial_\alpha \Gamma \{ \lambda, [\lambda, \beta] \} \Bigg\} +  \sum_{\lambda, \mu=1}^4 \Biggl\{\Gamma \{ \lambda, [\lambda, \beta] \} \Gamma \{ \mu, [\alpha, \beta] \} \\
& - \Gamma \{ \mu, [\lambda, \alpha] \} \Gamma \{ \lambda, [\mu, \beta] \} \Bigg\} =0, \end{split} \end{equation*}
where the $ \partial_\lambda = \frac{\partial}{\partial x^\lambda}$, the $x^\lambda$ are a system of four space-time coordinates whatsoever, and we have denoted by $\Gamma \{ \lambda, [\alpha, \beta] \}$ the $\Gamma^{\lambda}_{\alpha \beta}$  to stress the non-tensorial behaviour of the Christoffel symbols.

This ten equations are not independent because the Ricci Tensor satisfies the four Bianchi identities
\begin{equation*} \sum_{\lambda=1}^{4} \nabla_\lambda G^{\lambda \mu} \equiv 0, \end{equation*}
where $G^{\lambda \mu} \equiv R^{\lambda \mu} - \frac{1}{2} (g^{-1})^{\lambda \mu} R$ is the Einstein Tensor, and $R$ is the scalar of curvature.

The problem of determinism is here formulated for an exterior space-time in the form of the Cauchy problem relative to the system of partial differential equations $R_{\alpha \beta}=0$ and with initial data carried by any hypersurface $S$. 

The study of the values on $S$ of the partial derivatives of $g_{\alpha \beta}$ shows that, if $S$ is nowhere tangent to a characteristic manifold, and if the Cauchy data satisfy four given conditions, the Cauchy problem for $R_{\alpha \beta}=0$ admits in the analytic case a unique solution. 

Thus, if $S$ is defined by the equation $x^4=0$, the four conditions that the initial data must verify are the four equations $G^4_\lambda=0$ which are expressed in terms of the data only. We want to remark that $G^{4}_\lambda$ is obtained from the $(1, 1)$ form of the Einstein tensor, by fixing the controvariant index to the component 4 and letting to vary the covariant component.

It is possible to use the results of Chapter 5, since once a space-time and a hypersurface $S$ are given, there always exists a coordinate change $\tilde{x}^{\lambda}=f(x^\mu)$, with $\tilde{x}=0$ for $x^4=0$, so that every equation $R_{\alpha \beta}=0$ does not contain, in the new coordinates, second derivatives besides those of ${g}_{\alpha \beta}$ and the system of Einstein equations takes then the form of the systems studied in Chapter 5. 

The vacuum Einstein equations, in every coordinates (Levi-Civita \cite{levi2013n}) read as
\begin{equation*} R_{\alpha \beta} \equiv - \mathcal{G}_{\alpha \beta} - L_{\alpha \beta} =0 \end{equation*}
where $\mathcal{G}_{\alpha \beta}$ is
\begin{equation*} \mathcal{G}_{\alpha \beta} \equiv \frac{1}{2} \sum_{\lambda, \mu=1}^4 (g^{-1})^{\lambda \mu} \frac{\partial^2 g_{\alpha \beta}}{\partial x^\lambda \partial x^\mu} + H_{\alpha \beta} \end{equation*}
with $H_{\alpha \beta}$ as a polynomial of the $g_{\lambda \mu}$ and $g^{\lambda \mu}$; and $L_{\alpha \beta}$ is
\begin{equation} \label{eq:6.1} L_{\alpha \beta} \equiv \frac{1}{2} \sum_{\mu=1}^4 \bigg{[} g_{\beta \mu} \partial_{\alpha} F^{\mu} + g_{\alpha \mu} \partial_{\beta} F^{\mu} \bigg{]}. \end{equation}
We see that with a choice of coordinates, more precisely if $x^\lambda$ are four $\textit{isothermal}$ $\textit{coordinates}$, it is possible to assume, without restricting the generality of the hypersurface $S$, that the initial data satisfy, besides the four conditions $G^4_\lambda=0$, the so-called $\textit{conditions of isothermy}$:
\begin{equation} \label{eq:6.2} F^{\mu} \equiv \frac{1}{\sqrt{- g}} \sum_{\lambda=1}^4 \frac{ \partial (\sqrt{- g} (g^{-1})^{\lambda \mu})}{\partial x^\lambda}=0 \;{\rm for}\; x^4=0, \end{equation}
which are first-order partial differential equations satisfied by the potentials. 

Thus, as we desired, every equation $R_{\alpha \beta}=0$ does not contain second derivatives besides those of $g_{\alpha \beta}$. The reason why these coordinates are called $\textit{isothermal}$ is that they satisfy the wave equation associated with the metric. 

A function $u$ solving the Laplace equation in the Euclidean setting can be thought of as corresponding to a static solution to the heat equation, and the surfaces of constant $u$ are thus isothermal; thinking of the wave equation associated with the metric as the analogue of the Laplace equation, surfaces on which an isothermal coordinate is constant are thus $\textit{isothermal}$ with respect to that coordinates.   

We shall solve this Cauchy problem for the equations $\mathcal{G}_{\alpha \beta}=0$, verified by the potentials in isothermal coordinates, and we shall prove afterwards that the potentials obtained define indeed a space-time, related to isothermal coordinates, and verify the equations of gravitation $R_{\alpha \beta}=0$.

\subsection{Solution of the Cauchy Problem for the Equations $\mathcal{G}_{\alpha \beta}=0$}
We shall apply to the system
\begin{equation*} \mathcal{G}_{\alpha \beta} \equiv \sum_{\lambda, \mu=1}^4 (g^{-1})^{\lambda \mu} \frac{\partial^2 g_{\alpha \beta}}{\partial x^\lambda \partial x^\mu} + H_{\alpha \beta}=0 \end{equation*}
the results of Chapter 5, by setting $(g^{-1})^{\lambda \mu}=A^{\lambda \mu}$, $g_{\alpha \beta}=W_s$, $H_{\alpha \beta}= f_s$, whereas on the Cauchy data we should make two assumptions.

$\mathbf{Assumptions}$\\
In a domain $(d)$ of the initial surface $S$, $x^4=0$, defined by 
$$|x^i - \bar{x}^i | \leq d:$$
\begin{description}
\item[(1)] The Cauchy data $\varphi_s$ and $\psi_s$ possess partial derivatives continuous and bounded up to the orders five and four, respectively.
\item[(2)] The quadratic form $\sum_{\lambda, \mu=1}^4 (g^{-1})^{\lambda \mu} X_\lambda X_\mu$ is of normal hyperbolic form, i.e. $(g^{-1})^{44}>0$ and $\sum_{i,j=1}^3 (g^{-1})^{ij}X_i X_j$ is negative-definite. In particular, $g=det(g_{\lambda \mu}) \neq 0$.
\end{description}
We deduce from these assumptions the existence of a number $l$ such that for $|g_{\alpha \beta} - \bar{\varphi}_s| \leq l $ one has $g \neq 0$ and we see that, for some unknown functions $g_{\alpha \beta}=W_s$, the inequalities
\begin{equation} \label{eq:6.3} |W_s - \bar{\varphi}_s| \leq l, \hspace{0.5cm} \bigg{|}\frac{\partial W_s}{\partial x^i} - \frac{\partial \bar{\varphi}_s}{\partial x^i} \bigg{|}\leq l, \hspace{0.5cm} \bigg{|} \frac{\partial W_s}{\partial x^4} - \bar{\psi}_s \bigg{|} \leq l \end{equation}
are satisfied. The coefficients of the equations $\mathcal{G}_{\alpha \beta}=0$ (which are here independent of the variables $x^\alpha$) satisfy, as the Cauchy data, the assumptions of Chapter 5, i.e.:

\begin{description}
\item[(1)] The coefficients $A^{\lambda \mu}= (g^{-1})^{\lambda \mu}$ and $f_s=H_{\alpha \beta}$ are rational fractions with denominator $g$ of the $g_{\lambda \mu}= W_s$, and of the $g_{\lambda \mu}= W_s$ and $\frac{\partial W_s}{\partial x^\alpha}$, respectively and they admit partial derivatives with respect to all their arguments up to the fourth order continuous, bounded and satisfying Lipschitz conditions .
\item[(2)]  The quadratic form $\sum_{\lambda, \mu=1}^4 A^{\lambda \mu} X_\lambda X_\mu$ is of normal hyperbolic type, i.e. $A^{44}>0$ and $\sum_{i,j=1}^3 A^{ij}X_i X_j$ is negative-definite. 
\end{description}
Hence, we can apply to the system $\mathcal{G}_{\alpha \beta}=0$ the conclusions of Chapter 5.

$\mathbf{Conclusion}$ There exists a number $\epsilon (x^i_0) \neq 0$ such that, in the domain
\begin{equation*}    |x^i - \bar{x}^i | < d, \hspace{1cm} |x^4| \leq \epsilon(x^i_0) \end{equation*}
the Cauchy problem relative to the equations $\mathcal{G}_{\alpha \beta}=0$ admits a solution which has partial derivatives continuous and bounded up to the fourth order and which verifies the inequalities $(\ref{eq:6.3}$).

Once the solution has been found, it is left to prove that it verifies the conditions of isothermy. Thus, let us show that
\begin{description}
\item[(1°)] $\textit{The solution found of the system}$ $\mathcal{G}_{\alpha \beta}=0$ $\textit{verifies the four equations}$
\begin{equation*} \partial_4 F^\mu =0 \; {\rm for} \; x^4=0. \end{equation*}
Indeed, we have assumed that the initial data satisfy the conditions 
\begin{equation} \label{eq:6.4} G^4_\lambda=0 \; {\rm and} \; F^\mu=0 \hspace{1cm} {\rm for}\; x^4=0.\end{equation}
Hence, we have
\begin{equation*} \begin{split} G^4_\lambda \equiv & - \sum_{\mu=1}^4 (g^{-1})^{4 \mu} \Biggl\{ \mathcal{G}_{\lambda \mu} - \frac{1}{2} g_{\lambda \mu} \sum_{\alpha, \beta=1}^4 (g^{-1})^{\alpha \beta} \mathcal{G}_{\alpha \beta} + L_{\lambda \mu}\\
& - \frac{1}{2} g_{\lambda \mu}\sum_{\alpha, \beta=1}^4 (g^{-1})^{\alpha \beta} L_{\alpha \beta}\Biggr\}, \end{split}\end{equation*}
where $L_{\alpha \beta}$ is defined by ($\ref{eq:6.1}$).
Thus, the solution of the system $\mathcal{G}_{\alpha \beta}=0$ verifies the equations
\begin{equation*} - \frac{1}{2} \sum_{\alpha, \mu=1}^4(g^{-1})^{4 \mu}g_{\lambda \alpha} \partial_\mu F^\alpha - \frac{1}{2}\partial_\lambda F^4 + \frac{1}{2} \sum_{\alpha=1}^4 \delta^4_\lambda \partial_\alpha F^\alpha =0 \; {\rm for}\; x^4=0, \end{equation*}
from which, by virtue of $F^\mu=0$ and $\partial_\lambda F^\mu=0$, we have
\begin{equation*} - \frac{1}{2} (g^{-1})^{44} \sum_{\alpha=1}^4 g_{\lambda \alpha} \partial_4 F^\alpha=0. \end{equation*}
Eventually, we see that the solution found verifies the four equations $\partial_4 F^\mu=0$, for $x^4=0$.
\item[(2°)] $\textit{The solution found of}$ $\mathcal{G}^{\alpha \beta}=0$ $\textit{verifies}$
\begin{equation*} F^\mu =0. \end{equation*}
This property is going to result from the conservation conditions. Indeed, the metric components $g_{\alpha \beta}$ satisfy the four Bianchi identities
\begin{equation*} \sum_{\lambda=1}^4 \nabla_\lambda \bigg{(} R^{\lambda \mu} - \frac{1}{2} (g^{-1})^{\lambda \mu} R \bigg{)}=0, \end{equation*}
where $ R^{\lambda \mu}$ is the Ricci tensor corresponding to this metric. Thus, a solution of the system $\mathcal{G}_{\alpha \beta}=0$ verifies four equations
\begin{equation*} \sum_{\lambda=1}^4 \nabla_\lambda \bigg{(} L^{\lambda \mu} - \frac{1}{2} (g^{-1})^{\lambda \mu} L \bigg{)}=0, \end{equation*}
where $ L^{\lambda \mu} = \sum_{\alpha, \beta=1}^4 (g^{-1})^{\alpha \lambda}(g^{-1})^{\beta \mu}L_{\alpha \beta}$ and $L = \sum_{\alpha, \beta=1}^4 (g^{-1})^{\alpha \beta}L_{\alpha \beta}$.

It turns out from the expression $(\ref{eq:6.1}$) that these equations read as
\begin{equation*} \begin{split}& \frac{1}{2} \sum_{\alpha, \lambda=1}^4(g^{-1})^{\alpha \lambda} \nabla_\lambda (\partial_\alpha F^\mu) + \frac{1}{2} \sum_{\beta, \lambda=1}^4 (g^{-1})^{\beta \mu} \nabla_\lambda (\partial_\beta F^\lambda) \\
&- \frac{1}{2}\sum_{\alpha, \lambda=1}^4 (g^{-1})^{\lambda \mu} \nabla_\lambda (\partial_\alpha F^\alpha)=0, \end{split} \end{equation*}
from which, by developing and simplifying, we obtain 
\begin{equation*} \frac{1}{2} \sum_{\alpha, \lambda=1}^4 (g^{-1})^{\alpha \lambda} \frac{\partial^2 F^\mu}{\partial x^\alpha \partial x^\lambda} + P^\mu(\partial_\alpha F^\lambda)=0, \end{equation*}
where $P$ is a linear combination of the $\partial_\alpha F^\lambda$ whose coefficients are polynomials of the $(g^{-1})^{\alpha \beta}$, $g_{\alpha \beta}$ and of their first derivatives.

Hence, the four quantities $F^\mu$, formed with the $g_{\alpha \beta}$ solutions of $\mathcal{G}_{\alpha \beta}=0$, verify four partial differential equations of the type previously studied. 

The coefficients $A^{\lambda \mu}=(g^{-1})^{\lambda \mu}$ and $f_s=P^\mu$ verify, in $D$, the assumptions of Chapter 5. The quantities $F^\mu$ are by hypothesis vanishing on the domain $(d)$ of $x^4=0$, and we have proved that the same was true of their first derivatives $\partial_\alpha F^\mu$.
\end{description}
 Then, we deduce from the uniqueness theorem that, in $D$, we have
\begin{equation*} F^\mu=0, \hspace{0.5cm} {\rm and} \hspace{0.5cm} \partial_\alpha F^\mu=0. \end{equation*}
Therefore, the metric components verify effectively in $D$ the conditions of isothermy and represent the potentials of an Einstein space-time, solutions of the vacuum Einstein equations $R_{\alpha \beta}=0$.

\subsection{Uniqueness of the Solution}
In order to prove that there exists only one exterior space-time corresponding to the initial conditions given on $S$, one has to prove that every solution of the Cauchy problem formulated in such a way with respect to the equations $R_{\alpha \beta}=0$ can be deduced by a change of coordinates from the solution of this Cauchy problem relative to the equations $\mathcal{G}_{\alpha \beta}=0$. This last solution is unique.

Thus, let us consider a solution $g_{\alpha \beta}$ of the Cauchy problem relative to the equations $R_{\alpha \beta}=0$ and look for a transformation of coordinates
\begin{equation*} \tilde{x}^\alpha = f^\alpha (x^\beta). \end{equation*}
By conserving $S$ point-wise and in such a way that the potentials in the new system of coordinates $\tilde{g}_{\alpha \beta}$ verify the four equations
\begin{equation*} \tilde{F}^\lambda =0, \end{equation*}
we know that the four quantities $\tilde{F}^\lambda$ are invariant which verify the identities
\begin{equation*} \tilde{F}^\lambda \equiv \tilde{\Delta}_2 \tilde{x}^\lambda = \Delta_2 f^\lambda. \end{equation*}
In order for the equations $\tilde{F}^\lambda =0$ to be verified it is therefore necessary and sufficient that the functions $f^\alpha$ satisfy the equations
\begin{equation} \label{eq:6.5} \Delta_2 f^\alpha \equiv \sum_{\lambda, \mu=1}^4 (g^{-1})^{\lambda \mu} \bigg{(} \frac{\partial^2 f^\alpha}{\partial x^\lambda \partial x^\mu} - \sum_{\rho=1}^4 \Gamma \{ \rho, [\lambda, \mu] \} \frac{\partial f^\alpha}{\partial x^\rho} \bigg{)}=0 \end{equation}
which are partial differential equations of second-order, linear, normal hyperbolic in the domain $D$.

If we take for values of the functions $f^\alpha$ and of their first derivatives upon $S$, the following values
\begin{equation} \begin{split} \label{eq:6.6} & f^4=0, \hspace{1cm} \partial_\alpha f^4= \delta^4_\alpha, \\
& f^i=x^i, \hspace{1cm} \partial_\alpha f^i = \delta^i_\alpha, \end{split} \end{equation}
we see that the Cauchy problems formulated in such a way admit in $D$ solutions possessing their partial derivatives up to the fourth order continuous and bounded.

Thus, we have defined a change of coordinates $\tilde{x}^\lambda=f^\lambda (x^\alpha)$ such that, in the new system of coordinates, the potentials $\tilde{g}_{\alpha \beta}$ verify the conditions of isothermy $\tilde{F}^\lambda=0$. It remains to prove that this change of coordinates determines in a unique way the Cauchy data $\tilde{g}_{\alpha \beta}$ and $\tilde{\partial}_4 \tilde{g}_{\alpha \beta}$ for $x^4=0$, in terms of the original data $g_{\alpha \beta}$ and $\partial_4 g_{\alpha \beta}$ for $x^4=0$.

We know that $g_{\alpha \beta}$ are the components of a covariant rank-two tensor
\begin{equation} \label{eq:6.7} g_{\alpha \beta}= \sum_{\lambda, \mu=1}^4 \tilde{g}_{\lambda \mu} \bigg{(} \partial_\alpha f^\lambda \bigg{)} \bigg{(} \partial_\beta f^\mu \bigg{)}, \end{equation}
from which, by making use of $(\ref{eq:6.6}$), we have
\begin{equation*} g_{\alpha \beta}= \tilde{g}_{\alpha \beta}, \hspace{0.5cm} \partial_i g_{\alpha \beta}= \tilde{\partial}_i \tilde{g}_{\alpha \beta} \hspace{1cm} {\rm for}\; x^4= \tilde{x}^4= 0.\end{equation*}
It remains to evaluate the derivatives of the potentials with respect to $x^4$ and $\tilde{x}^4$ for $x^4= \tilde{x}^4=0$. Since $\varphi$ is an arbitrary function of a space-time point we have
\begin{equation*} \partial_4 \varphi= \sum_{\lambda=1}^4 \bigg{(} \tilde{\partial}_{\lambda} \varphi \bigg{)} \bigg{(} \partial_4 f^\lambda \bigg{)},\end{equation*}
from which
\begin{equation} \label{eq:6.8} \partial_4 \varphi= \tilde{\partial}_4 \varphi. \end{equation}
Furthermore, we find by differentiating the equality $(\ref{eq:6.7}$) with respect to $x^4$
\begin{equation*} \partial_4 g_{\alpha \beta}= \sum_{\lambda, \mu=1}^4 \bigg{[} (\partial_4 \tilde{g}_{\lambda \mu})(\partial_\alpha {f}^\lambda)(\partial_\beta f^\mu) + \tilde{g}_{\lambda \mu} \bigg{(} (\partial^2_{\alpha 4} f^\lambda)(\partial_\beta f^\mu) + (\partial^2_{\beta 4} f^\mu)(\partial_\alpha f^\lambda ) \bigg{)} \bigg{]}, \end{equation*}
from which
\begin{equation} \label{eq:6.9} \partial_4 g_{\alpha \beta}= \partial_4 \tilde{g}_{\alpha \beta} + \sum_{\lambda=1}^4 (\tilde{g}_{\lambda \beta}) \bigg{(} \partial^2_{\alpha 4} f^\lambda \bigg{)} + \sum_{\mu=1}^4 (\tilde{g}_{\mu \alpha}) \bigg{(} \partial^2_{\beta 4} f^\mu \bigg{)}. \end{equation}
We deduce also from the initial values $(\ref{eq:6.6}$):
\begin{equation*} \partial^2_{\alpha i} f^\lambda =0. \end{equation*}
The $f^\lambda$ verify on the other hand the conditions of isothermy $(\ref{eq:6.5}$), from which
\begin{equation*} (g^{-1})^{44}\partial^2_{44}f^\lambda= \sum_{\alpha, \beta=1}^4 (g^{-1})^{\alpha \beta} \Gamma \{ \lambda, [\alpha, \beta] \}. \end{equation*}
Hence, $\partial^2_{44}f^\lambda$ is determined in a unique way by the original Cauchy data; this is also equally true of $\partial_4 \tilde{g}_{\alpha \beta}$ for $x^4=0$.

Thus, we have
\begin{thm} Once a solution $g_{\alpha \beta}$ of the Cauchy problem is given in relation to the equations $R_{\alpha \beta}=0$, with the initial data satisfying upon $S$ the stated assumptions, there exists a change of coordinates, conserving $S$ point-wise, such that the potentials $\tilde{g}_{\alpha \beta}$ in the new system of coordinates verify everywhere the conditions of isothermy and represent the solution, unique, of a Cauchy problem, determined in a unique way, relative to the equations $\mathcal{G}_{\alpha \beta}=0$.
\end{thm}
Therefore, we can conclude that, in gravitational physics:
\begin{thm} There exists one and only one exterior space-time corresponding to the initial conditions assigned upon $S$. \end{thm}
Once we have proved that there exists a unique solution of the Cauchy Problem for Einstein Equations, we will proceed, in the next part of this Chapter, with the study of the causal structure of space-time.

\section{Causal Structure of Space-Time}
Given a space-time, from a physical point of view, it would seem reasonable to suppose that there is a local thermodynamic arrow of time defined continuously at every of its point, but for our purpose we shall only require that it should be possible to define continuously a division of non-spacelike vectors into two classes, which we arbitrarily label future-directed and past-directed. If this is the case, we shall say that space-time is $\textit{time-orientable}$. 

Thus, following Hawking-Ellis \cite{hawking1973large}, by letting $(M, g)$ be a space-time which is time-orientable as explained and given two sets $\mathcal{L}$ and $\mathcal{U}$, we can give the following definitions:

The $\textit{chronological future}$ $I^+(\mathcal{L}, \mathcal{U})$ of $\mathcal{L}$ relative to $\mathcal{U}$ is the set of all points in $\mathcal{U}$ which can be reached from $\mathcal{L}$ by a future-directed timelike curve in $\mathcal{U}$. We shall denote $I^+(\mathcal{L}, M)$ as $I^+(\mathcal{L})$ and it is an open set, since if $p \in M$ can be reached by a future-directed time-like curve from $\mathcal{L}$, then there is a small neighbourhood of $p$ which can be so reached. Hence, if $p \in M$ we can define:
\begin{equation} \label{eq:6.10} I^{+} (p) \equiv \{ q \in M : p <<q \}, \end{equation}
i.e. $I^{+}(p)$ is the set of all points $q$ of $M$ such that there is a future-directed timelike curve from $p$ to $q$. Similarly, one defines the $\textit{chronological past}$ of $p$
\begin{equation} \label{eq:6.11} I^{-}(p) \equiv \{ q \in M: q << p \}. \end{equation}
The $\textit{causal future}$ of $\mathcal{L}$ relative to $\mathcal{U}$ is denoted by $J^+(\mathcal{L}, \mathcal{U})$ and it is defined as the union of $\mathcal{L} \cap \mathcal{U}$ with the set of all points in $\mathcal{U}$ which can be reached from $\mathcal{L}$ by a future-directed non-spacelike curve in $\mathcal{U}$. We denote $J^+(\mathcal{L}, M)$ as $J^+(\mathcal{L})$ and it is the region of space-time which can be causally affected by events in $\mathcal{L}$. It is not necessarily a closed set even when $\mathcal{L}$ is a single point. Therefore, if $p \in M$ we can define
\begin{equation} \label{eq:6.12} J^{+}(p) \equiv \{ q \in M: p \leq q \}, \end{equation}
and similarly for the $\textit{causal past}$
\begin{equation} \label{eq:6.13} J^{-}(p) \equiv \{ q \in M: q \leq p \}, \end{equation}
where $a \leq b$ means that there exists a future-directed non-spacelike curve from $a$ to $b$. 
\begin{figure}
\centering
\includegraphics{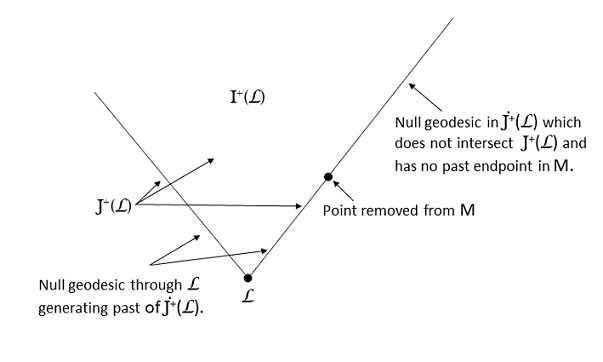} 
\caption{When a point has been removed from Minkowski space, the causal future $J^+(\mathcal{L})$ of a closed set $\mathcal{L}$ is not necessarily closed. Further parts of the boundary of the future of $\mathcal{L}$ may be generated by null geodesic segments which have no past endpoints in $M$.}\label{fig:6}
\end{figure}
A non-spacelike curve between two points which was not a null geodesic curve could be deformed into a timelike curve between two points. Thus, if $\mathcal{U}$ is an open set and $p$, $q$ and $r \in \mathcal{U}$, then we have
 \begin{equation*}
\left\{ \begin{array}{l}
q \in J^+(p, \mathcal{U}), r \in I^+(q, \mathcal{U}) \\
q \in I^+(p, \mathcal{U}), r \in J^+(q, \mathcal{U})
\end{array}\right\}
\end{equation*}
both imply $r \in I^+(p, \mathcal{U})$. From this follows that $\overline{I^+}(p, \mathcal{U})= \overline{J^+}(p, \mathcal{U})$ and $\dot{I^+}(p, \mathcal{U})= \dot{J^+}(p, \mathcal{U})$, where for any set $I$, $\bar{I}$ is the $\textit{closure}$ of $I$ and $\dot{I} \equiv \overline{I} \cap \overline{(M - I)}$ denotes the $\textit{boundary}$ of $I$. This example, illustrates a useful technique for constructing space-times with given causal properties: one starts with some simple space-time, such as Minkowski space, cuts out any closed set and, if desired, pastes it together in an appropriate way. The result is still a manifold with a Lorentz metric and therefore still a space-time even though it may look incomplete where the points have been cut out. This incompleteness can be resolved by a conformal transformation which sends the cut out points to infinity. For our purpose, we give a few more definitions.

\begin{defn}
The $\textit{future horismos}$ of $\mathcal{L}$ relative to $\mathcal{U}$, denoted by $E^+(\mathcal{L}, \mathcal{U})$, is defined has
\begin{equation} E^+(\mathcal{L}, \mathcal{U}) \equiv J^+(\mathcal{L}, \mathcal{U}) - I^+(\mathcal{L}, \mathcal{U}); \end{equation}
we write $E^+(\mathcal{L})$ for $E^+(\mathcal{L}, M)$.
\end{defn} 
If $\mathcal{U}$ is an open set, points of $E^+(\mathcal{L}, \mathcal{U})$ must lie on future-directed null geodesics from $\mathcal{L}$. Similary, we can define the $\textit{past horismos}$ $E^-(\mathcal{L}, \mathcal{U})$. 

\begin{defn} A point $p$ is a $\textit{future endpoint}$ of a future-directed non-spacelike curve $\lambda: F \rightarrow M$, if for every neighbourhood $V$ of $p$ there is a $t \in F$ such that $\lambda(t_1) \in V$ for every $t_1 \in F$ with $t_1 \geq t$. \end{defn}

\begin{defn} A non-spacelike curve is future-inextendible in a set $\mathcal{L}$ if it has no future endpoint in $\mathcal{L}$. \end{defn}

At this stage, to derive the properties of the boundaries we introduce the concepts of achronal and future sets.

\begin{defn}
A set $\mathcal{L}$ is said to be $\textit{achronal}$ if $I^+(\mathcal{L}) \cap \mathcal{L}$ is empty, in other words if no two points of $\mathcal{L}$ can be joined by a timelike curve.
\end{defn} 
\begin{defn}
A set $\mathcal{L}$ is said to be a $\textit{future set}$ if $I^+(\mathcal{L}) \subset \mathcal{L}$. Hence, $M - \mathcal{L}$ is a $\textit{past set}$.
\end{defn} 
Examples of future sets include $I^+(\mathcal{N})$ and $J^+(\mathcal{N})$ where $\mathcal{N}$ is any set.
The causal structure of $(M, g)$ is the collection of past and future sets at all points of $M$ together with their properties as shown in figure ($\ref{fig:6}$).

\begin{prop} \label{prop:6.2} If $\mathcal{L}$ is a future set then $\dot{\mathcal{L}}$ is a closed, imbedded, achronal three-dimensional $C^1$ submanifold.
\end{prop}
We shall call a set with the properties of $\dot{\mathcal{L}}$ an $\textit{achronal boundary}$. Such a set can be divided into four disjoint subset $\dot{\mathcal{L}}_N$, $\dot{\mathcal{L}}_+$, $\dot{\mathcal{L}}_-$ and $\dot{\mathcal{L}}_0$. For a point $q \in \dot{\mathcal{L}}$ there may or may not exist points $p, r \in \dot{\mathcal{L}}$ with $p \in E^-(q)-q$, $r \in E^+(q)-q$. The different possibilities define the subset of $\dot{\mathcal{L}}$ according to the scheme in table ($\ref{table:1}$).
\begin{table}
\caption{Scheme of $\dot{\mathcal{L}}$ subsets; see the discussion below.}
\label{table:1}
\center
$q \in$ \begin{tabular}{|c|c|c|}
\hline
$ \exists p$ & $\not \exists p$ & $\null$ \\
\hline
$\dot{\mathcal{L}}_N$ & $\dot{\mathcal{L}}_+$ & $\exists r$  \\
\hline
$\dot{\mathcal{L}}_-$ & $\dot{\mathcal{L}}_0$ & $\not \exists r$  \\
\hline
\end{tabular}

\end{table}
If $q \in \dot{\mathcal{L}}_N$, then $r \in E^+(p)$ since $r \in J^+(p)$ and $r \not \in I^+(p)$. This means that there is a null geodesic segment in $\dot{\mathcal{L}}$ through $q$. If $q \in \dot{\mathcal{L}}_+$ (respectively $\dot{\mathcal{L}}_-$) then $q$ is the future (past) endpoint of a null geodesic in $\dot{\mathcal{L}}$. The subset $\dot{\mathcal{L}}_0$ is spacelike. A useful condition for a point to lie in $\dot{\mathcal{L}}_N$, $\dot{\mathcal{L}}_+$ or $\dot{\mathcal{L}}_-$ is given by the following lemma due to Penrose:

\begin{lem} Let $W$ be a neighbourhood of $q \in \dot{\mathcal{L}}$ where $\mathcal{L}$ is a future set. Then
\begin{description}
\item[(i)] $I^+(q) \subset I^+ (\mathcal{L} - W)$ implies $q \in \dot{\mathcal{L}}_N \cup \dot{\mathcal{L}}_+$, 
\item[(ii)] $I^-(q) \subset I^+ (M - \mathcal{L} - W)$ implies $q \in \dot{\mathcal{L}}_N \cup \dot{\mathcal{L}}_-$.
\end{description}
\end{lem}
An example is given by $\dot{J}^+(K)= \dot{I}^+(K)$, that it the boundary of the future of a closed set $K$. It is an achronal manifold and by the above lemma, every point of $\dot{J}(K) - K$ belongs to $[\dot{J}^+(K)]_N$ or $[\dot{J}^+(K)]_+$. This means that $\dot{J}(K) - K$ is generated by null geodesic segments which may have future endpoints in $\dot{J}^+(K) - K$ but which, if they do have past endpoints, can have them only on $K$ itself. 

We shall say that an open set $\mathcal{U}$ is $\textit{causally simple}$ if for every compact set $K \subset U$,
\begin{equation*} \dot{J}^+(K)  \cap \mathcal{U} = E^+(K) \cap \mathcal{U} \hspace{0.5cm} {\rm and}  \hspace{0.5cm} \dot{J}^-(K)  \cap \mathcal{U} = E^-(K) \cap \mathcal{U}. \end{equation*}
This is equivalent to say that $\dot{J}^+(K)$ and $\dot{J}^-(K)$ are closed in $U$.

\subsection{Causality conditions}
Since the causality holds only locally the global question is left open. Thus we did not rule out the possibility that on large scale there might be closed timelike curves. However the existence of such curves would seem to lead the possibilities of logical paradoxes. Thus, we are more ready to believe that space-time satisfies the $\textit{chronology condition}$, i.e. there are no closed timelike curves. However, we must bear in mind the possibility that there might be points of space-time at which this condition does not hold. The set of all such points will be called the $\textit{chronology violating}$ set of $M$ and it is defined as follows:

\begin{prop} The chronology violating set of $M$ is the disjoint union of sets of the form $I^+(q) \cap I^-(q)$, with $q \in M$. \end{prop}

\begin{prop} If $M$ is compact, the chronology set of $M$ is non-empty. \end{prop}
From this result it would seem reasonable to assume that space-time is non-compact. Another argument against compactness is that any compact, four-dimensional manifold on which there is a Lorentz metric cannot be simply connected. Thus, a compact space-time is really a non-compact manifold in which points have been identified. It would seem physically reasonable to regard the covering manifold as representing space-time.

We shall say that the $\textit{causality condition}$ holds if there are no closed non-spacelike curves.

\begin{prop} The set of points at which the causality condition does not hold is the disjoint union of sets of the form $J^-(q) \cap J^+(q)$, with $q \in M$. \end{prop}
In particular, if the causality condition is violated at $q \in M$ but the chronology condition holds, there must be a closed null geodesic curve $\gamma$ through $q$. For physically realistic solutions, the causality and chronology conditions are equivalent. It would seem reasonable to exclude situations in which there were non-spacelike curves which returned arbitrarily close to their point of origin or which passed arbitrarily close to other non-spacelike curves which then passed arbitrarily close to the origin of the first curve and so on. We shall describe the first three of these conditions.

\begin{defn}The $\textit{future distinguishing condition}$ is said to hold at $p \in M$ if every neighbourhood of $p$ contains a neighbourhood of $p$ which no future directed non-spacelike curves from $p$ intersects more than once. An equivalent statement is that $I^+(q)=I^+(p)$ implies that $q=p$. \end{defn}
Similarly, it is possible to define the $\textit{past distinguishing condition}$ by exchanging the future with the past in the previous definition.

\begin{defn}The $\textit{strong causality condition}$ is said to hold at $p$ if every neighbourhood of $p$ contains a neighbourhood of $p$ which no non-spacelike curve intersects more than once. \end{defn}
Another definition of strong causality can be given, by following Penrose, if we exclude the null curves. It is defined as follows:

\begin{defn} $\textit{Strong causality}$ holds at $p \in M$ if arbitrarily small neighbourhoods of $p$ exist which each intersect no timelike curve in a disconnected set. \end{defn}

\begin{cor} The past and the future distinguishing conditions would also hold on $M$ since they are implied by strong causality. \end{cor}
Closely related to these three higher degree causality conditions is the phenomenon of $\textit{imprisonment}$.

A non spacelike curve $\lambda$ that is future-inextendible can do one of the three things as one follows it to the future. It can
\begin{description} 
\item[(i)] enter and remain within a compact set $\mathcal{L}$;
\item[(ii)] not remain in any compact set $\mathcal{L}$ and not re-enter a compact set $\mathcal{L}$;
\item[(iii)] not remain within any compact set $\mathcal{L}$ and not re-enter any such set more than a finite number of times.
\end{description}
In the third case, $\lambda$ can be thought as going off to the edge of space-time, that is either to infinity or a singularity. In the first and second cases we shall say that $\lambda$ is $\textit{totally}$ and $\textit{partially future imprisoned}$ in $\mathcal{L}$, respectively. Furthermore, we have the following result:

\begin{prop} \label{prop:6.6} If the strong causality condition holds on a compact set $\mathcal{L}$, there can be no future-inextendible non-spacelike curve totally or partially future imprisoned in $\mathcal{L}$. \end{prop}
and 
\begin{prop} If the future or past distinguishing condition holds on a compact set $\mathcal{L}$, there can be no future-inextendible non-spacelike curve totally future imprisoned in $\mathcal{L}$. \end{prop}
The causal relations on $(M, g)$ may be used to put a topology on $M$ called the $\textit{Alexandrov topology}$.

\begin{defn} The Alexandrov topology, is a topology in which a set is defined to be open if and only if it is the union of one or more sets of the form $I^+(q) \cap I^-(q)$, with $p, q \in M$. \end{defn}
As $I^+(q) \cap I^-(q)$ is open in the manifold topology, any set which is open in the Alexandrov topology will be open in the manifold topology, though the converse is not necessarily true.

\begin{thm}The following three requirements on a space-time $(M,g)$ are equivalent:
\begin{description}
\item[(1)] $(M,g)$ is strongly causal;
\item[(2)] the Alexandrov topology agrees with the manifold topology;
\item[(3)] the Alexandrov topology is Hausdorff.
\end{description} \end{thm}
However, suppose that the strong causality condition holds on $M$. Then, about any point $p \in M$ one can find a local causality neighbourhood $\mathcal{U}$. The Alexandrov topology of $(\mathcal{U}, {g|}_{\mathcal{U}})$ regarded as a space-time in its own right, is the same as the manifold topology of $\mathcal{U}$. Thus the Alexandrov topology of $M$ is the same as the manifold topology since $M$ can be covered by local causality neighbourhoods. This means that if the strong causality holds, one can determine the topological structure of space-time by observation of causal relationships. 

Even imposition of strong causality condition does not rule out all causal pathologies and to ensure that space-time is not just about to violate chronology condition. Thus, in order to be physically significant, a property of space-time ought to have some form of stability. The situation can be considerably improved if $\textit{stable causality condition}$ holds. To be able to define properly this concept, one has to define a topology on the set of all space-times, that is, all non-compact four-dimensional manifolds and all Lorentz metric on them. Essentially, three topologies seem of major interest: compact-open topology, open topology and fine topology.
\begin{description}
\item[(1)] $\textbf{Compact-Open Topology}$

$\forall n=0, 1, ..., r$, let $\epsilon_n$ be a set of continuous positive functions on $M$, $\mathcal{U} \subset M$ a compact set and $g$ the Lorentz metric under study. We then define: $G(\mathcal{U}, \epsilon_n, g)$ the set of all Lorentz metrics $\tilde{g}$ such that

\begin{equation} |g - \tilde{g}|_n < \epsilon_n \; {\rm on} \; \mathcal{U} \; \forall n, \end{equation}
where
\begin{eqnarray*} 
\; & \; & | g - \tilde{g}|_n \nonumber \\
& \equiv & \sqrt{\sum\limits_{a_i, b_j, r, s, u, v} \bigg{[} \nabla_{a_1} ...\nabla_{a_n} (g_{rs} - \tilde{g}_{rs}) \bigg{]}\bigg{[} \nabla_{b_1} ...\nabla_{b_n} (g_{uv} - \tilde{g}_{uv}) \bigg{]} h^{a_1 b_1}...h^{sv}}, \end{eqnarray*}
where $\nabla_a$ is the covariant derivative operator on $M$ and $\sum_{a,b=1}^4 h_{ab} dx^a$ $\otimes dx^b$ is any positive-definite metric on $M$.

In the compact-open topology, open sets are obtained from the $G(\mathcal{U}, \epsilon_i, g)$ through the operations of arbitrary union and finite intersection.

\item[(2)] $\textbf{Open Topology}$

We no longer require $\mathcal{U}$ to be compact, and we take $\mathcal{U}=M$ in section (1).

\item[(3)] $\textbf{Fine Topology}$

We define $H(\mathcal{U}, \epsilon_i, g)$ as the set of all Lorentz metrics $\tilde{g}$ such that
\begin{equation} |g - \tilde{g}|_i < \epsilon_i, \end{equation}
and $\tilde{g}=g$ out of the compact set $\mathcal{U}$. Moreover, we set $G'(\epsilon_i, g)$ \\$= \cup H(\mathcal{U}, \epsilon_i, g)$. A sub-basis for the fine topology is then given by the neighbourhoods $G'(\epsilon_i, g)$.
\end{description}
Now, the underlying idea for stable causality is that space-time must not contain closed timelike curves, and we still fail to find closed timelike curves if we open out the null cones. Thus, for our purpose, we are interested in the $C^0$ open topology.

\begin{defn} The stable causality condition holds on $M$ if the space-time metric $g$ has an open neighbourhood in the $C^0$ open topology such that there are no closed timelike curves in any metric belonging to the neighbourhood. \end{defn} 
In other words, what this condition means is that one can expand the light cones slightly at every point without introducing closed timelike curves.
The Minkowski, Friedmann-Robertson-Walker, Schwarzschild and Reissner-Nordström space-times are all stably causal. If stable causality condition holds, the differentiable and conformal structure can be determined from the causal structure, and space-time cannot be compact (because in a compact space-time there exist closed timelike curves). A very important characterization of stable causality is given by the following proposition:

\begin{prop} The stable causality condition holds everywhere on $M$ if and only if there is a function $f$ on $M$ whose gradient is everywhere timelike. \end{prop}
The function $f$ can be thought as a sort of $\textit{cosmic time}$ in the sense that it increases along every future-directed non-spacelike curve.

Now if the stable causality condition holds one can find a family of $C^r$ Lorentz metrics $h(a)$, with $a \in [0, 3]$, such that (Hawking-Ellis 1973):
\begin{description}
\item[(1)] $h(0)$ is the space-time metric $g$;
\item[(2)] there are no closed timelike curves in the metric $h(a)$ for each $a \in [0, 3]$;
\item[(3)] if $a_1$, $a_2 \in [0, 3]$ with $a_1 < a_2$, then every non-spacelike vector in the metric $h(a_1)$ is timelike in the metric $h(a_2)$.
\end{description}

\subsection{Cauchy developments}
In Newtonian theory there is instantaneous action-at-a-distance and hence in order to predict events at future points in space-time one has to know the state of the entire universe at the present time and also to assume the boundary conditions at infinity, such as that the potential goes to zero. 

On the other hand, in relativity theory, events at different points of space-time can be causally related only if they can be joined by a non-spacelike curve. 

Thus a knowledge of the appropriate data on a closed set $\mathcal{L}$ would determine events in a region $D^+(\mathcal{L})$ to the future of $\mathcal{L}$ called the $\textit{future Cauchy}$ $\textit{development}$ or $\textit{domain of dependence}$ of $\mathcal{L}$, and it is defined as the set of all points $p\in M$ such that every past-inextendible non-spacelike curve through $p$ intersects $\mathcal{L}$. Similarly, the $\textit{past Cauchy development}$, $D^-(\mathcal{L})$, is defined by exchanging the past with the future in the previous definition. The $\textit{total Cauchy development}$ is given by $D(\mathcal{L})=D^+(\mathcal{L}) \cup D^-(\mathcal{L})$.

Penrose defines the Cauchy development of $\mathcal{L}$ slightly differently, as the set of all points $p \in M$ such that every past-inextendible timelike curve through $p$ intersect $\mathcal{L}$. We shall denote this set $\tilde{D}^+(\mathcal{L})$. Thus, one has $\tilde{D}^+(\mathcal{L}) = \overline{D^+}(\mathcal{L})$.

The future boundary of ${D}^+(\mathcal{L})$, that is $\overline{D^+}(\mathcal{L}) - I^-({D}^+(\mathcal{L}))$, marks the limit of the region that can be predicted from knowledge of data on $\mathcal{L}$. We call this closed achronal set the $\textit{future Cauchy horizon}$ of $\mathcal{L}$ and denote it by $H^+(\mathcal{L})$.

\begin{defn}The future Cauchy horizon $H^+(\mathcal{L})$ of $\mathcal{L}$ is given by 
\begin{equation} H^+(\mathcal{L}) \equiv \{ X: X \in D^+(\mathcal{L}), I^+(X) \cap D^+(\mathcal{L})= \phi \}. \end{equation} 
Similarly, the past Cauchy horizon $H^-(\mathcal{L})$ is defined as
\begin{equation} H^-(\mathcal{L}) \equiv \{ X: X \in D^-(\mathcal{L}), I^-(X) \cap D^-(\mathcal{L})= \phi \}. \end{equation} 
 \end{defn}
The future Cauchy horizon of $\mathcal{L}$ will intersect $\mathcal{L}$ if $\mathcal{L}$ is null or if $\mathcal{L}$ has an edge. To make this precise we define the $\textit{edge}(\mathcal{L})$ as follows.

\begin{defn} The edge($\mathcal{L}$) for an achronal set $\mathcal{L}$ is the set of all points $q \in \bar{\mathcal{L}}$ such that in every neighbourhood $\mathcal{U}$ of $q$ there are points $p \in I^-(q, \mathcal{U})$ and $r \in I^+(q, \mathcal{U})$ which can be joined by a timelike curve in $\mathcal{U}$ which does not intersect $\mathcal{L}$. \end{defn}

It follows that if the edge($\mathcal{L}$) is empty for a non-empty achronal set $\mathcal{L}$, then $\mathcal{L}$ is a three-dimensional imbedded $C^1$-submanifold.

\begin{prop} For a closed achronal set $\mathcal{L}$,
\begin{equation*} {\rm edge}(H^+(\mathcal{L}))={\rm edge}(\mathcal{L}). \end{equation*}
\end{prop}

\subsection{Global Hyperbolicity}
Closely related to Cauchy developments is the property of global hyperbolicity. The notion of $\textit{global hyperbolicity}$ was introduced by Leray in order to deal with questions of existence and uniqueness of solutions of hyperbolic differential equations on a manifold. It plays a key role in developing a rigorous theory of geodesics in Lorentzian geometry, in proving singularity theorems and its ultimate meaning can be seen as requiring the existence of Cauchy surfaces, i.e. spacelike hypersurfaces which each non-spacelike curve intersects exactly once. We shall here follow Geroch \cite{geroch1970domain} and Hawking-Ellis \cite{hawking1973large}, defining and proving in part what follows.

\begin{defn} A space-time $(M, g)$ is said to be globally hyperbolic if
\begin{description}
\item[(1)]the strong causality assumption holds on $(M, g)$;
\item[(2)]if for any two points $p$, $q \in M$, $J^+(p) \cap J^-(q)$ is compact and contained in $M$.
\end{description} \end{defn} 
Condition $\textbf{(2)}$ can be thought of as saying that $J^+(p) \cap J^-(q)$ does not contain any points on the edge of space-time, i.e. at infinity or at a singularity. The reason for the nomenclature $\textit{global hyperbolicity}$ is that on $M$, the wave equation for a $\delta$-function source at $p \in M$ has a unique solution which vanishes outside $M- J^+(p, M)$.

Recall that $M$ is said to be causally simple if for every compact set $\mathcal{K}$ contained in $M$, $J^+(\mathcal{K}) \cap M$ and $J^-(\mathcal{K}) \cap M$ are closed in $M$.

\begin{prop} An open globally hyperbolic set $M$ is causally simple. \end{prop}
Leray did not give the above definition of global hyperbolicity but an equivalent one that is the following 

\begin{defn} Given two points $p$, $q \in M$ such that strong causality holds on $J^+(p) \cap J^-(q)$, we define $C(p, q)$ to be the space of all non-spacelike curves from $p$ to $q$, regarding two curves $\gamma(t)$ and $\lambda(u)$ as representing the same point of $C(p,q)$ if one is a reparametrization of the other, i.e. if there exists a continuous monotonic function $f(u)$ such that $\gamma(f(u))=\lambda(u)$. \end{defn}
The topology of $C(p, q)$ is defined by saying that a neighbourhood of $\gamma$ in $C(p, q)$ consists of all curves in $C(p, q)$ whose points in $M$ lie in a neighbourhood $W$ of the points of $\gamma$ in $M$. Leray's definition is that $M$ is globally hyperbolic if $C(p, q)$ is compact for all $p$, $q \in M$. These definitions are equivalent, as shown by the condition $\textbf{(2)}$ of the following theorem.

\begin{thm}In a globally hyperbolic space-time $(M, g)$, the following properties hold:
\begin{description}
\item[(1)] $J^+(p)$ and $J^-(p)$ are closed $\forall p \in M$;
\item[(2)] strong causality holds on $M$ such that
\begin{equation*} M = J^-(M) \cap J^+(M), \end{equation*}
and, $\forall p, q \in M$, the space $C(p, q)$ of all non-spacelike curves from $p$ to $q$ is compact in a suitable topology;
\item[(3)] there exist Cauchy surfaces.
\end{description} \end{thm}
\begin{proof}[$\textbf{Proof (1)}$] If $(X, F)$ is Hausdorff space and $A \subset X$ is compact, then $A$ is closed. In our case, this implies that $J^+(p) \cap J^-(q)$ is closed. Moreover, it is not difficult to see that $J^+(p)$ itself must be closed. In fact, otherwise we could find a point $r \in \overline{J^+(p)}$ such that $r \not \in J^+(p)$. 

Let us now choose $q \in I^+(r)$. We would then have $r \in \overline{J^+(p) \cap J^-(q)}$ but $r \not \in J^+(p) \cap J^-(q)$, which implies that $J^+(p) \cap J^-(q)$ is not closed, contradicting what we found before. Similarly we also prove that $J^-(p)$ is closed. \end{proof}

\begin{proof}[$\textbf{Proof (2)}$]
Suppose that $C(p, q)$ is compact. Let $r_n$ be an infinite sequence of points in $J^+(p) \cap J^-(q)$ and let $\lambda_n$ be a sequence of non-spacelike curves from $p$ to $q$ through the corresponding $r_n$. As $C(p, q)$ is compact, there will be a curve $\lambda$ to which some sequence $\lambda'_n$ converges in the topology on $C(p, q)$. 

Let $\mathcal{U}$ be a neighbourhood of $\lambda$ in $\mathcal{M}$ such that $\bar{\mathcal{U}}$ is compact. Then $\mathcal{U}$ will contain all $\lambda'_n$ and hence all $r'_n$ for $n$ sufficiently large, and so there will be a point $r \in \mathcal{U}$ which is a limit point of the $r'_n$. Clearly $r$ lies on $\lambda$. Thus, every infinite sequence in $J^+(p) \cap J^-(q)$ has a limit point in $J^+(p) \cap J^-(q)$. Therefore, $J^+(p) \cap J^-(q)$ is compact.

Conversely, suppose $J^+(p) \cap J^-(q)$ is compact. Let $\lambda_n$ be an infinite sequence of non-spacelike curves from $p$ to $q$. A lemma exists (see Hawking-Ellis 1973) which assures that given an open set, in our case $\mathcal{M} - q$, there will be a future-directed non-spacelike curve $\lambda$ from $p$ to $q$ which is inextendible in $\mathcal{M}- q$, and it is such that there is a subsequence $\lambda'_n$ which converges to $r$ for every $r \in \lambda$. The curve $\lambda$ must have a future endpoint at $q$ since by proposition it cannot be totally future imprisoned in the compact set $J^+(p) \cap J^-(q)$, and it cannot leave the set except at $q$. 

Let $\mathcal{U}$ be any neighbourhood of $\lambda$ in $\mathcal{M}$ and let $r_i$, with $1 \leq i \leq k$, be a finite set of points on $\lambda$ such that $r_1=p$, $r_k=q$ and each $r_i$ has a neighbourhood $V_i$ with $J^+(V_i) \cap J^-(V_{i+1})$ contained in $\mathcal{U}$. Then, for sufficiently large $n$, $\lambda'_n$ will be contained in $\mathcal{U}$. Thus, $\lambda'_n$ converges to $\lambda$ in the topology on $C(p, q)$ and so $C(p, q)$ is compact. \end{proof}

\begin{proof}[$\textbf{Proof (3)}$]
We put a measure $\mu$ on $M$ such that the total volume of $M$ in this measure is equal to 1. For $p \in M$, we define $f^+ : p \in M \rightarrow V$, to be the volume $V$ of $J^+(p, M)$ in the measure $\mu$. Clearly, $f^+(p)$ is a bounded function on $M$ which decreases along every future-directed non-spacelike curve. We shall show that global hyperbolicity implies that $f^+(p)$ is continuous on $M$. To do this, it will be sufficient to show that $f^+(p)$ is continuous on any non-spacelike curve $\lambda$.

Let $r \in \lambda$ and let $x_n$ be an infinite sequence of points on $\lambda$ strictly to the past of $r$. Let $T\equiv \cap_n J^+(x_n, M)$. Suppose that $f^+(p)$ was not upper semi-continuous on $\lambda$ at $r$. There would be a point $q \in T - J^+(r, M)$. Then $r \not \in J^-(q, M)$; but each $x_n \in J^-(q, M)$ and so $r \in \overline{J^-}(q, M)$, which is impossible as $J^-(q, M)$ is closed in $M$. The proof that it is lower semi-continuous is similar.

As $p$ is moved to the future along an inextendible non-spacelike curve $\lambda$ in $M$, the value of $f^+(p)$ must tend to zero. For suppose there were some point $q$ which lies to the future of every point of $\lambda$. Then the future-directed curve $\lambda$ would enter and remain within the compact set $J^+(r) \cap J^-(q)$ for every $r \in \lambda$ which would be impossible, by proposition ($\ref{prop:6.6}$), as the strong causality condition holds on $M$.  It becomes then trivial to prove the continuity of the function $f^+ : p \in M \rightarrow V^1$, where $V^1$ is the volume of $I^+(p, M)$. From now on, we shall mean by $f^+$ the volume function of $I^+(p, M)$. 

Now we consider a function $f(p)$ defined on $M$ by
\begin{equation*} f: p \in M \rightarrow f(p) \equiv \frac{f^-(p)}{f^+(p)}. \end{equation*}
Any surface of constant $f$ will be an acausal set and, by proposition ($\ref{prop:6.2}$), will be a three-dimensional $C^1$-manifold imbedded in $M$. The function $f(p)$ is also continuous and strictly decreasing along each past-directed timelike curve. Let $\mathcal{L}$ be the set of points at which $f=1$, since $f$ is strictly decreasing along timelike curves, $\mathcal{L}$ is achronal.

To show that $\mathcal{L}$ is also a $\textit{Cauchy surface}$, we shall prove the following:

\begin{prop} Let $\mathcal{L}$ be the set of points where $f=1$, and let $p \in M$ be such that $f(p) >1$ and $\gamma$ be any past-directed timelike curve, without a past endpoint, from $p$. Since $f$ is continuous, $\gamma$ must intersect $\mathcal{L}$, provided that  $p \in D^+(\mathcal{L}$). Similarily, if $f(p) <1$, $p \in D^-(p)$. \end{prop}
Eventually, the previous proposition, implies that $\mathcal{L}$ is indeed a Cauchy surface. Hence, we consider any past-directed timelike curve $\gamma$ without past endpoint from $p$. In view of the continuity of $f$, such a curve $\gamma$ must intersect $\mathcal{L}$, provided one can show that there exists $\epsilon \rightarrow 0^+: f|_{\gamma}= \epsilon$, where $\epsilon$ is arbitrary. Furthermore, given $q \in M$, e we denote a set $U \subset M$ such that $U \subset I^+(q)$. A subset $U$ of this form covers $M$. Moreover, any $U$ cannot be in $I^-(r)$, $ \forall r \in \gamma$. This is forbidden by global hyperbolicity. 

Suppose, on the contrary, that $q \in \cap_{ r \in \gamma} I^-(r)$. Then we choose a sequence of points $\{t_i\}$ on $\gamma$ such that $t_{i+1} \in I^-(t_i)$ and such that every point of $\gamma$ lies to the past of at least one $t_i$. For each $i$, draw a timelike curve $\gamma'$ which:
\begin{description}
\item[(a)]begins at $p$,
\item[(b)]  $\gamma'=\gamma$ to $t_i$,
\item[(c)] $\gamma'$ continues to $q$.
\end{description}
Since $M$ is globally hyperbolic, this sequence has a limit curve, $\Gamma$. The limit curve evidently contains $\gamma$. But this is impossible, for $\gamma$, if it were contained in a compact causal curve from $p$ to $q$, would then have a past endpoint. Hence, there must be some point $r$ of $\gamma$ such that $U \not \subset I^-(r)$. Since $M$ may be covered by such $U$'s, we conclude that $f^-(r)$ approaches zero as $r$ continues into the past on $\gamma$, and, therefore, that $\gamma$ intersect $\mathcal{L}$. We have shown that every past-directed timelike curve from $p$ intersect $\mathcal{L}$, i.e., that $p \in D^+(\mathcal{L})$. Similarly, if $f(p) < 1$, then $p \in D^-(\mathcal{L})$. Hence, $\mathcal{L}$ is a $\textit{Cauchy surface}$.  
\end{proof}

Global hyperbolicity is a stable property of space-times, i.e., arbitrary, sufficiently small variations in the metric will not destroy global hyperbolicity. The proof can be found in Geroch \cite{geroch1970domain}. An useful example of globally hyperbolic manifolds is given by the $\textit{Hyperbolic Riemannian manifolds}$ \cite{lichnerowicz2018republication}.

$\textbf{Example.}$ Let us consider an oriented differentiable manifold $V_n$ of dimension $n$ and class $C^{\infty}$, endowed with a volume element $\eta$ and let us introduce orthonormal frames, the elements of a principal fibre bundle $\mathcal{E}(V_n)$ over $V_n$, with structure group the Lorentz group $L(n)$. With respect to these frames $(e_0, e_A)$, where $A=1, .., (n-1)$ and $\alpha=0, 1, ..., (n-1)$, the metric can be written locally on an open neighbourhood as:
\begin{equation*} ds^2= g_{\alpha \beta} \theta^\alpha \theta^\beta, \end{equation*}
where the $\theta^\alpha$ are $1-forms$ and $g_{\alpha \beta}= \eta_{\alpha \beta}$ with $\eta_{\alpha \beta}=0$ for $\alpha \neq \beta$, $\eta_{00}=1$, $\eta_{AA}=-1$. We assume that the metric $ds^2$ of $V_\eta$ is normal hyperbolic. This metric defines in the tangent space at each point $x$ a convex cone of second order $C_x$.
 
If $A=(A^{\lambda'}_{\alpha})$ is a matrix in $L(n)$, the time signature $\rho_{\lambda}$ of the matrix $A$ is equal to $\pm 1$ depending on the sign of $A^{0'}_0$.  A $\textit{time orientation}$ $\rho$ is defined on $V_n$, with respect to the frames $y \in \mathcal{E}(V_n)$, by a indicator $\rho_y= \pm 1$ such that, if $y=y'A$, one has
\begin{equation*} \rho_y=\rho_{y'} \rho_A. \end{equation*}
We all assume that $V_n$ admits a time orientation $\rho$.

A vector $e_0$, with $e^2_0=1$, is $\textit{future-oriented}$ if the component of $\rho$ with respect to the orthonormal frames $(e_0, e_A)$ is equal to 1. Similarly, a vector $e_0$ is $\textit{past-oriented}$ if the component of $\rho$ with respect to the orthonormal frames $(e_0, e_A)$ is equal to $-1$. Thus, the time orientation $\rho$ makes it possible to distinguish the half-cones of $C$, the $\textit{future half-cone}$ $C^+$ and $\textit{the past half-cone}$ $C^-$. We want to stress that an orientable hyperbolic manifold may not admit a time orientation.

A timelike path of $V_n$ is a path whose tangent at every point $x$ lies within or on $C_x$. If $\mathcal{U}$ is a set in $V_n$, the $\textit{future}$ ${I}^+(\mathcal{U})$ is the set of points on timelike paths emanating from the points $x$ of $\mathcal{U}$ and lying in the future of $x$, the $\textit{past}$ ${I}^-(\mathcal{U})$ being the set of points on timelike paths leading to the points $x$ of $\mathcal{U}$ and lying in the past of $x$.

These definitions hold in particular in the case $\mathcal{U}= \{x'\}$. The emission ${I}(x')$ of a point $x'$ is the union of its future ${I}^+(x')$ and its past ${I}^-(x')$. The boundary $\partial {I}(x')$ of this emission is characteristic with respect to the field of cones, i.e. it is tangent at each of its points $x$ to the cone $C_x$. The boundary $\partial {I}(x')= \Gamma_{x'}$ is said to be the $\textit{characteristic conoid of vertex x'}$. This conoid consists of bicharacteristics or null geodesics emanating from $x'$.

By the use of geodesic normal coordinates centred at $x'$, one finds that as one approaches its vertex, the conoid $\Gamma_{x'}$ is diffeomorphic to a neighbourhood of the vertex of a cone, the bicharacteristics corresponding to the generators of the cone. That is no longer so away from the vertex $x'$, even under the global assumptions made below; in particular, the null geodesics emanating from $x'$ can intersect.

In the theory of hyperbolic linear systems, Leray has introduced some global assumptions which ensure the existence of elementary solutions, even in the presence of singularities of the characteristic conoid. According to Leray and Madame Choquet-Bruhat, a hyperbolic manifold $V_n$ satisfying the previous assumptions is said to be $\textit{globally hyperbolic}$ if the set of timelike paths joining two points is always either empty or compact: from every infinite set of timelike paths joining the two points, one can always extract a sequence that converges to a timelike path. It this condition is satisfied, no timelike line can ever be closed.  

On a globally hyperbolic manifold, a set $\mathcal{U}$ is said to be $\textit{compact towards}$ $\textit{the past}$ if the intersection of $\mathcal{U}$ with ${I}^-(x)$ is compact or empty for all $x$; ${I}^+(\mathcal{U})$ and every closed subset of ${I}^+(\mathcal{U})$ are then also compact towards the past. 

Similarly, one can say that $\mathcal{U}$ is $\textit{compact towards the future}$ if the intersection of $\mathcal{U}$ with ${I}^+(x)$ is compact or empty for all $x$; ${I}^-(\mathcal{U})$ and every closed subset of ${I}^-(\mathcal{U})$ are then also compact towards the future. From a fundamental lemma of Leray, it turns out that if $\mathcal{U}$ is compact towards the past and ${\mathcal{U}}'$ is compact, the intersection ${I}^+(\mathcal{U}) \cap {I}^-({\mathcal{U}}')$ is compact. 

Every point of a locally hyperbolic manifold admits a neighbourhood $\Omega$ homeomorphic to an open ball and globally hyperbolic, in such a way that the previous results hold on $\Omega$.

This example is interesting because it also provides an alternative definition of the $\textit{characteristic conoid}$ to that given in the first chapters. Its interest lies in the use of causal structure concepts and hence can be seen as more fundamental.

Eventually, global hyperbolicity plays a key role in proving singularity theorems because of the following proposition:

\begin{prop} Let $p$ and $q$ lie in a globally hyperbolic set  $M$ and $q \in J^+(p)$. Then, there exists a non-spacelike geodesic from $p$ to $q$ whose length is greater than or equal to that of any other non-spacelike curve from $p$ to $q$. \end{prop}.

\chapter{Application: Green functions of Gravitational Radiation Theory}
\chaptermark{Green functions of Gravitational Radiation Theory}
\epigraph{The heavens and all the constellations rung, \\
The planets in their station listening stood.}{John Milton, Paradise Lost}
In the previous Chapters, it has been shown how the Riemann function solves a characteristic initial-value problem. Our aim is to use this method to study gravitational radiation in black hole collisions at the speed of light. More precisely, to analyse the  Green function for the perturbative field equations by studying the corresponding second-order hyperbolic operator with variable coefficients. After reduction to canonical form of this hyperbolic operator, the integral representation of the solution in terms of the Riemann kernel is obtained. The study of the axisymmetric collision of two black holes at the speed of light is useful in order to understand the more realistic collision of two black holes with a large but finite incoming Lorentz factor $\gamma$. The curved radiative region of the space-time, produced after the two incoming impulsive plane-fronted shock waves have collided, is treated using perturbation theory. To proceed with the study of the Green functions of the gravitational radiation in black hole collisions at the speed of light, following D'Eath \cite{d1992gravitational1, d1992gravitational2}, we make an introduction about its main features.

\section{Black Hole Collisions at the speed of light}
Since the time when general relativity was originally formulated by Einstein there is no analytic solution which does not possess a large number of simplifying symmetries. To study the generation of gravitational radiation by realistic physical sources it is necessary to consider isolated gravitating systems that are time dependent and which can have no simplifying features apart from axisymmetry. This can be done by making use of approximation procedures. There are two alternatives which are $\textit{numerical simulation}$ and $\textit{perturbation theory}$, respectively. In this last case, one assumes that the space-time metric differs only very slightly from some fixed background. The field equations for the metric perturbations are linear in the lowest order and mathematically tractable owing to the simple nature of the background metric. However, since the time-dependent perturbations must be small, the gravitational radiation produced is almost always correspondingly weak.  To deduce the behaviour of gravitating systems when the perturbations are not small, it is necessary to perform the weak-field limit which can provide physical insight but not quantitative results. In fact, there is only one physical process in which perturbation methods have proved successful in describing truly strong-field gravitational radiation that is the high-speed collision of two black holes. The success of perturbation theory in these space-times is due to certain special features of their geometry. 

More precisely, owing to special-relativistic effects, the gravitational field of a black hole travelling close to the speed of light becomes concentrated in the vicinity of its trajectory, which lies close to a null plane in the surrounding nearly Minkowskian space-time. At precisely the speed of light, the black hole turns into a particular sort of impulsive gravitational plane-fronted wave. Then the curvature is zero except on the null plane of its trajectory, and there is a massless particle travelling along the axis of symmetry at the center of this null plane. 

An important property of this sort of gravitational shock wave is that geodesics crossing it are not only bent inwards, but also undergo an instantaneous translation along the null surface that describes the trajectory of the wave. The nature of this translation is such that geodesics crossing the shock close to the axis of symmetry are delayed relative to those which cross the shock far out from the axis. Hence, when two such waves pass through each other in a head-on collision, the far-field region of each wave is given a large head start over its near-field counterpart, in addition to being bent slightly inwards. Because of this, the self-interaction of the far field of each wave as it propagates out towards null infinity takes place without interference from the highly nonlinear region near the axis of symmetry; and because gravity is weak in the far-field region, perturbation theory can be used to study this process. However, the radiation produced in the forward and backward null directions is not weak, for although the far fields contain only a fraction of the total energy, the solid angle into which they are focused is small, and hence the energy flux per unit solid angle in these directions is not small. Thus, the perturbation methods can successfully describe the generation of truly strong-field gravitational radiation in these space-times. 

There are two different perturbation methods that one can use to treat these high-speed collisions. In one approach, the collision was studied by large but finite $\gamma$, where $\gamma$ is the Lorentz factor of the incoming holes. It was shown that the metric of a single high-speed hole, and hence also the precollision metric in the high-speed collision, can be expressed as a perturbation series in ${\gamma}^{-1}$. Then, it is possible to use a method of matched asymptotic expansions to investigate the space-time geometry to the future of the collision. It is necessary to use a number of different asymptotic expansions  to allow for the various length and time scales characteristic of the gravitational field in different parts of the space-time. One expects that expansions holding in adjacent regions will match smoothly on to each other; the regions to the past thereby providing boundary conditions for those neighbouring regions to the future. 

Following this approach, it is possible to calculate the radiation on angular scales of $O({\gamma}^{-1})$ produced by the focusing of the far fields of the waves as they pass through each other during the collision. In this region the news function has an asymptotic expansion of the form
\begin{equation} \label{eq:7.1} c_0(\bar{\tau}, \hat{\theta}={\gamma}^{-1}\psi) \sim \sum_{n=0}^{\infty} {\gamma}^{-2n} \mathcal{Q}_{2n}(\bar{\tau}, \psi) \end{equation}
valid as $\gamma \rightarrow \infty$ with $\bar{\tau}$, $\psi$ fixed, where $\bar{\tau}$ is a suitable $\textit{retarded time}$ coordinate and $\hat{\theta}$ is the angle from the symmetry axis in the center-of-mass frame. The calculus of the leading term $\mathcal{Q}_0(\bar{\tau}, \psi)$ shows that this does not vanish and it is a regular term of $\bar{\tau}$. Since $\mathcal{Q}_0(\bar{\tau}, \psi)$ is not dumped by any power of ${\gamma}^{-1}$, the news function is of order 1, and therefore describes truly strong-field gravitational radiation. 

On angular scales of order 1, the news function should have an asymptotic expansion of the form
\begin{equation} \label{eq:7.2}c_0(\hat{\tau}, \hat{\theta}={\gamma}^{-1}\psi) \sim \sum_{n=0}^{\infty} {\gamma}^{-n} {S}_{n}(\hat{\tau}, \hat{\theta}) \end{equation}
valid as $\gamma \rightarrow \infty$ with $\hat{\tau}$, $\hat{\theta}$ fixed. The retarded time variables in $(\ref{eq:7.1})$ and $(\ref{eq:7.2})$ are not the same, since they refer to varying time delays suffered by different parts of the shocks when they collide. Here ${S}_{0}(\hat{\tau}, \hat{\theta})$ must be the news function for the collision at the speed of light, i.e. $\gamma=\infty$. If the two asymptiotic expansions $(\ref{eq:7.1})$ and $(\ref{eq:7.2})$ both hold in the intermediate region where ${\gamma}^{-1} << \hat{\theta} << 1$, then matching enables us to gain information about the angular dependence of ${S}_{0}(\hat{\tau}, \hat{\theta})$ near the axis $\hat{\theta}=0$. Furthermore, if ${S}_{0}(\hat{\tau}, \hat{\theta})$ is sufficiently regular it will possess a convergent series of the form
\begin{equation}  \label{eq:7.3}{S}_{0}(\hat{\tau}, \hat{\theta})= \sum_{n=0}^{\infty} a_{2n} (\hat{\tau}) (sin(\hat{\theta}))^{2n}, \end{equation}
since it is symmetrical about $\hat{\theta}= \frac{\pi}{2}$ in the center-of-mass frame. Since $\hat{\theta}={\gamma}^{-1}\psi$ in Eq. $(\ref{eq:7.1})$, the ${\hat{\theta}}^{2n}$ part of Eq. $(\ref{eq:7.3})$ will be found from the $({\gamma}^{-1}\psi)^{2n}={\gamma}^{-2n}\psi^{2n}$ part of $(\ref{eq:7.1})$, and then finding $\mathcal{Q}_{2n}(\bar{\tau}, \psi) $ enables one to determine the coefficient $a_{2n}(\hat{\tau})$ of $(sin(\hat{\theta}))^{2n}$ in Eq. ($\ref{eq:7.3})$. In this way, $a_{2n}(\hat{\tau})$ was found, given by the limiting form of $\mathcal{Q}_{0}(\bar{\tau}, \psi)$ as $\psi \rightarrow \infty$. Hence, perturbation methods can be used to determine the entire news function of the highly nonlinear speed-of-light collision. But to calculate high-order $\mathcal{Q}_{2n}(\bar{\tau}, \psi)$ requires the solution of inhomogeneous flat-space wave equations with complicated source terms, and it is not possible to determine the nonisotropic part of ${S}_{0}(\hat{\tau}, \hat{\theta})$.

We will follow another way of calculating ${S}_{0}(\hat{\tau}, \hat{\theta})$ using perturbation methods, which deals with the collision at the speed of light. Starting with the speed-of-light collision of two shocks which each have energy $\mu$, then we make a large Lorentz boost away from the center-of-mass frame. There the energy $\nu=\mu e^\alpha$ of the incoming shock 1, which initially lies on the hyperplane $z+t=0$ between two portions of Minkowski space, obeys $\nu >> \lambda$, where $\lambda= \mu e^{- \alpha}$ is the energy of the incoming shock 2, which initially lies on the hypersurface $z-t=0$. In the boosted frame, the metric describing the scattering of the weak shock off the strong one possesses a perturbation expansion in powers of $\frac{\lambda}{\nu}$, that is
\begin{equation} \label{eq:7.4} g_{ab} \sim \nu^2 \bigg{[} \eta_{ab} + \sum_{i=1}^\infty \bigg{(}\frac{\lambda}{\nu} \bigg{)}^i h_{ab}^{(i)} \bigg{]}, \end{equation}
with respect to suitable coordinates, where $\eta_{ab}$ is the Minkowski metric. The problem of solving the Einstein field equations becomes a singular perturbation problem of finding $h_{ab}^{(1)}$, $h_{ab}^{(2)}$, $\dots$, by successively solving the linearized field equations at first, second, ... order in $\frac{\lambda}{\nu}$, given the characteristic initial data on the surface $u=0$ just to the future of the strong shock 1. 

On boosting back to the center-of-mass frame, one finds that the perturbation series $(\ref{eq:7.4})$ gives an accurate description of the space-time geometry in the region in which gravitational radiation propagates at small angles away from the forward symmetry axis $\hat{\theta}=0$. By reflection symmetry, an analogous series also give a good description near the backward axis $\hat{\theta}=\pi$. The news function $c_0(\hat{\tau}, \hat{\theta})$, which describes the gravitational radiation arriving at future null infinity $J^+$ in the center-of-mass frame, is expected to have the convergent series expansion
\begin{equation} \label{eq:7.5} {c}_{0}(\hat{\tau}, \hat{\theta})= \sum_{n=0}^{\infty} a_{2n} \bigg{(}\frac{\hat{\tau}}{\mu}\bigg{)} (sin(\hat{\theta}))^{2n}, \end{equation}
where $\hat{\tau}$ is a suitable retarded time coordinate and where we replaced $a_{2n}(\hat{\tau})$ with $a_{2n}\big{(}\frac{\hat{\tau}}{\mu}\big{)}$, since $\hat{\tau}$ will always appear as an argument in the dimensionless combination $\frac{\hat{\tau}}{\mu}$. The first-order perturbation calculation of $h_{ab}^{(1)}$, on boosting back to the center-of-mass frame, yields $a_{0} \big{(}\frac{\hat{\tau}}{\mu}\big{)}$, in agreement with the expression of the isotropic part of the news function of the collision of two black holes at large but finite incoming Lorentz factor $\gamma$ on angular scales of order 1. The second-order calculation of $h_{ab}^{(2)}$, which consists in solving the second-order field equations in the boosted frame which take the form of inhomogeneous flat-space wave equations with complicated source terms, gives an integral expression for the first nonisotropic coefficient $a_{2} \big{(}\frac{\hat{\tau}}{\mu}\big{)}$ which cannot be evaluated numerically. Then, in what follows, we are going to show how the calculation of $a_{2} \big{(}\frac{\hat{\tau}}{\mu}\big{)}$ can be simplified analytically so as to enable us to compute this function numerically. 

This is of our interest since, if all the gravitational radiation in the space-time is accurately described by Eq. ($\ref{eq:7.5})$, then the mass of the assumed final static Schwarzshild black hole remaining after the collision can be determined from knowledge only of $a_{0} \big{(}\frac{\hat{\tau}}{\mu}\big{)}$ and $a_{2} \big{(}\frac{\hat{\tau}}{\mu}\big{)}$.

To begin the process of finding a simpler form for $a_{2} \big{(}\frac{\hat{\tau}}{\mu}\big{)}$, we note that because of the conformal symmetry at each order in perturbation theory, the field equations obeyed by the metric perturbations $h_{ab}^{(1)}$, $h_{ab}^{(2)}$, $\dots$ in Eq. ($\ref{eq:7.4})$ may all be reduced to equations in two independent variables. Indeed, a conformal transformation does not effect the intrinsic nature of the perturbation problem, it merely alters the value of the perturbation parameter. Then, once a conformal transformation is performed, in an appropriate gauge, the field equations for the $h_{ab}^{(i)}$ are all of the form 
\begin{equation} \label{eq:7.6} \Box h_{ab}^{(i)} = S_{ab}^{(i)}, \end{equation}
where $S_{ab}^{(i)}$ is a function of $h_{ab}^{(i-1)}$, ..., $h_{ab}^{(1)}$ and their derivatives. Since each $h_{ab}^{(i)}$ is, at this stage, of the form $fn(q, r){\rho}^{-k}$, its corresponding $S_{ab}^{(i)}$ must be of the form $fn(q, r){\rho}^{-(k+2)}$. This indicates that it is possible to eliminate $\rho$ from the field equations by separation of variables, thereby reducing them to two-dimensional differential equations. 
 
\section{Reduction to two dimensions}
Let us now perform the reduction to two dimensions explicitly (D'Eath \cite{d1992gravitational2}), starting with the first-order perturbations $h_{ab}^{(1)}$. These are particular cases of the general system given by the flat-space wave equation
\begin{equation} \label{eq:7.7} \Box \psi \equiv 2 \frac{\partial^2 \psi}{\partial u \partial v} + \frac{1}{\rho} \frac{\partial}{\partial \rho} \bigg{[} \rho \frac{\partial \psi}{\partial \rho} \bigg{]} + \frac{1}{\rho^2} \frac{\partial^2 \psi}{\partial \phi^2}=0, \end{equation}
supplemented by the boundary condition 
\begin{equation} \label{eq:7.8} \begin{split}
&\psi|_{u=0} = e^{im\phi} \rho^{-n}f [8 ln(v \rho) - \sqrt{2} v ], \\
&f(x)=0, \hspace{1cm} \forall x <0,
\end{split} \end{equation}
where $m$ and $n$ are integers and, apart from the above restriction, $f(x)$ is arbitrary. We know from our previous arguments that $\psi$ must be of the form $e^{i m \phi} \rho^{-n}\chi(q, r)$ if $u \geq 0$, where
\begin{equation} \begin{split} \label{eq:7.9}
& q \equiv u \rho^{-2}, \\
& r \equiv 8 log(v \rho) - \sqrt{2} v. \end{split} \end{equation}
From Eq. ($\ref{eq:7.9})$ we find
\begin{equation} \begin{split} \label{eq:7.10}
& \bigg{[} \frac{\partial}{\partial v} \bigg{]}_{v, \rho, \phi}= \frac{1}{\rho^2} \bigg{[} \frac{\partial}{\partial q} \bigg{]}_{r, \rho, \phi}, \\
&  \bigg{[} \frac{\partial}{\partial u} \bigg{]}_{u, \rho, \phi}= - \sqrt{2} \bigg{[} \frac{\partial}{\partial r} \bigg{]}_{q, \rho, \phi}, \\
&  \bigg{[} \frac{\partial}{\partial \rho} \bigg{]}_{u, v, \phi}=  \bigg{[} \frac{\partial}{\partial \rho} \bigg{]}_{q, r, \phi} - \frac{2 q}{\rho} \bigg{[} \frac{\partial}{\partial q} \bigg{]}_{r, \rho, \phi} - \frac{8}{\rho}  \bigg{[} \frac{\partial}{\partial r} \bigg{]}_{q, \rho, \phi},
\end{split} \end{equation} 
and therefore
\begin{equation}\begin{split}\label{eq:7.11} 2 \frac{\partial^2}{\partial u \partial v} + \frac{1}{\rho} \frac{\partial}{\partial \rho} \bigg{[} \rho \frac{\partial}{\partial \rho} \bigg{]} + \frac{1}{\rho^2}\frac{\partial^2}{\partial \phi^2} =& \frac{1}{\rho^2} \Biggl\{ - 2\sqrt{2} \frac{\partial^2}{\partial q \partial r} + \bigg{[} \rho \frac{\partial}{\partial \rho} - 2q \frac{\partial}{\partial q} \\
 &+ 8 \frac{\partial}{\partial r} \bigg{]} \bigg{[} \rho \frac{\partial}{\partial \rho} - 2q \frac{\partial}{\partial q} + 8 \frac{\partial}{\partial r} \bigg{]} + \frac{\partial^2}{\partial \phi^2} \Biggr\}. \end{split} \end{equation}
Thus, $\chi$ is the solution of
\begin{equation} \begin{split}\label{eq:7.12} \mathcal{L}_{m, n} \chi \equiv & \Biggl\{ - 2 \sqrt{2} \frac{\partial^2}{\partial q \partial r} + \bigg{[} -n -2q \frac{\partial}{\partial q} + 8 \frac{\partial}{\partial r} \bigg{]} \bigg{[} -n - 2q \frac{\partial}{\partial q} \\
&+ 8 \frac{\partial}{\partial r} \bigg{]} - m^2 \Biggr\} \chi=0, \end{split} \end{equation}
where the boundary condition is $\chi|_{q=0}= f(r)$.

For the homogeneous wave equation where the solution has a simple integral form, there is no advantage in eliminating $\rho$ and $\phi$ from the differential equation.  
However, the higher-order metric perturbations $h_{ab}^{(i)}$ with $i \geq 2$ turn out to obey inhomogeneous flat-space wave equations of the form
\begin{equation} \label{eq:7.13} \Box \psi = S \end{equation}
where $S$ is a source term given by $S=e^{im \phi}\rho^{- (n +2)}H(q, r)$ and the boundary condition may be taken to be $\psi|_{u=0}=0$. This leads to the following equation for $\chi\equiv e^{-im \phi} \rho^n \psi$:
\begin{equation} \label{eq:7.14} \mathcal{L}_{m, n} \chi(q, r) = H(q, r),  \end{equation}
where $\mathcal{L}_{m, n}$ is a $\textit{hyperbolic operator}$ in the independent variables $q$ and $r$. 
In contrast with the homogeneous case, the benefits  gained in the reduction of Eq. ($\ref{eq:7.13})$ to Eq. ($\ref{eq:7.14})$ are not insignificant. Previously, to calculate the solution at some space-time point $P$ we would have had to integrate the source term $S$, suitably weighted, over the past null cone of $P$. Now, once the Green's function for the differential operator $\mathcal{L}_{m, n}$, defined by
\begin{equation}\label{eq:7.15} \mathcal{L}_{m, n}G_{m, n}(q, r; q_0, r_0)= \delta(q-q_0)\delta(r-r_0),\end{equation}
(where $\mathcal{L}_{m, n}$ acts on the $(q, p)$ part of $G_{m, n}$) has been found; we need simply to integrate the product of $H$ and the Green's function for the differential operator $\mathcal{L}_{m, n}$ over some two-dimensional region in the $(q, r)$-plane, i.e.
\begin{equation} \label{eq:7.16} \chi(q, r)= \int \int G_{m, n} (q, r; q_0, r_0) H(q_0, r_0)dq_0 dr_0, \end{equation}
subject to suitable boundary conditions. 

This makes it much easier to estimate the various contributions to the solution from different parts of the integration region.

\section{Reduction to canonical form and the Riemann function}
It is more convenient to reduce Eq. ($\ref{eq:7.14})$ to canonical form, and then to find an integral representation of the solution. But first, we want to demonstrate that the differential operator $\mathcal{L}_{m, n}$ is hyperbolic. Hence, we define new coordinates 
\begin{equation} \label{eq:7.17} x=x(q, r) \hspace{2cm} y=y(q, r). \end{equation}
Now,
\begin{equation} \label{eq:7.18} \mathcal{L}_{m, n} = - ( 2 \sqrt{2} + 32 q) \frac{\partial^2}{\partial q \partial r} + 4q^2 \frac{\partial^2}{\partial q^2} + 64 \frac{\partial^2}{\partial r^2} + 4(n+1)q \frac{\partial}{\partial q} - 16n \frac{\partial}{\partial r} + n^2 - m^2. \end{equation}
We choose $x$ and $y$ so that  the coefficients $\frac{\partial^2}{\partial x^2}$ and $\frac{\partial^2}{\partial y^2}$ vanish and $\mathcal{L}_{m, n}$ is transformed to normal hyperbolic form, in which (see Chapter 1, Eq. ($\ref{eq:60})$)
\begin{equation} \label{eq:7.19} \mathcal{L}_{m, n}= f(x, y) \frac{\partial^2}{\partial x \partial y} + g(x, y) \frac{\partial}{\partial x} + h(x, y) \frac{\partial}{\partial y} + n^2 - m^2. \end{equation}
Expressing $\mathcal{L}_{m, n}$ in terms of $\frac{\partial}{\partial x}$ and $\frac{\partial}{\partial y}$ we find that
\begin{equation} \begin{split} \label{eq:7.20} \mathcal{L}_{m, n}= & \Biggl\{ - (2 \sqrt{2} + 32 q) \bigg{[}\frac{\partial x}{\partial q} \bigg{]}\bigg{[}\frac{\partial x}{\partial r} \bigg{]} + 4q^2 \bigg{[}\frac{\partial x}{\partial q} \bigg{]}^2 + 64\bigg{[}\frac{\partial x}{\partial r} \bigg{]}^2 \Biggr\} \frac{\partial^2}{\partial x^2} \\
& + \Biggl\{  - (2 \sqrt{2} + 32 q) \bigg{[}\frac{\partial y}{\partial q} \bigg{]}\bigg{[}\frac{\partial y}{\partial r} \bigg{]} + 4q^2 \bigg{[}\frac{\partial y}{\partial q} \bigg{]}^2 + 64\bigg{[}\frac{\partial y}{\partial r} \bigg{]}^2 \Biggr\} \frac{\partial^2}{\partial y^2} \\
&+ \Biggl\{ - (2 \sqrt{2} + 32 q) \bigg{[} \bigg{[}\frac{\partial x}{\partial q} \bigg{]}\bigg{[}\frac{\partial y}{\partial r} \bigg{]} + \bigg{[}\frac{\partial y}{\partial q} \bigg{]}\bigg{[}\frac{\partial x}{\partial r} \bigg{]} \bigg{]}+ 8q^2 \bigg{[}\frac{\partial x}{\partial q} \bigg{]}\bigg{[}\frac{\partial y}{\partial q} \bigg{]} \\
&+ 128\bigg{[}\frac{\partial x}{\partial r} \bigg{]}\bigg{[}\frac{\partial y}{\partial r} \bigg{]} \Biggr\} \frac{\partial^2}{\partial y \partial x} + ... \end{split}\end{equation}
where we have omitted the terms of first and zeroth order in $\frac{\partial}{\partial x}$ and $\frac{\partial}{\partial y}$.

In order that Eq. ($\ref{eq:7.19})$ be satisfied, we must have
\begin{equation} \label{eq:7.21} - (2 \sqrt{2} + 32 q)\bigg{[}\frac{\partial x}{\partial q} \bigg{]}\bigg{[}\frac{\partial x}{\partial r} \bigg{]} + 4q^2 \bigg{[}\frac{\partial x}{\partial q} \bigg{]}^2 + 64  \bigg{[}\frac{\partial x}{\partial r} \bigg{]}^2 =0, \end{equation}
\begin{equation} \label{eq:7.22} - (2 \sqrt{2} + 32 q)\bigg{[}\frac{\partial y}{\partial q} \bigg{]}\bigg{[}\frac{\partial y}{\partial r} \bigg{]} + 4q^2 \bigg{[}\frac{\partial y}{\partial q} \bigg{]}^2 + 64  \bigg{[}\frac{\partial y}{\partial r} \bigg{]}^2 =0. \end{equation}
This means that $\frac{\partial x }{\partial q}/ \frac{\partial x}{ \partial r}$ and $\frac{\partial y }{\partial q}/ \frac{\partial y}{ \partial r}$ must be the two real roots of the quadratic equation
\begin{equation} \label{eq:7.23} 4 q^2 x^2 - (2 \sqrt{2} + 32 q) x + 64 =0. \end{equation}
The discriminant of this equation is positive, hence $\mathcal{L}_{m, n}$ is hyperbolic and its characteristic coordinates $x$ and $y$ satisfy
\begin{equation} \label{eq:7.24} \bigg{[}\frac{\partial x}{\partial q}\bigg{]} = \bigg{[} \frac{ 1 + 8q \sqrt{2} + \sqrt{(1 + 16 q \sqrt{2})}}{2 q^2 \sqrt{2}} \bigg{]} \bigg{[} \frac{\partial x}{\partial r} \bigg{]}, \end{equation}
and 
\begin{equation} \label{eq:7.25} \bigg{[}\frac{\partial y}{\partial q}\bigg{]} = \bigg{[} \frac{ 1 + 8q \sqrt{2} + \sqrt{(1 + 16 q \sqrt{2})}}{2 q^2 \sqrt{2}} \bigg{]} \bigg{[} \frac{\partial y}{\partial r} \bigg{]}, \end{equation}
where we have arbitrarily assigned the plus sign to $x$ and the minus sign to $y$. Hence, we have shown the hyperbolic nature of $\mathcal{L}_{m, n}$ and we have reduced Eq. ($\ref{eq:7.14})$ to the form ($\ref{eq:7.19})$. For ease of calculation we now choose \cite{Esposito:2001ry}
\begin{equation} \label{eq:7.26} \frac{\partial x}{\partial r}=1, \hspace{2cm} \frac{\partial y}{\partial r}=1. \end{equation}
If we solve Eqs. $(\ref{eq:7.24})$ and $(\ref{eq:7.25})$, by making use of $(\ref{eq:7.26})$, we find
\begin{equation} \label{eq:7.27} x = r + 8 ln \bigg{[} \frac{\sqrt{(1+ 16 q \sqrt{2})} -1}{2} \bigg{]} - \frac{8}{\big{[}\sqrt{(1+16q\sqrt{2})} - 1\big{]}}-4, \end{equation}
\begin{equation} \label{eq:7.28} y = r + 8 ln \bigg{[} \frac{\sqrt{(1+ 16 q \sqrt{2})} +1}{2} \bigg{]} + \frac{8}{\big{[}\sqrt{(1+16q\sqrt{2})} + 1\big{]}}-4, \end{equation}
where the constants of integration have been chosen for future convenience. To simplify these formulae we define
\begin{equation} \label{eq:7.29} t\equiv \sqrt{1+16 q \sqrt{2}} = t(x, y). \end{equation}
Then Eqs. ($\ref{eq:7.27})$ and ($\ref{eq:7.28})$ read as
\begin{equation} \label{eq:7.30} x = r + 8 ln \bigg{(} \frac{t -1}{2} \bigg{)} - \frac{8}{\big{(}t - 1\big{)}}-4, \end{equation}
\begin{equation} \label{eq:7.31} y = r + 8 ln \bigg{(} \frac{t +1}{2} \bigg{)} + \frac{8}{\big{(}t+ 1\big{)}}-4. \end{equation}
If we subtract Eq. ($\ref{eq:7.31})$ to ($\ref{eq:7.30})$, we have
\begin{equation*} (x-y)= 8ln \bigg{(} \frac{t-1}{t+1} \bigg{)} - 8 \frac{2t}{(t^2 - 1)}.\end{equation*}
From which, we find
\begin{equation*} ln \bigg{(} \frac{t-1}{t+1} \bigg{)} - \frac{2t}{(t^2 - 1)} = \frac{(x-y)}{8}, \end{equation*}
that can be written in the form
\begin{equation*}  \frac{(t-1)}{(t+1)}e^{\frac{2t}{(1-t^2)}} = e^{\frac{(x-y)}{8}}. \end{equation*}
If we define 
\begin{equation*} \omega \equiv \frac{(t-1)}{(t+1)} \rightarrow t= \frac{(1 + \omega )}{(1 - \omega)} \end{equation*}
we have to solve the transcendental equation
\begin{equation*} \omega e^{\frac{(\omega^2 -1)}{2 \omega}}= e^\frac{(x-y)}{8} \end{equation*}
to obtain $\omega = \omega (x-y)$, from which we have $t= t(x-y)$. 

Now, if we exploit the formulae
\begin{equation} \label{eq:7.32} \frac{\partial x}{\partial q}= \frac{64 \sqrt{2}}{(t-1)^2} , \end{equation}
\begin{equation} \label{eq:7.33}\frac{\partial y}{\partial q}= \frac{64 \sqrt{2}}{(t+1)^2} , \end{equation}
we find that the coefficients $f(x, y)$, $g(x, y)$ and $h(x, y)$ of Eq. ($\ref{eq:7.19})$ are
\begin{equation} \begin{split} \label{eq:7.34} f(x, y) &=- (2 \sqrt{2} + 32 q) \bigg{(} \frac{\partial x}{\partial q} + \frac{\partial y}{\partial q} \bigg{)} + 8 q^2 \frac{\partial x}{\partial q}\frac{\partial y}{\partial q} + 128 \\
&= 256 \bigg{[} 1 - \frac{ 2 t^2 (t^2 + 1)}{(t-1)^2 (t+1)^2} \bigg{]}, \end{split} \end{equation}
\begin{equation} \label{eq:7.35} g(x, y)= 4(n+1) q \frac{\partial x}{\partial q} -16n= 16 \bigg{[} 1 + \frac{2(n+1)}{(t-1)} \bigg{]}, \end{equation}
\begin{equation} \label{eq:7.36} h(x, y)= 4(n+1) q \frac{\partial y}{\partial q}- 16n= 16 \bigg{[} 1 - \frac{2(n+1)}{(t+1)} \bigg{]}. \end{equation}
The resulting canonical form of Eq. ($\ref{eq:7.14})$ is
\begin{equation} \label{eq:7.37} \mathcal{L}[\chi] = \bigg{(} \frac{\partial^2}{\partial x \partial y} + a(x, y) \frac{\partial}{\partial x} + b(x, y) \frac{\partial}{\partial y} + c(x, y) \bigg{)} \chi (x, y) = \tilde{H}(x, y), \end{equation}
where
\begin{equation} \label{eq:7.38} a(x, y) \equiv \frac{g(x, y)}{f(x, y)} = \frac{1}{16} \frac{(1-t)(t+1)^2 (2n + 1 + t)}{(t^4 + 4t^2 - 1)}, \end{equation}
\begin{equation} \label{eq:7.39} b(x, y) \equiv \frac{h(x, y)}{f(x, y)} = \frac{1}{16} \frac{(t+1)(t-1)^2 (2n + 1 + t)}{(t^4 + 4t^2 - 1)},\end{equation}
\begin{equation} \label{eq:7.40} c(x, y) \equiv \frac{n^2 - m^2}{f(x, y)} = \frac{(m^2 - n^2)}{256} \frac{(t-1)^2 (t+1)^2}{(t^4 + 4t^2 - 1)},\end{equation}
 \begin{equation} \label{eq:7.41} \tilde{H}(x, y) \equiv \frac{H(x, y)}{f(x, y)} = -\frac{H(x, y)}{256} \frac{(t-1)^2 (t+1)^2}{(t^4 + 4t^2 - 1)}.\end{equation}
Note that $a(-t) =b(t)$, $b(-t)=a(t)$, $c(-t)=c(t)$ and $\tilde{H}(-t)=\tilde{H}(t)$.

For a hyperbolic equation in the form ($\ref{eq:7.37})$, we can use the Riemann integral representation of the solution. For this purpose, on denoting by $\mathcal{L}^{*}$ the adjoint of the operator $\mathcal{L}$ in $(\ref{eq:7.37})$, which acts according to
\begin{equation} \label{eq:7.42} \mathcal{L}^{*} [R]= \frac{\partial^2 R}{\partial x \partial y} - \frac{\partial (a R)}{\partial x} - \frac{\partial (b R)}{\partial y} + c R \end{equation}
we have to find the Riemann kernel $R(x, y; \xi, \eta)$ ($(\xi, \eta)$ are the coordinates of a point $P$ such that the characteristics through it intersect a curve $C$ at points $A$ and $B$) subject to the following conditions:
\begin{description}
\item[(a)] As a function of $x$ and $y$, $R$ satisfies the adjoint equation
\begin{equation} \label{eq:7.43} {\mathcal{L}^*}_{(x, y)}[R]=0. \end{equation}
\item[(b)] $R_x=bR$ on $AP$, i.e.
\begin{equation} \label{eq:7.44} \frac{\partial R}{\partial x}(x, y; \xi, \eta)=b(x, \eta)R(x, y; \xi, \eta) \hspace{1cm} {\rm on}\; y= \eta, \end{equation}
and $R_y=aR$ on $BP$, i.e.
\begin{equation} \label{eq:7.45} \frac{\partial R}{\partial y}(x, y; \xi, \eta)=a(\xi, y)R(x, y; \xi, \eta) \hspace{1cm} {\rm on} \; x= \xi. \end{equation}
\item[(c)] $R=1$ at $P$, i.e.
\begin{equation} \label{eq:7.46} R(\xi, \eta; \xi, \eta)=1. \end{equation}
\end{description}
\begin{figure}
\centering
\includegraphics{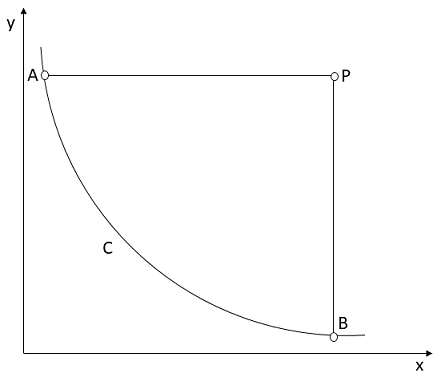}
\caption{Geometry of the characteristic initial-value problem in two independent variables.}\label{fig:7}
\end{figure}
Then, according to the formula ($\ref{eq:83})$ we have obtained in Chapter 1, it is possible to express the solution of Eq. ($\ref{eq:7.37})$ in the form
\begin{equation}\begin{split} \label{eq:7.47} \chi(P)= &\frac{ \chi(A)R(A) + \chi(B)R(B)}{2} + \int_A^B \bigg{(}\bigg{[}bR \chi +\frac{1}{2} \bigg{(} R \frac{\partial \chi}{\partial x} - \frac{\partial R}{\partial x}\chi \bigg{)} \bigg{]}dx \\
&- \bigg{[} aR \chi +  \frac{1}{2}\bigg{(} R \frac{\partial \chi}{\partial y} - \frac{\partial R}{\partial y} \chi \bigg{)} \bigg{]}dy \bigg{)}+ \int\int_\Omega R(x, y; \xi, \eta)\tilde{H}(x, y) dx dy \end{split}\end{equation}
where the path of integration is the one in fig. ($\ref{fig:7})$ and $\Omega$ is a domain with boundary. We note that the main difference between Eq. ($\ref{eq:83})$ and Eq. ($\ref{eq:7.37})$ is that ($\ref{eq:83})$ refers to the equation $L[z]=f$ with $f=0$, whereas in our case $f=\tilde{H} \neq 0$.

Eqs. ($\ref{eq:7.44})$ and $(\ref{eq:7.45})$ are ordinary differential equations for the Riemann kernel $R(x, y; \xi, \eta)$ along the characteristics parallel to the coordinate axes. By virtue of ($\ref{eq:7.46})$, their integration yields 
\begin{equation} \label{eq:7.48} R(x, \eta; \xi, \eta)= exp \bigg{(}\int_\xi^x b(\lambda, \eta) d\lambda\bigg{)}, \end{equation}
\begin{equation} \label{eq:7.49} R(\xi, y; \xi, \eta)= exp \bigg{(}\int_\eta^y a(\xi, \lambda) d\lambda \bigg{)},\end{equation}
which are the values of $R$ along the characteristics through $P$. Instead, Eq. ($\ref{eq:7.47})$ yields the solution of Eq. ($\ref{eq:7.37})$ for arbitrary initial values given along an arbitrary non-characteristic curve $C$, by means of a solution $R$ of the adjoint equation ($\ref{eq:7.43})$ which depends on $x$, $y$ and two parameters $\xi$ and $\eta$. Unlike $\chi$, $R$ solves a characteristic initial-value problem. 

\section{Goursat problem for the Riemann function}
The reduction to canonical form of Eq. ($\ref{eq:7.14})$ previously performed is based on novel features with respect to the analysis of D'Eath \cite{d1992gravitational2}, since Eq. ($\ref{eq:7.47})$ also contains the integral along $AB$ and the term $\frac{1}{2}[\chi(A)R(A) + \chi(B)R(B)]$. This representation of the solution might be more appropriate for the numerical purposes, but the task of finding the Riemann function $R$ remains extremely difficult. However, it is possible to use approximate methods for solving Eq. ($\ref{eq:7.43})$. For this purpose, by virtue of Eq. ($\ref{eq:7.42})$, equation ($\ref{eq:7.43})$ is an equation of the form \cite{Esposito:2001ry}
\begin{equation} \label{eq:7.50} \bigg{(} \frac{\partial^2}{\partial x \partial y} - a \frac{\partial}{\partial x} - b \frac{\partial}{\partial y} + c - \frac{\partial a}{\partial x} - \frac{\partial b}{\partial y}\bigg{)} R(x, y; \xi, \eta)=0. \end{equation}
Eq. ($\ref{eq:7.50})$ can be written in the form of a canonical hyperbolic equation 
\begin{equation} \label{eq:7.51} \bigg{(} \frac{\partial^2}{\partial x \partial y} +A \frac{\partial}{\partial x} +B \frac{\partial}{\partial y} + C \bigg{)} R(x, y; \xi, \eta)=0, \end{equation}
where
\begin{equation}  \left\{ \label{eq:7.52} \begin{array}{l}
A \equiv - a = - \frac{1}{16} \frac{(1-t)(t+1)^2 (2n + 1 + t)}{(t^4 + 4t^2 - 1)};\\
B \equiv -b = - \frac{1}{16} \frac{(t+1)(t-1)^2 (2n + 1 + t)}{(t^4 + 4t^2 - 1)};\\
C \equiv c - \frac{\partial a}{\partial x} - \frac{\partial b}{\partial y}=  \frac{(m^2 - n^2)}{256} \frac{(t-1)^2 (t+1)^2}{(t^4 + 4t^2 - 1)} - \frac{\partial}{\partial x} \bigg{[}\frac{1}{16} \frac{(1-t)(t+1)^2 (2n + 1 + t)}{(t^4 + 4t^2 - 1)}\bigg{]} \\
 - \frac{\partial}{\partial y} \bigg{[} \frac{1}{16} \frac{(t+1)(t-1)^2 (2n + 1 + t)}{(t^4 + 4t^2 - 1)}\bigg{]} .
\end{array}\right.  \end{equation}
Therefore, on defining
\begin{equation} \left\{ \label{eq:7.53} \begin{array}{l}
U \equiv R,\\
V \equiv \frac{\partial R}{\partial x} + BR,
\end{array}\right. \end{equation}
we have 
\begin{equation*} V \equiv \frac{\partial U}{\partial x} + B U \rightarrow \frac{\partial U}{\partial x} = V - BU. \end{equation*}
If we replace this expression of $\frac{\partial U}{\partial x}$ in Eq. $(\ref{eq:7.51})$ we have
\begin{equation*} \begin{split}& \frac{\partial}{\partial y} \frac{\partial U}{\partial x} + A \frac{\partial U}{\partial x} + B \frac{\partial U}{\partial y} + CU = \frac{\partial}{\partial y} (V - BU) + A(V- BU) + B \frac{\partial U}{\partial y} + CU\\
&= \frac{\partial V}{\partial y} - \frac{\partial B}{\partial y} U - B \frac{\partial U}{\partial y} + AV - ABU + B\frac{\partial U}{\partial y} + CU \\
&= \frac{\partial V}{\partial y} - \frac{\partial B}{\partial y}U + AV - ABU + CU=0. \end{split}\end{equation*}
Then, Eq. ($\ref{eq:7.51})$ is equivalent to the hyperbolic canonical system
\begin{equation} \label{eq:7.54} \frac{\partial U}{\partial x} = f_1(x, y)U + f_2(x, y)V, \end{equation}
\begin{equation} \label{eq:7.55} \frac{\partial V}{\partial y} = g_1(x, y)U + g_2(x, y)V, \end{equation}
where
\begin{equation} \label{eq:7.56} \left\{ \begin{array} {l}
f_1 \equiv - B= b, \\
f_2 \equiv 1, \\
g_1 \equiv AB - C + \frac{\partial B}{\partial y} = ab - c + \frac{\partial a}{\partial x}, \\
g_2 \equiv -A= a. 
\end{array} \right. \end{equation}
An existence and uniqueness theorem holds for the system described by Eqs. ($\ref{eq:7.54})$ and ($\ref{eq:7.55})$ with boundary data ($\ref{eq:7.48})$ and ($\ref{eq:7.49})$ and hence we can exploit the finite-difference method to find approximate solutions for the Riemann function $R(x, y; \xi, \eta)$ and eventually $\chi(P)$ by Eq. ($\ref{eq:7.47})$.

\section{Solution of the characteristic initial-value problem for the homogeneous hyperbolic equation}
At this stage, we have to solve a characteristic initial-value problem for a homogeneous hyperbolic equation in canonical form in two independent variables, for which we have developed formulae to be used for numerical solution with the help of a finite-differences scheme. For this purpose, we study the canonical system ($\ref{eq:7.54})$ and ($\ref{eq:7.55})$ written as
\begin{equation} \label{eq:7.57} \frac{\partial U}{\partial x} = F(x, y, U, V),\end{equation}
\begin{equation} \label{eq:7.58} \frac{\partial V}{\partial y}= G(x, y, U, V). \end{equation}
in the rectangle $R \equiv \{x, y: x \in [x_0, x_0 + a], y \in [y_0, y_0 + b] \}$ with known values of $U$ on the side $AD$ where $x=x_0$, and known values of $V$ on the side $AB$ where $y=y_0$. Then, the segments $AB$ and $AD$ are divided into $m$ and $n$ equal parts, respectively. By setting $\frac{a}{m}\equiv h$ and $\frac{b}{n} \equiv k$, the original differential equations become equations relating values of $U$ and $V$ at three intersection points of the resulting lattice, i.e.
\begin{equation} \label{eq:7.59} \frac{U(P_{r, s+1}) - U(P_{r, s})}{h}=F, \end{equation}
\begin{equation} \label{eq:7.60} \frac{V(P_{r+1, s}) - V(P_{r, s})}{k}=G. \end{equation}
It is now convenient to set $U_{r,s}\equiv U(P_{r, s})$ and $V_{r,s} \equiv V(P_{r,s})$, hence these equations read
\begin{equation} \label{eq:7.61} U_{r, s+1}= U_{r, s}+ hF(P_{r,s}, U_{r, s}, V_{r, s}), \end{equation}
\begin{equation} \label{eq:7.62} V_{r+1, s}= V_{r, s}+ kG(P_{r,s}, U_{r, s}, V_{r, s}). \end{equation}
Then, if both $U$ and $V$ are known at $P_{r, s}$, it is possible to evaluate $U$ at $P_{r, s+1}$ and $V$ at $P_{r+1, s}$. The evaluation at subsequent intersection points of the lattice goes on along horizontal or vertical segments. In the former case, the resulting algorithm is
\begin{equation} \label{eq:7.63} U_{r, s}= U_{r, 0}+ h \sum_{i=1}^{s-1}F(P_{r, i}, U_{r, i}, V_{r, i}), \end{equation}   
\begin{equation} \label{eq:7.64} V_{r, s}= V_{r-1, s}+ kG(P_{r-1, s}, U_{r-1, s}, V_{r-1, s}), \end{equation}
while in the latter case we have the algorithm expressed by the equations
\begin{equation} \label{eq:7.65} V_{r, s}= V_{0, s}+ k \sum_{i=1}^{r-1}G(P_{i, s}, U_{i, s}, V_{i, s}), \end{equation}
\begin{equation} \label{eq:7.66} U_{r, s}= U_{r, s-1}+ hF(P_{r, s-1}, U_{r, s-1}, V_{r, s-1}). \end{equation}
The stability of such solutions is closely linked with the geometry of the associated characteristics.

$\mathbf{Conclusion.}$ It is possible to evaluate the coefficient $a_2$ which appears in the news function ($\ref{eq:7.5}$) by solving the equation
\begin{equation*}\omega e^{\frac{(\omega^2 -1)}{2 \omega}}= e^\frac{(x-y)}{8} \end{equation*}
numerically for $\omega= \omega(x-y)$, from which it is possible to obtain $t(x-y)$. This yields $a$, $b$, $c$ and $H$ as functions of $(x, y)$ according to $(\ref{eq:7.38})$, $(\ref{eq:7.39})$, $(\ref{eq:7.40})$ and $(\ref{eq:7.41})$, and hence $A$, $B$ and $C$ in the equation for the Riemann function are obtained according to $(\ref{eq:7.52})$, where the derivatives with respect to $x$ and $y$ are evaluated numerically. Eventually, the system given by $(\ref{eq:7.54})$ and $(\ref{eq:7.55})$ is solved according to the finite-differences scheme with $F= f_1 U + f_2 V$ and $G=g_1U + g_2V$.

Once the Riemann function $R=U$ is obtained with the desired accuracy, numerical evaluation of the integral $(\ref{eq:7.47})$ yields $\chi(P)$, and $\chi(q, r)$ is obtained upon using equations $(\ref{eq:7.30})$ and $(\ref{eq:7.31})$ for the characteristic coordinates.

\chapter*{Conclusions}
\addcontentsline{toc}{chapter}{Conclusions}
\markboth{}{}
The study of the Fourès-Bruhat proof of existence and uniqueness of the solution of Cauchy's problem for Einstein vacuum field equations has been the main aim of the present work. 

This has been shown by first considering systems of $m$ partial differential equations in $n$ unknown functions of $n + 1$ independent variables for which we have given the definition of $\textit{characteristic manifolds}$ and introduced the concept of wavelike propagation. 

Then, we have introduced the theory of hyperbolic equations giving the definition of hyperbolic equation, first on a vector space and then on a manifold, and hence we have considered a second-order linear hyperbolic equation in two variables to discuss Riemann's method. More precisely, we have given the proof of existence of $\textit{Riemann's kernel function}$ and stressed its importance in solving hyperbolic equations that obey characteristic initial-value problems. 

Therefore, our argumentation proceeds in studying the $\textit{fundamental}$ \\ $\textit{solutions}$ and their relation with Riemann's kernel. A first definition of $\textit{characteristic conoid}$ has been given by noticing that the fundamental solution is singular not only at a point but along a certain surface. Since any singular surface of a solution of a linear differential equation must be a characteristic, such singular surface must hence satisfy a first order differential equation. Among the solutions, the one we have considered has a given point as a conic point and it is called the $\textit{characteristic conoid}$ and then, upon introducing on a connected, four-dimensional, Hausdorff four-manifold $M$ the $\textit{characteristic polynomial}$ of a linear partial differential operator $L$, it has been defined as the cone in the cotangent space at $x \in M$. Moreover, we have studied the fundamental solution with an algebraic singularity and introduced the concept of $\textit{geodesic}$ as auto-parallel curves. 

To conclude the discussion upon the fundamental solution, we have seen how to build fundamental solutions, by showing some examples with odd or even number of variables and by studying the case of scalar wave equation. 

The discussion moves towards the study of linear systems of normal hyperbolic form. We have seen that every solution of a system $[E]$ of $n$ second order partial differential equations, with $n$ unknown functions and four variables $x$, hyperbolic and linear, which possesses in a domain $D$ first partial derivatives with respect to the four variables $x$ continuous and bounded
$$ \sum_{\lambda, \mu=1}^4 A^{\lambda \mu} \frac{\partial^2 u_s}{\partial x^\lambda \partial x^\mu} + \sum_{s=1}^n \sum_{\mu=1}^4 {B^s_r}^\mu \frac{\partial u_s}{\partial x^\mu} + f_r=0 \hspace{5cm} [E]$$
verifies some $\textit{Kirchhoff formulae}$. We have then obtained a system of integral equations verified in a domain $D_0$ by these solutions. 

Then, we have considered a system $[F]$ of non-linear, second-order, hyperbolic partial differential equations with $n$ unknown functions $W_s$ and four variables $x^\alpha$
$$\sum_{\lambda, \mu=1}^4 A^{\lambda \mu}(x^\alpha, W_s, W_{s\alpha}) \frac{\partial^2 W_s}{\partial x^\mu \partial x^\lambda} + f_s(x^\alpha, W_s, W_{s\alpha}) =0 \hspace{3cm} [F]$$
to show under which assumptions it is possible to turn it into a linear system for which the results previously obtained for them hold. For this purpose, we have considered the functions $W_s$ as functions of the four variables $x^\alpha$; the coefficients $A^{\lambda \mu}$ and $f_s$ are then functions of these four variables. We apply these results to the equations $[F']$ obtained by differentiating five times with respect to the variables $x^\alpha$ the given equations $[F]$. Thus, we obtain a system of integral equations whose left-hand side are the unknown functions $W_s$, their partial derivatives with respect to the $x^\alpha$ up to the fifth order and some auxiliary functions $X$, $\Omega$, and whose right-hand sides contain only these functions and the integration parameters. 

Then, in order to solve the Cauchy problem for the nonlinear equations $[F]$ we tried to solve the system of integral equations verified by the solutions. Some difficulties arise since the quantities occurring under the integral sign must be continuous and bounded upon assuming differentiability of the coefficients $A^{\lambda \mu}$, viewed as a function of the variables $x^\alpha$. This does not hold when the functions $W_s$, $W_{s \alpha}$, ..., $U_S$ are independent, thus the quantity $[A^{ij}]\frac{\partial^2 \sigma}{\partial x^i \partial x^j} J_{x \lambda}$ fails to be bounded and continuous. 

Moreover, we have passed through the intermediate stage of approximate equations $[F_1]$, where the coefficients $A^{\lambda \mu}$ are some functions of $x^\alpha$. Therefore, we tried to solve the integral equations and to show that their solution is a solution of the equations $[F_1]$, but we have noticed that the obtained solution $W_s$ is only five times differentiable and our method is applicable only if the $A^{\lambda \mu}$ depend on the $W_s$ and not on the $W_{s \alpha}$. Hence, we have solved the Cauchy problem for the system $[G]$
$$ \sum_{\lambda, \mu=1}^4 A^{\lambda \mu}(x^\alpha, W_s) \frac{\partial^2 W_s}{\partial x^\lambda \partial x^\mu}+ f_s(x^\alpha, W_s, W_{s\alpha}) =0 \hspace{4cm} [G]$$
where the coefficients $A^{\lambda \mu}$ do not depend on the $W_{s \alpha}$. It is enough to apply the results for equations $[E]$ to the equations $[G']$ deduced from equations $[G]$ by four differentiations with respect to the variables $x^\alpha$ in order to obtain a system of integral equations whose right-hand sides contain only functions that are the same as those occurring on the left-hand sides. 

The integral equations $[J]$, verified by the bounded solutions and with bounded first derivatives of equations $[G']$, only involve the coefficients $A^{\lambda \mu}$ and $B^{T \lambda}_S$ and their partial derivatives up to the orders four and two, respectively, and the coefficients $F_S$. To solve the integral equations $[J]$ we have found the same difficulty as in the general case.

Hence, to solve the Cauchy problem we have studied the approximated system $[G_1]$ of $[G]$, by making the substitution in $A^{\lambda \mu}$ (and not in $f_s$) of the $W_s$ with some approximate values ${W_s}^{(1)}$. 

Then, we have studied the equations $[G_1']$, obtained by differentiation of $[G_1]$ five times with respect to the variables $x^\alpha$, viewed as linear equation of type $[E]$ in the unknown functions $U_S$, and we have proved that its corresponding system of integral equations $[J_1]$, admits of a unique, continuous and bounded solution in a domain $D$.

Eventually, since the solution of the Cauchy problem given for the equations $[G_1]$ defines a representation of the space of the functions ${W_S}^{(1)}$ into itself, we have proved that this representation admits a fixed point, belonging to the space. 

The corresponding $W_s$ are solutions of the given equations $[G]$. This solution is unique and possesses partial derivatives continuous and bounded up to the fourth order.

At this stage, once we have shown the existence and uniqueness of the solution of the Cauchy problem for systems of linear and non-linear equations, we have  seen how finally these results can be used to solve the Cauchy problem for the Einstein field equations. 

The Cauchy problem for the vacuum field equations, $R_{\alpha \beta}=0$ with initial data on a hypersurface $S$ has been formulated and it has been shown under which conditions this problem admits, in the analytic case, a solution and this solution is unique. 

Therefore, we refer to the vacuum field equations written for any coordinates and, by making use of $\textit{isothermal coordinates}$, we have seen that they are of the type of the nonlinear equations previously studied, i.e
$$ \mathcal{G}_{\alpha \beta}= \sum_{\lambda, \mu=1}^4(g^{-1})^{\lambda \mu} \frac{\partial^2 g_{\alpha \beta}}{\partial x^\lambda \partial x^\mu} + H_{\alpha \beta}=0.$$
Thus, the Cauchy problem for Einstein vacuum field equations can be solved, if we identify $(g^{-1})^{\lambda \mu} =A^{\lambda \mu}$, $g_{\alpha \beta}= W_s$ and $H_{\alpha \beta}=f_s$, by using the same method.

Eventually, we have studied the causal structure of space-time giving the conditions in order that causality holds locally, hence we have given the definition of $\textit{strong causality}$, $\textit{stable causality}$ and $\textit{global hyperbolicity}$. 

Moreover, we have seen the relation between global hyperbolicity and the existence of the Cauchy surfaces and hence we have given an alternative, and more fundamental, definition of the $\textit{characteristic conoid}$ that comes strictly from the causal structure of space-time.

To conclude our argumentation, we have studied, as an application of Riemann's kernel, the axisymmetric black hole collisions at the speed of light. 
More precisely, we have analyzed the Green function for the perturbative field equations  by studying the corresponding second-order hyperbolic operator with variable coefficients. Then, we have seen that the inverse of the hyperbolic operator for the inhomogeneous wave equations occurring in the perturbative analysis can be accomplished with the help of the Riemann integral representation, after solving the equation for the Riemann function. 

Hence, it is necessary to solve a characteristic initial-value problem for a homogeneous hyperbolic equation in canonical form in two independent variables, for which we have developed formulae to be used for the numerical solution with the help of a finite-differences scheme.

\appendix

\chapter{Sobolev Spaces}
\section{Introduction}
Let us consider the following problem \cite{brezis1986analisi}. Given a function $f \in C([a, b])$, we have to find a function $u(x)$ which verifies
\begin{equation} \label{eq:A.1}\left\{ \begin{array} {l}
- u'' + u = f \hspace{0.5cm}  {\rm on} \;[a, b] \\
u(a)=u(b)=0
\end{array}\right. \end{equation} 
A $\textit{classical solution}$, or $\textit{strong solution}$, is a function $C^2$ on $[a, b]$ that verifies the previous problem in the usual sense.

Upon multiplying $(\ref{eq:A.1})$ by $\varphi \in C^1([a, b])$ and after integration by parts; we have
\begin{equation} \label{eq:A.2} \int_a^b u' \varphi' + \int_a^b u \varphi= \int_a^b f \varphi, \; \forall \varphi \in C^1([a, b]), \end{equation}
with $\varphi(a)= \varphi(b)=0$. The Eq. $(\ref{eq:A.2})$ has meaning only if $u \in C^1([a, b])$ whereas Eqs. ($\ref{eq:A.1})$ hold if $u \in C^2([a, b])$.

A function $u$ of class $C^1$ verifying ($\ref{eq:A.2})$ is a $\textit{weak solution}$ of $(\ref{eq:A.1})$. The $\textit{weak solutions}$ involve the $\textit{Sobolev spaces}$ that are a basic tool. Therefore, we will give a more precise characterization of these spaces which are essential for the application of the $\textit{variational method}$ in the theory of partial differential equations.

\section{Sobolev Space $W^{1, p}(\Omega)$}
Let $\Omega$ be an open set and $p$ an integer such that $1 \leq p \leq \infty$.
\begin{defn} The Sobolev space $W^{1, p}(\Omega)$ is defined by
\begin{equation*} \begin{split} W^{1, p}(\Omega)= \Biggl\{ & u \in L^p(\Omega) : \exists g_1, g_2, ..., g_N \in L^p(\Omega) \\
& :\int_{\Omega} u \frac{\partial \varphi}{\partial x_i} = - \int_\Omega g_i \varphi , \forall \varphi \in C^\infty_c(\Omega), \forall i=1, 2, ..., N \Biggr\}\end{split}\end{equation*}\end{defn}
Let us define $H^1(\Omega)=W^{1, 2}(\Omega)$ and, for $u \in W^{1, p}(\Omega)$, we define
\begin{equation*} \frac{\partial u}{\partial x_i}={g_i}^{(2)} \; {\rm and} \; \nabla u = \bigg{(} \frac{\partial u}{\partial x_1}, \frac{\partial u}{\partial x_2}, ..., \frac{\partial u}{\partial x_N}\bigg{)} = {\rm grad} \; u. \end{equation*}
The space $W^{1, p}$ has the norm
\begin{equation*} ||u||_{W^{1, p}} = ||u||_{L^p} + \sum_{i=1}^N \bigg{|}\bigg{|}\frac{\partial u}{\partial x_i} \bigg{|} \bigg{|}_{L^p}. \end{equation*} 
or the equivalent norm
\begin{equation*} \bigg{(} ||u||^p_{L^p} + \sum_{i=1}^N \bigg{|} \bigg{|}\frac{\partial u}{\partial x_i}\bigg{|} \bigg{|}^p_{L^p} \bigg{)}^{\frac{1}{p}}. \end{equation*}
The space $H^1(\Omega)$ has the inner product
\begin{equation*} (u, v)_{H^1}=(u, v)_{L^2} + \sum_{i=1}^N \bigg{(}\frac{\partial u}{\partial x_i}, \frac{\partial v}{\partial x_i} \bigg{)}_{L^2}; \end{equation*}
the associated norm 
\begin{equation*} ||u||_{H^1} = \bigg{(}||u||^2_{L^p} + \sum_{i=1}^N||\frac{\partial u}{\partial x_i}||^2_{L^p}\bigg{)}^{\frac{1}{2}}, \end{equation*}
is equivalent to the norm of $W^{1, 2}$.

\begin{prop} The space $W^{1, p}$ is a Banach space for $1 \leq p \leq \infty$. The space $W^{1, p}$ is reflexive for $1 < p < \infty$ and separable for $1 \leq p < \infty$. The space $H^1$ is a separable Hilbert space. \end{prop}
The definition of $W^{1, p}$ states that $\varphi$ is a $\textbf{test function}$. Thus, we can use both $C^1_c(\Omega)$ and $C^{\infty}_c(\Omega)$, indifferently, as set of test functions.
Furthermore, if $u \in C^1(\Omega) \cap L^p(\Omega)$ and if $\frac{\partial u}{\partial x_i} \in L^p(\Omega)$ $\forall i=1, 2, ...,N$ (which are the usual partial derivatives of $u$), then $u \in W^{1, p}(\Omega)$. The usual partial derivatives of $u$ coincides with the derivatives of $u$ in the sense of $W^{1, p}$. In particular, if $\Omega$ is limited, hence
\begin{equation*} C^1({\Omega}) \subset W^{1, p}(\Omega), \end{equation*}
for $1 \leq p \leq \infty$. Conversely, if $u \in W^{1, p}(\Omega) \cap C(\Omega)$, with $1 \leq p \leq \infty$, and if $\frac{\partial u}{\partial x_i} \in C(\Omega)$, $\forall i=1, 2, ..., N$ (that are the partial derivatives in the sense of $W^{1, p}$), then $u \in C^1(\Omega)$.

\begin{obs} Let $u \in L^1_{loc}(\Omega)$; $\textit{distributions' theory}$ makes it possible to regard $\frac{ \partial u}{\partial x_i}$ as an element of distribution space $\mathcal{D}(\Omega)$ (which contains also $L^1_{loc}(\Omega)$). Making use of this theory, it is possible to define $W^{1, p}$ as the set of functions $u \in L^p(\Omega)$ such that all partial derivatives $\frac{\partial u}{\partial x_i}$, for $1 \leq i \leq N$, belong to $L^p(\Omega)$.\end{obs}
If $\Omega=\mathbf{R}^N$ and $p=2$, it is even possible to define Sobolev spaces making use of the Fourier transform. 

Given the Sobolev space $W^{1, p}$ the following results hold:
\begin{description}
\item[(a)] Let $u_n$ be a sequence of $W^{1, p}$ such that $u_n \rightarrow u$ in $L^p$ and $\nabla u$ converges towards a limit in $(L^p)^N$, then $u \in W^{1, p}$ and $||u_n - u||_{W^{1, p}} \rightarrow 0$. If $1 < p \leq \infty$ it is sufficient to know that $u_n \rightarrow u$ in $L^p$ and that $\nabla u_n$ remains bounded in $(L^p)^N$ to conclude that $u \in W^{1, p}$.
\item[(b)]Given a function $f$ defined on $\Omega$, we indicate with $\bar{f}$ its extension to zero outside of $\Omega$, that is
\begin{equation*} \bar{f}(x)= \left\{ \begin{array} {l}
f(x) \hspace{2cm} {\rm if} \; x \in \Omega \\
0 \hspace{2.5cm} {\rm if} \; x \in \mathbf{R}^N \setminus \Omega
\end{array}\right. \end{equation*}
\end{description}
Let $u \in W^{1, p}(\Omega)$ and $\alpha \in C^1_c(\Omega)$. Then

\begin{equation*} \overline{\alpha u} \in W^{1, p}(\mathbf{R}^N) \; {\rm and} \; \frac{\partial}{\partial x_i} (\overline{\alpha u}) = \alpha \; \overline{\bigg{(}\frac{\partial u}{\partial x_i} + \frac{\partial \alpha}{\partial x_i} u \bigg{)}}. \end{equation*}
Indeed, if $\varphi \in C^1_c(\mathbf{R}^N)$, then we have:

\begin{equation*} \begin{split} \int_{\mathbf{R}^N} \overline{\alpha u} \frac{\partial \varphi}{ \partial x_i} &=  \int_\Omega {\alpha u} \frac{\partial \varphi}{ \partial x_i} = \int_\Omega u \bigg{[} \frac{\partial}{\partial x_i} (\alpha \varphi) - \frac{\partial \alpha}{\partial x_i} \varphi \bigg{]}\\
&= - \int_\Omega \bigg{(} \frac{\partial u}{\partial x_i} \alpha \varphi + u \frac{\partial \alpha}{\partial x_i} \varphi \bigg{)} = - \int_{\mathbf{R}^N} \overline{ \bigg{(} \alpha \frac{\partial u}{\partial x_i} + \frac{\partial \alpha}{\partial x_i} u \bigg{)}} \varphi. \end{split} \end{equation*}
The same conclusion remains valid if, instead of assuming that $\alpha \in C^1_c(\Omega)$, we take $\alpha \in C^1(\mathbf{R}^N) \cap L^{\infty}(\mathbf{R}^N)$ with $\nabla \alpha \in L^{\infty}(\mathbf{R}^N)^N$ and ${\rm Supp}\; \alpha \subset \mathbf{R}^N \setminus \Gamma$.

\begin{thm} $\mathbf{(Friedrichs).}$ Let $u \in W^{1, p}(\Omega)$ with $1 \leq p < \infty$. Thus, there exists a sequence $u_n \in C^\infty_c(\mathbf{R}^N)$ such that
\begin{description}
\item[(1)]$ u_n|_{\Omega} \rightarrow u \hspace{3cm} \in L^p(\Omega) $
\item[(2)] $\nabla u_n|_{\omega} \rightarrow \nabla u|_{\omega} \hspace{2cm} \in L^p(\omega)^N, \forall \omega \subset \subset \Omega $
\end{description} \end{thm}
where $\omega \subset \subset \Omega$ means that $\omega$ is an open set such that $\overline{\omega} \subset \Omega$ and $\overline{\omega}$ is a compact set.
\proof Let us define
\begin{equation*} \bar{u}(x)=\left\{\begin{array} {l}
u(x) \hspace{2cm} {\rm if} \; x \in \Omega \\
0 \hspace{2cm} {\rm if} \; x \in \mathbf{R}^N \setminus \Omega 
\end{array}\right. \end{equation*}
and $v_n= \rho_n * \bar{u}$, where $\rho_n$ is a regularizing function. We know that $v_n \in C^{\infty}(\mathbf{R}^N)$ and $v_n \rightarrow \bar{u} \in L^p(\mathbf{R}^N)$. We prove that $\nabla v_n \_\omega \rightarrow \nabla u|_\omega \in L^p(\omega)$, $\forall \omega \subset \subset \Omega$. Since  
$\omega \subset \subset \Omega$, we consider a function $\alpha \in C^1_c(\Omega)$, $0 \leq \alpha \leq 1$, such that $\alpha=1$ in a neighbourhood of $\omega$. For $n$ large enough we have
\begin{equation*} (3) \hspace{2cm} \rho_n * \overline{\alpha u} = \rho_n * \bar{u} \hspace{2cm} {\rm on} \; \omega. \end{equation*}
Indeed
\begin{equation*} \begin{split}& {\rm Supp}(\rho_n * \overline{\alpha u} - \rho_n* \overline{u})={\rm Supp}[\rho_n*(1 - \bar{\alpha}) \bar{u}] \\
&{\rm Supp} \rho_n + {\rm Supp}(1 - \bar{\alpha})\bar{u} \subset B \bigg{(} 0, \frac{1}{n} \bigg{)} + {\rm Supp}(1 - \bar{\alpha}) \subset \mathbf{R}^N \setminus \omega \end{split} \end{equation*}
for $n$ sufficiently large.

Making use of the following lemma
\begin{lem} Let $\rho \in L^1(\mathbf{R}^N)$ and $v \in W^{1, p}(\mathbf{R}^N)$ with $1 \leq p \leq \infty$. Then
\begin{equation*} \rho * v \in W^{1, p} (\mathbf{R}^N) \; {\rm and}\; \frac{\partial}{\partial x_i}(\rho * v)= \rho * \frac{\partial v}{\partial x_i}, \; \forall i=1, 2, ..., N. \end{equation*} \end{lem}
and of the result $(b)$, we have
\begin{equation*} \frac{\partial}{\partial x_i} (\rho_n * \overline{\alpha u}) = \rho_n * \overline{\bigg{(} \alpha \frac{\partial u}{\partial x_i} + \frac{\partial \alpha}{\partial x_i} u \bigg{)}}, \end{equation*}
and therefore
\begin{equation*} \frac{\partial}{\partial x_i} (\rho_n * \overline{\alpha u}) \rightarrow \overline{\bigg{(} \alpha \frac{\partial u}{\partial x_i} + \frac{\partial \alpha}{\partial x_i} u \bigg{)}}. \end{equation*}
In particular
\begin{equation*}  \frac{\partial}{\partial x_i} (\rho_n * \overline{\alpha u}) \rightarrow \frac{\partial u}{\partial x_i} \in L^p(\omega), \end{equation*}
and, because of $(3)$
\begin{equation*} \frac{\partial}{\partial x_i} (\rho_n * \overline{ u}) \rightarrow \frac{\partial u}{\partial x_i} \in L^p(\omega). \end{equation*}
Eventually, we set $u_n= \xi_nv_n^{(1)}$, and it is easily verified that the sequence $u_n$ satisfies the desired properties, i.e. $u_n \in C^\infty_c(\mathbf{R}^N)$, $u_n \rightarrow u \in L^p(\Omega)$ and $\nabla u_n \rightarrow \nabla u \in L^p(\omega)^N$.
 \endproof

\begin{prop} Let $u \in L^p(\Omega)$ with $1 < p \leq \infty$. The following properties are equivalent
\begin{description}
\item[(i)] $u \in W^{1, p}(\Omega)$
\item[(ii)]There exists a constant $C$ such that
\begin{equation*}\bigg{|}\int_\Omega u \frac{\partial \varphi}{\partial x_i} \bigg{|} \leq C||\varphi ||_{L^p}, \hspace{2cm} \forall \varphi \in C^\infty_c(\Omega), \; \forall i=1, 2, ..., N \end{equation*}
\item[(iii)]There exists a constant $C$ such that for every open set $\omega \subset \subset \Omega$ and $h \in \mathbf{R}^N$, with $|h| < dist(\omega, \mathbf{R}^N \setminus \Omega)$ we have
\begin{equation*} || \tau_h u - u ||_{L^p} \leq C|h|. \end{equation*}
\end{description} 
Furthermore, we can choose $C=||\nabla u ||_{L^p}$ in $(ii)$ and $(iii)$.\end{prop}
If $p=1$ the following implication still hold
$$ (i) \rightarrow (ii) \iff (iii). $$
The functions that verify (ii), or (iii), with $p=1$ are the $\textit{functions with bounded}$ $\textit{variation}$, which are functions of $L^1$ and whose first derivatives, in the sense of distributions, are bounded measures.

\section{Sobolev Space $W^{m, p}(\Omega)$}
Let $m \geq 2$ be an integer and $p$ be a real number such that $1 \leq p \leq \infty$. We define by recurrence
$$ W^{m, p}(\Omega)=\{ u \in W^{m-1, p}; \; \frac{\partial u}{\partial x_i} \in W^{m-1, p}(\Omega) \; \forall i=1, 2, ..., N \}.$$
This is equivalent to the definition
\begin{equation*} \begin{split} W^{m, p}(\Omega)=\Biggl\{& u \in L^p(\Omega): \forall \alpha \; {\rm with} \; |\alpha| \leq m \; \exists g_\alpha \in L^p(\Omega) \; {\rm such} \; {\rm that} \\  
& \int_\Omega u D^\alpha \varphi= (-1)^{|\alpha|}\int_\Omega g_\alpha \varphi; \forall \varphi \in C^\infty_c(\Omega) \Biggr\} .\end{split}\end{equation*}
We set $D^\alpha u = g_\alpha$.

The space $W^{m, p}(\Omega)$ with the norm
\begin{equation*} ||u||_{W^{m, p}} = \sum_{0 \leq |\alpha| \leq m} ||Du||_{L^p} \end{equation*}
is a Banach space.

We set $H^m(\Omega)=W^{m, 2}(\Omega)$; $H^m(\Omega)$ with inner product
$$ (u, v)_{H^m}=\sum_{0 \leq |\alpha| \leq m } (D^\alpha u, D^\alpha v)_{L^2} $$
is a Hilbert space.

If $\Omega$ is sufficiently regular with $\Gamma= \partial \Omega$ bounded, then the norm of $W^{m, p}(\Omega)$ is equivalent to the norm
$$||u||_{L^p} + \sum_{|\alpha|=m} ||D^\alpha u ||_{L^p} $$
More precisely, for every $\alpha$ with $0 < \alpha \leq m$ and $\forall \epsilon >0$ there exists a constant $C$, which depends on $\Omega$, $\epsilon$ and $\alpha$, such that
$$||D^\alpha u||_{L^p} \leq \epsilon \sum_{|\beta|=m} ||D^\beta u||_{L^p} + C||u||_{L^p} \hspace{1cm} \forall u \in W^{m, p} (\Omega).$$

\section{The space $W^{1, p}_0(\Omega)$}
Let $p$ be $1 \leq p < \infty$; $W^{1, p}_0(\Omega)$ is the closure of $C^1_c(\Omega)$ in $W^{1, p}(\Omega)$. Let us set
$$ H^1_0(\Omega)=W^{1, 2}_0 (\Omega)$$
The space $W^{1, p}_0$ with the norm induced by $W^{1, p}$ is a separable Banach space; if $1 < p < \infty$ it is reflexive. $H^1_0$ is a Hilbert space for the inner product of $H^1$.

Since $C^1_c(\mathbf{R}^N)$ is dense in $W^{1, p}(\mathbf{R}^N)$, we have
$$W^{1, p}_0(\mathbf{R}^N)= W^{1, p}(\mathbf{R}^N). $$
Conversely, if $\Omega \subset \mathbf{R}^N$, then in general $W^{1, p}_0(\Omega) \neq W^{1, p}(\Omega)$. However, if $\mathbf{R}^N \setminus \Omega$ is enough small and $p <N$, we have $W^{1, p}_0(\Omega)=W^{1, p}(\Omega)$. Furthermore, $C^\infty_c(\Omega)$ is dense in $W^{1, p}_0(\Omega)$, then we can give the definition of $W^{1, p}(\Omega)$ making use of $C^\infty_c(\Omega)$ or $C^1_0(\Omega)$, indifferently.

The functions of $W^{1, p}_0(\Omega)$ are the functions of $W^{1, p}(\Omega)$ that vanish on $\Gamma= \partial \Omega$.

\begin{lem}Let $u \in W^{1, p}(\Omega)$, $1 \leq p < \infty$, with {\rm Supp}$u$ compact and belonging to $\Omega$. Then $u \in W^{1, p}_0(\Omega)$. \end{lem}
\proof Given an open set $\omega$ such that ${\rm Supp}u \subset \omega \subset \subset \Omega$ and by choosing $\alpha \in C^1_c(\omega)$ such that $\alpha=1$ on ${\rm Supp}u$, then $\alpha u = u$. On the other hand, the Friederichs theorem states the existence of a sequence $u_n \in C^\infty_c(\mathbf{R}^N)$ such that $u_n \rightarrow u$ in $L^p(\Omega)$ and $\nabla u_n \rightarrow \nabla u$ in $L^p(\omega)^N$. Consequently, $\alpha u_n \rightarrow \alpha u$ in $W^{1, p}(\Omega)$ and $\alpha u \in W^{1, p}_0(\Omega)$. Hence, $u \in W^{1, p}_0(\Omega)$. \endproof
 
\begin{thm} Let us suppose that $\Omega$ is of class $C^1$. Let
$$ u \in W^{1, p}(\Omega) \cap C(\overline{\Omega}) \hspace{1cm} {\rm with} \; 1 \leq p < \infty.$$
Then the following properties are equivalent:
\begin{description}
\item[(i)] $u=0$ on $\Gamma$,
\item[(ii)] $u \in W^{1, p}_0(\Omega)$.
\end{description}\end{thm}
This theorem specifies what is meant by "sufficiently regular" in section A.3.

\section{The dual space of $W^{1, p}_0(\Omega)$}
Let us denote with $W^{-1, p'}(\Omega)$ the space which is dual to $W^{1, p}_0(\Omega)$, $1 \leq p < \infty$ and with $H^{-1}$ the space dual to $H^1_0(\Omega)$.

We identify $L^2(\Omega)$ with its dual, but the same does not hold for $H^1_0(\Omega)$ with its dual. Hence, we have the following scheme:
\begin{equation*}{H_0^1}(\Omega) \subset L^2(\Omega) \subset H^{-1}(\Omega)\end{equation*}
with continuous and dense immersions. 

If $\Omega$ is bounded, we have
\begin{equation*}W^{1, p}_0(\Omega) \subset L^2(\Omega) \subset W^{-1, p'}(\Omega) \hspace{1cm} {\rm if} \; \frac{2N}{N+2} \leq p < \infty \end{equation*}
with continuous and dense immersions.

If $\Omega$ is unbounded, we have
\begin{equation*}W^{1, p}_0(\Omega) \subset L^2(\Omega) \subset W^{-1, p'}(\Omega) \hspace{1cm} {\rm if} \; \frac{2N}{N+2} \leq p \leq 2. \end{equation*}
It is possible to characterize the elements of $W^{-1, p'}(\Omega)$ by making use of the following proposition.
\begin{prop}Let $F \in W^{-1, p'}(\Omega)$, then there exist $f_0$, $f_1$, ..., $f_N \in L^{p'}(\Omega)$ such that
\begin{equation*} < F, v> = \int f_0 v + \sum_{i=1}^N \int f_i \frac{\partial v}{\partial x_i} \hspace{1cm} \forall v \in W^{1, p}_0 \end{equation*}
and
\begin{equation*} {\rm Max}_{0 \leq i \leq N} ||f_i||_{L^{p'}}= ||F||. \end{equation*}
\end{prop}
If $\Omega$ is bounded, it is possible to choose $f_0=0$.

\chapter{Kasner Space-times}
The Kasner spacetimes were discovered by Kasner (1925). They have attracted considerable interest for the study of the behaviour of space-times near the initial singularity. The Kasner models are built with the isometry group $G$ being the Abelian group $\mathbf{R}^3$. All structure constants are zero.

\section{Kasner solutions}
Following Choquet-Bruhat \cite{choquet2009general}, let $x^i$ be arbitrary Cartesian coordinates on $\mathbf{R}^3$; the differentials $dx^i$ are a basis of invariant 1-forms on $\mathbf{R}^3$. We can choose them at each time $t$ so that they are orthogonal in the metric $g$ of the corresponding orbit. Thus, this metric takes the diagonal form
\begin{equation} \label{B.1} g \equiv \sum_{i=1, 2, 3}a_i(t)(dx^i)^2. \end{equation}
The vacuum Einstein equations 
\begin{equation*} R_{\alpha \beta}\equiv \partial_\lambda \Gamma \{\lambda, [\alpha, \beta]\} - \partial_\alpha \Gamma\{\lambda, [\beta, \lambda]\} + \Gamma \{\lambda, [\alpha, \beta] \} \Gamma \{\mu, [\lambda, \mu]\} - \Gamma \{ \mu, [\beta, \lambda] \} =0 \end{equation*}
reduce to ordinary differential equations by making use of the metric $(\ref{B.1})$, and we have
\begin{equation}\label{B.2} R_{00} \equiv - \frac{1}{4}(\partial_0 log(a_i))^2 - \frac{1}{2} \partial_0 \partial_0 log(a_i)=0 \end{equation}
\begin{equation} \label{B.3}  R_{ij} \equiv 0 \hspace{2cm} {\rm if} \; i \neq j \end{equation}
\begin{equation} \label{B.4}  R_{ii} \equiv - \Biggl\{\frac{1}{2} (\partial_0 log(a_i)) \partial_0 a_i - \frac{1}{4}a_i \sum_{p=1, 2, 3} (\partial_0 log(a_p)) - \frac{1}{2} {\partial_{00}}^2 a_i \Biggr\} =0 \end{equation}
We set $\frac{1}{2}\partial_0 log(a_i)=b_i$, then the $R_{00}$ equation reads as
\begin{equation}\label{B.5} R_{00} \equiv \sum_{i=1, 2, 3} (b_i^2 + \partial_0 b_i)=0. \end{equation}
By using the identity
\begin{equation}\label{B.6} \frac{1}{2} {a_i}^{-1} {\partial_{00}}^2 a_i \equiv \partial_0 \bigg{(} \frac{1}{2}{a_i}^{-1} \partial_0 a_i \bigg{)} + \frac{1}{2} {a_i}^{-2} (\partial_0 a_i)^2, \end{equation}
the equation $(\ref{B.4})$ reads as
\begin{equation} \label{B.7} {a_i}^{-1}R_{ii} \equiv b_i \bigg{(} \sum_{j=1, 2, 3} b_j \bigg{)} + \partial_0 b_i =0, \end{equation}
and hence it follows that
\begin{equation}\label{B.8} g^{ij}R_{ij}\equiv \bigg{(}\sum_{i=1, 2, 3}b_i \bigg{)}^2 + \sum_{i=1, 2, 3}(\partial_0 b_i )=0. \end{equation}
The two equations $(\ref{B.5})$ and $(\ref{B.8})$ imply the constraint
\begin{equation} \label{B.9} G_{00}\equiv \frac{1}{2} (R_{00} + g^{ij} R_{ij} ) \equiv - \bigg{(} \sum_{i=1, 2, 3} b_i \bigg{)}^2 + \sum_{i=1, 2, 3}{b_i}^2 =0. \end{equation}
Letting $K$ denote the extrinsic mean curvature of the space sections $t=cost.$, we evaluate
\begin{equation} K \equiv - \sum_{i=1, 2, 3}b_i \equiv - {X}^{-1} \partial_0 X \end{equation}
with 
\begin{equation} X \equiv (det(g))^{\frac{1}{2}} \equiv (a_1 a_2 a_3)^{\frac{1}{2}}. \end{equation}
The identity ($\ref{B.6})$ gives
\begin{equation} \partial_0K=K^2 \rightarrow {\partial_{00}}^2 X=0. \end{equation}
The solution whose volume tends to zero, when $t$ tends to zero and $K$ becomes infinite, takes the form
\begin{equation}\label{B.13} X=t, \hspace{1cm} K= - \frac{1}{t}. \end{equation}
Therefore, Eq. ($\ref{B.7}$) become a diagonal system of first-order differential equations for the functions $b_i$, i.e.
\begin{equation} R_{ii}{a_i}^{-1} \equiv \frac{b_i}{t} + \partial_0 b_i=0. \end{equation}
The general solution of this equation becomes infinite for $t=0$. It takes the form $b_i= \frac{p_i}{t}$, with $p_i={\rm const.}$. Hence
\begin{equation} a_i = t^{2p_i}. \end{equation}
The $\textit{Kasner exponents}$ $p_i$ must verify, due to $(\ref{B.9}$) and $(\ref{B.13})$
\begin{equation} \label{B.16} \bigg{(} \sum_{i=1, 2, 3}p_i \bigg{)}^2 = \sum_{i=1, 2, 3} {p_i}^2 = 1. \end{equation}
Then, the vacuum Einstein equations are all satisfied. The Kasner space-time metric is
\begin{equation} - dt^2 + t^{2p_1}(dx^1)^2 + t^{2p_2}(dx^2)^2 + t^{2 p_3}(dx^3)^2 \end{equation}
where the Kasner exponents $p_i$ lie in the Kasner circle ($\ref{B.16})$, the intersection of a 2-sphere and a plane.

One of the Kasner solutions has two of the Kasner exponents vanishing. In this case, the space-time metric is $\textit{locally flat}$, as can be seen by evaluating the Ricci tensor of the 2-metric 
\begin{equation} - dt^2 + t^2 dx^2 .\end{equation}
The space-time with such a Kasner metric supported by the manifold $\mathbf{R}^3 \times \mathbf{R}^+$ is in fact isometric to the wedge $t >|x|$ on the Minkowski space-time. 

For all Kasner solutions the volume of $M_t$ expands from zero to infinity as $t$ increases from zero to infinity, since $X \equiv (det g)^{\frac{1}{2}} = t$.

If two of the exponents are not zero the ($\ref{B.16})$ shows that one at least must be negative. Suppose $p_1 <0$, $p_2 >0$, and $p_3 >0$. Then, as $t$ tends to zero the space-time shrinks in the direction of $x^2$ and $x^3$ while it expands indefinitely in the direction of $x^1$. The opposite happens at $t$ tends to infinity; in both time directions the Universe is very anisotropic, while it is much less so at intermediate times.

\bibliographystyle{unsrt}
\bibliography{bibliografiatesi}
\nocite{*}
\addcontentsline{toc}{chapter}{Bibliography}

\end{document}